\documentclass[letterpaper,11pt]{article}
\usepackage{cite}
\usepackage[utf8]{inputenc}
\usepackage[letterpaper,hmargin=1in,vmargin=1in,footskip=36pt]{geometry}
\usepackage{mathptmx} % small math font
\DeclareMathAlphabet\mathcal{OMS}{cmsy}{m}{n}
\SetMathAlphabet\mathcal{bold}{OMS}{cmsy}{b}{n}

\usepackage[utf8]{inputenc}
\usepackage[english]{babel}
\usepackage{amsthm}
\usepackage{enumerate}
\usepackage[fleqn]{amsmath}
\usepackage{amsfonts}
\usepackage{amssymb}
\usepackage{amsbsy}
\usepackage[svgnames]{xcolor}
\usepackage{colortbl}
\usepackage{lineno}
\usepackage{xspace}
\usepackage{bm}
\usepackage{scalerel}
\usepackage{stackengine,wasysym}
\usepackage{bbm}
\usepackage{mathrsfs}
\usepackage{graphics}
\usepackage{diagbox}

\usepackage{algorithm}  
\usepackage{algpseudocode}  
\usepackage{amsmath}

\def\ve#1{\mathchoice{\mbox{\boldmath$\displaystyle\bf#1$}}
	{\mbox{\boldmath$\textstyle\bf#1$}}
	{\mbox{\boldmath$\scriptstyle\bf#1$}}
	{\mbox{\boldmath$\scriptscriptstyle\bf#1$}}}

\makeatletter
\let\bfseries=\undefined
\DeclareRobustCommand\bfseries
{\not@math@alphabet\bfseries\mathbf
	\boldmath\fontseries\bfdefault\selectfont}
\makeatother

\def\Orthant_j{{\mathcal O}_{j}}

\newcommand\vev{{\ve v}}

\newcommand{\OO}{{\mathcal{O}}}

\usepackage[textsize=tiny,textwidth=2cm,shadow,color=green]{todonotes}
\usepackage{paralist, tabularx}

\newenvironment{psmallmatrix}{\left(\smallmatrix}{\endsmallmatrix\right)}

\newcommand\FourBlockBig[5][\relax]{\begin{pmatrix}#2& #3\\#4&#5 \end{pmatrix}\ifx#1\relax\else^{(#1)}\fi}
\newcommand\FourBlock[5][\relax]{\begin{psmallmatrix}#2& #3\\#4&#5 \end{psmallmatrix}\ifx#1\relax\else{^{(#1)}}\fi}
\newcommand\TwoBlock[3][\relax]{\begin{psmallmatrix}#2\\#3 \end{psmallmatrix}\ifx#1\relax\else{^{(#1)}}\fi}

\newtheorem{theorem}{Theorem}
\newtheorem{claim}{Claim}
\newtheorem{corollary}{Corollary}
\newtheorem{lemma}{Lemma}
\newtheorem{definition}{Definition}
\newtheorem{observation}{Observation}

\makeatletter
\newtheorem*{rep@theorem}{\rep@title}
\newcommand{\newreptheorem}[2]{%
	\newenvironment{rep#1}[1]{%
		\def\rep@title{#2 \ref{##1}}%
		\begin{rep@theorem}}%
		{\end{rep@theorem}}}
\makeatother

\newreptheorem{theorem}{Theorem}
\newreptheorem{corollary}{Corollary}
\newreptheorem{lemma}{Lemma}
\newreptheorem{proposition}{Proposition}

\title{Improved Approximation Schemes for (Un-)Bounded Subset-Sum and Partition}

\author{Xiaoyu Wu\thanks{Zhejiang University.
		\texttt{xiaoyu\_wu@zju.edu.cn}.}
	\ \ \
	Lin Chen\thanks{Texas Tech University.
		\texttt{chenlin198662@gmail.com}. }
}

\begin{document}
	%	\linenumbers
	%\linespread{1}
	\maketitle
	
	\thispagestyle{empty}
	
	\begin{abstract}
		We consider the SUBSET SUM problem and its important variants in this paper. In the SUBSET SUM problem, a (multi-)set $X$ of $n$ positive numbers and a target number $t$ are given, and the task is to find a subset of $X$ with the maximal sum that does not exceed $t$. It is well known that this problem is NP-hard and admits fully polynomial-time approximation schemes (FPTASs). In recent years, it has been shown that there does not exist an FPTAS of running time $\tilde\OO( 1/\epsilon^{2-\delta})$ for arbitrary small $\delta>0$ assuming ($\min$,+)-convolution conjecture~\cite{bringmann2021fine}. However, the lower bound can be bypassed if we relax the constraint such that the task is to find a subset of $X$ that can slightly exceed the threshold $t$ by $\epsilon$ times, and the sum of numbers within the subset is at least $1-\tilde\OO(\epsilon)$ times the optimal objective value that respects the constraint. Approximation schemes that may violate the constraint are also known as weak approximation schemes. For the SUBSET SUM problem, there is a randomized weak approximation scheme running in time $\tilde\OO(n+ 1/\epsilon^{5/3})$ [Mucha et al.'19]. For the special case where the target $t$ is half of the summation of all input numbers, weak approximation schemes are equivalent to approximation schemes that do not violate the constraint, and the best-known algorithm runs in $\tilde\OO(n+1/\epsilon^{{3}/{2}})$ time [Bringmann and Nakos'21].
		
		In this paper, we substantially improve the state-of-art results. We derive a deterministic weak approximation scheme of running time $\tilde\OO(n+1/\epsilon^{3/2})$ for the SUBSET SUM problem, which is the first deterministic approximation scheme of subquadratic running time. For unbounded SUBSET SUM where each input number can be used arbitrarily many times, we obtain an $\tilde\OO(n+1/\epsilon)$-time deterministic weak approximation scheme, which is the best possible. For PARTITION, we improve the existing result by establishing an $\tilde\OO(n+ 1/\epsilon^{5/4})$-time deterministic approximation scheme. 
		
		These results are built upon our main technical contributions: i). a number-theoretic rounding mechanism that leverages a number-theoretic property of integers called smoothness, and ii). a divide-and-conquer based framework that combines FFT (Fast Fourier Transform) with the number-theoretic property.

	\end{abstract}
	
	\vspace{2mm}
	\hspace{3.5mm}\textbf{Keywords:} Approximation scheme; Combinatorial optimization; Partition; Subset-Sum
	
	\clearpage
	\setcounter{page}{1}

	\section{Introduction}\label{Intro}
	We study approximation algorithms for the fundamental problem SUBSET SUM and its special cases in this paper. We first introduce the problems.

Let $\Sigma(Y)$ denote the sum of elements in a (multi-)set $Y$. SUBSET SUM is defined as follows:

\begin{definition}[SUBSET SUM]\label{def:subset-sum}
	Given a (multi-)set $X$ of $n$ positive integers and a target $t>0$, find a subset $X'\subset X$ which achieves the maximum sum among all subsets summing up to $t$. Formally, the task is to find $X'\subset X$ such that 
	$$
	\Sigma(X') = \max\{\Sigma{(Y)} \ | \ Y \subset   X, \Sigma(Y) \le t\},
	$$
\end{definition}

SUBSET SUM problem is a fundamental optimization problem in computer science and is one of Karp's initial list of 21 NP-complete problems~\cite{DBLP:books/daglib/p/Karp10}. An important field of study on NP-hard problems is finding efficient approximation algorithms. In particular, a Fully Polynomial Time Approximation Scheme (FPTAS) for a maximization problem is an algorithm that, given an instance of size $n$ and a parameter $\epsilon>0$, returns a solution whose value is at least $1-\epsilon$ times the optimal solution. Importantly, the run-time of an FPTAS is polynomial of $n$ and $1/\epsilon$. 

SUBSET SUM is one of the first NP-hard problems shown to possess FPTASs and there is a long line of research on finding faster FPTASs for SUBSET SUM. The first published FPTAS was designed by Ibarra and Kim~\cite{karthik2018parameterized} requires $\OO(n/\epsilon^2)$-time. Lawler~\cite{lawler1979fast} subsequently proposed an FPTAS
with improved time bound $\OO(n+1/\epsilon^4)$. Later, Gens and Levner obtained further improved schemes running in time $\OO(n/\epsilon)$~\cite{gens1978approximation,gens1979computational} and in time $\OO(\min\{n/\epsilon,n+1/\epsilon^3\})$~\cite{gens1994fast}, which was further improved by Kellerer et al.~\cite{DBLP:conf/isaac/KellererPS97,DBLP:journals/jcss/KellererMPS03} to a $\OO(\min\{n/\epsilon,n+1/\epsilon^2 \log(1/\epsilon)\})$-time algorithm. %, where the $\tilde\OO$ notation hides polylogarithmic factors in $n$ and $1/\epsilon$.
Very recently, Bringmann~\cite{bringmann2021fine} showed a conditioned lower bound on approximating SUBSET SUM: assuming the $(\min,+)$-convolution conjecture, SUBSET SUM has no approximation scheme in time $\OO((n+1/\epsilon)^{2-\delta})$ for any constant $\delta>0$. This strong lower bound relies on that the constraint of $\Sigma(Y) \le t$ is strict. If the constraint can be slightly violated, then substantially faster approximation schemes exist. Such approximation schemes are called weak approximation schemes. Formally, 

\begin{definition}[WEAK-APX for SUBSET SUM]\label{def:weak_apx}
	Let $X^*$ be an optimal solution for an instance $(X, t)$ of the SUBSET SUM problem. Given $(X,t)$, a weak $(1-\epsilon)$-approximation
	algorithm (or a weak approximation scheme) for SUBSET SUM returns $Y\subset X$ such that 
	$$(1-\epsilon)\Sigma(X^*)\le\Sigma(Y) \le (1+\epsilon)t .$$
\end{definition}

Recently, Mucha et al.~\cite{DBLP:conf/soda/MuchaW019} designed a randomized weak approximation scheme for SUBSET SUM with a strongly subquadratic running time of $\tilde\OO(n+1/\epsilon^{{5}/{3}})$, where the $\tilde\OO$ notation hides polylogarithmic factors in $n$ and $1/\epsilon$. Bringmann and Nakos~\cite{bringmann2021fine} mentioned in their paper that their technique should yield a randomized weak approximation scheme of running time $\tilde\OO(n+1/\epsilon^{{3}/{2}})$ for SUBSET SUM, but there is no formal proof. So far all existing subquadratic approximation schemes are randomized. It remains as an important open problem whether there is a deterministic weak subquadratic time approximation scheme for SUBSET SUM. 

%A mildly subquadratic randomized approximation scheme was also presented in~\cite{bringmann2021fine}. 
%Subquadratic reduction from Min-Plus-Convolution to SUBSET SUM showed in Bringmann~\cite{bringmann2021fine} yields a mildly subquadratic randomized approximation scheme.

We also study PARTITION, which is a fundamental special case of SUBSET SUM where $t$ is fixed to $\Sigma(X)/2$. Formally, 
\begin{definition}[PARTITION]\label{def:parti}
	Given a (multi-)set $X$ of $n$ positive integers, find a subset $X'\subset X$ that achieves the maximum sum among all subsets whose sums do not exceed $\Sigma(X)/2$. Formally, the task is to find $X'\subset X$ such that 
	$$\Sigma(X') = \max\{\Sigma(Y)\ | \ Y \subset X, \Sigma(Y) \le \Sigma(X)/2\}.$$
\end{definition}

PARTITION has many practical applications, including scheduling~\cite{coffman1993probabilistic}, minimization of circuit sizes and cryptography~\cite{DBLP:journals/tit/MerkleH78}, as well as game theory~\cite{hayes2002computing,DBLP:books/ox/06/Mertens06}. As a special case of SUBSET SUM, all SUBSET SUM algorithms also apply to PARTITION. A particularly important observation raised by Mucha et al.~\cite{DBLP:conf/soda/MuchaW019} is that any weak approximation scheme for SUBSET SUM is an approximation scheme for PARTITION.
Designing approximation algorithms specifically for PARTITION also has a long history. In 1980, Gens and Levner~\cite{gens1980fast} proposed an FPTAS with running time $\OO(\min\{n/\epsilon,n+1/\epsilon^2\})$.  Recently, Bringmann and Nakos~\cite{bringmann2021fine} obtained a deterministic FPTAS with a running time of $\tilde\OO(n+1/\epsilon^{{3}/{2}})$. Regarding the lower bound, PARTITION does not admit any approximation scheme of running time $\OO(poly(n)/\epsilon^{1-\delta})$ for any constant $\delta>0$, assuming the SETH~\cite{abboud2022seth} or the SetCover conjecture~\cite{DBLP:journals/talg/CyganDLMNOPSW16}.

An important observation~\cite{DBLP:conf/soda/MuchaW019} is that, since the target is $\Sigma(X)/2$ in PARTITION, a weak approximation scheme for PARTITION is also an approximation scheme. Hence, from an upper bound perspective, PARTITION and SUBSET SUM differ significantly: PARTITION admits a strongly subquadratic running time approximation scheme while SUBSET SUM does not.
%$poly(n)\epsilon^{0.99}$

The third problem we study in this paper is	the UNBOUNDED SUBSET SUM, which is another special case of SUBSET SUM, where every element has infinitely many copies. Formally,
\begin{definition}[UNBOUNDED SUBSET SUM]\label{def:u_subset_sum}
	Given a set $X=\{x_1,x_2,\cdots,x_n\}$ of $n$ different positive integers and a target $t>0$, find the set of non-negative integers $(m_1,m_2,\cdots, m_n)$ such that 
	$$
	\sum^{n}_{i=1}x_{i}m_i = \max\{\sum^{n}_{i=1} x_i y_i \le t \ | \ (y_1,y_2,\cdots,y_n) \in \mathbb{Z}^{n}_{\ge 0}  \}.
	$$
\end{definition}

UNBOUNDED SUBSET SUM can be reduced to (bounded) SUBSET SUM~\cite{DBLP:books/daglib/0010031}, thus algorithms designed for SUBSET SUM also work with UNBOUNDED SUBSET SUM. The best known (deterministic) approximation scheme is due to Jansen and Kraft~\cite{jansen2018faster}, running in time $\tilde{\OO}(n+1/\epsilon^2 )$. In terms of weak approximation schemes, Bringmann et al.~\cite{bringmann2022faster} showed a randomized $\tilde\OO(n+1/\epsilon^{{3}/{2}})$-time algorithm for a more general problem UNBOUNDED KNAPSACK. It is not clear whether UNBOUNDED SUBSET SUM admits a better approximation algorithm.
%However, UNBOUNDED SUBSET SUM has particular structural properties that lead to significant differences from the bounded version in the upper bound of the approximation algorithm.

\paragraph{Related Work.}
%\textbf{0/1 KNAPSACK and UNBOUNDED KNAPSACK.} %In the 0/1 KNAPSACK we are given a set of items $I$ with given integer weights and values ($(w_i,v_i)$) $i\in I$, along with a knapsack capacity $W$. The goal is to find the maximum total value of a subset $I'\subset I$ such that $\Sigma_{i\in I'} w_i \le W$. If we are allowed to take multiple copies of a single item, then we obtain the UNBOUNDED KNAPSACK. 
SUBSET SUM is a special case of KNAPSACK. There is a long line of research on approximation schemes for KNAPSACK, see, e.g.,~\cite{ibarra1975fast,lawler1979fast,pisinger1998knapsack,gens1980fast}. %,DBLP:conf/icalp/Jin19,chan2018approximation
Very recently, Chan~\cite{chan2018approximation} proposed an FPTAS for KNAPSACK with a running time of $\tilde\OO(n+1/\epsilon^{{12}/{5}})$~\cite{chan2018approximation}.  Later, Jin~\cite{DBLP:conf/icalp/Jin19} obtained an improved algorithm of running time $\tilde\OO(n+1/\epsilon^{{9}/{4}})$, which is the best-known so far.  Whether an $\tilde\OO(n+1/\epsilon^{2})$ time FPTAS exists is still a crucial open problem. %Notice that the time spent in the rounding process is not negligible, we consider to improve Jin's number-theoretic construction to obtain a (almost) linear time rounding method.
Conditional lower bound has been obtained for KNAPSACK~\cite{bringmann2021fine}: an FPTAS with $\OO(n+1/\epsilon^{2-\delta}) \ (\delta >0)$ running time would refute the $(\min,+)$-convolution conjecture. %Rounding plays an important role when designing approximation schemes. It seems that existing rounding methods are difficult to bring about algorithm improvements, and researchers hope to achieve better algorithms by introducing new rounding techniques. The latest breakthrough improvements on approximating 0/1 KNAPSACK come from some number-theoretic constructs that can be used for rounding. 

%\textbf{Number-Theoretic Construction.} 
%\textbf{Pseudo-polynomial Algorithms.}  
While we focus on approximation algorithms in this paper, it is worth mentioning that exact algorithms for SUBSET SUM have also received extensive studies. %Above knapsack-type problems are NP-hard but not in a strong sense: there are pseudo-polynomial time algorithms for solving them.  
Bellman showed in 1957 that KNAPSACK, and hence SUBSET SUM, can be solved in time $\OO(nt)$ by dynamic programming~\cite{bellman1957dynamic}. Important progress has been achieved in recent years for SUBSET SUM. In 2019, Koiliaris and Xu~\cite{koiliaris2019faster} obtained an $\tilde\OO(\sqrt{n}t,t^{4/3})$-time algorithm for SUBSET SUM, followed by a randomized $\tilde\OO(n+t)$ algorithm~\cite{bringmann2017near,karthik2018parameterized}. %, followed by~\cite{karthik2018parameterized} with further improvements on the logarithmic factors. %On the other hand, SUBSET SUM does not admit an  $\OO(poly(n)t^{1-\delta})$ algorithm for any constant $\delta>0$, assuming the SETH ~\cite{abboud2022seth} or the SetCover conjecture~\cite{DBLP:journals/talg/CyganDLMNOPSW16}. %These pseudo-polynomial time algorithms would become impractical when $t$ is exponentially large in the size of the input. One way to avoid this drawback is looking for approximation solutions rather than the optimum. 
For UNBOUNDED SUBSET SUM, Bringmann~\cite{bringmann2017near} gave an exact algorithm with a running time of $\tilde{\OO}(t) $. Jansen and Rohwedder~\cite{DBLP:conf/innovations/JansenR19} showed later that there is an exact algorithm with a running time $\tilde{\OO}(a_n)$ for the stronger parameter $a_n$, where $a_n$ refers to the largest input number. Klein~\cite{DBLP:conf/soda/Klein22} showed an exact algorithm with a running time of $\tilde{\OO}(a_0^2)$ where $a_0$ refers to the smallest input number. In terms of lower bounds, assuming the strong exponential time hypothesis (SETH), Abboud et al.~\cite{abboud2022seth} proved that there is no algorithm for SUBSET SUM or UNBOUNDED SUBSET SUM with a running time of $\OO(t^{1-\delta})$ for any $\delta > 0$.

This paper is motivated by the open problem proposed by Bringmann and Nakos~\cite{bringmann2021fine}. They mentioned the possibility of extending their approximation scheme for PARTITION to a  randomized $\tilde\OO(n+\epsilon^{-\frac{3}{2}})$-time weak approximation for SUBSET SUM. However, they did not provide proof and left an open problem on the existence of a subquadratic deterministic weak approximation algorithm for SUBSET SUM.

\paragraph{Our contributions.} The main contribution of this paper is to obtain substantial improvement on the approximation schemes for three closely related fundamental problems:  SUBSET SUM, PARTITION and UNBOUNDED SUBSET SUM. More precisely, 
\begin{itemize}		
	\item We obtain a deterministic $\tilde\OO(n+\epsilon^{-\frac{3}{2}})$-time weak $(1-\epsilon)$-approximation algorithm for SUBSET SUM. %improving upon the existing $\tilde\OO(n+\epsilon^{-\frac{5}{3}})$-time randomized weak $(1-\epsilon)$-approximation algorithm. 
	This gives the first deterministic subquadratic algorithm for weak approximating SUBSET SUM, and resolves the open problem raised by Bringmann and Nakos~\cite{bringmann2021fine}.
	\item We obtain an $\tilde\OO(n+\epsilon^{-\frac{5}{4}})$-time deterministic FPTAS for PARTITION, improving upon the existing $\tilde\OO(n+\epsilon^{-\frac{3}{2}})$-time deterministic FPTAS. Furthermore, unlike the prior algorithm that is highly tailored to PARTITION, our algorithm is also an $\tilde\OO(n+\epsilon^{-\frac{5}{4}})$-time weak $(1-\epsilon)$-approximation algorithm for SUBSET SUM when the target $t=\Theta(\Sigma(X))$. This indicates that for SUBSET SUM, the special case PARTITION where $t=\Sigma(X)/2$, is perhaps not too much different from the case, say, $t=\Sigma(X)/3$, despite that a weak approximation algorithm for PARTITION is naturally an approximation algorithm.  % This indicates that PARTITION, or $t=\Sigma(X)/2$, is perhaps not too much different from the case, say, $t=\Sigma(X)/3$, despite that a weak approximation algorithm for $t=\Sigma(X)/2$ is naturally an approximation algorithm.  
	\item We obtain a deterministic $\tilde\OO(n+\epsilon^{-1})$-time weak $(1-\epsilon)$-approximation algorithm for UNBOUNDED SUBSET SUM. Note that by taking $\epsilon=\frac{1}{t+1}$, the weak approximation algorithm implies an exact algorithm, therefore the lower bound on the running time of exact algorithms also applies. In particular, for arbitrary small constant $\delta>0$, there does not exist any $\tilde\OO(n+\epsilon^{-1+\delta})$-time weak $(1-\epsilon)$-approximation algorithm for UNBOUNDED SUBSET SUM assuming SETH~\cite{abboud2022seth}, hence our algorithm is essentially the best possible. %Taking $\epsilon=1/(t+1)$, this directly implies an exact algorithm wit 
\end{itemize}

%algorithm a deterministic $\tilde\OO(n+\epsilon^{\frac{3}{2}})$-time weak $(1-\epsilon)$-approximation algorithm for Subset-Sum and a deterministic $\tilde\OO(n+\epsilon^{-1})$-time $(1-\epsilon)$-approximation algorithm for Unbounded Subset-Sum. Furthermore, for Subset-Sum with target is some constant fraction of the total summation of elements, our approximation scheme for Partition can be extended to an deterministic $\tilde\OO(n+\epsilon^{-\frac{5}{4}})$-time weak $(1-\epsilon)$-approximation algorithm for Subset Sum. {\color{red} our running time vs. existing running time.}

%In terms of techniques, our main contribution is a number-theoretic rounding mechanism that can simultaneously round all the input numbers into (semi-)smooth numbers with small prime factors in almost linear time, and a divide-and-conquer based framework that allows FFT to leverage the smoothness of numbers. Our rounding mechanism may be of separate interest to other optimization problems. 

\paragraph{Overview of our techniques.} The study of approximation schemes for subset-sum dates back to the 1970s. Recent breakthrough results~\cite{DBLP:conf/soda/MuchaW019,bringmann2021fine} that break the barrier of the quadratic running time, $\tilde{\OO}(\frac{1}{\epsilon^2})$, all rely crucially on FFT (Fast Fourier Transform). However, FFT suffers from that the input numbers can be very large. Using the standard rounding technique, we can obtain a rounded instance with $\OO(\frac{1}{\epsilon})$ integers of value $\OO(\frac{1}{\epsilon})$, and a straightforward FFT only yields an $\OO(\frac{1}{\epsilon^2})$-time PTAS. Thus, FFT has to be carefully combined with some additive combinatoric result~\cite{DBLP:conf/soda/MuchaW019} or some sparsification technique~\cite{bringmann2021fine} to achieve a better running time. We improve FFT through a completely new observation. We observe that, if many input numbers share a large common divisor, then we can first scale down these numbers by dividing the common divisor, use FFT to compute their scaled-down subset-sums, and then scale up these subset-sums. Following this natural idea, we observe that FFT will benefit if all the input numbers are very ``smooth". In number theory, a $y$-smooth number is a number whose prime factors do not exceed $y$. If all the input numbers (of value $\OO(\frac{1}{\epsilon})$) are $\frac{1}{\epsilon^{\theta}}$-smooth, then we can write them into the form $h_1h_2\cdots h_d$ where $h_i$'s are $\OO(\frac{1}{\epsilon^{\theta}})$. Consequently, there are at most $\OO(\frac{1}{\epsilon^{\theta}})$ prime factors for these  $\OO(\frac{1}{\epsilon})$ input numbers to share, which means that many numbers will share a large common divisor. We can then establish a divide-and-conquer based framework that allows FFT to leverage the smoothness.     

The question is, if the input numbers are not smooth, can we round them to smooth numbers with a small loss? This is a very challenging problem in number theory. While it is generally believed that for arbitrary small $\theta>0$ and any $x\in [N]$ (where $[N]:=\{0,1,\cdots,N\}$) there should always exist $x^{\theta}$-smooth numbers that are sufficiently close to $x$, so far it is only known that $x^{\theta}$-smooth numbers are guaranteed to exist within $[x-\Theta({\sqrt{x}}),x+\Theta({\sqrt{x}})]$~\cite{matomaki2016multiplicative}. Moreover, there do not seem to exist good algorithms for computing such a smooth number except a straightforward bruteforce~\cite{granville2008smooth}. Both the error and running time, which are $\Theta(\sqrt{x})$, do not suffice for our needs. Therefore, we establish a number-theoretic lemma (Lemma~\ref{lemma:smooth_appro}) that substantially extends the result of Jin~\cite{DBLP:conf/icalp/Jin19}. We relax the notion of smoothness to ``semi-smoothness" in the sense that for any $\lambda>0$ and any set $X$ of numbers in $\Theta(\frac{1}{\epsilon^{2+\lambda}})$, we can always find $x_j'$ for each $x_j\in X$ such that $|x_j-x_j'|\le \epsilon x_j$, $x_j'=x_j^{common}\cdot x_j^{smooth}$ where $x_j^{smooth}$ is $\frac{1}{\epsilon^{\theta}}$-smooth for arbitrary small $\theta$, and there are only $\log|X|(\log\frac{1}{\epsilon})^{\OO(1)}$ distinct $x_j^{common}$'s. That is, we can always round the input numbers and then divide them into only logarithmically many groups such that by modulo the common divisor of numbers in each group, they all become smooth numbers. More importantly, the overall running time for rounding all the numbers is almost linear in $|X|+\frac{1}{\epsilon}$. We then combine FFT with such a number theoretic construction carefully to obtain improved algorithmic results for (un-)bounded SUBSET SUM and its special case PARTITION. Our method may be of separate interest to other related optimization problems, and also adds to the list of algorithmic applications of smooth numbers surveyed in~\cite{granville2008smooth}.  %{\color{red} Ce's result can be viewed as a special case when $\lambda=...$? (is it correct? highlight why ours is stronger and requires non-trivial technique.)

	\paragraph{Organization of the paper.} 
	In Section~\ref{n_d} and Section~\ref{subsec:prep}, we give the notations and the definitions which are frequently used in this paper. In Section~\ref{sec:compute_sumset_}, we derive exact and approximation algorithms for computing the sum of multisets (see Definition~\ref{def:sumset}). In Section~\ref{sec:theoretic_lemma}, we introduce a number-theoretic rounding lemma (i.e., Lemma~\ref{lemma:smooth_appro}), which implies that input numbers can be effectively rounded to semi-smooth numbers. In Section~\ref{sec:alg-smooth}, we consider smooth numbers, and derive algorithms for computing (capped) subset-sums of smooth numbers. Then we consider general input instances. In Section~\ref{sec:proce}, we establish a lemma (see Lemma~\ref{obs:obs_pre}) which decomposes an arbitrary SUBSET SUM instance into a logarithmic number of sub-instances with a much simplified structure, and therefore it suffices to develop an (weak) approximation scheme for each sub-instance. The subsequent Section~\ref{sec:5/4_main} and Section~\ref{sec:3/2} are dedicated to designing improved approximation algorithms for PARTITION and SUBSET SUM, respectively. Finally in Section~\ref{sec:1/3_weak_apx}, we present the weak approximation scheme for UNBOUNDED SUBSET SUM.
	
	%We remark that the major new technical ingredients lie in Section~\ref{sec:theoretic_lemma} and Section~\ref{sec:alg-smooth}. However, these techniques have to be carefully combined with techniques developed in previous papers~\cite{DBLP:conf/soda/MuchaW019,bringmann2021fine}, particularly the insight from additive combinatoric to outperform the existing algorithms.  
	
	%The algorithms designed in this paper for (un-)bounded SUBSET SUM are divide-and-conquer based frameworks, in which we need to frequently compute (or approximate) sumset and give the corresponding backtracking oracle. 
	
	%frequently compute (or approximate) sumset and give the corresponding oracle for backtracking. In the following, we introduce exact and approximation algorithms for computing (capped) sumset, 

	\subsection{Notations and Definitions}\label{n_d}
	We present notations and definitions that will be used throughout this paper. They mostly follow from prior works\cite{DBLP:conf/soda/MuchaW019}.
	
	Let $\mathbb{Z}$ be the set of all integers. Let $\mathbb{N} = \{0,1,2,\cdots\}$ be the set of all natural numbers and let $\mathbb{N}_{+}$ be the set of all positive integers. Let $\mathbb{R}$ be the set of all real numbers and $\mathbb{R}_{\ge 0}$ be the set of all non-negative real numbers. For a number $x\in \mathbb{R}_{\ge 0}$, we define $pow(x)$ as the largest power of 2 not exceeding $x$ (i.e. $2^{pow(x)} \le x < 2^{pow(x)+1}$).%, {\color{blue}if $0\le x < 2$ we set $pow(x) = 0$.}
	
	A \textit{multiset} is a set-like, unordered collection in which repetition of elements is allowed. Given a finite multiset $X$, we call the number of all elements contained in $X$ its cardinality and use $|X|$ to represent it.  We let $set(X)$ denote the set of all {\it distinct} elements in $X$. 
	
	Some \textit{operations} on multisets are defined in the following.
	
	\begin{definition}\label{def:subset}
		Given a finite multiset $X$ and an element $x$. If $x\in X$, we define $card_{X}[x]$ as the multiplicity of $x$ in $X$ and we let $card_{X}[x] = 0$ if $x\notin X$.  
		
		We call $Y$ a subset of the multiset $X$ if $set(Y)$ is a subset of $set(X)$ and $card_{Y}[x] \le card_{X}[x]$ holds for every $x\in X$. We use $Y \subset X$ to represent that $Y$ is a subset of $X$.
	\end{definition}

	\begin{definition}[Multiset Complement]\label{def:complement}
		Given $Y\subset X$, we define a multiset $X\backslash Y$ as follows: 
		for every $x\in X$, if $card_{X}[x]-card_{Y}[x] \neq 0$, then $x \in X\backslash Y$ and $card_{X\backslash Y}[x] = card_{X}[x]-card_{Y}[x]$; else if $card_{X}[x]-card_{Y}[x] = 0$, then $x \notin X\backslash Y$. We call $X \backslash Y$ the \textit{multiset-complement} of $Y$ in $X$. 
	\end{definition}
	
	\begin{definition}[Scalar Multiplication] Given a multiset $X:=\{x_1,x_2,\cdots, x_n\}$ and a constant $k$. We define $kX: = \{kx_1,kx_2,\cdots,kx_n\}$, which is the multiset obtained by multiplying each element in $X$ by $k$. 
	\end{definition}
	
	\begin{definition}[Multisets Union] Given multisets $X_1,X_2,\cdots, X_{\ell}$, where $\ell \ge 2$. %\subset \mathbb{R}_{\ge 0}
		We denote $X_1 \cup X_2 \cup \cdots \cup X_{\ell}$ as a set containing all distinct elements in $X_1,X_2,\cdots,X_{\ell}$, that is, each element in $X_1 \cup X_2 \cup \cdots \cup X_{\ell}$ is distinct and $set(X_{i}) \subset X_1 \cup X_2 \cup \cdots \cup X_{\ell}$ for $i=1,2,\cdots, \ell$. We call 
		$X_1 \cup X_2 \cup \cdots \cup X_{\ell}$ the {\it union} of $X_i$'s.

		Moreover, we denote $X_1 \dot{\cup} X_2 \dot{\cup} \cdots  \dot{\cup} X_{\ell}$ as a multiset containing all elements in $X_1,X_2,\cdots X_{\ell}$. To be specific, $X_1 \dot{\cup} X_2 \dot{\cup} \cdots  \dot{\cup} X_{\ell}$ satisfies the followings: (1). $X_{i} \subset X_1 \dot{\cup} X_2 \dot{\cup} \cdots  \dot{\cup} X_{\ell}$ for $i=1,2,\cdots, \ell$; (2). the multiplicity of an element in $X_1 \dot{\cup} X_2 \dot{\cup} \cdots  \dot{\cup} X_{\ell}$ is the sum of the multiplicity of this element in each $X_i$, i.e.,
		$card_{X_1 \dot{\cup} X_2 \dot{\cup} \cdots  \dot{\cup} X_{\ell}}[x] = \sum^{\ell}_{i=1} card_{X_i}[x]$ for any $x\in X_1 \dot{\cup} X_2 \dot{\cup} \cdots  \dot{\cup} X_{\ell}$. We call $X_1 \dot{\cup} X_2 \dot{\cup} \cdots  \dot{\cup} X_{\ell}$ the {\it multiset-union} of $X_i$'s. 
		
		For simplicity, we sometimes abbreviate $X_1 {\cup} X_2 {\cup} \cdots  {\cup} X_{\ell}$ and $X_1 \dot{\cup} X_2 \dot{\cup} \cdots  \dot{\cup} X_{\ell}$ as ${\cup}^{\ell}_{i=1}X_i$ and $\dot{\cup}^{\ell}_{i=1}X_i$, respectively. 
	\end{definition}

	\begin{definition}
		For a multiset $X \subset \mathbb{R}_{\ge 0}$, we denote the maximum (resp. minimum) element in $X$ as $X^{max}:= \max\{x\ | \ x\in X\}$ (resp. $X^{min}:= \min\{x\ | \ x\in X\}$), and the sum of all elements in $X$ as $\Sigma(X):= \sum_{x\in X} x$.
	\end{definition}
	
	\begin{definition}[SUBSET-SUMS]\label{def:setsum}
		For a multiset $X \subset \mathbb{R}_{\ge 0}$, we define $S(X):= set\{\Sigma(Y)\ | \ Y \subset  X\}$ as the set of all possible subset sums of $X$, and call $S(X)$ the {\it subset-sums} of $X$. We call two multisets $X_1,X_2 \subset \mathbb{R}_{\ge 0}$ equivalent if $S(X_1) = S(X_2)$. 
		
		Given $[a,b]\subset \mathbb{R}_{\ge 0}$, we use $S(X;[a,b])$ to denote $ S(X)\cap [a,b]$. Specifically, for $S(X;[0,b])$, we call it the {\it $b$-capped subset-sums} of $X$. 
	\end{definition}
	%Let $[\![x:y]\!]:= \{\lfloor x \rfloor+1,\lfloor x \rfloor+2,\cdots, \lfloor y \rfloor \}$ be the set of integers in the interval $[x,y]$. 
	%Let $2^{X}: = \{Y \ | \ Y\subset X\}$ denote the set of all subsets of $X$. 
	
	\begin{definition}[Multisets Sum]\label{def:sumset} Given multisets $X_1,X_2,\cdots, X_{\ell} \subset \mathbb{R}_{\ge 0}$, where $\ell \ge 2$. We define $X_1 \oplus X_2 \oplus \cdots \oplus X_{\ell}:= set\{\sum^{\ell}_{i=1} x_i \ | \ x_i\in X_i\cup\{0\} \text{ \ for every \ } 1\le i \le \ell \}$ , and call it the {\it sumset} of $X_i$'s. For simplicity, we sometimes abbreviate $X_1 \oplus X_2 \oplus \cdots \oplus X_{\ell}$ as $\oplus^{\ell}_{i=1}X_i$.  Given $b \in \mathbb{R}_{\ge0}$, we call $(\oplus^{\ell}_{i=1}X_i) \cap [0,b]$ the {\it $b$-capped sumset} of $X_i$'s.
		%{\color{blue}In addition, we define $X_1 \oplus_{\omega} X_2 \oplus_{\omega} \cdots \oplus_{\omega} X_{\ell} :=  (X_1 \oplus X_2 \oplus \cdots \oplus X_{\ell}) \cap [0,{\omega}]$, and call it the $\omega$-CAPPED-SUMSET of $X_i$'s.}
		
		%Then $X_1 \oplus_t X_2 \oplus_t \cdots \oplus_t X_{\ell}= \oplus^{\ell}_{i=1}X_i \cap[0,t].$
	\end{definition}
	
	Given $F$ as a set of functions, by $F \subset \Theta(g)$ we mean that $f= \Theta(g)$ holds for every $f\in F$. 	
	Unless otherwise specified, we use the notation $\tilde\OO(T)$ to denote a function $\OO(T (\log T)^c)$ for any constant $c>0$, i.e., $\tilde\OO$ suppresses polylogarithmic factors.
	
	Throughout this paper, $\epsilon$ refers to an arbitrarily small positive number.

	\subsection{Definitions of Approximate Set and Backtracking Oracle.}\label{subsec:prep} 
	%说一下这节的内容概括，引入一下

	%\subsubsection{Definitions of approximate set and backtracking oracle.}
	%说明decision问题和optimiazation问题的区别
	%We first introduce some definitions about approximation. 
	%{\color{pink}For technical reason, we introduce 	$T$-time oracle for Backtracking in this subsection. This concept has been used implicitly in prior works like~\cite{DBLP:conf/soda/MuchaW019}. We remark that (weak) approximation schemes in this paper need to: (i).  }

	For technical reason, we introduce two concepts in this subsection, the approximate set and the backtracking oracle. Roughly speaking, since the multiset of input numbers may be difficult to deal with directly, we will build an alternative set with a much simplified structure but is sufficient for the purpose of weak approximation. Such a set is the approximate set of the input multiset, further parameterized $r$ and $u$ as we specify in Definition~\ref{def:e_appr}. Our algorithm will work on the approximate set instead of the input multiset, but then the solution obtained from the approximate set needs to be transformed back to the input multiset. Such a transformation is achieved via the backtracking oracle. We remark that if we only want to ``approximately determine" SUBSET SUM or PARTITION, that is, to answer ``Yes" or ``No" instead of returning a feasible solution when the answer is ``Yes", then the backtracking oracle can be safely ignored. 
	
	%We first introduce the following two definitions, then explain how these two definitions will be used.% for finding $(1-\epsilon)$-approximation solution of SUBSET SUM problem
	
	\begin{definition}[$(r,u)$-APX-SET]\label{def:e_appr}
		Given a multiset $A \subset \mathbb{R}_{\ge 0}$, we call $C$ an $(r,u)$-approximate set of $A$ if and only if $C$ satisfies all of the following conditions: 
		\begin{itemize}
			\item [(i).]$C \subset [0, (1+r)u]$;
			\item [(ii).] for any $c \in C$, there exists $a\in A$ such that $|c-a| \le r u$;
			\item [(iii).] for any $a'\in A\cap [0,u]$, there exists $c'\in C$ such that $|c'-a'| \le r u$. 
		\end{itemize}
		In particular, we call $C$ an $r$-approximate set of $A$ if and only if $C$ is an $(r,A^{max})$-approximate set of $A$.
		
		We call $C$ an $(r,u)$-approximate set of $A$ with an additive error of $ERR$	if $C$ satisfies all of the following conditions: 
		\begin{itemize}
			\item [(i).]$C \subset [0, (1+r)u]$;
			\item [(ii).] for any $c \in C$, there exists $a\in A$ such that $|c-a| \le r u+ERR$;
			\item [(iii).] for any $a'\in A\cap [0,u]$, there exists $c'\in C$ such that $|c'-a'| \le r u + ERR$, where $ERR>0$ is an additive error. 
		\end{itemize}
	\end{definition}
	
	%\noindent\textbf{Remark.} Given a multiset $A \subset \mathbb{R}_{\ge 0}$, if set $C \subset \mathbb{R}_{\ge 0}$ satisfying (i). $C \subset [0, (1+r)u]$; (ii). for any $c \in C$, there exists $a\in A$ such that $|c-a| \le r u+ERR$; (iii). for any $a'\in A\cap [0,u]$, there exists $c'\in C$ such that $|c'-a'| \le r u + ERR$, where $ERR>0$ is an additive error. Then we call $C$ an $(r,u)$-approximate set of $A$ with an additive error of $ERR$. 
	
	\begin{definition}[T-time ORACLE for Backtracking]\label{def:def_for_oracle}
		We give the definitions of \textbf{T-time oracle} for four cases. %~\newline
		\begin{itemize}
			\item[Def-1.] Given multisets $X_1,X_2,\cdots,X_{\ell}\subset \mathbb{R}_{\ge 0}$ and a set $C \subset \oplus^{\ell}_{i=1} X_i$, a \textbf{T-time oracle} for backtracking from $C$ to $\dot{\cup}X^{\ell}_{i=1}$ is an algorithm that given any $c\in C$, in $T$ processing time, it will return $(x_1,x_2,\cdots,x_{\ell})$ satisfying $\sum^{\ell}_{i=1}x_i= c$, where $x_i\in X_{i}\cup \{0\}$ for every $ i =1,2,\cdots, \ell$.  
			\item[Def-2.] Given a multiset $X\subset \mathbb{R}_{\ge 0}$ and a set $C'\subset S(X)$, a \textbf{T-time oracle} for backtracking from $C'$ to $X$  is an algorithm that given any $c'\in C'$, in $T$ processing time, it will return $X'\subset X$ satisfying $\Sigma(X')= c'$.
			\item[Def-3.] Given multisets $X_1,X_2,\cdots,X_{\ell}\subset \mathbb{R}_{\ge 0}$. Let $B$ be an $(r,u)$-approximate set of $\oplus^{\ell}_{i=1} X_i$ with an additive error of $ERR$. A \textbf{T-time oracle} for backtracking from $B$ to  $\dot{\cup}X^{\ell}_{i=1}$  is an algorithm that given any $b\in B$, in $T$ processing time, it will return $(x_1,x_2,\cdots,x_{\ell})$ satisfying $|b-\Sigma^{\ell}_{i=1}x_i|\le ru+ERR$, where $x_i\in X_{i}\cup \{0\}$ for every $ i =1,2,\cdots, \ell$.
			\item [Def-4.] Given a multiset $X\subset \mathbb{R}_{\ge 0}$. Let $B'$ be an $(r,u)$-approximate set of $S(X)$ with an additive error of $ERR$, an $(r,u)$-approximate \textbf{T-time oracle} for backtracking from $B'$ to  $X$  is an algorithm that given any $b'\in B'$, in $T$ processing time, it will return $X'\subset X$ satisfying $| b'-\Sigma(X') |\le ru+ERR$.
		\end{itemize}
	\end{definition}
	
	Now we give a very high-level description on how we leverage the two concepts introduced above in our algorithms. Towards that, we first present a simple observation.		
	%In the following, we present a useful preliminary Lemma.	
	\begin{lemma}\label{lemma:prea}
		Given an instance $(X,t)$ of SUBSET SUM. Let $OPT$ be the optimal objective value of $(X,t)$ and let $X_t=\{x\in X \ | \ x\le t\}$.	If $\Sigma(X_t) <  t/2$, then $OPT <  t/2$ and $X_t$ is the optimal solution of $(X,t)$. Else if $\Sigma(X_t) <  t/2$, we can assert that $OPT \ge t/2$.
		
	\end{lemma}
	\begin{proof}
		Given an instance $(X,t)$ of SUBSET SUM, let $X^*\subset X$ and $OPT$ be the optimal solution and optimal objective value of $(X,t)$, respectively.  Define $X_t := \{x\in X \ | \ x\le t\}$. It is easy to see that if $\Sigma(X_{t}) < t/2$, then apparently $OPT < t/2$ and $X_{t}$ is the optimal solution of $(X,t)$.  We claim that if $\Sigma(X_{t}) \ge t/2$, then $OPT \ge t/2$. 
		
		Assume that $\Sigma(X_{t}) \ge t/2$ and $OPT < t/2$. Recall that $X^*\subset X$ is the optimal solution of $(X,t)$. The following two observations show that $x >t$ holds for every $x \in X\backslash X^*$.
		\begin{itemize}
			\item if there exists $x\in X \backslash X^*$ with $x \le t/2$, then $\{x\} \dot{\cup}  X^*$ is a better solution;
			\item if there exists $x\in X\backslash X^*$ with $ t/2< x \le t $, then $\{x\}$ is a better solution. 
		\end{itemize}
		Hence $X_t = \{x\in X : x\le t\} \subset X \backslash (X \backslash X^*) = X^* $. Since $OPT < t/2$, we have $\Sigma(X_{t}) < t/2$, which contradicts the fact that  $\Sigma(X_{t}) \ge t/2$.  Thus if $\Sigma(X_{t}) \ge t/2$, then $OPT \ge t/2$. 
		
		Note that $X_t := \{x\in X \ | \ x\le t\}$ and $\Sigma(X_t)$ can be obtained in linear time, so far the proof of  Lemma~\ref{lemma:prea} is completed.\qed
	\end{proof}
	% To see why, let $X^*\subset X$ be an optimal solution. Consider the case of $OPT < t/2$, we claim that $x >t$ holds for every $x \in X\backslash X^*$:  i). if there exists $x\in X\backslash X^*$ with $x \le t/2$ , then $\{x\} \dot{\cup}  X^*$ is a better solution; and ii). if there exists $x\in X\backslash X^*$ with $ t/2< x \le t $, then $\{x\}$ is a better solution. Thus the case that $OPT < t/2$ is trivial as we simply output $\{x\in X \ | \ x\le t\}$ in linear time as the optimal solution of $(X,t)$.

	%Taking advantage of the second assumption in Lemma~\ref{lemma:prea}, i.e., the objective value of $(X,t)$ is at least $t/2$, to develop an algorithm that returns a $(1-\tilde\OO(\epsilon))$-approximation solution for $(X,t)$ in $T$ processing time, it is sufficient to compute a $(\tilde\OO(\epsilon),t)$-approximate set with cardinality of $\OO(T)$ for $S(X)$ in $T$ processing time, while building a $T$-time oracle for backtracking from this approximate set to $X$.
	%\noindent\textbf{Remark:} 
	\noindent\textbf{How Definition~\ref{def:e_appr} and Definition~\ref{def:def_for_oracle} are used in our algorithms?} Consider any instance $(X,t)$ of SUBSET SUM. Let $X^*\subset X$ and $OPT$ be the optimal solution and optimal objective value of $(X,t)$, respectively. Taking advantage of Lemma~\ref{lemma:prea}, we only need to consider the case that $OPT \ge t/2$. Let $B_{\epsilon}$ be an $(\epsilon, t)$-approximate set of $S(X)$ and assume that there is a \textbf{T-time oracle} for backtracking from $B_{\epsilon}$ to $X$. Then according to Definition~\ref{def:e_appr} and Definition~\ref{def:def_for_oracle},  it is straightforward that a weak $(1-\OO(\epsilon))$-approximate solution of $(X,t)$ can be determined within $\OO(|B_{\epsilon}|+T)$ processing time.  %Given any instance $(X,t)$ of SUBSET SUM and let $X^*$ be its optimal solution. Let $B_{\epsilon}$ be an $(\epsilon/2, t)$-approximate set of $S(X)$, and assume that there is a \textbf{T-time oracle} for backtracking from $B_{\epsilon}$ to $X$, then according to Definition~\ref{def:e_appr} and Definition~\ref{def:def_for_oracle}, in $\OO(|B_{\epsilon}|+T)$ processing time, one can determine a weak $(1-\epsilon)$-approximate solution of $(X,t)$, i.e.,  a subset $X'\subset X$ satisfying $(1-\epsilon)\Sigma(X^*) \le \Sigma(X') \le (1+\epsilon)t$.  
	Thus towards designing a  weak approximation scheme for SUBSET SUM with a running time of $\tilde\OO(n+1/\epsilon^{1+\delta})$, it suffices to design an $\tilde\OO(n+1/\epsilon^{1+\delta})$-time algorithm such that given any instance $(X,t)$ of SUBSET SUM, the algorithm can return: (i). an $(\tilde\OO(\epsilon),t)$-approximate set with cardinality of $\tilde\OO(n+1/\epsilon^{1+\delta})$ for $S(X)$, where recall that $S(X)$ is the set of all subset-sums of $X$; and (ii). an $\tilde\OO(n+1/\epsilon^{1+\delta})$-time oracle for backtracking from this approximate set to $X$.  Therefore, when designing the approximation algorithm for SUBSET SUM, we mainly consider two problems: one is to find the approximate set of $S(X)$ of a small cardinality, and the other is to build an efficient oracle for backtracking. 
	
	\section{Computing Sumset.} % and building the corresponding backtracking oracle
	
	%In the next, we present some Lemmas that are frequently used in this paper.
	The algorithms designed in this paper for (un-)bounded SUBSET SUM follow a general divide-and-conquer framework, where we need to frequently compute (approximate) sumset and give the corresponding oracle for backtracking.  In this section, we introduce exact and approximation algorithms for computing (capped) sumset, respectively. Both exact and approximation algorithms will be used in our subsequent analysis. In particular, the exact algorithm is (almost) linear in the largest integer in the input, and will thus be used when the input only involves very small numbers (which may be $o(1/\epsilon)$). On the other hand, the approximation algorithm is linear in $1/\epsilon$, which will be used when the input numbers are large.
	
	Note that sumset is always a set. Throughout this section, when we say compute a sumset, we mean to specify every element in this sumset.

	\subsection{Exact Algorithms for  Computing Sumset.}\label{sec:compute_sumset_}
	
	One basic approach to compute the sumset is to use Fast Fourier Transform (FFT). In particular, we can use the following lemma extended from \cite{bringmann2021fine}. % 提一下他们有什么，没有什么，没有给出backtrak oracle. for the completeness of the paper,我们把证明完整的呈现
	\begin{lemma}\label{lemma:sumset}%[CF. Lemma 5.3 from \cite{bringmann2021fine}]
		Given sets %multisets 
		$X_1, X_2,\cdots,X_{\ell} \subset \mathbb{N}$, in $\OO(\sigma\log\sigma \log\ell)$ processing time, where $\sigma:= X^{max}_1+X^{max}_2+\cdots+X^{max}_{\ell}$, we can compute $X_1 \oplus X_2 \oplus \cdots \oplus X_{\ell}$ and build an $\OO(\sigma \log \sigma\log\ell)$-time oracle for backtracking from $\oplus^{\ell}_{i=1}X_{i}$ to $\dot{\cup}^{\ell}_{i=1}X_i$.  
	\end{lemma}
	\noindent\textbf{Remark.} Lemma~\ref{lemma:sumset} implies an algorithm for computing subset-sums of a  multiset. This is because that for any multiset $X\subset \mathbb{N}$, we have $S(X) = \oplus_{x \in set(X)} \{x\}$ and $X = \dot{\cup}_{x\in set(X)} \{x\}$. Then by Lemma~\ref{lemma:sumset}, in $\OO(\Sigma(X) \log \Sigma(X)\log |X|)$ processing time, we can compute $S(X)$  and build an $\OO(\Sigma(X) \log \Sigma(X)\log |X|)$-time oracle for backtracking from $S(X)$ to $X$.

	Bringmann et al. \cite{bringmann2021fine} have derived an algorithm for computing $X_1 \oplus X_2 \oplus \cdots \oplus X_{\ell}$. In the following, we show that this algorithm can be extended so that it also gives an oracle for backtracking from $\oplus^{\ell}_{i=1}X_{i}$ to $\dot{\cup}^{\ell}_{i=1}X_i$. We first present the following observation.
	\begin{observation}\label{obs:exa_sum}
		Given sets %{\color{red} multisets} 
		$B_1,B_2\subset \mathbb{N}$, in $\OO((B^{max}_{1}+B^{max}_2)\log (B^{max}_{1}+B^{max}_2))$ processing time, we can compute $B_1 \oplus B_2$ and meanwhile build an $\OO((B^{max}_{1}+B^{max}_2)\log (B^{max}_{1}+B^{max}_2))$-time oracle for backtracking from $B_1 \oplus B_2$ to $B_1 \dot{\cup} B_2$.
	\end{observation}
	\begin{proof}
		We first construct arrays $L_1$ and $L_2$ with a length of $B^{max}_{1}+B^{max}_2$. Then we store $B_i (i=1,2)$ in arrays $L_i $ such that (1). $L_i [0] = 1$; (2). $L_i [j]=1$ if $B_i$ contains $j$; (3). $L_i [j]=0$ otherwise.  %we build auxiliary array $L'_i(i=1,2)$ such that $L'_i [j] = \{x \in B_i \cup \{0\} : x = j\}$. 
		Note that the length of $L_i (i=1,2)$ is $B^{max}_{1}+B^{max}_2$. Perform FFT (Fast Fourier Transform) on $L_1$ and $L_2$, we can obtain $B_1 \oplus B_2$ in $\OO((B^{max}_{1}+B^{max}_2)\log (B^{max}_{1}+B^{max}_2))$ processing time.
		
		%We first construct arrays $L_1$ and $L_2$ with a length of $B^{max}_{1}+B^{max}_2$. %Let $L_{i}[j]$ denote the $j$-th element in array $L_i$
		%In the meanwhile, we build auxiliary array $L'_1$ and $L'_2$ with a length of $B^{max}_{1}+B^{max}_2$. Then we store $B_i (i=1,2)$ in arrays $L_i $ and $L'_i $  such that (1). $L_i [0] = 1$ and $L'_i [0] = \{0\}$; (2). $L_i [j]=1$ and $L'_i [j] = \{x \in B_i \cup \{0\} \ | \ x = j\}$ if $B_i$ contains $j$; (3). $L_i [j]=0$  and  $L'_i [j] = \emptyset$ otherwise.  %we build auxiliary array $L'_i(i=1,2)$ such that $L'_i [j] = \{x \in B_i \cup \{0\} : x = j\}$. 
		%Note that the lengths of $L_i (i=1,2)$ and $L'_i(i=1,2)$ are both $B^{max}_{1}+B^{max}_2$. Perform FFT (Fast Fourier Transform) on $L_1$ and $L_2$, we can obtain $B_1 \oplus B_2$ in $\OO((B^{max}_{1}+B^{max}_2)\log (B^{max}_{1}+B^{max}_2))$ processing time.
		
		For the second part of the observation, the oracle built for backtracking from $ B_1 \oplus B_2$ to $B_1 \dot{\cup} B_2$ works as follows: given any $c \in B_1 \oplus B_2$, consider every $b$ satisfying $L_1[b]=1$, i.e., $b\in B_1 \cup \{0\}$. If $L_2 [c-b] = 1$, i.e., $c-b \in B_2 \cup \{0\}$, let $b_1 = b$ and $b_2= c-b$, it follows $c = b_1+b_2$. It is easy to see that the time to determine $b_1$ and $b_2$ is $\OO((B^{max}_{1}+B^{max}_2)\log (B^{max}_{1}+B^{max}_2))$.\qed
	\end{proof}
	Now we are ready to prove Lemma~\ref{lemma:sumset}.
	\begin{proof}[Proof of Lemma~\ref{lemma:sumset}]
		We design an iterative approach to compute $\oplus^{\ell}_{i=1} X_i$ and build an oracle for backtracking. %Without loss of generality, we can assume that $\ell \ge 2$. 
		
		Define $X_j:= \{0\}$ for $j=\ell+1,\ell+2,\cdots,2^{pow(\ell)+1}$. It holds that $\oplus^{\ell}_{i=1} X_i = \oplus^{2^{pow(\ell)+1}}_{i=1} X_i$. We build a tree structure of $\OO(pow(\ell))$ layers as follows:
		\begin{itemize}
			\item At the beginning, we create $2^{pow(\ell)+1}$ leaf nodes and let the $i$-th leaf node contain $X_i$.
			\item  At iteration-1, we use Observation~\ref{obs:exa_sum} to compute pairwise sumset $X_{2j_1 -1} \oplus X_{2j_1}$ and derive an $\OO( (X^{max}_{2j_1 -1}+ X^{max}_{2j_1})\log(X^{max}_{2j_1 -1}+ X^{max}_{2j_1}))$-time oracle for backtracking from $X_{2j_1 -1} \oplus X_{2j_1}$ to $X_{2j_1 -1} \dot{\cup} X_{2j_1}$, where $j_1 = 1,2,\cdots, 2^{pow(\ell)}$. For two nodes containing $X_{2j_1 -1}$ and $X_{2j_1}$ separately,  we create a parent node of these two nodes, and let the parent node contain $X_{2j_1 -1} \oplus X_{2j_1}$ and the oracle for backtracking from $X_{2j_1 -1} \oplus X_{2j_1}$ to $X_{2j_1 -1} \dot{\cup} X_{2j_1}$. The processing time for one pair $X_{2j_1 -1} \oplus X_{2j_1}$ is $\OO((X^{max}_{2j_1 -1}+ X^{max}_{2j_1})\log (X^{max}_{2j_1 -1}+ X^{max}_{2j_1}))$, thus the total  processing time at iteration-1 is $\OO(\sigma\log\sigma)$. We then obtain a reduced instance with size of $2^{pow(\ell)}$, i.e., $X_{1} \oplus X_{2}, X_{3} \oplus X_{4}, \cdots, X_{2^{pow(\ell)+1}-1} \oplus X_{2^{pow(\ell)+1}}$. 
			\item Using the same approach in iteration-1 recursively, and iteratively create tree nodes. After $pow(\ell)+1$ such rounds we have built a tree structure whose root node containing $\oplus^{2^{pow(\ell)+1}}_{i=1} X_i$ and an oracle for backtracking from  $\oplus^{2^{pow(\ell)+1}}_{i=1} X_i$ to $(\oplus^{2^{pow(\ell)}}_{i=1} X_i) \dot{\cup} (\oplus^{2^{pow(\ell)+1}}_{i=2^{pow(\ell)}+1} X_i)$. 
		\end{itemize}
		Observe that the total processing time at each iteration is $\OO(\sigma\log\sigma)$, since there are $pow(\ell)+1 = \OO(\log \ell)$ iterations, thus the overall processing time to build this tree is $\OO(\sigma\log\sigma \log\ell)$. 
		
		Note that we have obtained $\oplus^{2^{pow(\ell)+1}}_{i=1} X_i$, which is contained in the root node. It remains to show that we have designed an $\OO(\sigma \log\sigma \log\ell)$-time oracle for backtracking from $\oplus^{2^{pow(\ell)+1}}_{i=1} X_i$ to $\dot{\cup}^{\ell}_{i=1} X_i$. The oracle works as follows: given any $c\in \oplus^{2^{pow(\ell)+1}}_{i=1} X_i$, through backtracking recursively from the root node to leaf nodes, a simple calculation shows that in total $\OO(\sigma \log\sigma \log\ell)$-processing time, one can determine $(c_1,c_2,\cdots, c_{\ell})$ such that $c = \Sigma^{\ell}_{i=1} c_i$, where $c_i \in X_i\cup \{0\}$ for every $ i=1,2,\cdots,\ell, \cdots,2^{pow(\ell)+1}$. \qed \end{proof}

	Sometimes we only care about computing $\omega$-capped sumset, i.e., $(\oplus^{\ell}_{i=1}X_i)\cap [0, \omega]$, and hope to design a customized algorithm with running time decreases as $\omega$ decreases. Towards this, we develop the following Lemma~\ref{lemma:cap_sumset}.
	
	%Towards computing capped sumset, i.e., $(\oplus^{\ell}_{i=1}X_i)\cap [0, \omega]$ with $\omega \in \mathbb{N}$, we develop the following Lemma~\ref{lemma:cap_sumset}.
	%algorithm derived in Observation~\ref{obs:exa_sum} can return $(\oplus^{\ell}_{i=1}X_i)\cap [0, \omega]$ in $\OO((\oplus^{\ell}_{i=1}X_i)^{max} \log (\oplus^{\ell}_{i=1}X_i)^{max}\log\ell)$ processing time, and meanwhile build an $\OO((\oplus^{\ell}_{i=1}X_i)^{max} \log (\oplus^{\ell}_{i=1}X_i)^{max}\log\ell)$-time oracle for backtracking from $(\oplus^{\ell}_{i=1}X_i)\cap [0, \omega]$ to $ \dot{\cup}^{\ell}_{i=1} X_i$. However, when $(\oplus^{\ell}_{i=1}X_i)^{max}$ is large and $\omega$ is small, the algorithm is not as efficient. Instead, 
	
	\begin{lemma}\label{lemma:cap_sumset}
		Given sets $X_1, X_2,\cdots,X_{\ell} \subset \mathbb{N}$. For any $\omega \in \mathbb{N}$, in $\OO(\ell \omega \log \omega)$ processing time, we can compute 
		$(\oplus^{\ell}_{i=1}X_i)\cap [0, \omega]$ and meanwhile build an $\OO(\ell \omega \log \omega)$-time oracle for backtracking from $(\oplus^{\ell}_{i=1}X_i)\cap [0, \omega]$ to $ \dot{\cup}^{\ell}_{i=1} X_i$.
	\end{lemma}
	\begin{proof}
		Given any sets $A,B \subset \mathbb{N}$, note that $(A \oplus B) \cap [0, \omega] = \left((A\cap [0, \omega]) \oplus (B \cap [0, \omega])\right) \cap [0, \omega]$. Observation~\ref{obs:exa_sum} guarantees that in $\OO(\omega \log \omega)$ processing time, we can compute $A \oplus B \cap [0, \omega]$, and meanwhile build an $\OO(\omega\log \omega)$-time oracle for backtracking from $(A \oplus B) \cap [0, \omega]$ to $A \dot{\cup} B$.
		
		We design an iterative approach to compute $(\oplus^{\ell}_{i=1} X_i)\cap [0, \omega]$ and build an oracle for backtracking.

		Define $X_j:= \{0\}$ for $j=\ell+1,\ell+2,\cdots,2^{pow(\ell)+1}$. It holds that $\oplus^{\ell}_{i=1} X_i = \oplus^{2^{pow(\ell)+1}}_{i=1} X_i$. We build a tree structure of $\OO(pow(\ell))$ layers as follows:
		\begin{itemize}
			\item At the beginning, we create $2^{pow(\ell)+1}$ leaf nodes and let the $i$-th leaf node contain $X_i$.
			\item At iteration-1, we use Observation~\ref{obs:exa_sum} to compute 
			pairwise $\left(X_{2j_1 -1} \oplus X_{2j_1} \right) \cap [0, \omega]$%$\left((X_{2j_1 -1} \cap [0, \omega]) \oplus (X_{2j_1} \cap [0,\omega]) \right)\cap [0, \omega] = \left(X_{2j_1 -1} \oplus X_{2j_1} \right) \cap [0, \omega]$
			and derive an $\OO(\omega\log\omega)$-time oracle for backtracking from $\left(X_{2j_1 -1} \oplus X_{2j_1} \right) \cap [0, \omega]$ to $X_{2j_1 -1} \dot{\cup} X_{2j_1}$, where $j_1 = 1,2,\cdots, 2^{pow(\ell)}$. For two nodes containing $X_{2j_1 -1}$ and $X_{2j_1}$ separately, we create a parent node of these two nodes, and let the parent node contain $\left(X_{2j_1 -1} \oplus X_{2j_1} \right) \cap [0, \omega]$ and the oracle for backtracking from $\left(X_{2j_1 -1} \oplus X_{2j_1} \right) \cap [0, \omega]$ to $X_{2j_1 -1} \dot{\cup} X_{2j_1}$. The processing time for one pair $\left(X_{2j_1 -1} \oplus X_{2j_1} \right) \cap [0, \omega]$ is $\OO(\omega \log \omega)$, thus the total 
			processing time at iteration-1 is $\OO(2^{pow(\ell)}\omega \log \omega)$.  We then obtain a reduced instance with size of $2^{pow(\ell)}$, i.e., $\left(X_{1} \oplus X_{2}\right)\cap [0,\omega], \left(X_{3} \oplus X_{4}\right)\cap [0,\omega], \cdots, \left(X_{2^{pow(\ell)+1}-1} \oplus X_{2^{pow(\ell)+1}}\right)\cap [0,\omega]$. 
			\item Using the same approach in iteration-1 recursively, and iteratively create tree nodes. After $pow(\ell)+1$ such rounds we have built a tree structure whose root node containing  $\oplus^{2^{pow(\ell)+1}}_{i=1} X_i \cap [0,\omega]$ and an $\OO(\omega\log\omega)$-time oracle for backtracking from  $\oplus^{2^{pow(\ell)+1}}_{i=1} X_i \cap [0,\omega]$ to $(\oplus^{2^{pow(\ell)}}_{i=1} X_i) \dot{\cup} (\oplus^{2^{pow(\ell)+1}}_{i=2^{pow(\ell)}+1} X_i)$.  
		\end{itemize}
		%The overall processing time of building this tree is $\OO((1+2^1 +2^2+\cdots+2^{pow(\ell)})\omega \log\omega ) = \OO(\ell \omega \log\omega)$.
		Observe that the total processing time at iteration-$i$ is $\OO(2^{pow(\ell)+1-i}\omega \log \omega)$, since there are $pow(\ell)+1$ iterations, thus the overall processing time to build this tree is $\OO((\sum^{pow(\ell)+1}_{i=1}2^{pow(\ell)+1-i})\omega \log \omega ) = \OO(\ell \omega \log \omega )$. 
		
		Note that we have obtained $\oplus^{2^{pow(\ell)+1}}_{i=1} X_i \cap [0,\omega]$, which is contained in the root node. It remains to show that we have designed an 
		$\OO(\ell \omega \log \omega )$-time oracle for backtracking from $\oplus^{2^{pow(\ell)+1}}_{i=1} X_i \cap [0,\omega]$ to $\dot{\cup}^{\ell}_{i=1} X_i$. The oracle works as follows: given any $c\in \oplus^{2^{pow(\ell)+1}}_{i=1} X_i \cap [0,\omega]$, through backtracking recursively from the root node to leaf nodes, a simple calculation shows that in total $\OO(\ell \omega  \log \omega)$-processing time, one can determine $(c_1,c_2,\cdots,c_{\ell})$ such that $c = \Sigma^{\ell}_{i=1} c_i$, where $c_i \in X_{i} \cup \{0\}$ for every $i=1,2,\cdots,\ell$.  \qed\end{proof}

	\subsection{Approximation Algorithms for  Computing Sumset.}
	%强调一下这里是我们把exact的扩展，是我们的结果
	\subsubsection{Approximating Sumset.}
	In this section, we aim to prove the following lemma, which implies an approxmation algorithm for computing sumset.
	\begin{lemma}\label{lemma:tree-fashion_}
		Given multisets $C_1, C_2,\cdots, C_{\ell} \subset \mathbb{R}_{\ge 0}$, in $\OO(\sum^{\ell}_{i=1}|C_i| +\frac{\ell}{\epsilon}\log \frac{1}{\epsilon})$ processing time, we can 
		\begin{itemize}
			\item[(i).] Compute an $\OO(\epsilon \log \ell + \epsilon^2 \ell)$-approximate set with cardinality of $\OO(\frac{1}{\epsilon})$ for $ C_1 \oplus C_2 \oplus \cdots \oplus C_{\ell}$;
			\item[(ii).] Meanwhile build an $\OO(\frac{\ell}{\epsilon}\log\frac{1}{\epsilon})$-time oracle for backtracking from this approximate set to $\dot{\cup}^{\ell}_{i=1} C_i$.
		\end{itemize}
		%Furthermore, given multisets $X_1, X_2,\cdots,X_{\ell} \subset \mathbb{R}$. Let $C_i$ be an $\tilde\OO(\epsilon)$-approximate set of $S(X_i)$ with an additive error of $ERR_i$, where $i = 1,2,\cdots, \ell$. Assume that for each $C_i$, there is a $T_i$-time oracle for backtracking from $C_i$ to $X_i$. Then in $\OO(\sum^{\ell}_{i=1}|C_i| \log |C_i|+\frac{\ell}{\epsilon}\log \frac{1}{\epsilon})$ processing time, within an additive error of $(1+\OO(\epsilon \log \ell + \epsilon^2 \ell)) \Sigma^{\ell}_{i=1}ERR_i$, we can compute an $\OO(\epsilon \log \ell + \epsilon^2 \ell)$-approximate set with cardinality of $\OO(\frac{1}{\epsilon})$ for $S(\dot{\cup}^{\ell}_{i=1}X_i)$, while building an $\OO(\sum^{\ell}_{i=1}T_i+\frac{\ell}{\epsilon}\log\frac{1}{\epsilon})$-time oracle for backtracking from this approximate set to $\dot{\cup}^{\ell}_{i=1}X_i$.
	\end{lemma}
	%proof overview。
	%介绍一下为啥需要如下两个observations
	Before proving Lemma~\ref{lemma:tree-fashion_}, we first present the following observations.
	\begin{observation}\label{obs:cc_1}
		Given multisets $B_1,B_2\subset \mathbb{R}_{\ge 0}$, in $\OO(|B_1|
		+|B_2|+\frac{1}{\epsilon}\log \frac{1}{\epsilon})$ processing time, we can compute an ${\epsilon}$-approximate set with cardinality of $\OO(\frac{1}{\epsilon})$ for $B_1 \oplus B_2$, and meanwhile build an $\OO(\frac{1}{\epsilon}\log\frac{1}{\epsilon})$-time oracle for backtracking from this approximate set to $B_1 \dot{\cup} B_2$.

		%Denote this approximate set by $C$, at the same time, 
		%build a structure which allows us to solve the following problem in $\OO(\frac{1}{\epsilon})$ time: given any $c$ in this approximate set, determine $b_i \in B_i\cup\{0\}(i=1,2)$ such that $|c-(b_1+b_2)| \le \epsilon(B^{max}_1+B^{max}_2)$.
	\end{observation}
	\begin{proof}
		Let $\sigma: = B^{max}_1+B^{max}_2$ and we define $\tilde{B}_{i} = set \left \{\left \lfloor \frac{x}{{(\epsilon \sigma)}/{2}} \right\rfloor :  x\in B_i \right\}$ for $i=1,2$. We first construct array $L_i$ with a length of $ \left \lfloor\frac{2}{\epsilon} \right\rfloor$ for $i=1,2$. 
		Then we store $\tilde{B}_{i}(i=1,2)$ in array $L_i$ such that (1). $L_i[0] =1$; (2). $L_i[j] = 1$ if $\tilde{B}_{i}$ contains $j$; (3). $L_i[j] = 0$ otherwise.  In the meantime, we build auxiliary array $L'_{i}$ with a length of $ \left \lfloor\frac{2}{\epsilon} \right\rfloor$ for $i=1,2$.  
		Then we  store ${B}_{i}(i=1,2)$ in array $L'_i$ such that (1). $L'_i=\{0\}$ ; (2). $L'_i[j] = \{x \in B_i\cup\{0\} :  \left \lfloor \frac{x}{{(\epsilon \sigma)}/{2}} \right\rfloor = j \}$ if $\tilde{B}_{i}$ contains $j$; (3). $L'_i[j] = \emptyset$ otherwise.  % and $L'_i(i=1,2)$ 
		
		Perform FFT on $L_1$ and $L_2$, we can obtain $\tilde{B}_{1} \oplus \tilde{B}_{2}$. % in $\OO(\frac{1}{\epsilon}\log\frac{1}{\epsilon})$ processing time.
		%Then we store $\frac{\epsilon u}{2} (\tilde{B}_{1} \oplus \tilde{B}_{2})$ in a array.
		Let $C := \frac{\epsilon \sigma}{2} (\tilde{B}_{1} \oplus \tilde{B}_{2})$. We first show that $C$ is an $\epsilon$-approximate set of $B_1 \oplus B_2$. It suffices to observe the followings:
		\begin{itemize}
			\item  $0\le C^{\min} \le C^{\max} \le B^{\max}_1+B^{\max}_2$, thus $C \subset [0,(1 +\epsilon )(B_1 \oplus B_2)^{\max} ]$.
			
			\item  Consider any $c=c_1+c_2 \in C$, where $c_i \in \frac{\epsilon \sigma}{2}(\tilde{B}_i\cup\{0\})$ for $i=1,2$. For each $c_i$, there exists $c'_i\in B_i\cup\{0\}$ satisfying $\left\lfloor \frac{c'_i}{{(\epsilon \sigma)}/{2}}\right\rfloor = \frac{c_i}{{(\epsilon \sigma)}/{2}}$. Let $c' = c'_1+c'_2$, it follows that $c' \in B_1 \oplus B_2$ and 
			$|c-c'| = \left| (c_1+c_2)-(c'_1+c'_2) \right| \le \frac{\epsilon \sigma}{2} \left(\left| \frac{c_1}{{(\epsilon \sigma)}/{2}}- \frac{c'_1}{{(\epsilon \sigma)}/{2}} \right|+\left| \frac{c_2}{{(\epsilon \sigma)}/{2}}-\frac{c'_2}{{(\epsilon \sigma)}/{2}} \right|\right) \le \epsilon \sigma.$
			
			\item Consider any $b = b_1+b_2 \in B_1 \oplus B_2$, where $b_i \in B_i\cup\{0\}$ for $i=1,2$. For each $b_i$, it holds that $\left\lfloor \frac{b_i}{{(\epsilon \sigma)}/{2}}\right\rfloor \in \tilde{B}_i\cup\{0\}$. Let $b'_i = \frac{\epsilon \sigma}{2}\left\lfloor \frac{b_i}{{(\epsilon \sigma)}/{2}}\right\rfloor (i=1,2)$ and let $b' = b'_1+b'_2$, it follows that $b' \in \frac{(\epsilon \sigma)}{2} (\tilde{B}_1 \oplus \tilde{B}_2)$ and 
			$|b'-b| = |(b_1+b_2)-(b'_1+b'_2)|\le  \frac{\epsilon \sigma}{2} \left( \left|
			\left\lfloor \frac{b_1}{{(\epsilon \sigma)}/{2}}\right\rfloor-
			\frac{b_1}{{(\epsilon \sigma)}/{2}}\right|
			+\left|\left \lfloor \frac{b_2}{{(\epsilon \sigma)}/{2}}\right\rfloor- \frac{b_2}{{(\epsilon \sigma)}/{2}}\right| 
			\right) \le  {\epsilon} \sigma.$
			
		\end{itemize}
		
		Now we show that an $\OO(\frac{1}{\epsilon} \log \frac{1}{\epsilon})$-time oracle for backtracking from $C$ to $B_1 \dot{\cup} B_2$ has been built and works as follows: given any $c\in C$, consider every $b$ satisfying $L_1[b] =1$, i.e., $b \in \tilde{B}_1\cup\{0\}$. If $L_2[\frac{2c}{\epsilon \sigma}-b] = 1$, i.e., $\frac{2c}{\epsilon \sigma}-b \in \tilde{B}_2\cup\{0\}$, choose any $b_1 \in L'_1[b]$ and any $b_2 \in L'_2[\frac{2c}{\epsilon \sigma}-b]$, it follows that $b_1+b_2 \in B_1\oplus B_2$ and 
		$|c-(b_1+b_2)| \le \frac{\epsilon \sigma}{2}\left(\left|b- \frac{2b_1}{\epsilon \sigma}\right|+\left|(\frac{2c}{\epsilon \sigma}-b)- \frac{2b_2}{\epsilon \sigma}\right|\right)\le \epsilon \sigma$.  It is easy to see that the time to determine $b_1$  and $b_2$ is $\OO(\frac{1}{\epsilon} \log \frac{1}{\epsilon})$.
		
		%Note that the lengths of $L_1$ and $L_2$ are both $ \left \lfloor\frac{2}{\epsilon} \right\rfloor$. 
		It remains to prove that the total processing time is $\OO(|B_1|+|B_2|+\frac{1}{\epsilon}\log \frac{1}{\epsilon})$.  One can easily prove that the time to build $L_i(i=1,2)$ and $L'_i(i=1,2)$ is $\OO(|B_1|+|B_2|)$. Recall that the length of $L_i(i=1,2)$ is  $\left \lfloor\frac{2}{\epsilon} \right\rfloor$, thus the FFT to compute $\tilde{B}_{1} \oplus \tilde{B}_{2}$ runs in $\OO(\frac{1}{\epsilon}\log\frac{1}{\epsilon})$ time. Note that  $\tilde{B}_{1} \oplus \tilde{B}_{2}$ is a set and $\tilde{B}_{1} \oplus \tilde{B}_{2} \subset [0,\frac{4}{\epsilon}]$, thus we can obtain $C$ in $\OO(\frac{1}{\epsilon})$ time. To summarize, the total processing time is $\OO(|B_1|
		+|B_2|+\frac{1}{\epsilon}\log \frac{1}{\epsilon})$. 	\qed\end{proof}

	\begin{observation}\label{obs:cc_2}
		Given multisets $A_1, A_2 \subset \mathbb{R}_{\ge 0}$, let $B_1$ and $B_2$ be $f(\epsilon)$-approximate sets of $A_1$ and $A_2$, respectively, where $f$ is an arbitrary computable function. Then any $\epsilon$-approximate set of $B_1 \oplus B_2$ is an $(\epsilon+(1+\epsilon)f(\epsilon))$-approximate set of $A_1 \oplus A_2$.
	\end{observation}
	\begin{proof}
		Given multisets $A_1, A_2 \subset \mathbb{R}_{\ge 0}$, let $B_1$ and $B_2$ be $f(\epsilon)$-approximate sets of $A_1$ and $A_2$, respectively, where $f$ is an arbitrary computable function.  We have $B^{\max}_{i} \le (1+f(\epsilon))A^{\max}_{i}$ for $i=1,2$. Denote by $C$ an $\epsilon$-approximate set of $B_1 \oplus B_2$. It suffices to observe the followings:
		\begin{itemize}
			\item $C^{max} \le (1+\epsilon)(B^{max}_1+B^{max}_2)\le (1+\epsilon)(1+f(\epsilon))(A^{max}_1+A^{max}_2) = (1+\epsilon+(1+\epsilon)f(\epsilon))(A^{max}_1+A^{max}_2)$;
			\item Consider any $c \in C$, since $C$ is an $\epsilon$-approximate set of $B_1 \oplus B_2$, there exists $c'=c'_1+c'_2 \in B_1 \oplus B_2$ where $c'_i\in B_i\cup\{0\}$, such that $|c-c'| \le \epsilon (B^{max}_1+B^{max}_2) \le \epsilon(1+f(\epsilon))(A^{max}_1+A^{max}_2)$. For each $c'_i$,
			since $B_i$ is an $f(\epsilon)$-approximate set of $A_i$, there exists $c''_i \in A_i \cup\{0\}$ such that $|c''_i-c'_i| \le f(\epsilon)A^{max}_i$. Let $c'' = c''_1+c''_2$, it follows that $c'' \in A_1 \oplus A_2 $ and $|c''-c'| \le f(\epsilon)(A^{max}_1+A^{max}_2)$. Then we have
			$|c'' -c| \le |c''-c'|+|c'-c| \le (\epsilon+ (1+\epsilon)f(\epsilon))(A^{max}_1+A^{max}_2).$
			\item Consider any sum $a = a_1+a_2 \in A_1 \oplus A_2$ where $a_i \in A_i \cup\{0\}$. For each $a_i$, since $B_i$ is an $f(\epsilon)$-approximate set of $A_i$, there exists $a'_i \in B_i \cup\{0\}$ such that $|a_i-a'_i| \le f(\epsilon)A^{max}_i$, it follows that $
			|(a'_1+a'_2) -a| \le f(\epsilon)(A^{max}_1+A^{max}_2 )$. Let $a' = a'_1+a'_2$, note that $a'\in B_1 \oplus B_2$, since $C$ is an $\epsilon$-approximate set of $B_1 \oplus B_2$, thus there exists $a'' \in C$ such that $|a''-a'| \le \epsilon (B^{max}_1+ B^{max}_2) \le \epsilon (1+f(\epsilon))(A^{max}_1+ A^{max}_2)$. Then we have
			$|a''-a|\le |a''-a'|+|a'-a| \le (\epsilon+(1+\epsilon)f(\epsilon))(A^{max}_1+ A^{max}_2).$
		\end{itemize} 
		\qed\end{proof}
	Now we are ready to prove Lemma~\ref{lemma:tree-fashion_}. 
	
	%证明中的重点要highlight。
	
	\begin{proof}[Proof of Lemma~\ref{lemma:tree-fashion_}]
		Given multisets $C_1, C_2,\cdots,C_{\ell} \subset \mathbb{R}_{\ge 0}$. In the following, we design an iterative approach to compute an approximate set for $\oplus^{\ell}_{i=1}C_i$ and build an oracle for backtracking.
		
		Define $C_j := \{0\}$ for $j = \ell+1, \ell+2,2^{pow(\ell)+1}$. It holds that $\oplus^{\ell}_{i=1}C_i = \oplus^{2^{pow(\ell)+1}}_{i=1}C_i$. We build a tree structure of $\OO(pow(\ell))$ layers as follows:
		\begin{itemize}
			\item  At the beginning, we create $2^{pow(\ell)+1}$ leaf nodes and let the $i$-th leaf node contain $C_i$.
			\item At iteration-1, for each pair $C_{2j_{1}-1}$ and $ C_{2j_{1}}$, where $j_{1} = 1,2,\cdots,2^{pow(\ell)}$, we use Observation~\ref{obs:cc_1} to compute an $\epsilon$-approximate set with cardinality of $\OO(\frac{1}{\epsilon})$ for $C_{2j_{1}-1}\oplus C_{2j_{1}}$ and derive an $\OO(\frac{1}{\epsilon}\log\frac{1}{\epsilon})$-time oracle for backtracking from this approximate set to $C_{2j_{1}-1}\dot{\cup} C_{2j_{1}}$.  Denote by $U^1_{j_1}$ this approximate set. For two nodes containing $C_{2j_{1}-1}$ and $ C_{2j_{1}}$ separately, we create a parent node of these two nodes, and let the parent node contain $U^1_{j_1}$ and the oracle for backtracking from $U^1_{j_1}$ to $C_{2j_{1}-1}\dot{\cup} C_{2j_{1}}$. 
			The processing time for one pair $C_{2j_{1}-1}$ and $ C_{2j_{1}}$ is $\OO(|C_{2j_{1}-1}|+|C_{2j_{1}}|+\frac{1}{\epsilon}\log \frac{1}{\epsilon})$ by Observation~\ref{obs:cc_1}, thus the total processing time at iteration-1 is $\OO(\Sigma^{\ell}_{i=1}|C_i|+2^{pow(\ell)}\log \frac{1}{\epsilon})$.  A reduced instance with size of $2^{pow(\ell)}$ , i.e., $U^1_{1},U^1_{2}, \cdots,U^1_{2^{pow(\ell)}}$, is obtained.

			\item Before proceeding to iteration $h$ where $h \ge 2$, we assume the following things have been done:
			\begin{itemize}
				\item Let $U^0_{j_0}$ denote $C_{j_0}$, where $j_0 = 1,2,\cdots,2^{pow(\ell)+1}$. 
				
				We have obtained $\{U^{k}_{j_{k}}  |   j_{k}=1,2,\cdots, 2^{pow(\ell)+1-k} \text{\ and\ } k=1, 2,3,\cdots, h-1\}$, where $U^{k}_{j_{k}}$ is an $\epsilon$-approximate set with cardinality of $\OO(\frac{1}{\epsilon})$ for $U^{k-1}_{2j_{k}-1} \oplus U^{k-1}_{2j_{k}}$.
				
				\item For each $U^{k}_{j_{k}}$, where $j_{k}=1,2,\cdots, 2^{pow(\ell)+1-k} \text{\ and\ } k=1,2,3,\cdots, h-1$, we have built an $\OO(\frac{1}{\epsilon}\log\frac{1}{\epsilon})$-time oracle for backtracking from $U^{k}_{j_{k}}$ to $U^{k-1}_{2j_{k}-1} \dot{\cup} U^{k-1}_{2j_{k}}$. Moreover, a node is created for containing $U^{k}_{j_{k}}$ and the oracle for backtracking from  $U^{k}_{j_{k}}$ to $U^{k-1}_{2j_{k}-1} \dot{\cup} U^{k-1}_{2j_{k}}$ .
			\end{itemize}
			Now we start iteration-$h$. For each pair $U^{h-1}_{2j_{h}-1}$ and $U^{h-1}_{2j_{h}}$, where $j_{h} = 1,2,\cdots,2^{pow(\ell)+1-h}$, we use Observation~\ref{obs:cc_1} to compute an $\epsilon$-approximate set with cardinality of $\OO(\frac{1}{\epsilon})$ for $U^{h-1}_{2j_{h}-1} \oplus U^{h-1}_{2j_{h}}$ and derive an $\OO(\frac{1}{\epsilon}\log\frac{1}{\epsilon})$-time oracle for backtracking from  this approximate set to $U^{h-1}_{2j_{h}-1}\dot{\cup} U^{h-1}_{2j_{h}}$.  Denote by $U^{h}_{j_h}$ this approximate set. For two nodes containing $U^{h-1}_{2j_{h}-1}$ and $U^{h-1}_{2j_{h}}$ separately, we create a parent node of these two nodes, and let the parent node contain $U^{h}_{j_h}$ and the oracle for backtracking from  $U^h_{j_h}$ to $U^{h-1}_{2j_{h}-1}\dot{\cup} U^{h-1}_{2j_{h}}$. Notice that 
			$|U^{h-1}_{2j_{h}-1}|= \OO(\frac{1}{\epsilon})$ and $|U^{h-1}_{2j_{h}}|= \OO(\frac{1}{\epsilon})$ for every $j_{h} = 1,2,\cdots,2^{pow(\ell)+1-h}$, then  the processing time for one pair $U^{h-1}_{2j_{h}-1}$ and $U^{h-1}_{2j_{h}}$ is $\OO(|U^{h-1}_{2j_{h}-1}|+|U^{h-1}_{2j_{h}}|+\frac{1}{\epsilon}\log \frac{1}{\epsilon}) = \OO(\frac{1}{\epsilon}\log \frac{1}{\epsilon})$  by Observation~\ref{obs:cc_1}, thus the total processing time at iteration-$h$ is $\OO(2^{pow(\ell)+1-h} \frac{1}{\epsilon}\log\frac{1}{\epsilon})$.   A reduced instance with size of $2^{pow(\ell)+1-h}$, i.e., $U^{h}_{1},U^{h}_{2},\cdots,U^{h}_{2^{pow(\ell)+1-h}}$, is obtained.

			\item Using the same approach in iteration-$h$ recursively, and iteratively create tree nodes. After $pow(\ell)+1$ such rounds, we stop and have built a tree structure whose root node contains (i). $U^{pow(\ell)+1}_1$, which is an $\epsilon$-approximate set with cardinality of $\OO(\frac{1}{\epsilon})$ for $U^{pow(\ell)}_1 \oplus U^{pow(\ell)}_2$; (ii).  an $\OO(\frac{1}{\epsilon}\log\frac{1}{\epsilon})$-time oracle for backtracking from $U^{pow(\ell)+1}_1$ to $U^{pow(\ell)}_1 \dot{\cup} U^{pow(\ell)}_2$.

			%If $h \le pow(\ell)$, we proceed to the next iteration. If $h = pow(\ell)+1$, we stop iterating.
			
			To summarize, the total processing time is $$\OO(\Sigma^{\ell}_{i=1}|C_i|+(2^1+2^2+\cdots+2^{pow(\ell)}) \frac{1}{\epsilon}\log \frac{1}{\epsilon}) = \OO(\Sigma^{\ell}_{i=1}|C_i|+\frac{\ell }{\epsilon}\log \frac{1}{\epsilon}).
			$$  
		\end{itemize}

		%Let $f^{1}(\epsilon) = \epsilon$. Assume that $U^{k}_{j_{k}}$ is an $f^{k}(\epsilon)$-approximate set of $\oplus^{2^k}_{t=1}C_{2^{k}(j_{k}-1)+t}$, where $j_{k}=1,2,\cdots, 2^{pow(\ell)+1-k}$ and $k=1,2,\cdots, h-1$, then according to Observation~\ref{obs:cc_2}, $U^{h}_{j_h}$ is an $(\epsilon+(1+\epsilon)f^{h-1}(\epsilon))$-approximate set of
		%$$\oplus^{2^h}_{t=1}C_{2^{h}(j_{h}-1)+t}=\left(\oplus^{2^{h-1}}_{t=1}C_{2^{h-1}(2j_{h}-2)+t}\right) \oplus \left(\oplus^{2^{h-1}}_{t=1}C_{2^{h-1}(2j_{h}-1)+t}\right)$$
		%Let $f^{h}(\epsilon) = (\epsilon+(1+\epsilon)f^{h-1}(\epsilon))$.  
		%Consider the following recurrence relation: $f^{1}(\epsilon) = \epsilon$ and $f^{h}(\epsilon) = \epsilon+(1+\epsilon)f^{h-1}(\epsilon) \text{ \ for \ } 1\le h \le pow(\ell)+1.$ A simple calculation shows that $f^{pow(\ell)+1}(\epsilon) = \OO(\epsilon \log \ell + \epsilon^2 \ell).$ Thus $U^{pow(\ell)+1}_1$ is an $\OO(\epsilon \log \ell + \epsilon^2 \ell)$-approximate set for $\oplus^{2^{pow(\ell)+1}}_{i=1}C_i$.
		
		Consider the functions in $\{f^h (\epsilon) \ | \ h=1,2,\cdots,  pow(\ell)+1 \}$, whch are defined by the following  recurrence relation: $f^{1} (\epsilon) = \epsilon$ and 
		$f^{h}(\epsilon) = \epsilon+(1+\epsilon)f^{h-1}(\epsilon) \text{ \ for \ } 2\le h \le pow(\ell)+1$.  According to Observation~\ref{obs:cc_2}, given integer $1\le k \le pow(\ell)$,
		if $U^{k}_{j_{k}}$ is an $f^{k}(\epsilon)$-approximate set of $\oplus^{2^k}_{t=1}C_{2^{k}(j_{k}-1)+t}$ for every $j_{k}=1,2,\cdots, 2^{pow(\ell)+1-k}$,  then $U^{k+1}_{j_{k+1}}$ is an $(\epsilon+(1+\epsilon)f^{k}(\epsilon))$-approximate set of $\oplus^{2^{k+1}}_{t=1}C_{2^{k+1}(j_{k+1}-1)+t}$, where  $j_{k+1}=1,2,\cdots, 2^{pow(\ell)-k}$. Recall that $U^1_{j_1}$ is an $f^1 (\epsilon)$-approximate set of $\oplus^{2}_{t=1}C_{2(j_{1}-1)+t}$ for every $j_{1}=1,2,\cdots, 2^{pow(\ell)}$, it can be proved by recursion that  $U^{h}_{j_{h}}$ is an $f^{h}(\epsilon)$-approximate set of $\oplus^{2^h}_{t=1}C_{2^{h}(j_{h}-1)+t}$, where  $j_{h}=1,2,\cdots, 2^{pow(\ell)+1-h}$ and $h =1,2,\cdots,   pow(\ell)+1$. In particular, $U^{pow(\ell)+1}_1$ is an $f^{pow(\ell)+1}(\epsilon)$-approximate set of $\oplus^{2^{pow(\ell)+1}}_{i=1}C_i$ and a simple calculation shows that $f^{pow(\ell)+1}(\epsilon) = \OO(\epsilon \log \ell + \epsilon^2 \ell)$.

		With the help of this tree structure,  an $\OO(\frac{\ell}{\epsilon}\log\frac{1}{\epsilon})$-time oracle for backtracking from $U^{pow(\ell)+1}_1$ to $\dot{\cup}^{\ell}_{i=1} C_i$ is derived and works as follows:
		\begin{itemize}
			\item For any $c \in U^{pow(\ell)+1}_1$, we backtrace from the root of this tree. Note that root node contains an $\OO(\frac{1}{\epsilon}\log\frac{1}{\epsilon})$-time oracle for backtracking from $U^{pow(\ell)+1}_1$ to  $U^{pow(\ell)}_1 \dot{\cup} U^{pow(\ell)}_2$.
			%Recall that $U^{pow(\ell)}_{1}$ and $U^{pow(\ell)}_{2}$ are  $f^{pow(\ell)}(\epsilon)$-approximate sets of $\oplus^{2^{pow(\ell)}}_{t=1}C_{t}$ and $\oplus^{2^{pow(\ell)}}_{t=1}C_{2^{pow(\ell)}+t}$, respectively.  
			Thus in $\OO(\frac{1}{\epsilon}\log\frac{1}{\epsilon})$ processing time, we can determine
			$y^{(pow(\ell);1)} \in U^{pow(\ell)}_1\cup\{0\}$ and $y^{(pow(\ell);2)} \in U^{pow(\ell)}_2\cup\{0\}$ such that 
			$$|c-(y^{(pow(\ell);1)}+y^{(pow(\ell);1)})| \le \epsilon  (U^{pow(\ell)}_1+U^{pow(\ell)}_2)^{max}\le \epsilon(1+f^{pow(\ell)}(\epsilon))\Sigma^{\ell}_{i=1} C^{max}_i.$$ 
			\item Backtrace recursively. 
			
			Given $1\le h \le pow(\ell)$, assume that we have determined $y^{(h;j_{h})} \in U^{h}_{j_{h}}\cup\{0\}$ for every $ j_{h} = 1,2,3, \cdots, 2^{pow(\ell)+1-h}$, such that %$\{ y^{(h;j_{h})} \in U^{h}_{j_{h}}\cup\{0\} \ | \ j_{h} = 1,2,3, \cdots, 2^{pow(\ell)+1-h}\}$ such that 
			$$
			|c-\Sigma^{2^{pow(\ell)+1-h}}_{j_{h}=1}y^{(h;j_{h})}| \le \epsilon \Sigma^{pow(\ell)}_{k=h} (1+f^{k}(\epsilon))\Sigma^{\ell}_{i=1} C^{max}_i.
			$$
			For each $y^{(h;j_{h})}\in U^{h}_{j_{h}} \cup \{0\}$, if $y^{(h;j_{h})} = 0$, then we let $y^{(h-1;2j_{h}-1)} = 0$ and $y^{(h-1;2j_{h})} = 0$. If $y^{(h;j_{h})} \in  U^{h}_{j_{h}}$, %then we let $c_i = 0$ for $i = 2^{h}(j_h-1)+1,2^{h}(j_h-1)+2,\cdots, 2^{h}j_h$ and stop backtracking from $U^{h}_{j_{h}}$ to $\dot{\cup}^{2^h}_{\iota=1}C_{2^{h}(j_{h}-1)+\iota}$. Furthermore, we let $y^{(h-1;2j_{h}-1)} = 0$ and $y^{(h-1;2j_{h})} = 0$. If $y^{(h;j_{h})} \in  U^{h}_{j_{h}}$,
			%Note that given $\{y^{(h;j_{h})} \in U^{h}_{j_{h}} \ | \ j_{h} = 1,2,3, \cdots, 2^{pow(\ell)+1-h}\}$
			note that the node containing $U^{h}_{j_{h}}$ also contains an $\OO(\frac{1}{\epsilon}\log\frac{1}{\epsilon})$-time oracle for backtracking from $U^{h}_{j_{h}}$ to $U^{h-1}_{2j_h -1} \dot{\cup} U^{h-1}_{2j_h}$. Thus
			in $\OO(\frac{1}{\epsilon}\log\frac{1}{\epsilon})$ processing time, we can determine $y^{(h-1;2j_h -1)} \in U^{h-1}_{2j_h -1} \cup\{0\}$ and $y^{(h-1;2j_h)} \in U^{h-1}_{2j_h}\cup\{0\}$ such that 
			$$|y^{(h;j_h)} -(y^{(h-1;2j_h -1)} + y^{(h-1;2j_h)} )| \le \epsilon( (U^{h-1}_{2j_h -1})^{max}+ (U^{h-1}_{2j_h})^{max}).
			$$
			Then in $\OO(2^{pow(\ell)+1-h}\frac{1}{\epsilon}\log\frac{1}{\epsilon})$ processing time, we can determine $y^{(h-1;j_{h-1})} \in U^{h-1}_{j_{h-1}}\cup\{0\}$ for every $j_{h-1} = 1,2,3, \cdots, 2^{pow(\ell)+2-h}$, %$\{ y^{(h-1;j_{h-1})} \in U^{h-1}_{j_{h-1}}\cup\{0\} \ | \ j_{h-1} = 1,2,3, \cdots, 2^{pow(\ell)+2-h}\}$ 
			such that 
			$$|\Sigma^{2^{pow(\ell)}+1-h}_{j_h = 1}y^{(h;j_{h})} - \Sigma^{2^{pow(\ell)}+2-h}_{j_{h-1} = 1}y^{(h-1;j_{h-1})}|
			\le \epsilon\Sigma^{2^{pow(\ell)+2-h}}_{j_{h-1}=1}(U^{h-1}_{j_{h-1}})^{max} \le \epsilon(1+f^{h-1}(\epsilon))\Sigma^{\ell}_{i=1}C^{max}_i.$$
			It follows that 
			$$
			|c-\Sigma^{2^{pow(\ell)+2-h}}_{j_{h-1}=1}(y^{(h-1;j_{h-1})})| \le \epsilon \Sigma^{pow(\ell)}_{k=h-1}  (1+f^{k}(\epsilon))\Sigma^{\ell}_{i=1}C^{max}_i.
			$$
			\item After $pow(\ell)$ such rounds, we stop and have determined $c_i \in C_i\cup\{0\}$ for every 
			$i = 1,2,3, \cdots, 2^{pow(\ell)+1}$, such that			%$\{ c_i \in C_i\cup\{0\}\ | \ i = 1,2,3, \cdots, 2^{pow(\ell)+1}\}$ such that
			$$
			|c- \Sigma^{\ell}_{i=1}c_i| \le \epsilon \left(1+ \Sigma^{pow(\ell)}_{k=1}  (1+f^{k}(\epsilon))\right)\Sigma^{\ell}_{i=1}C^{max}_i=  \OO(\epsilon \log \ell + \epsilon^2 \ell)\Sigma^{\ell}_{i=1}C^{max}_i.$$
			To summarize, the total processing time is $\OO((1+2^1 + 2^2 + \cdots+ 2^{pow(\ell)+1})\frac{1}{\epsilon}\log\frac{1}{\epsilon}) = \OO(\frac{\ell}{\epsilon} \log\frac{1}{\epsilon})$.
		\end{itemize}

		In conclude, within $\OO(\Sigma^{\ell}_{i=1}|C_i|+\frac{\ell }{\epsilon}\log \frac{1}{\epsilon})$ processing time, we will build a tree structure whose root node contains an $\OO(\epsilon \log \ell + \epsilon^2 \ell)$-approximate set with cardinality of $\OO(\frac{1}{\epsilon})$ for $\oplus^{2^{pow(\ell)+1}}_{i=1}C_i$. Meanwhile, with the help of this tree structure, an $\OO(\frac{\ell}{\epsilon}\log\frac{1}{\epsilon})$-time oracle for backtracking from this approximate set to $\dot{\cup}^{\ell}_{i=1} C_i$ is derived.		
		%and this tree structure is an oracle for backtracking from this approximate set to$\dot{\cup}^{2^{pow(\ell)+1}}_{i=1}C_i$
		\qed\end{proof}
	
	Lemma~\ref{lemma:tree-fashion_} implies the following corollary, which allows us to build the approximate set of $\dot{\cup}^{\ell}_{i=1}X_i$ from the approximate set of each $X_i$.
	\begin{corollary}\label{coro:tree-fashion_}
		Given multisets $X_1, X_2,\cdots,X_{\ell} \subset \mathbb{R}$. Let $C_i$ be an $\tilde\OO(\epsilon)$-approximate set of $S(X_i)$ with an additive error of $ERR_i$, where $i = 1,2,\cdots, \ell$. Assume that for each $C_i$, there is a $T_i$-time oracle for backtracking from $C_i$ to $X_i$. Then in $\OO(\sum^{\ell}_{i=1}|C_i| +\frac{\ell}{\epsilon}\log \frac{1}{\epsilon})$ processing time, within an additive error of $(1+\OO(\epsilon \log \ell + \epsilon^2 \ell)) \Sigma^{\ell}_{i=1}ERR_i$, we can 
		\begin{itemize}
			\item[(i).] Compute an $\tilde\OO(\epsilon \log \ell + \epsilon^2 \ell)$-approximate set with cardinality of $\OO(\frac{1}{\epsilon})$ for $S(\dot{\cup}^{\ell}_{i=1}X_i)$;
			\item[(ii).] Meanwhile build an $\OO(\sum^{\ell}_{i=1}T_i+\frac{\ell}{\epsilon}\log\frac{1}{\epsilon})$-time oracle for backtracking from this approximate set to $\dot{\cup}^{\ell}_{i=1}X_i$.
		\end{itemize}
		
	\end{corollary}
	\begin{proof}%[Proof of Corollary~\ref{coro:tree-fashion_}.]

		%In the next, we will prove the second part of Lemma~\ref{lemma:tree-fashion_}. 
		Given multisets $X_1, X_2,\cdots,X_{\ell} \subset \mathbb{R}$. Let $C_i$ be an $\tilde\OO(\epsilon)$-approximate set of $S(X_i)$ with an additive error of $ERR_i$, where $i = 1,2,\cdots, \ell$. Assume that for each $C_i$, there is a $T_i$-time oracle for backtracking from $C_i$ to $X_i$. The approach to approximate $S(\dot{\cup}^{\ell}_{i=1}X_i)$ and build an oracle for backtracking is very similar to the one we have designed in the proof of Lemma~\ref{lemma:tree-fashion_}. The only difference is that when creating leaf nodes, we let the $i$-th node contain not only $C_i$ but also the $T_i$-time oracle for backtracking from $C_i$ to $X_i$. Then we use the above iterative approach to approximate $\oplus^{\ell}_{i=1}C_i$, in $\OO(\Sigma^{\ell}_{i=1}|C_i|+\frac{\ell }{\epsilon}\log \frac{1}{\epsilon})$ processing time, we will build an augmented tree structure whose root node contains an $\OO(\epsilon \log \ell + \epsilon^2 \ell)$-approximate set with cardinality of $\OO(\frac{1 }{\epsilon})$ for $\oplus^{\ell}_{i=1}C_i$.  Denote by $U$ this approximate set.
		Meanwhile, with the help of this tree structure, an $\OO(\frac{\ell}{\epsilon}\log \frac{1}{\epsilon})$-time oracle for backtracking from $U$ to $\dot{\cup}^{\ell}_{i=1}C_i$ is derived.  
		
		We claim that within an additive error of $(1+\OO(\epsilon \log \ell + \epsilon^2 \ell)) \Sigma^{\ell}_{i=1}ERR_i$, $U$ is an $\tilde\OO(\epsilon \log \ell + \epsilon^2 \ell)$-approximate set of $S(\dot{\cup}^{\ell}_{i=1}X_i)$. Towards the claim, it is sufficient to observe the followings:
		\begin{itemize}
			\item Since $U$ is an $\OO(\epsilon \log \ell + \epsilon^2 \ell)$-approximate set  of $\oplus^{\ell}_{i=1}C_i$, then $U^{\max} \le  (1+\OO(\epsilon \log \ell + \epsilon^2 \ell)) \Sigma^{\ell}_{i=1}C^{\max}_i$. For each $S(X_i)$, note that $C_i$ is an $\tilde\OO(\epsilon)$-approximate set of $S(X_i)$ with an additive error of $ERR_i$, we have $C^{max}_i \le (1+\tilde\OO(\epsilon))\Sigma(X_i)+ERR_i$ and 
			\begin{align*}
				U^{\max} &\le  (1+\OO(\epsilon \log \ell + \epsilon^2 \ell)) \Sigma^{\ell}_{i=1}C^{\max}_i\\
				& \le  (1+\OO(\epsilon \log \ell + \epsilon^2 \ell))( (1+\tilde\OO(\epsilon))\sum^{\ell}_{i=1}\Sigma(X_i)+\sum^{\ell}_{i=1}ERR_i )\\
				&\le (1+\tilde\OO(\epsilon \log \ell + \epsilon^2 \ell))\sum^{\ell}_{i=1}\Sigma(X_i)+ (1+\OO(\epsilon \log \ell + \epsilon^2 \ell)) \sum^{\ell}_{i=1}ERR_i
			\end{align*}
			
			\item Consider any $s\in S(\dot{\cup}^{\ell}_{i=1}X_i)$. There exists $(s_1,s_2,\cdots,s_{\ell})$ where $s_i \in S(X_i)$ for $i=1,2,\cdots,\ell$, such that $s = \Sigma^{\ell}_{i=1}s_i$. For each $S(X_i)$, recall that $C_i$ is an $\tilde\OO(\epsilon)$-approximate set of $S(X_i)$ with an additive error of $ERR_i$, then $C^{max}_i \le (1+\tilde\OO(\epsilon))\Sigma(X_i)+ERR_i$. Moreover, for each $s_i \in S(X_i)$, there exits $s'_i\in C_i$ such that $|s'_i-s_i|\le \tilde\OO(\epsilon)\Sigma(X_i)+ERR_i$, it follows that $|s-\Sigma^{\ell}_{i=1}s'_i|\le \tilde\OO(\epsilon)\sum^{\ell}_{i=1}\Sigma(X_i)+\Sigma^{\ell}_{i=1}ERR_i.$
			Let $s' = \Sigma^{\ell}_{i=1}s'_i$. Note that $s'\in \oplus^{\ell}_{i=1} C_i$ and $U$ is an $\OO(\epsilon \log \ell + \epsilon^2 \ell)$-approximate set of  $\oplus^{\ell}_{i=1}C_i$, thus there exists $s''\in U$ such that 
			$|s''-s'|\le \OO(\epsilon \log \ell + \epsilon^2 \ell)\Sigma^{\ell}_{i=1}C^{max}_i \le \tilde\OO(\epsilon \log \ell + \epsilon^2 \ell)\sum^{\ell}_{i=1}\Sigma(X_i) + \OO(\epsilon \log \ell + \epsilon^2 \ell) \Sigma^{\ell}_{i=1}ERR_i.$
			%To summarize, for any $s\in S(\dot{\cup}^{\ell}_{i=1}X_i)$, there exists $s''\in U$ such that 
			Furthermore, we have
			$$|s''-s|\le |s''-s'|+|s'-s|\le  \tilde\OO(\epsilon \log \ell + \epsilon^2 \ell)\sum^{\ell}_{i=1}\Sigma(X_i) + (1+\OO(\epsilon \log \ell + \epsilon^2 \ell)) \Sigma^{\ell}_{i=1}ERR_i.$$
			\item Consider any $y\in U$. Since $U$ is an $\OO(\epsilon \log \ell + \epsilon^2 \ell)$-approximate set  of $\oplus^{\ell}_{i=1}C_i$, there exists $(c_1,c_2,\cdots,c_{\ell})$ where $c_i\in C_i\cup\{0\}$ for $i=1,2,\cdots$, such that 	%$\{c_i\in C_i\cup\{0\} \ | \ i=1,2,\cdots,\ell\}$ 
			$|y-\Sigma^{\ell}_{i=1}c_i|\le  \OO(\epsilon \log \ell + \epsilon^2 \ell)\Sigma^{\ell}_{i=1}C^{max}_i \le \tilde\OO(\epsilon \log \ell + \epsilon^2 \ell)\sum^{\ell}_{i=1}\Sigma(X_i) + \OO(\epsilon \log \ell + \epsilon^2 \ell) \Sigma^{\ell}_{i=1}ERR_i$. For each $c_i\in C_i\cup\{0\}$, if $c_i = 0$, let $X'_i = \emptyset$, apparently $X'_i \subset X_i$ and $|c_i -\Sigma(X'_i)|\le \tilde\OO(\epsilon) \Sigma(X_i)+ERR_i$. Else if $c_i\in C_i$, note that $C_i$ is an $\tilde\OO(\epsilon)$-approximate set of $S(X_i)$ with an additive error of $ERR_i$, there exists  $X'_i \subset X_i$ such that $|c_i -\Sigma(X'_i)|\le \tilde\OO(\epsilon) \Sigma(X_i)+ERR_i$.  Furthermore, we have

			%To summarize,  there exists $X'_i \subset X_i$ for every $i=1,2,\cdots,\ell$ such that 
			
			%This follows that $|\Sigma^{\ell}_{i=1}c_i -\Sigma^{\ell}_{i=1}\Sigma(X'_i)| \le \tilde\OO(\epsilon)\Sigma^{\ell}_{i=1}\Sigma(X_i)+\Sigma^{\ell}_{i=1}ERR_i$. To summarize, in total $\OO(\Sigma^{\ell}_{i=1}T_i+\ell/\epsilon)$  processing time, with the help of the augmented tree structure, we can  determine $\{X'_i \subset X_i \ | \ i=1,2,\cdots,\ell\}$ such that 
			\begin{align*}
				|y- \sum^{\ell}_{i=1}\Sigma(X'_i)|&\le |y- \sum^{\ell}_{i=1}c_i|+|\sum^{\ell}_{i=1}c_i-\sum^{\ell}_{i=1}\Sigma(X'_i)| \\
				&\le  \tilde\OO(\epsilon \log \ell + \epsilon^2 \ell)\sum^{\ell}_{i=1}\Sigma(X_i) + \OO(\epsilon \log \ell + \epsilon^2 \ell) \Sigma^{\ell}_{i=1}ERR_i + \tilde\OO(\epsilon)\Sigma^{\ell}_{i=1}\Sigma(X_i)+\Sigma^{\ell}_{i=1}ERR_i\\
				&= \tilde\OO(\epsilon \log \ell + \epsilon^2 \ell)\sum^{\ell}_{i=1}\Sigma(X_i) + (1+\OO(\epsilon \log \ell + \epsilon^2 \ell)) \sum^{\ell}_{i=1}ERR_i.
			\end{align*}		
			It is easy to see that $\dot{\cup}^{\ell}_{i=1} X'_{i} \subset \dot{\cup}^{\ell}_{i=1} X_i$ and $\Sigma(\dot{\cup}^{\ell}_{i=1} X'_{i} ) = \sum^{\ell}_{i=1}\Sigma(X'_i) \in S(\dot{\cup}^{\ell}_{i=1} X_i)$.
		\end{itemize}
		
		%Thus with an additive error of $(1+\OO(\epsilon \log \ell + \epsilon^2 \ell)) \sum^{\ell}_{i=1}ERR_i$,  $U$ is an $\tilde\OO(\epsilon \log \ell + \epsilon^2 \ell)$-approximate set of $S(\Sigma^{\ell}_{i=1}\dot{\cup}X_i)$. moreover, an $\OO(\Sigma^{\ell}_{i=1}T_i+\frac{\ell}{\epsilon}\log\frac{1}{\epsilon})$-time oracle for backtracking from $U$ to $\dot{\cup}^{\ell}_{i=1}X_i$.
		
		In the following, we show that with the help of the above augmented tree structure, an $\OO(\Sigma^{\ell}_{i=1}T_i+\frac{\ell}{\epsilon}\log\frac{1}{\epsilon})$-time oracle for backtracking from $U$ to $\dot{\cup}^{\ell}_{i=1} X_i$ can be derived. The oracle works as follows. Given any $y\in U$, note that $U$ is an $\OO(\epsilon\log\ell+\epsilon^2 \ell)$-approximate set of $\oplus^{\ell}_{i=1}C_i$ and with the help of augmented tree structure, an $\OO(\frac{\ell}{\epsilon}\log\frac{1}{\epsilon})$-time oracle for backtracking from $U$ to $\dot{\cup}^{\ell}_{i=1} C_i$ has been derived. Thus in $\OO(\frac{\ell}{\epsilon}\log\frac{1}{\epsilon})$ processing time, we can determine $c_i\in C_i\cup\{0\}$ for every $i=1,2,\cdots,\ell$ such that $|y-\Sigma^{\ell}_{i=1}c_i|\le  \OO(\epsilon \log \ell + \epsilon^2 \ell)\Sigma^{\ell}_{i=1}C^{max}_i \le \tilde\OO(\epsilon \log \ell + \epsilon^2 \ell)\sum^{\ell}_{i=1}\Sigma(X_i) + \OO(\epsilon \log \ell + \epsilon^2 \ell) \Sigma^{\ell}_{i=1}ERR_i$. For each $c_i\in C_i\cup\{0\}$, if $c_i = 0$, let $X'_i = \emptyset$, apparently $X'_i \subset X_i$ and $|c_i -\Sigma(X'_i)|\le \tilde\OO(\epsilon) \Sigma(X_i)+ERR_i$. Else if $c_i\in C_i$, recall that the leaf node containing $C_i$ also contains a $T_i$-time oracle for backtracking from $C_i$ to $X_i$, thus in $T_i$ processing time, we can determine $X'_i \subset X_i$ such that $|c_i -\Sigma(X'_i)|\le \tilde\OO(\epsilon) \Sigma(X_i)+ERR_i$. To summarize, with the help of augmented tree structure, in total $\OO(\Sigma^{\ell}_{i=1}T_i+\frac{\ell}{\epsilon}\log\frac{1}{\epsilon})$  processing time, we can determine $X'_i \subset X_i$ for every $i=1,2,\cdots,\ell$ such that   
		$$|y- \sum^{\ell}_{i=1}\Sigma(X'_i)| \le \tilde\OO(\epsilon \log \ell + \epsilon^2 \ell)\sum^{\ell}_{i=1}\Sigma(X_i) + (1+\OO(\epsilon \log \ell + \epsilon^2 \ell)) \sum^{\ell}_{i=1}ERR_i.$$ 
		Note that $ \dot{\cup}^{\ell}_{i=1} X'_i \subset  \dot{\cup}^{\ell}_{i=1} X_i$ and $\Sigma( \dot{\cup}^{\ell}_{i=1} X'_i) = \sum^{\ell}_{i=1} \Sigma(X'_i)$, hence Corollary~\ref{coro:tree-fashion_} is proved.		
		%This follows that $|\Sigma^{\ell}_{i=1}c_i -\Sigma^{\ell}_{i=1}\Sigma(X'_i)| \le \tilde\OO(\epsilon)\Sigma^{\ell}_{i=1}\Sigma(X_i)+\Sigma^{\ell}_{i=1}ERR_i$. To summarize, in total $\OO(\Sigma^{\ell}_{i=1}T_i+\ell/\epsilon)$  processing time, with the help of the augmented tree structure, we can  determine $\{X'_i \subset X_i \ | \ i=1,2,\cdots,\ell\}$ such that 
		\qed\end{proof}
	
	\subsubsection{Approximating Capped Sumset.}
	Sometimes, we only care about approximating  $\omega$-capped sumset, e.g., $(\oplus^{\ell}_{i=1}X_i)\cap [0, \omega]$. Towards this, we develop the following Lemma~\ref{lemma:sub_sum_top}. %For this purpose, we design a customized algorithm whose running time decreases as $\omega$ decreases.
	
	\begin{lemma}\label{lemma:sub_sum_top}
		Given multisets $C_1,C_2,\cdots,C_{\ell} \subset  \mathbb{R}_{\ge 0}$ % and a fixed constant $\omega \in \mathbb{R}_{> 0}$. 
		and a parameter $\omega >0$. 
		In $\OO(\sum^{\ell}_{i=1}|C_i| + \frac{\ell}{\epsilon}\log \frac{1}{\epsilon})$ processing time, we can 
		\begin{itemize}
			\item[(i).] Compute a $(\OO(\epsilon \ell + \epsilon^2 \ell^2), \omega)$-approximate set with cardinality of $\OO(\frac{1}{\epsilon})$ for $C_1 \oplus C_2 \oplus \cdots \oplus C_{\ell}$;
			\item[(ii).] Meanwhile build an $\OO(\frac{\ell}{\epsilon}\log\frac{1}{\epsilon})$-time oracle for backtracking from this approximate set to  $\dot{\cup}^{\ell}_{i=1}C_i$.
		\end{itemize}
	\end{lemma}

	The proof idea of Lemma~\ref{lemma:sub_sum_top} is similar to the proof of Lemma~\ref{lemma:tree-fashion_}. We iteratively construct a tree structure such that the root node contains the desired approximate set of $\oplus^{\ell}_{i=1} C_i$ and the corresponding backtracking oracle. %for backtracking from this approximate set to $\dot{\cup}^{\ell}_{i=1} C_i$.
	When building the tree structure, the biggest difference is that in the proof of Lemma~\ref{lemma:tree-fashion_}, each layer-$h$ node is obtained by computing the $\epsilon$-approximate set of the sumset of its two child nodes, while in the proof of Lemma~\ref{lemma:sub_sum_top}, each layer-$h$ node is obtained by computing the $(\epsilon,(1+2f^{h}(\epsilon))\omega)$-approximate set of the sumset of its two child nodes, where $f^h(\epsilon)$'s are defined by the following recurrence relation: $f^0(\epsilon) = 0$  and $f^{t}(\epsilon) = \epsilon+2(1+\epsilon)f^{t-1}(\epsilon) \text{ \ for \ } 1\le t \le  h-1$. 
	
	Before proving Lemma~\ref{lemma:sub_sum_top}, we first present the following observations. In particular, Observation~\ref{obs:cap_apx_ss} provides an efficient algorithm for computing a $(\epsilon,(1+2f^{h}(\epsilon))\omega)$-approximate set of the sumset of given two multisets, and Observation~\ref{obs:apx_apx_ss}  guarantees that the set in root node is actually the $(\OO(\epsilon\ell+\epsilon^2 \ell^2),\omega)$-approxiamte set of $\oplus^{\ell}_{i=1}C_i$. 
	
	%in which each layer-$h$ node is obtained by approximating the sumset of its two child nodes in layer-$(h+1)$. %The leaf nodes contain $C_i$'s respectively and the root node contains the desired approximate set of $\oplus^{\ell}_{i=1} C_i$.
	
	%finally the root node contains the desired approximate set of $\oplus^{\ell}_{i=1} C_i$ and the corresponding oracle for backtracking from this approximate set to $\dot{\cup}^{\ell}_{i=1} C_i$.
	
	%The proof idea of Lemma~\ref{lemma:sub_sum_top} is similar to the proof of Lemma~\ref{}, which is designing an iterative approach to (i). compute an $(\OO(\epsilon \ell + \epsilon^2 \ell^2), \omega)$-approximate set for $\oplus^{\ell}_{i=1}C_i$ and (ii). meanwhile build an oracle for backtracking. The most different is that in each iteration, 
	
	%Note that when $\omega = \sum^{\ell}_{i=1} C_i^{\max}$, Lemma~\ref{lemma:sub_sum_top}  

	%Before proving Lemma~\ref{lemma:sub_sum_top}, we first present the following observations.
	\begin{observation}\label{obs:cap_apx_ss}
		Given multisets $B_1,B_2\subset \mathbb{R}_{\ge 0}$ %and a fixed constant $\omega > 0$, 
		and a parameter $\omega > 0$. In $\OO(|B_1| +|B_2|+\frac{1}{\epsilon}\log \frac{1}{\epsilon})$ processing time, we can compute an $(\epsilon,\omega)$-approximate set with cardinality of $\OO(\frac{1}{\epsilon})$ for $B_1 \oplus B_2$, while building an  $\OO(\frac{1}{\epsilon}\log\frac{1}{\epsilon})$-time oracle for backtracking from this approximate set to $B_1 \dot{\cup} B_2$.
	\end{observation}
	\begin{proof}
		Define $B^{\omega}_{i} = \{b \ | \ b \in B_{i} \text{ \ and \ } b \le \omega \}$ for $i=1,2$.  Recall Observation~\ref{obs:cc_1}, in $\OO(|B^{\omega}_{1} |+|B^{\omega}_{2} |+\frac{1}{\epsilon}\log\frac{1}{\epsilon})$ processing time, we can compute an $\frac{\epsilon}{2}$-approximate set with cardinality of $\OO(\frac{1}{\epsilon})$ for $B^{\omega}_{1} \oplus B^{\omega}_{2}$, and meanwhile build an $\OO(\frac{1}{\epsilon}\log \frac{1}{\epsilon})$-time oracle for backtracking from this approximate set to $B^{\omega}_1 \dot{\cup} B^{\omega}_2$. Denote by $C$ this approximate set. Let $C^{\omega}= C\cap[0,(1+\epsilon)\omega]$.   We claim that $C^{\omega}$ is an $(\epsilon,\omega)$-approximate set of $B_1 \oplus B_2$. It is sufficient to observe the followings:
		\begin{itemize}	
			\item $C^{\omega} = C\cap[0,(1+\epsilon)\omega] \subset [0, (1+\epsilon)\omega]$;
			\item Consider any $c \in C^{\omega}$. Note that $C^{\omega}\subset C$ and $C$ is an $\frac{\epsilon}{2}$-approximate set of $B^{\omega}_{1} \oplus B^{\omega}_{2}$, thus there exist $c' \in B^{\omega}_{1} \oplus B^{\omega}_{2}$ such that $|c-c'| \le \frac{\epsilon}{2} ( B^{\omega}_{1} \oplus B^{\omega}_{2})^{\max} \le \epsilon \omega$. %Note that $c'_1+c'_2 \in B_1 \oplus B_2$.
			\item Consider any $b \in (B_1 \oplus B_{2} )\cap [0,\omega]$. It is easy to see that $b \le \omega$ and $b \in B^{\omega}_{1} \oplus B^{\omega}_{2}$. By the fact that $C$ is an $\frac{\epsilon}{2}$-approximate set of $B^{\omega}_{1} \oplus B^{\omega}_{2}$, there exists $b'\in C$ such that $|b-b'|\le \frac{\epsilon}{2} ( B^{\omega}_{1} \oplus B^{\omega}_{2})^{\max} \le \epsilon \omega$. It follows that $b' \le b +\epsilon \omega \le (1+\epsilon)\omega$, thus $b' \in C^{\omega}$.
		\end{itemize}
		
		Moreover, note that $C^{\omega} \subset C$ and $B^{\omega}_1 \dot{\cup} B^{\omega}_2 \subset B_1 \dot{\cup} B_2$, the $\OO(\frac{1}{\epsilon}\log\frac{1}{\epsilon})$-time oracle built for backtracking from $C$ to $B^{\omega}_1 \dot{\cup} B^{\omega}_2$ diretly yields an  $\OO(\frac{1}{\epsilon}\log\frac{1}{\epsilon})$-time oracle for backtracking from $C^{\omega}$ to $B_1 \dot{\cup} B_2$.
		
		Note that the time to obtain $B^{\omega}_{1}$ and $B^{\omega}_2$ is $\OO(|B_{1}|+|B_{2}|)$. Since $|C| = \OO(\frac{1}{\epsilon})$, the time to obtain $C^{\omega}$ from $C$ is $\OO(\frac{1}{\epsilon})$. To summarize, within $\OO(|B_{1}|+|B_{2}|+\frac{1}{\epsilon}\log\frac{1}{\epsilon})$ processing time,  we can compute an $(\epsilon,\omega)$-approximate set with cardinality of $\OO(\frac{1}{\epsilon})$ for $B_1 \oplus B_2$, while building an  $\OO(\frac{1}{\epsilon}\log\frac{1}{\epsilon})$-time oracle for backtracking from this approximate set to $B_1 \dot{\cup} B_2$.
		\qed\end{proof}
	
	\begin{observation}\label{obs:apx_apx_ss}
		Given multisets $A_1, A_2 \subset \mathbb{R}_{\ge 0}$ %and a fixed constant $\omega>0$. 
		and a parameter $\omega>0$. Let $B_1$ and $B_2$ be $(f(\epsilon),\omega)$-approximate sets of $A_1$ and $A_2$, respectively, where $f$ is an arbitrary computable function. Then any $(\epsilon,(1+2f(\epsilon))\omega)$-approximate set of $B_1 \oplus B_2$ is an $(\epsilon+2(1+\epsilon)f (\epsilon),\omega)$-approximate set of $A_1 \oplus A_2$.
		
	\end{observation}
	\begin{proof}
		Given multisets $A_1, A_2 \subset \mathbb{R}_{\ge 0}$ %and a fixed constant $\omega>0$. 
		and a parameter $\omega>0$. Let $B_1$ and $B_2$ be $(f(\epsilon),\omega)$-approximate sets of $A_1$ and $A_2$, respectively, where $f$ is an arbitrary computable function.  We have $B^{\max}_{i} \le (1+f(\epsilon))\omega$ for $i=1,2$. Denote by $C$ an $(\epsilon,(1+2f(\epsilon))\omega)$-approximate set of $B_1 \oplus B_2$. Towards Observation~\ref{obs:apx_apx_ss}, it is sufficient to observe the followings:
		
		\begin{itemize}
			\item Since $C$ is an $(\epsilon,(1+2f(\epsilon))\omega)$-approximate set of $B_1 \oplus B_2$, we have  $C \subset [0, (1+\epsilon)(1+2f(\epsilon))\omega]= [0, (1+\epsilon+2(1+\epsilon)f (\epsilon))\omega].$
			\item  Consider any $c\in C$. Since $C$ is an $(\epsilon,(1+2f(\epsilon))\omega)$-approximate set of $B_1 \oplus B_2$,  there exists $c' = c'_1 + c'_2 \in B_1 \oplus B_2$ where $c'_i \in B_i \cup\{0\}$, such that $|c-c'|\le \epsilon (1+2f(\epsilon))\omega$. Recall that $B_i (i=1,2)$ is an $(f (\epsilon),\omega)$-approximate set of $A_i$, then for each $c'_i$, there exists $c''_i \in A_i \cup\{0\}$ such that $|c''_i -c'_i | \le f (\epsilon)\omega$. Let $c'' = c''_1+c''_2$, it follows that $c''\in A_1 \oplus A_2$ and  $|c''-c'| = |(c''_1+c''_2)-(c'_1+c'_2)|\le 2f (\epsilon) \omega$. Furthermore, we have $
			|c''-c|\le |c''-c'|+|c'-c| \le 2f (\epsilon)\omega+ \epsilon (1+2f(\epsilon))\omega = (\epsilon+2(1+\epsilon)f (\epsilon))\omega.$
			
			\item Consider any $a \in (A_1 \oplus A_2)\cup [0,\omega]$. There exist $a_1 \in (A_{1}\cup\{0\})\cap[0,\omega]$ and  $a_2 \in (A_{2}\cup\{0\})\cap[0,\omega]$ such that $a = a_1+a_2$. For each $a_i$, recall that $B_i$ is an $(f (\epsilon),\omega)$-approximate set of $A_i$, there exists $a'_i \in B_i \cup\{0\}$ such that $|a_i -a'_i|\le f(\epsilon)\omega$, it follows that $|(a'_1+a'_2)-a| \le |a'_1-a_1|+|a'_2-a_2| \le 2f(\epsilon)\omega. $
			Let $a' = a'_1+a'_2$, it holds that
			$a' = a'_1+a'_2\le a+2f(\epsilon)\omega \le (1+2f(\epsilon))\omega$, thus $a' \in (B_1 \oplus B_2)\cap [0,(1+2f(\epsilon))\omega]$. Recall that $C$ is an $(\epsilon,(1+2f(\epsilon))\omega)$-approximate set of $B_1 \oplus B_2$,  thus for $a'$, there exists $a'' \in C$ such that $|a''-a'| \le \epsilon (1+2f(\epsilon))\omega$. To summarize, we have  $|a''-a|\le |a''-a'|+|a'-a| \le  \epsilon (1+2f(\epsilon))\omega+2f(\epsilon)\omega = (\epsilon+2(1+\epsilon)f (\epsilon))\omega.$
		\end{itemize}		
		\qed\end{proof}
	
	Now we are ready to prove Lemma~\ref{lemma:sub_sum_top}. %To be specific, we will design an iterative approach to compute an $(\OO(\epsilon \ell + \epsilon^2 \ell^2), \omega)$-approximate set for $\oplus^{\ell}_{i=1}C_i$ and build a oracle for backtracking.

	\begin{proof}[Proof of Lemma~\ref{lemma:sub_sum_top}] 
		Given multisets $C_1,C_2,\cdots,C_{\ell}\subset \mathbb{R}_{\ge 0}$ %and a fixed constant $\omega \in \mathbb{R}_{>0}$. 
		and a parameter $\omega \in \mathbb{R}_{>0}$. In the following, we design an iterative approach to compute an approximate set for $\oplus^{\ell}_{i=1}C_i$ and build an oracle for backtracking.
		
		Define $C_j := \{0\}$ for $j = \ell+1, \ell+2,2^{pow(\ell)+1}$. It holds that $\oplus^{\ell}_{i=1}C_i = \oplus^{2^{pow(\ell)+1}}_{i=1}C_i$. We build a tree structure of $\OO(pow(\ell))$ layers as follows:
		\begin{itemize}
			\item At the beginning, we create $2^{pow(\ell)+1}$ leaf nodes and let the $i$-th node contain $C_i$.  Let $f^{0} (\epsilon) := 0$.
			\item 	At iteration-1, for each pair $C_{2j_{1}-1}$ and $C_{2j_{1}}$, where $j_{1} = 1,2,\cdots,2^{pow(\ell)}$, we use Observation~\ref{obs:cap_apx_ss} to compute an $(\epsilon,\omega)$-approximate set with cardinality of $\OO(\frac{1}{\epsilon})$ for $C_{2j_{1}-1}\oplus C_{2j_{1}}$ and derive an $\OO(\frac{1}{\epsilon}\log\frac{1}{\epsilon})$-time oracle for backtracking from this approximate set to $C_{2j_{1}-1}\dot{\cup} C_{2j_{1}}$. Denote by $U^1_{j_1}$ this approximate set.
			For two nodes containing $C_{2j_{1}-1}$ and $ C_{2j_{1}}$ separately, we create a parent node of these two nodes, and let the parent node contain $U^1_{j_1}$ and the oracle for backtracking from $U^1_{j_1}$ to $C_{2j_{1}-1}\dot{\cup} C_{2j_{1}}$. The processing time for one pair $C_{2j_{1}-1}$ and $ C_{2j_{1}}$ is $\OO(|C_{2j_{1}-1}|+|C_{2j_{1}}|+\frac{1}{\epsilon}\log \frac{1}{\epsilon})$ by Observation~\ref{obs:cap_apx_ss}, thus the total processing time at iteration-1 is $\OO(\Sigma^{\ell}_{i=1}|C_i| +2^{pow(\ell)} \frac{1}{\epsilon}\log \frac{1}{\epsilon})$. 
			A reduce instance with size of $2^{pow({\ell})}$, i.e., $U^1_{1}, U^1_{2}, \cdots, U^1_{2^{pow(\ell)}}$, is obtained.   Let $f^{1}(\epsilon) := \epsilon$.
			\item Before proceeding to iteration $h$ where $h\ge 2$, we assume the following things have been done:
			\begin{itemize}
				\item 
				Let $U^0_{j_0}$ denote $C_{j_0}$, where $j_0 = 1,2,\cdots, 2^{pow(\ell)+1}$. 
				
				We have obtained $\{U^{k}_{j_{k}} \ |  \ j_{k}=1,2,\cdots, 2^{pow(\ell)+1-k} \text{\ and\ } k=1,2,3,\cdots, h-1\}$, where $U^{k}_{j_{k}}$ is an $(\epsilon, (1+2f^{(k-1)}(\epsilon))\omega)$-approximate set with cardinality of $\OO(\frac{1}{\epsilon})$ for $U^{k-1}_{2j_{k}-1} \oplus U^{k-1}_{2j_{k}}$.
				
				Functions in $\{f^{t}(\epsilon) \ | \ t=0,1,\cdots,h-1 \}$ are defined by the following recurrence relation: $f^0(\epsilon) = 0$, $f^{1}(\epsilon) = \epsilon$ and
				$f^{t}(\epsilon) = \epsilon+2(1+\epsilon)f^{t-1}(\epsilon) \text{ \ for \ } 2\le t \le  h-1.$
				\item For each $U^{k}_{j_{k}}$, where $j_{k}=1,2,\cdots, 2^{pow(\ell)+1-k} \text{\ and \ } k=1,2,3\cdots, h-1$, we have built an $\OO(\frac{1}{\epsilon}\log\frac{1}{\epsilon})$-time oracle for backtracking from $U^{k}_{j_{k}}$ to $U^{k-1}_{2j_{k}-1} \dot{\cup} U^{k-1}_{2j_{k}}$. Moreover, a node is created for containing $U^{k}_{j_{k}}$ and the oracle for backtracking from  $U^{k}_{j_{k}}$ to $U^{k-1}_{2j_{k}-1} \dot{\cup} U^{k-1}_{2j_{k}}$.
			\end{itemize}
			
			Now we start iteration-$h$. For each pair $U^{h-1}_{2j_{h}-1}$ and $U^{h-1}_{2j_{h}}$, where $j_{h} = 1,2,\cdots,2^{pow(\ell)+1-h}$, we use Observation~\ref{obs:cap_apx_ss} to compute an $(\epsilon, (1+2f^{h-1}(\epsilon))\omega)$-approximate set with cardinality of $\OO(\frac{1}{\epsilon})$ for $U^{h-1}_{2j_{h}-1} \oplus U^{h-1}_{2j_{h}}$ and build an $\OO(\frac{1}{\epsilon}\log\frac{1}{\epsilon})$-time oracle for backtracking from $U^{h}_{j_h}$ to $U^{h-1}_{2j_{h}-1} \dot{\cup} U^{h-1}_{2j_{h}}$. Denote by $U^{h}_{j_h}$ this approximate set. For two nodes containing $U^{h-1}_{2j_{h}-1}$ and $U^{h-1}_{2j_{h}}$ separately, we create a parent node of these two nodes, and let the parent node contain $U^{h}_{j_h}$ and the oracle for backtracking from $U^{h}_{j_h}$ to $U^{h-1}_{2j_{h}-1} \dot{\cup} U^{h-1}_{2j_{h}}$. Notice that $|U^{h-1}_{2j_{h}-1}| = \OO(\frac{1}{\epsilon})$ and $|U^{h-1}_{2j_{h}}|=\OO(\frac{1}{\epsilon})$ for every $j_h = 1,2,\cdots, 2^{pow(\ell)+1-h}$, then
			the processing time for one pair $U^{h-1}_{2j_{h}-1}$ and $U^{h-1}_{2j_{h}}$ is $\OO(\frac{1}{\epsilon}\log\frac{1}{\epsilon})$ by Observation~\ref{obs:cap_apx_ss}, thus  the total processing time at iteration-$h$ is $\OO(2^{pow(\ell)+1-h} \frac{1}{\epsilon}\log\frac{1}{\epsilon})$.  A reduced instance with size of $2^{pow(\ell)+1-h}$, i.e., $U^{h}_1, U^{h}_2, \cdots, U^{h}_{2^{pow(\ell)+1-h}}$, is obtained. Let $f^{h}(\epsilon) := \epsilon+2(1+\epsilon)f^{h-1}(\epsilon)$.
			
			\item Using the same approach in iteration-$h$ recursively, and iteratively create tree nodes. After $pow(\ell)+1$ such rounds we stop and have built a tree structure whose root node contains (i). $U^{pow(\ell)+1}_1$, which is an $(\epsilon, (1+2f^{pow(\ell)}(\epsilon))\omega)$-approximate set with cardinality of $\OO(\frac{1}{\epsilon})$ for $U^{pow(\ell)}_1\oplus U^{pow(\ell)}_2$; (ii). an $\OO(\frac{1}{\epsilon}\log\frac{1}{\epsilon})$-time oracle for backtracking from  $U^{pow(\ell)+1}_1$ to $U^{pow(\ell)}_1\dot{\cup} U^{pow(\ell)}_2$. 
			
			To summarize, the total processing time is 
			$$\OO(\Sigma^{\ell}_{i=1}|C_i|+(2^1+2^2+\cdots+2^{pow(\ell)}) \frac{1}{\epsilon}\log \frac{1}{\epsilon}) = \OO(\Sigma^{\ell}_{i=1}|C_i|+\frac{\ell }{\epsilon}\log \frac{1}{\epsilon}).$$ 
		\end{itemize}
		
		\noindent Consider the functions in $\{f^{h}(\epsilon)\ | \ h=0,1,2,\cdots,pow(\ell)+1\}$, which are defined by the following recurrence relation:  $f^0(\epsilon) = 0$, $f^{1}(\epsilon) = \epsilon$ and
		$f^{h}(\epsilon) = \epsilon+2(1+\epsilon)f^{h-1}(\epsilon) \text{ \ for \ } 2\le h \le pow(\ell)+1$. According to Observation~\ref{obs:apx_apx_ss}, given integer $1\le k \le pow(\ell)$, if $ U^{k}_{j_{k}}$ is an $(f^{k}(\epsilon),\omega)$-approximate set of $\oplus^{2^k}_{t = 1}C_{2^{k}(j_{k}-1)+t}$, where $j_{k}=1,2,\cdots, 2^{pow(\ell)+1-k}$ and $k=1,2,\cdots, h-1$, then $U^{k+1}_{j_{k+1}}$ is an $(\epsilon+2(1+\epsilon)f^{k}(\epsilon),\omega)$-approximate set of 
		$ \oplus^{2^{k+1}}_{t = 1}C_{2^{h}(j_{k+1}-1)+t}=\left(\oplus^{2^{k}}_{t = 1}C_{2^{k}(2j_{k+1}-2)+t}\right)\oplus \left(\oplus^{2^{k}}_{t = 1}C_{2^{k}(2j_{k+1}-1)+t}\right)$, where $j_{k+1} = 1,2,\cdots, 2^{pow(\ell)-k}$. Recall that  $U^1_{j_1}$ is
		an $(f^{1}(\epsilon),\omega)$-approximate set  of $C_{2j_{1}-1}\oplus C_{2j_{1}}$ for every $j_1 = 1,2,\cdots, 2^{pow(\ell)}$, it can be proved by recursion that $U^h_{j_h}$ is an $(f^{h}(\epsilon),\omega)$-approximate set  of $\oplus^{2^h}_{t = 1}C_{2^{h}(j_{h}-1)+t}$, where $j_{h}=1,2,\cdots, 2^{pow(\ell)+1-h}$ and $h=1,2,\cdots, pow(\ell)+1$. In particular, $U^{pow(\ell)+1}_1$ is an $(f^{pow(\ell)+1}(\epsilon),\omega)$-approximate set of $\oplus^{2^{pow(\ell)+1}}_{i = 1}C_{i}$ and a simple calculation shows that  $f^{pow(\ell)+1}(\epsilon) = \OO(\epsilon \ell+ \epsilon^2 \ell^2)$.

		%To summarize, when the iteration stops, we have built a tree structure whose root node contains an $(f^{pow(\ell)+1}(\epsilon),\omega)$-approximate set with cardinality of $\OO(1/\epsilon)$ for $\oplus^{2^{pow(\ell)+1}}_{\iota = 1}C_{\iota}$. Note that there are $pow(\ell)$ iterations, thus the total processing time is $$\OO(\Sigma^{\ell}_{i=1}(|C_i|\log |C_i|)+(2^1+2^2+\cdots+2^{pow(\ell)}) \frac{1}{\epsilon}\log \frac{1}{\epsilon}) = \OO(\Sigma^{\ell}_{i=1}(|C_i|\log |C_i|)+\frac{\ell }{\epsilon}\log \frac{1}{\epsilon}).$$ 
		
		With the help of this tree structure, an $\OO(\frac{\ell}{\epsilon}\log \frac{1}{\epsilon})$-time oracle for backtracking from $U^{pow(\ell)+1}_1$ to $\dot{\cup}^{\ell}_{i=1}C_i$ is derived and works as follows:
		\begin{itemize}
			\item For any $c \in U^{pow(\ell)+1}_1$, we back trace from the root of the tree. Note that the root node contains an $\OO(\frac{1}{\epsilon} \log \frac{1}{\epsilon})$-time oracle for backtracking from $U^{pow(\ell)+1}_1$ to $U^{pow(\ell)}_1 \dot{\cup}U^{pow(\ell)}_2$. Thus in $\OO(\frac{1}{\epsilon} \log \frac{1}{\epsilon})$ processing time, we can find numbers
			$y^{(pow(\ell);1)} \in U^{pow(\ell)}_1 \cup\{0\}$ and $y^{(pow(\ell);2)} \in U^{pow(\ell)}_2 \cup\{0\}$ such that $$|c-(y^{(pow(\ell);1)}+y^{(pow(\ell);1)})| \le \epsilon(1+2f^{pow(\ell)}(\epsilon)) \omega.$$ 
			
			\item Back track recursively.
			
			Given $1\le h \le pow(\ell)$, assume that we have determined $y^{(h;j_{h})} \in U^{h}_{j_{h}} \cup\{0\}$ for every $ j_{h} = 1,2,3, \cdots, 2^{pow(\ell)+1-h}$  such that 
			$$
			|c-\Sigma^{2^{pow(\ell)+1-h}}_{j_{h}=1}y^{(h;j_{h})}| \le \epsilon \Sigma^{pow(\ell)}_{k=h}2^{pow(\ell)-k} (1+2f^{k}(\epsilon))\omega.
			$$
			For each $y^{(h;j_{h})}\in U^{h}_{j_{h}} \cup \{0\}$, if $y^{(h;j_{h})} = 0$, then we let $y^{(h-1;2j_{h}-1)} = 0$ and $y^{(h-1;2j_{h})} = 0$. If $y^{(h;j_{h})} \in  U^{h}_{j_{h}}$, %then we let $c_i = 0$ for $i = 2^{h}(j_h-1)+1,2^{h}(j_h-1)+2,\cdots, 2^{h}j_h$ and stop backtracking from $U^{h}_{j_{h}}$ to $\dot{\cup}^{2^h}_{\iota=1}C_{2^{h}(j_{h}-1)+\iota}$. Furthermore, we let $y^{(h-1;2j_{h}-1)} = 0$ and $y^{(h-1;2j_{h})} = 0$. If $y^{(h;j_{h})} \in  U^{h}_{j_{h}}$,
			note that the node containing $U^{h}_{j_{h}}$ also contains an $\OO(\frac{1}{\epsilon}\log \frac{1}{\epsilon})$-time oracle for backtracking from $U^{h}_{j_{h}}$ to $U^{h-1}_{2j_h-1} \dot{\cup} U^{h-1}_{2j_h}$, thus in $\OO(\frac{1}{\epsilon}\log \frac{1}{\epsilon})$ processing time, we can determine $y^{(h-1;2j_h -1)} \in U^{h-1}_{2j_h -1}\cup\{0\}$ and $y^{(h-1;2j_h)} \in U^{h-1}_{2j_h} \cup\{0\}$ such that 
			$$|y^{(h;j_h)} -(y^{(h-1;2j_h -1)} + y^{(h-1;2j_h)} )| \le \epsilon(1+2f^{h-1}(\epsilon))\omega.
			$$
			Then in $\OO(2^{pow(\ell)+1-h}\frac{1}{\epsilon}\log \frac{1}{\epsilon})$ processing time, we can determine $y^{(h-1;j_{h-1})} \in U^{h-1}_{j_{h-1}}\cup\{0\}$ for every $j_{h-1} = 1,2,3, \cdots, 2^{pow(\ell)+2-h}$  such that 
			$$|\Sigma^{2^{pow(\ell)}+1-h}_{j_h = 1}y^{(h;j_{h})} - \Sigma^{2^{pow(\ell)}+2-h}_{j_{h-1} = 1}y^{(h-1;j_{h-1})}|\le 2^{pow(\ell)+1-h} \epsilon(1+2f^{h-1}(\epsilon))\omega.$$
			It follows that 
			$$
			|c-\Sigma^{2^{pow(\ell)+2-h}}_{j_{h-1}=1}y^{(h-1;j_{h-1})}| \le \epsilon \Sigma^{pow(\ell)}_{k=h-1} 2^{pow(\ell)-k}(1+2f^{k}(\epsilon))\omega.
			$$
			
			\item After $pow(\ell)$ such rounds, we stop and have determined $c_i \in C_i\cup\{0\}$ for every $ i = 1,2,3, \cdots, {\ell}$  such that
			$$
			|c- \sum^{\ell}_{i=1}c_i| \le \epsilon \Sigma^{pow(\ell)}_{k=0} 2^{pow(\ell)-k} (1+2f^{k}(\epsilon))\omega = \OO(\epsilon \ell + \epsilon^2 \ell^2)\omega.$$ 
			
			To summarize, the total processing time is $\OO((1+2^1 + 2^2 + \cdots+ 2^{pow(\ell)})\frac{1}{\epsilon}\log \frac{1}{\epsilon} ) = \OO(\frac{\ell}{\epsilon} \log \frac{1}{\epsilon} )$ .
		\end{itemize} 
		
		In conclude, within $ \OO(\Sigma^{\ell}_{i=1}|C_i|+\frac{\ell }{\epsilon}\log \frac{1}{\epsilon})$ processing time, we will build a tree structure whose root node contains an $( \OO(\epsilon \ell+ \epsilon^2 \ell^2),\omega )$-approximate set with cardinality of $\OO(\frac{1}{\epsilon})$ for $\oplus^{\ell}_{i = 1}C_{i}$ . Meanwhile, with the help of this tree structure, an $\OO(\frac{\ell}{\epsilon} \log \frac{1}{\epsilon} )$-time oracle for backtracking from this approximate set to $\dot{\cup}^{\ell}_{i=1}C_i$ is derived.			
		\qed\end{proof}
	
	Lemma~\ref{lemma:sub_sum_top} implies the following corollary, which allows us to build the approximate set of $S(\dot{\cup}^{\ell}_{i=1}X_i)$ from the approximate set of each $S(X_i)$.  
	\begin{corollary}\label{coro:sub_sum_top}
		Given multisets $X_1,X_2,\cdots,X_{\ell} \subset  \mathbb{R}_{\ge 0}$. Let $C_i$ be an $(\tilde\OO(\epsilon), \omega)$-approximate set of $S(X_i)$, where $i=1,2,\cdots, \ell$. Assume that for each $C_i$, there is a $T_i$-time oracle for backtracking from $C_i$ to $X_i$. Then in  $\OO(\Sigma^{\ell}_{i=1}|C_i| + \frac{\ell}{\epsilon}\log \frac{1}{\epsilon})$ processing time, we can 
		\begin{itemize}
			\item[(i).] Compute an $(\tilde\OO(\epsilon \ell + \epsilon^3 \ell^3), \omega)$-approximate set with cardinality of $\OO(\frac{1}{\epsilon})$ for $S(\dot{\cup}^{\ell}_{i=1}X_i)$;
			\item[(ii).] Meanwhile build an $\OO(\Sigma^{\ell}_{i=1}T_i + \frac{\ell}{\epsilon}\log\frac{1}{\epsilon})$-time oracle for backtracking from this approximate set to $\dot{\cup}^{\ell}_{i=1}X_i$.
		\end{itemize}
	\end{corollary}
	\begin{proof}
		Given multisets $X_1, X_2,\cdots,X_{\ell} \subset \mathbb{R}_{\ge 0}$. Let $C_1,C_2,\cdots,C_{\ell}$ be $(\tilde\OO(\epsilon), \omega)$-approximate sets of $S(X_1), S(X_2),\cdots,S(X_{\ell})$, respectively. Assume that for each $C_i$, there is a $T_i$-time oracle for backtracking from $C_i$ to $X_i$. The approach to approximate $S(\dot{\cup}^{\ell}_{i=1}X_i)$ and build oracle for backtracking is very similar to the one we have designed in the proof of Lemma~\ref{lemma:sub_sum_top}. The only difference is that when creating leaf nodes, we let the $i$-th node contain not only $C_i$ but also the $T_i$-time oracle for backtracking from $C_i$ to $X_i$. Let $\tilde{\omega} = (1+\tilde\OO(\epsilon\ell))\omega$. Then we use the above iterative approach to approximate $\oplus^{\ell}_{i=1}C_i$, in $\OO(\Sigma^{\ell}_{i=1}|C_i| + \frac{\ell}{\epsilon} \log \frac{1}{\epsilon})$ processing
		time, we will build an augmented tree structure whose root node contains a  $(\OO(\epsilon\ell+\epsilon^2 \ell^2), \tilde{\omega})$-approximate set with cardinality of $\OO(\frac{1}{\epsilon})$ for $\oplus^{\ell}_{i=1} C_i$. Denote by  $U$ this approximate set. Meanwhile, with the help of this tree structure, an $\OO(\frac{\ell}{\epsilon}\log \frac{1}{\epsilon})$-time oracle for backtracking from $U$ to $\dot{\cup}^{\ell}_{i=1}C_i$ is derived.
		
		We claim that $U$ is an $(\tilde\OO(\epsilon \ell + \epsilon^3 \ell^3), \omega)$-approximate set of $S(\dot{\cup}^{\ell}_{i=1}X_i)$. Towards the claim, it is sufficient to observe the followings:
		
		%this augmented tree structure is the $\OO(\Sigma^{\ell}_{i=1}T_i+\frac{\ell}{\epsilon}\log \frac{1}{\epsilon})$-time oracle for backtracking from $C$ to $\dot{\cup}^{\ell}_{i=1}X_i$.

		\begin{itemize}
			\item Since $U$ is an $(\OO(\epsilon\ell+\epsilon^2 \ell^2), \tilde{\omega})$-approximate set of $\oplus^{\ell}_{i=1} C_i$, we have $U \subset [0, (1+\OO(\epsilon\ell+\epsilon^2 \ell^2))\tilde{\omega}] = [0, (1+\tilde\OO(\epsilon \ell + \epsilon^3 \ell^3)))\omega]$.
			\item Consider any $s \in S(\dot{\cup}^{\ell}_{i=1}X_i)\cap[0,\omega]$. There exists $(s_1,s_2,\cdots,s_{\ell})$ where  $s_i\in S(X_i)\cap [0,\omega]$ for $i=1,2,\cdots,\ell$ such that $s = \sum^{\ell}_{i=1} s_i$. For each $S(X_i)$, recall that $C_i$ is an $(\tilde\OO(\epsilon), \omega)$-approximate set of $S(X_i)$, thus for $s_i \in S(X_i) \cap [0,\omega]$, there exists  $s'_i\in C_i$ such that $|s_i-s'_i|\le \tilde\OO(\epsilon) \omega$. It follows that  $|s -\Sigma^{\ell}_{i=1}s'_i|\le \tilde\OO(\epsilon\ell)\omega$. Let $s' = \Sigma^{\ell}_{i=1}s'_i$, we have  $s' \le s+ \tilde\OO(\epsilon\ell)\omega \le  (1+\tilde\OO(\epsilon\ell))\omega =\tilde{\omega}$. Thus $s'\in (\oplus^{{\ell}}_{i=1}C_i)\cap [0,\tilde{\omega}]$. Recall that $U$ is an $(\OO(\epsilon\ell+\epsilon^2 \ell^2), \tilde{\omega})$-approximate set for $\oplus^{\ell}_{i=1} C_i$, thus there exists $s''\in U$ such that $|s''-s'|\le \OO(\epsilon\ell+\epsilon^2 \ell^2)\tilde{\omega}  = \tilde\OO(\epsilon \ell + \epsilon^3 \ell^3)\omega$. Furthermore, we have 
			$|s-s''|\le \tilde\OO(\epsilon \ell + \epsilon^3 \ell^3)\omega.$
			
			\item Consider any $c\in U$. Since $U$ is an $(\OO(\epsilon\ell+\epsilon^2 \ell^2), \tilde{\omega})$-approximate set of $\oplus^{\ell}_{i=1} C_i$, there exists $(c'_1,c'_2,\cdots,c'_{\ell})$ where $c'_i \in C_i \cup\{0\}$ for $i=1,2,\cdots,\ell$ such that $|c-\Sigma^{\ell}_{i=1}c'_i|\le \OO(\epsilon\ell+\epsilon^2 \ell^2) \tilde{\omega} = \tilde\OO(\epsilon \ell+\epsilon^3 \ell^3)\omega$. For each $c'_i \in C_i \cup\{0\}$, if $c'_i = 0$, let $X'_i = \emptyset$, apparently $X'_i\subset X_i$ and $|c'_i -\Sigma(X'_i)| \le \tilde\OO(\epsilon)\omega$. Else if $c'_i \in C_i$, note that $C_i$ is an $(\tilde\OO(\epsilon), \omega)$-approximate set of $S(X_i)$, then there exists $X'_i \subset X_i$ such that $|c'_i -\Sigma(X'_i)| \le \tilde\OO(\epsilon)\omega$. Furthermore, we have 
			\begin{align*}
				|c-\sum^{\ell}_{i=1}\Sigma(X'_i)| &\le|c-\sum^{\ell}_{i=1}c'_i|+|\sum^{\ell}_{i=1}\Sigma(X'_i)-\sum^{\ell}_{i=1}c'_i| \\
				& \le \tilde\OO(\epsilon \ell+\epsilon^3 \ell^3)\omega + \tilde\OO(\epsilon \ell){\omega} =
				\tilde\OO(\epsilon \ell + \epsilon^3 \ell^3)\omega
			\end{align*}
			It is easy to see that $\dot{\cup}^{\ell}_{i=1} X'_i \subset \dot{\cup}^{\ell}_{i=1} X_i$ and $\Sigma(\dot{\cup}^{\ell}_{i=1} X'_i) = \sum^{\ell}_{i=1}\Sigma(X'_i) \in S(\dot{\cup}^{\ell}_{i=1} X_i)$.

			%Recall that $U$ is an $(\OO(\epsilon\ell+\epsilon^2 \ell^2), \tilde{\omega})$-approximate set of $\oplus^{\ell}_{i=1} C_i$, and the augmented tree structure is an $\OO(\frac{\ell}{\epsilon}\log \frac{1}{\epsilon})$-time oracle for backtracking from $C$ to $\oplus^{\ell}_{i=1}C_i$. Thus in $\OO(\ell/\epsilon)$-processing time, we can determine $\{c_i\in C_i \cup\{0\} | i=1,2,\cdots,\ell\} $ such that $|c-\Sigma^{\ell}_{i=1}c_i|\le \OO(\epsilon\ell+\epsilon^2 \ell^2) \tilde{\omega} = \tilde\OO(\epsilon \ell+\epsilon^3 \ell^3)\omega$. For each $c_i$, if $c_i = 0$, let $X'_i = \emptyset$, apparently $X'_i\subset X_i$ and $|c_i -\Sigma(X'_i)| \le \tilde\OO(\epsilon)\omega$. If $c_i \in C_i$, recall that $C_i$ is an $(\tilde\OO(\epsilon), \omega)$-approximate set of $S(X_i)$ and the leaf node containing $C_i$ also contains a $T_i$-time oracle for backtracking from $C_i$ to $X_i$. Thus in $T_i$ time, we can determine $X'_i \subset X_i$ such that $|c_i -\Sigma(X'_i)| \le \tilde\OO(\epsilon)\omega$, it follows that $|\sum^{\ell}_{i=1}\Sigma(X'_i)-\sum^{\ell}_{i=1}c_i| \le \tilde\OO(\epsilon \ell){\omega}$. To summarize, with the help of the augmented tree structure, in total $\OO(\Sigma^{\ell}_{i=1}T_i + \ell/\epsilon)$ processing time, we can determine $\{X'_i \subset X_i | i=1,2,\cdots,\ell \}$ such that
			%$|c-\sum^{\ell}_{i=1}\Sigma(X'_i)| \le |c-\Sigma^{\ell}_{i=1}c_i|+|\sum^{\ell}_{i=1}\Sigma(X'_i)-\sum^{\ell}_{i=1}c_i|  \le\tilde\OO(\epsilon \ell + \epsilon^3 \ell^3)\omega.$
			
		\end{itemize} 
		
		In the following, we show that with the help of the above augmented tree structure, an $\OO(\Sigma^{\ell}_{i=1}T_i + \frac{\ell}{\epsilon}\log \frac{1}{\epsilon})$-time oracle for backtracking from $U$ to $\dot{\cup}^{\ell}_{i=1}X_i$ can be derived. The oracle works as follows.  Given any $c\in U$, note that $U$ is a  $(\OO(\epsilon\ell+\epsilon^2 \ell^2), \tilde{\omega})$-approximate set of $\oplus^{\ell}_{i=1} C_i$ and  with the help of augmented structure, an $\OO(\frac{\ell}{\epsilon}\log \frac{1}{\epsilon})$-time oracle for backtracking from $U$ to $ \dot{\cup}^{\ell}_{i=1} C_i$ has been derived. Thus in $\OO(\frac{\ell}{\epsilon}\log \frac{1}{\epsilon})$-processing time, we can determine $c_i\in C_i \cup\{0\}$ for every $i=1,2,\cdots,\ell$  such that $|c-\Sigma^{\ell}_{i=1}c_i|\le \OO(\epsilon\ell+\epsilon^2 \ell^2) \tilde{\omega} = \tilde\OO(\epsilon \ell+\epsilon^3 \ell^3)\omega$. For each $c_i$, if $c_i = 0$, let $X'_i = \emptyset$, apparently $X'_i\subset X_i$ and $|c_i -\Sigma(X'_i)| \le \tilde\OO(\epsilon)\omega$. Else if $c_i \in C_i$, recall that $C_i$ is an $(\tilde\OO(\epsilon), \omega)$-approximate set of $S(X_i)$ and the leaf node containing $C_i$ also contains a $T_i$-time oracle for backtracking from $C_i$ to $X_i$, thus in $T_i$ time, we can determine $X'_i \subset X_i$ such that $|c_i -\Sigma(X'_i)| \le \tilde\OO(\epsilon)\omega$. % it follows that $|\sum^{\ell}_{i=1}\Sigma(X'_i)-\sum^{\ell}_{i=1}c_i| \le \tilde\OO(\epsilon \ell){\omega}$. 
		To summarize, with the help of augmented tree structure, in total $\OO(\Sigma^{\ell}_{i=1}T_i +\frac{\ell}{\epsilon}\log \frac{1}{\epsilon})$ processing time, we can determine $X'_i \subset X_i$ for every $i=1,2,\cdots,\ell$ such that 
		$$|c-\sum^{\ell}_{i=1}\Sigma(X'_i)| \le |c-\Sigma^{\ell}_{i=1}c_i|+|\sum^{\ell}_{i=1}\Sigma(X'_i)-\sum^{\ell}_{i=1}c_i|  \le\tilde\OO(\epsilon \ell + \epsilon^3 \ell^3)\omega.$$
		Note that $\dot{\cup}^{\ell}_{i=1}X'_i \subset \dot{\cup}^{\ell}_{i=1}X_i$ and $\Sigma(\dot{\cup}^{\ell}_{i=1}X'_i) = \sum^{\ell}_{i=1} \Sigma(X'_i)$, hence Corollary~\ref{coro:sub_sum_top} is proved.
		\qed\end{proof}

	\section{Number-theoretic construction}\label{sec:theoretic_lemma}

	In this section, we introduce one of our main technical contributions, namely, the following number-theoretic rounding lemma.
	\begin{lemma}[Number-Theoretic Rounding Lemma]\label{lemma:smooth_appro}
		Given a multiset $X\subset [\frac{1}{\epsilon^{2+\lambda}},\frac{2}{\epsilon^{2+\lambda}}]$, where   $\lambda \in [-1,+\infty)$ is a parameter.%$\lambda \in \mathbb{R}$
		For any $\alpha \in [ 1+\lambda, 2+\lambda)$ and any $d \in \mathbb{N}_{+}$, 
		let $\bar{d}$ be the integer such that $2+\lambda-\alpha \in (\frac{\bar{d}}{d},\frac{\bar{d} +1}{d}]$. Then there exists a set $\Delta \subset \mathbb{R}$ with $\Delta\subset \Theta (\epsilon^{-\alpha})$
		and  $|\Delta| =  \OO(\log(|X|) \cdot (\log \frac{1}{\epsilon})^{\OO(\bar{d}+1)})$, such that each $x\in X$ can be rounded to the form $\rho h_1 h_2 \cdots h_{\bar{d} + 1}$, where $\rho \in \Delta$ and $h_i$'s satisfying the following conditions: 
		\begin{subequations}
			\begin{align}
				&  h_1 h_2 \cdots h_{\bar{d} + 1} \in \mathbb{N}_{+}\cap [\frac{1}{4\cdot \epsilon^{2+\lambda-\alpha}}, \frac{1}{\epsilon^{2+\lambda-\alpha}}].  \label{eq:cond_0} \\
				& h_{\bar{d}+1} \in \mathbb{N}_{+} \cap [\frac{1}{2\cdot \epsilon^{2+\lambda-\alpha-\frac{\bar{d}}{d}}},\frac{2}{\epsilon^{2+\lambda-\alpha-\frac{\bar{d}}{d}}}]\ \text{and if \ } \bar{d} >0 \text{ \ then \ } 
				h_{i} \in \mathbb{N}_{+} \cap [\frac{1}{2\cdot \epsilon^{\frac{1}{d}}},\frac{2}{\epsilon^{\frac{1}{d}}}] \text{ \ for \ } i = 1,2,\cdots,\bar{d}. \label{eq:cond_1}\\
				& |x-\rho h_1 h_2 \cdots h_{\bar{d}+1}| \le \epsilon^{2+\lambda -\alpha} x. \label{eq:cond_2}
			\end{align}
		\end{subequations}		
		Moreover, there exists a deterministic algorithm, which can return $\Delta$ and round every $x \in X$ to the form $\rho h_1 h_2\cdots h_{\bar{d}+1}$ in $\OO((|X|+\frac{\alpha(\alpha+1)}{\epsilon}) \cdot
		\log(|X|)\cdot (\log\frac{1}{\epsilon} )^{\OO(\bar{d}+1)})$ time.
	\end{lemma}
	%{\color{blue} Is $\delta$ integer? Is $\Delta$ integer set? If so, specify. Does the running time depends on $\lambda$? If not, just write $\lambda\in \mathbb{N}_{+}$. Also give a remark that "In particular, if we take $\alpha=...$, then...". This particular value of alpha is used later.}
	
	\noindent\textbf{Remark 1.}  Since $\epsilon$ is a sufficiently small positive number, we can assume that $\epsilon^{-\frac{1}{d}} \ge 2$ (by setting $\epsilon \le 2^{-{d}}$). It follows that $\mathbb{N}_{+} \cap [\frac{1}{2\cdot \epsilon^{\frac{1}{d}}},\frac{2}{\epsilon^{\frac{1}{d}}}] \neq \emptyset$ %and $h_i \ge 1$ 
	holds for every $i = 1,2,\cdots,\bar{d}$.  Moreover, note that $2+\lambda-\alpha \in (0,1]$, then $2+\lambda-\alpha \in (\frac{\bar{d}}{d},\frac{\bar{d} +1}{d}]$ implies (i).  $0\le \bar{d} \le d-1$; %$\bar{d}$ is within $[0,d)$;(ii).  $2+\lambda-\alpha- \frac{\bar{d}}{d} \in (0, \frac{1}{d}]$ and 
	(ii).$ \frac{1}{\epsilon^{2+\lambda-\alpha-\frac{\bar{d}}{d}}} \ge 1$, hence $\mathbb{N}_{+} \cap [\frac{1}{2\cdot \epsilon^{2+\lambda-\alpha-\frac{\bar{d}}{d}}},\frac{2}{\epsilon^{2+\lambda-\alpha-\frac{\bar{d}}{d}}}] \neq \emptyset$.% and $h_{\bar{d}+1} \ge 1$

	%Moreover, since $2+\lambda-\alpha \in (\frac{\bar{d}}{d},\frac{\bar{d} +1}{d}]$, then $2+\lambda-\alpha- \frac{\bar{d}}{d} \in (0, \frac{1}{d}]$ and $ \frac{1}{\epsilon^{2+\lambda-\alpha-\frac{\bar{d}}{d}}} \ge 1$. It follows that $\mathbb{N} \cap [\frac{1}{2\cdot \epsilon^{2+\lambda-\alpha-\frac{\bar{d}}{d}}},\frac{2}{\epsilon^{2+\lambda-\alpha-\frac{\bar{d}}{d}}}] \neq \emptyset$ and $h_{\bar{d}+1} \ge 1$.

	\noindent\textbf{Remark 2.} The rounding procedure we derived in the proof of Lemma~\ref{lemma:smooth_appro} guarantees that numbers with the same value in multiset $X$ will be rounded to the same number.

	Roughly speaking, Lemma~\ref{lemma:smooth_appro} states that any multiset of given numbers within an interval of $[b,2b]$ for some $b>0$ can be rounded into ``semi-smooth" numbers (relative error within $\OO(\epsilon^{2+\lambda-\alpha})$), which is a multiplication of two components -- a common divisor component (i.e., $\rho$) and a smooth number component (i.e., $h_1h_2\cdots h_{\bar{d}+1}$). Moreover, there are only a few distinct common divisors (i.e., $|\Delta|$ is small). Here $\alpha$ is a parameter that can be adjusted depending on the application. In particular, we will take two specific values for $\alpha$ when we apply Lemma~\ref{lemma:smooth_appro} in the future, namely $1+\lambda$ and $3/2$. 
	
	%{\color{blue} Compare with Jin's result. Jin's can be viewed as a special case when $\lambda=...$? (is it correct? highlight why ours is stronger and requires non-trivial technique.)}
	
	%\textbf{Remark.} %Let $T_{\Delta}$ denote the time of computing $\Delta$. In the following Lemma~\ref{lemma:tree-fashion_} and Lemma~\ref{lemma:sub_sum_top} we will show that the total processing time of our algorithms rely crucially on  $\max\{\frac{|\Delta|}{\epsilon}, T_{\Delta}\}$. Lemma~\ref{lemma:smooth_appro} guarantees that $\max\{\frac{|\Delta|}{\epsilon}, T_{\Delta}\} = \tilde\OO(n+\epsilon^{-1})$ when $\alpha$ is chosen carefully, which is the key to our improved running time. %gives us hope for obtaining algorithms with time bounds of $\tilde\OO(n+\epsilon^{-1})$. 

	It is worth mentioning that a similar number-theoretic result has been obtained by Jin~\cite{DBLP:conf/icalp/Jin19} for solving the knapsack problem. However, Jin's result is not sufficient for SUBSET SUM related problems when our target is a running time below $\OO(\epsilon^{-2})$. We briefly explain the bottleneck. Using Jin's number-theoretic construction, one can also compute a subset $\Delta' $ and round every $x\in X$ to the form $\rho' h'_1 h'_2 \cdots h'_{\bar{d}+1}$, where $\delta'\in \Delta'$ and $h'_i$'s satisfying above conditions~\eqref{eq:cond_1} and \eqref{eq:cond_2}. However, there is a trade-off between the size of $\Delta'$ and the time for constructing it. That is, for any $\gamma\in (0,1)$, if $|\Delta'|=\tilde{\OO}(\epsilon^{\gamma-1})$, then the time to compute $\Delta'$ is $\tilde{\OO}(\epsilon^{-\gamma-1})$. For the purpose of SUBSET SUM, we want $\gamma$ to approach to $1$ (so that $\Delta'$ has logarithmic cardinality), then Jin's method requires $\tilde{\OO}(\epsilon^{-2})$ time, which is too large. %To address the issue, our method strengthens Jin's method  

	The rest of this section is dedicated to proving Lemma~\ref{lemma:smooth_appro}. The proof is divided into two steps: we first prove the existence of $\Delta$, and then derive an algorithm which computes $\Delta$ and rounds every $x \in X$ to the form $\rho h_1 h_2\cdots h_{d}$. 
	%{remove once finished. (\color{red}The first step, existence proof, utilizes more Jin's idea, highlight the different part.)}
	Step 1, the existence proof, utilizes the idea of Chan~\cite{chan2018approximation} and Jin~\cite{DBLP:conf/icalp/Jin19}, which constructs the set $\Delta$ by iteratively determining the existence an integer that can be a common divisor to a subset of numbers that are close to sufficiently many input numbers. We generalize their method in a parameterized way that allows us to control how large the common divisor and the smooth component are, which will also facilitate Step 2, the computation of $\Delta$. %It is worth mentioning that Jin's method only gives a running time of..., while... %such that most remaining numbers can be rounded to the product of this the common divisor and some smooth number. At each iteration, Jin's method decreases the remaining unrounded numbers by $\OO(\frac{\epsilon^{\alpha-\lambda-1}}{(\log{\epsilon^{-1}})^{\OO(d)}})$, which makes the set $\Delta'$ of common divisors obtained when the iteration stops containing $\tilde\OO(\epsilon^{1+\lambda-\alpha})$ elements. We prefer $|\Delta'| = \tilde\OO(1)$, then the time of Jin's method to compute $\Delta'$ will be  $\tilde\OO(\epsilon^{-2})$, which contributes little to obtain better approximation scheme. In the following, we derive an iterative method such that at each iteration, the remaining unrounded numbers decreases by $\OO(\frac{1}{(\log{\epsilon^{-1}})^{\OO(d)}})$, then the set $\Delta$ of common divisors obtained when the iteration stops containing $\tilde\OO(1)$ elements. Moreover, the time to compute $\Delta$ is $\tilde\OO(|X|+\epsilon^{-1})$.
	
	\begin{proof}[Proof of Lemma~\ref{lemma:smooth_appro}]~{}
		\paragraph{Step 1 - Existence of $\Delta$.}
		We proceed with a constructive proof that finds the elements in $\Delta$ in an iterative way. In each iteration, we show that the current numbers can be rounded such that a significant fraction (i.e., a fraction of $1/(\log\frac{1}{\epsilon^{2+\lambda-\alpha}})^{\OO(\bar{d}+1)}$) of them share one large common divisor. Hence, $\OO(\log|X| (\log \frac{1}{\epsilon})^{\OO(\bar{d}+1)})$ iterations suffice, which is the size of $\Delta$.
		
		Let $X = \{x_1,x_2,\cdots, x_{m}\}\subset [\frac{1}{\epsilon^{2+\lambda}},\frac{2}{\epsilon^{2+\lambda}}]$, where $m = |X|$ and  $\lambda \in [-1,+\infty)$ is a parameter. 
		Given $\alpha \in [1+\lambda, 2+\lambda)$, we have $2+\lambda-\alpha \in (0 ,1]$. Then for each $x_i \in X$, there exists $k_i\in \mathbb{N}$ %{\color{blue} $k_i\in \mathbb{Z}$} 
		such that $x_i \in [(1+\epsilon^{2+\lambda-\alpha})^{k_i}, (1+\epsilon^{2+\lambda-\alpha})^{k_i+1})$, we round $x_i$ down to $\bar{x}_i=(1+\epsilon^{2+\lambda-\alpha})^{k_i}$. Let $\overline{X} = \{\bar{x}_1,\bar{x}_2, \cdots, \bar{x}_m \}$. 
		It follows that 
		%$\bar{x}_i \le x_i < (1+\epsilon^{2+\lambda-\alpha})\bar{x} \text{ \ and \ } \frac{1}{\epsilon^{2+\lambda}(1+\epsilon^{2+\lambda-\alpha})} \le \bar{x}_i \le \frac{2}{\epsilon^{2+\lambda}} \text{\ \  for \ \  } i=1,2,\cdots,m.$
		\begin{subequations}
			\begin{align}
				&\bar{x}_i \le x_i < (1+\epsilon^{2+\lambda-\alpha})\bar{x}_i \text{\ \  for \ \  } i=1,2,\cdots,m. \label{eq:x-val-1}\\
				&\frac{1}{\epsilon^{2+\lambda}(1+\epsilon^{2+\lambda-\alpha})} \le \bar{x}_i \le \frac{2}{\epsilon^{2+\lambda}} \text{\ \  for \ \  } i=1,2,\cdots,m. \label{eq:x-val-2}
			\end{align}
		\end{subequations}
		Given $d\in \mathbb{N}_{+}$, let $0\le \bar{d} < d$ be the integer such that $2+\lambda-\alpha \in (\frac{\bar{d}}{d},\frac{\bar{d} +1}{d}]$, then $2+\lambda-\alpha- \frac{\bar{d}}{d} \in (0, \frac{1}{d}]$.
		
		Consider positive integer $H$ that can be written as $H = h_{1}h_{2}\cdots h_{\bar{d} +1}$, where $h_{\bar{d}+1} \in \mathbb{N}_{+}\cap [\frac{1}{2\cdot \epsilon^{2+\lambda-\alpha-\frac{\bar{d}}{d}}},\frac{2}{\epsilon^{2+\lambda-\alpha-\frac{\bar{d}}{d}}}]$ and if $\bar{d} >0$ then $h_{i} \in \mathbb{N}_{+}\cap[\frac{1}{2\cdot \epsilon^{\frac{1}{d}}},\frac{2}{\epsilon^{\frac{1}{d}}}]$ for $i =1,2,\cdots,\bar{d}$. Denote by $\mathscr{H}$ the set of all different integers with this property within $[\frac{1}{4\cdot \epsilon^{2+\lambda-\alpha}},\frac{1}{\epsilon^{2+\lambda-\alpha}}]$. It is easy to see that $|\mathscr{H}| \le \frac{1}{\epsilon^{2+\lambda-\alpha}}$. Next we show that $|\mathscr{H}|$ is relatively big, more precisely, we claim the following.
		
		\begin{claim}\label{claim:number1}
			$|\mathscr{H}| \ge \frac{1}{4\epsilon^{2+\lambda-\alpha}}/ \left(\log\frac{1}{\epsilon^{2+\lambda-\alpha}}\right)^{\OO(\bar{d}+1)} $. 
		\end{claim}
		Before proving Claim~\ref{claim:number1}, we first import the following lemma from \cite{DBLP:conf/icalp/Jin19}.
		\begin{lemma}[CF. Lemma 12 from \cite{DBLP:conf/icalp/Jin19}]\label{lemma:jin_ce}
			Let $T_{1}, T_{2}, \ldots, T_{\ell}$ be positive real numbers satisfying $T_{1} \geq 2$ and $T_{i+1} \geq 2 T_{i}$. There exist at least $T_{\ell} /\left(\log T_{\ell}\right)^{\OO(\ell)}$ integers $t$ satisfying the following condition: $t$ can be written as a product of integers $t=n_{1} n_{2} \cdots n_{\ell}$, such that $n_{1} n_{2} \cdots n_{i} \in\left(T_{i} / 2, T_{i}\right]$ for every $1 \leq i \leq \ell$.
		\end{lemma}
		Now we are ready to prove Claim~\ref{claim:number1}.
		\begin{proof}[Proof of Claim~\ref{claim:number1}]
			If $\bar{d} = 0$, which implies that $\mathscr{H} = \mathbb{N}_{+} \cap [\frac{1}{2\cdot \epsilon^{2+\lambda-\alpha}},\frac{1}{\epsilon^{2+\lambda-\alpha}}]$, then we have $|\mathscr{H}| \ge \frac{1}{4\cdot\epsilon^{2+\lambda-\alpha}}$. %\ge \frac{1}{4\epsilon^{2+\lambda-\alpha}}/ \left(\log\frac{1}{\epsilon^{2+\lambda-\alpha}}\right)^{\OO(\bar{d}+1)}   %Recall that $\frac{1}{\epsilon^{2+\lambda-\alpha-\frac{\bar{d}}{d}}} \ge 1$ and $\bar{d} \ge 1$, we have $\frac{2^{\bar{d}}}{\epsilon^{2+\lambda-\alpha-\frac{\bar{d}}{d}}} \ge 2$, which implies that $T_{\bar{d}+1} \ge 2T_{\bar{d}}$.
			Otherwise $\bar{d} \ge 1$. Recall Lemma~\ref{lemma:jin_ce}, let $T_{\bar{d}+1} = \frac{1}{\epsilon^{2+\lambda-\alpha}}= {\epsilon^{-\frac{\bar{d}}{d}}} \cdot \frac{1}{\epsilon^{2+\lambda-\alpha-\frac{\bar{d}}{d}}}$ and
			$T_i = \epsilon^{-\frac{i}{d}}$ for $1\le i\le \bar{d}$. Since $\epsilon$ is a sufficiently small positive number, we can assume that $\epsilon^{-\frac{1}{d}} \ge 2$, then $T_1 =\epsilon^{-\frac{1}{d}} \ge 2$. Next, we will deal with  $\frac{1}{\epsilon^{2+\lambda-\alpha-\frac{\bar{d}}{d}}} \ge 2$ and $\frac{1}{\epsilon^{2+\lambda-\alpha-\frac{\bar{d}}{d}}} <2$ separately. 
			
			First we consider the case that $\bar{d} \ge 1$ and  $\frac{1}{\epsilon^{2+\lambda-\alpha-\frac{\bar{d}}{d}}} \ge 2$. 
			It is easy to see that $T_{\bar{d}+1} ={\epsilon^{-\frac{\bar{d}}{d}}} \cdot \frac{1}{\epsilon^{2+\lambda-\alpha-\frac{\bar{d}}{d}}}\ge 2T_{\bar{d}}$. Moreover, the following two observations guarantee that $T_1 \ge 2$ and $T_{i+1}\ge 2T_{i}$ holds for every $1\le i \le\bar{d}$. 
			\begin{itemize}
				\item If $\bar{d} = 1$, then apparently $T_1 \ge 2$ and $T_{i+1}\ge 2T_{i}$ holds for every $1\le i \le\bar{d}$.
				\item If $\bar{d} \ge 2$, we have  $T_1 \ge 2$, $T_{\bar{d}+1} \ge 2T_{\bar{d}}$ and $T_{i+1} = \epsilon^{-\frac{i+1}{d}} = \epsilon^{-\frac{1}{d}}\cdot \epsilon^{-\frac{i}{d}} \ge 2 T_{i}$ holds for every $1\le i\le \bar{d}-1$. 
			\end{itemize}
			According to Lemma~\ref{lemma:jin_ce}, there are at least $
			T_{\bar{d}+1}/(\log T_{\bar{d}+1})^{\OO(\bar{d}+1)} $ different integers $t$ satisfying the following condition: $t$ can be written as a product of integers $t=n_{1} n_{2} \cdots n_{\bar{d}+1}$, such that $n_{1} n_{2} \cdots n_{i} \in\left(T_{i} / 2, T_{i}\right]$ for every $1 \leq i \leq \bar{d}+1$. These imply that $(\frac{1}{2\epsilon^{2+\lambda-\alpha}},\frac{1}{\epsilon^{2+\lambda-\alpha}}]$ contains at least $
			\frac{1}{\epsilon^{2+\lambda-\alpha}}/ \left(\log \frac{1}{\epsilon^{2+\lambda-\alpha}} \right)^{\OO(\bar{d}+1)}$ different integers $t$ satisfying the following condition: $t$ can be written as a product of integers $t=n_{1} n_{2} \cdots n_{\bar{d}+1}$, where $ n_{\bar{d}+1} \in  \mathbb{N}_{+}\cap[\frac{1}{2\cdot \epsilon^{2+\lambda-\alpha-\frac{\bar{d}}{d}}},\frac{2}{\epsilon^{2+\lambda-\alpha-\frac{\bar{d}}{d}}}]$ and $n_i \in \mathbb{N}_{+}\cap[\frac{1}{2\cdot \epsilon^{\frac{1}{d}}},\frac{2}{\epsilon^{\frac{1}{d}}}] \ \text{for} \ 1\le i \le \bar{d}$.
			By the definition of $\mathscr{H}$, $\mathscr{H}$ contains at least $
			\frac{1}{\epsilon^{2+\lambda-\alpha}}/ \left(\log \frac{1}{\epsilon^{2+\lambda-\alpha}} \right)^{\OO(\bar{d}+1)}$ elements. %\ge \frac{1}{2\epsilon^{2+\lambda-\alpha}}/ \left(\log\frac{1}{\epsilon^{2+\lambda-\alpha}}\right)^{\OO(\bar{d}+1)}
			%where $n_i$'s satisfying 
			%$$ n_{\bar{d}+1} \in  \mathbb{N}\cap[\frac{1}{2\cdot \epsilon^{2+\lambda-\alpha-\frac{\bar{d}}{d}}},\frac{2}{\epsilon^{2+\lambda-\alpha-\frac{\bar{d}}{d}}}] \  \text{and} \  n_i \in \mathbb{N}\cap[\frac{1}{2\cdot \epsilon^{\frac{1}{d}}},\frac{2}{\epsilon^{\frac{1}{d}}}] \ \text{for} \ 1\le i \le \bar{d}.$$
			\iffalse{\begin{itemize}
					\item $n_{\bar{d}+1} \in  \mathbb{N}\cap[\frac{1}{2\cdot \epsilon^{2+\lambda-\alpha-\frac{\bar{d}}{d}}},\frac{2}{\epsilon^{2+\lambda-\alpha-\frac{\bar{d}}{d}}}]$ and  $n_i \in \mathbb{N}\cap[\frac{1}{2\cdot \epsilon^{\frac{1}{d}}},\frac{2}{\epsilon^{\frac{1}{d}}}]$ for $1\le i \le \bar{d}$;
					\item $n_i \in \mathbb{N}\cap[\frac{1}{2\cdot \epsilon^{\frac{1}{d}}},\frac{2}{\epsilon^{\frac{1}{d}}}]$ for $1\le i \le \bar{d}$.
			\end{itemize}}\fi

			Now we consider the case that $\bar{d} \ge 1$ and $\frac{1}{\epsilon^{2+\lambda-\alpha-\frac{\bar{d}}{d}}} < 2$ (i.e., $\frac{1}{4\epsilon^{2+\lambda-\alpha}}\le \frac{\epsilon^{-\frac{\bar{d}}{d}}}{2}$).  Recall that $ \frac{1}{ \epsilon^{2+\lambda-\alpha-\frac{\bar{d}}{d}}} \ge 1$,  %by the fact that  $\frac{1}{\epsilon^{2+\lambda-\alpha-\frac{\bar{d}}{d}}} < 2$, 
			we have $ [ \frac{\epsilon^{-\frac{\bar{d}}{d}}}{2}, {\epsilon^{-\frac{\bar{d}}{d}}}] \subset [\frac{1}{4\cdot \epsilon^{2+\lambda-\alpha}},\frac{1}{\epsilon^{2+\lambda-\alpha}}]$ and $1 \in  \mathbb{N}_{+}\cap [\frac{1}{2\cdot \epsilon^{2+\lambda-\alpha-\frac{\bar{d}}{d}}},\frac{2}{\epsilon^{2+\lambda-\alpha-\frac{\bar{d}}{d}}}]$. Thus to prove that $\mathscr{H}$ contains at least $
			\frac{1}{4\epsilon^{2+\lambda-\alpha}}/ \left(\log \frac{1}{\epsilon^{2+\lambda-\alpha}} \right)^{\OO(\bar{d}+1)}$ elements, we only need to prove that 
			$[ \frac{\epsilon^{-\frac{\bar{d}}{d}}}{2}, {\epsilon^{-\frac{\bar{d}}{d}}}] $ contains at least  $\frac{\epsilon^{-\frac{\bar{d}}{d}}}{2}/ \left(\log \epsilon^{-\frac{\bar{d}}{d}} \right)^{\OO(\bar{d}+1)}$ different integers $t$ satisfying the following condition: $t$ can be written as a product of integers $t=n_{1} \cdots n_{\bar{d}} n_{\bar{d}+1}$, where $n_{\bar{d}+1}=1$ and $n_i \in \mathbb{N}_{+}\cap[\frac{1}{2\cdot \epsilon^{\frac{1}{d}}},\frac{2}{\epsilon^{\frac{1}{d}}}]$ for $1\le i \le \bar{d}$.  Towards this, it is sufficient to observe the followings:\begin{itemize}
				\item If $\bar{d} = 1$, then $\frac{\epsilon^{-\frac{\bar{d}}{d}}}{2} = \frac{\epsilon^{-\frac{1}{d}}}{2}$. Note that $[ \frac{\epsilon^{-\frac{\bar{d}}{d}}}{2}, {\epsilon^{-\frac{\bar{d}}{d}}}]$ contains at least   $\frac{\epsilon^{-\frac{\bar{d}}{d}}}{2}/ \left(\log \epsilon^{-\frac{\bar{d}}{d}} \right)^{\OO(\bar{d}+1)}$ different integers, and any integer $t$ in $ [\frac{\epsilon^{-\frac{\bar{d}}{d}}}{2},\epsilon^{-\frac{\bar{d}}{d}}]$, i.e., in $[\frac{\epsilon^{-\frac{{1}}{d}}}{2},\epsilon^{-\frac{{1}}{d}}]$, satisfies the condition: $t$ can be written as a product of integers $t=n_{1} n_{2}\cdots  n_{\bar{d}+1}=n_{1}n_{2}$, where $n_{\bar{d}+1}=n_{2}=1$ and $n_i =t \in \mathbb{N}_{+}\cap[\frac{1}{2\cdot \epsilon^{\frac{1}{d}}},\frac{2}{\epsilon^{\frac{1}{d}}}]$ for $1\le i \le \bar{d}$.
				\item If $\bar{d} \ge 2$, recall that $T_1 \ge 2$ and $T_{i+1} = \epsilon^{-\frac{i+1}{d}} = \epsilon^{-\frac{1}{d}}\cdot \epsilon^{-\frac{i}{d}} \ge 2 T_{i}$ holds for every $1\le i\le \bar{d}-1$.  According to Lemma~\ref{lemma:jin_ce}, there are at least $
				T_{\bar{d}}/(\log T_{\bar{d}})^{\OO(\bar{d})} =
				\epsilon^{-\frac{\bar{d}}{d}}/ \left(\log\epsilon^{-\frac{\bar{d}}{d}} \right)^{\OO(\bar{d})}$ different integers $t$ satisfying the following condition: $t$ can be written as a product of integers $t=n_{1} n_{2} \cdots n_{\bar{d}}$, such that $n_{1} n_{2} \cdots n_{i} \in\left(T_{i} / 2, T_{i}\right]$ for every $1 \leq i \leq \bar{d}$. These imply that $ [ \frac{\epsilon^{-\frac{\bar{d}}{d}}}{2}, {\epsilon^{-\frac{\bar{d}}{d}}}]$ contains at least $
				\epsilon^{-\frac{\bar{d}}{d}}/ \left(\log \epsilon^{-\frac{\bar{d}}{d}} \right)^{\OO(\bar{d})}\ge \frac{\epsilon^{-\frac{\bar{d}}{d}}}{2}/ \left(\log \epsilon^{-\frac{\bar{d}}{d}} \right)^{\OO(\bar{d}+1)}$ different integers $t$ satisfying the following condition: $t$ can be written as a product of integers $t=n_{1} n_{2}\cdots n_{\bar{d}}n_{\bar{d}+1}$, where $n_{\bar{d}+1} =1$ and $n_i \in \mathbb{N}_{+}\cap[\frac{1}{2\cdot \epsilon^{\frac{1}{d}}},\frac{2}{\epsilon^{\frac{1}{d}}}]$ for $1\le i \le \bar{d}$. 
			\end{itemize}
			To summarize, $\mathscr{H}$ contains at least $
			\frac{1}{4\epsilon^{2+\lambda-\alpha}}/ \left(\log \frac{1}{\epsilon^{2+\lambda-\alpha}} \right)^{\OO(\bar{d}+1)}$ elements. \qed\end{proof}
		
		Let $\mathscr{H} = \{H_1,H_2,\cdots,H_{\tau}\}$, where $\frac{1}{4\epsilon^{2+\lambda-\alpha}} /(\log \frac{1}{\epsilon^{2+\lambda-\alpha}})^{\OO(\bar{d}+1)} \le \tau \le \frac{1}{\epsilon^{2+\lambda-\alpha}}$. For each $H_v \in \mathscr{H}$, there exists $s_v\in \mathbb{N}$ such that $H_v \in [(1+\epsilon^{2+\lambda-\alpha})^{s_v}, (1+\epsilon^{2+\lambda-\alpha})^{s_v+1})$, we round $H_v$ down to $\overline{H}_v=(1+\epsilon^{2+\lambda-\alpha})^{s_v}$. Let $\overline{\mathscr{H}}= \{\overline{H}_1,\overline{H}_2,\cdots,\overline{H}_{\tau}\}$. It follows that 
		\begin{subequations}
			\begin{align}
				&\overline{H}_{v}\le H_{v} < (1+\epsilon^{2+\lambda-\alpha})\overline{H}_v \text{\ \ for \ \  } v= 1,2,\cdots,\tau. \label{eq:h-val-1}\\
				&\frac{1}{(1+\epsilon^{2+\lambda-\alpha})\cdot 4\cdot \epsilon^{2+\lambda-\alpha}} \le \overline{H}_{v} \le  \frac{1}{\epsilon^{2+\lambda-\alpha}} \text{\ \  for \ \  } v = 1,2,\cdots,\tau. \label{eq:h-val-2}
			\end{align}
		\end{subequations}
		%$\frac{1}{(1+\epsilon^{2+\lambda-\alpha})4\cdot \epsilon^{2+\lambda-\alpha}} \overline{H}_{i} \le  \frac{1}{\epsilon^{2+\lambda-\alpha}} \text{ \ and  \ } \overline{H}_{i}\le H_{i} < (1+\epsilon^{2+\lambda-\alpha})\overline{H}_i \text{\ \ for \ \  } i = 1,2,\cdots,\tau.$
		Notice that elements in $\overline{\mathscr{H}}$ are different from each other. This is because that if there exist two different integers $H_{v_1},H_{v_2} \in \mathscr{H}$ satisfying $H_{v_1}, H_{v_2} \in [(1+\epsilon^{2+\lambda-\alpha})^{s},(1+\epsilon^{2+\lambda-\alpha})^{s+1})$, then we have $1 \le |H_{v_1}-H_{v_2}| < \epsilon^{2+\lambda-\alpha}\cdot (1+\epsilon^{2+\lambda-\alpha})^{s} \le 
		\epsilon^{2+\lambda-\alpha}\cdot  \frac{1}{\epsilon^{2+\lambda-\alpha}} = 1,$ which is impossible. %{\color{blue} We maintain a one-to-one correspondence $\phi$ between $\mathscr{H}$ and $\overline{\mathscr{H}}$: $\phi(H_i) = \overline{H}_i \ (\forall i)$.}

		Consider the following $\tau \times m $ table, denote it by $Table^{(1)}$.
		
		\begin{table}[H]
			\renewcommand{\arraystretch}{1.1}
			\centering
			\begin{tabular}{|c|cccc|}
				\hline
				\diagbox{$\overline{\mathscr{H}}$}{$\overline{X}$} 
				&   $\bar{x}_1 = (1+\epsilon^{2+\lambda-\alpha})^{k_1} $  & $\bar{x}_2 = (1+\epsilon^{2+\lambda-\alpha})^{k_2}$
				& $\cdots$  &  $\bar{x}_{m} = (1+\epsilon^{2+\lambda-\alpha})^{k_{m}}$ \\
				\hline
				$\overline{{H}}_1 = (1+\epsilon^{2+\lambda-\alpha})^{s_{1}}$ &${k_{1}-s_{1}}$ &${k_{2}-s_{1}}$   &$\cdots$   &${k_{m}-s_{1}}$ \\
				$\overline{{H}}_2 = (1+\epsilon^{2+\lambda-\alpha})^{s_{2}}$ &${k_{1}-s_{2}}$ &${k_{2}-s_{2}}$   &$\cdots$   &${k_{m}-s_{2}}$ \\
				$\vdots$ &$\vdots$   & $\vdots$ & $\vdots$  &$\vdots$ \\
				$\overline{{H}}_{\tau}  =(1+\epsilon^{2+\lambda-\alpha})^{s_{\tau}}$ &${k_{1}-s_{\tau}}$ &${k_{2}-s_{\tau}}$  &$\cdots$   &${k_{m}-s_{\tau}}$ \\
				\hline
			\end{tabular}
			%\caption{ }
			%\label{tab1}
		\end{table}
		\noindent Since elements in $\overline{\mathscr{H}}$ are different from each other, the elements of the same column in $Table^{(1)}$ are different from each other, i.e., integers in $\{k_i -s_v \ | \ v = 1,2,\cdots, \tau  \}$ are different from each other for every $i=1,2,\cdots,m$. Let $\mathscr{C} = set \{k_i -s_v\ | \ i = 1,2,\cdots,m \text{ \ and \ } v = 1,2,\cdots, \tau \}$, which is the set of all distinct entries in $Table^{(1)}$. Recall \eqref{eq:x-val-2} and \eqref{eq:h-val-2}, for every $i=1,2,\cdots,m$ and every $v= 1,2,\cdots,\tau$, we have  $(1+\epsilon^{2+\lambda-\alpha})^{k_i - s_v} = \frac{\bar{x}_i }{\overline{H}_v} \in [\frac{1}{(1+\epsilon^{2+\lambda-\alpha})\epsilon^{\alpha}}, \frac{8\cdot (1+\epsilon^{2+\lambda-\alpha})}{\epsilon^{\alpha}}]$, thus $|\mathscr{C}| = \OO(\frac{1}{\epsilon^{2+\lambda-\alpha}})$.
		%thus $\mathscr{C}$ contains at most $\OO(\frac{1}{\epsilon^{2+\lambda-\alpha}})$ elements.
		Let $ \mathscr{C} = \{c_1, c_2, \cdots,c_{\mu}\}$, where $\mu = \OO(\frac{1}{\epsilon^{2+\lambda-\alpha}}).$   
		
		We first find in $\mathscr{C}$ the element with the most occurrences in $Table^{(1)}$ (we break tie arbitrarily), let's say $c_1$. Then we find the columns containing $c_1$ in $Table^{(1)}$ and build $Table^{(2)}$ by dropping these columns from $Table^{(1)}$.  Let $n^{(1)}_{i}$ be the number of columns containing $c_{i}$ in $Table^{(1)}$ for every $i = 1,2,\cdots,\mu$. Note that elements of the same column in $Table^{(1)}$ are different from each other, thus we have $
		\sum_{i=1}^{\mu} n^{(1)}_{i} = \tau  m.$ Since $c_1$ is the element with the most occurrences in $Table^{(1)}$, it follows that $n^{(1)}_{1} \ge \frac{\tau  m}{\mu}$ and the number of columns in $Table^{(2)}$ is 
		$$
		m - n^{(1)}_{1} \le (1-\frac{\tau }{\mu}) m.
		$$

		Repeat the above operations iteratively. Assume that we have constructed $Table^{(1)}, \cdots, Table^{(j)}$. Besides, assume that we have found $c_1,\cdots,c_{j}$, which are the elements with the most occurrences in $Table^{(1)},\cdots, Table^{(j)}$, respectively. % found the most frequent elements in $Table^{(1)},\cdots, Table^{(j)}$, say $c_1,\cdots,c_{j}$, respectively. 
		Let $n^{(t)}_{i}$ denote the number of columns containing $c_{i}$ in $Table^{(t)}$, where $i = t,t+1,\cdots, \mu $ and $t = 1,2,\cdots, j$. It can be inductively proved that the number of columns in $Table^{(t)}$ is 
		$$m-\sum^{t-1}_{i=1}n^{(i)}_{i} \le (1-\frac{\tau}{\mu})^{t-1}  m, \text{\ where \ }t = 1,2,\cdots,j.$$
		
		We stop building $Table^{(j+1)}$ as soon as $m-\sum^{j}_{i=1}n^{(i)}_i = 0$. Then $(1-\frac{\tau}{\mu})^{j-1}  m \ge m-\sum^{j-1}_{i=1}n^{(i)}_i \ge 1$. Recall that $\mu  =\OO(\frac{1}{\epsilon^{2+\lambda-\alpha}})$ and $\frac{1}{4\epsilon^{2+\lambda-\alpha}} / (\log \frac{1}{\epsilon^{2+\lambda-\alpha}})^{\OO(\bar{d}+1)}\le\tau \le \frac{1}{\epsilon^{2+\lambda-\alpha}}$, a careful calculation shows that $j = \OO(\log(m) \cdot (\log \frac{1}{\epsilon})^{\OO(\bar{d}+1)}) = \OO(\log(|X|) \cdot (\log \frac{1}{\epsilon})^{\OO(\bar{d}+1)})$. %Thus $j \le \OO(\log m (\log \frac{1}{\epsilon})^{\OO(\bar{d})})$ when $m-\sum^{j}_{i=1}n^{(i)}_i = 0$. 
		In the meantime, we have obtained $c_1,\cdots,c_{j}$, which are the elements with the most occurrences in $Table^{(1)},\cdots, Table^{(j)}$, respectively. We claim that
		$\{(1+\epsilon^{2+\lambda-\alpha})^{c_i} \ |\ i=1,2,\cdots,j\}$ is a required $\Delta$ in the Lemma~\ref{lemma:smooth_appro}. The claim is guaranteed by the followings:
		\begin{itemize}
			\item  $(1+\epsilon^{2+\lambda-\alpha})^{c_i}  \in  [\frac{1}{(1+\epsilon^{2+\lambda-\alpha})\epsilon^{\alpha}}, \frac{8\cdot (1+\epsilon^{2+\lambda-\alpha})}{\epsilon^{\alpha}}],$ i.e.,  $(1+\epsilon^{2+\lambda-\alpha})^{c_i} = \Theta(\epsilon^{-\alpha})$, holds for every $1\le i\le j$.
			%we have $\{(1+\epsilon^{2+\lambda-\alpha})^{c_i} \ |\ i=1,2,\cdots,j\}\subset \Theta(\epsilon^{-\alpha})$.
			\item Consider any $x_i\in X$, whose rounded value is $\bar{x}_i = (1+\epsilon^{2+\lambda-\alpha})^{k_i}$. According to \eqref{eq:x-val-1}, we have 
			$$(1+\epsilon^{2+\lambda-\alpha})^{k_i} \le x_{i} < (1+\epsilon^{2+\lambda-\alpha})^{k_i+1}.$$
			Since $m-\sum^{j}_{i=1}n^{(i)}_i = 0$, the column corresponding to $\bar{x}_i$ %=(1+\epsilon^{2+\lambda-\alpha})^{k_i}
			in $Table^{(1)}$ must be dropped from some $Table^{(t)}\ (t \le j)$. Recall that $c_t$ is the element with the most occurrences in $Table^{(t)}$ and all columns dropped from $Table^{(t)}$ contain $c_t$, thus $c_t $ is in the column corresponding to $\bar{x}_i$. Hence  $ (1+\epsilon^{2+\lambda-\alpha})^{k_i-c_{t}} =\frac{\bar{x}_i}{(1+\epsilon^{2+\lambda-\alpha})^{c_{t}}}\in \overline{\mathscr{H}}$. Note that $(1+\epsilon^{2+\lambda-\alpha})^{k_i-c_{t}}$ is a rounded value of some $h_1 h_2 \cdots h_{\bar{d}+1} \in \mathscr{H}$, it thus follows that $h_1 h_2 \cdots h_{\bar{d}+1} \in [(1+\epsilon^{2+\lambda-\alpha})^{k_i-c_{t}}, (1+\epsilon^{2+\lambda-\alpha})^{k_i-c_{t}+1})$, then we have $ \frac{x_i}{1+\epsilon^{2+\lambda-\alpha}}\le (1+\epsilon^{2+\lambda-\alpha})^{k_i} \le  (1+\epsilon^{2+\lambda-\alpha})^{c_{t}}\cdot h_1 h_2\cdots h_{\bar{d}+1} \le (1+\epsilon^{2+\lambda-\alpha})^{k_i+1}\le (1+\epsilon^{2+\lambda-\alpha}) x_{i}.$
			Thus $x_i$ can be rounded to $(1+\epsilon^{2+\lambda-\alpha})^{c_t}\cdot h_1 h_2,\cdots h_{\bar{d}+1}$ and it holds that $ |x_i - (1+\epsilon^{2+\lambda-\alpha})^{c_t}\cdot h_1 h_2,\cdots h_{\bar{d}+1}|\le \epsilon^{2+\lambda-\alpha}x_{i}.$ 
		\end{itemize}

		\paragraph{Step 2 - $\tilde\OO(|X|+\frac{\alpha(\alpha+1)}{\epsilon})$-time algorithm to compute $\Delta$ and round every $x \in X$ to the form $\rho h_1 h_2\cdots h_{\bar{d}+1}$.} According to the discussion above, we can obtain $\Delta$ and round every $x \in X$ to the form $\rho h_1 h_2\cdots h_{\bar{d}+1}$ by iteratively solving the following two problems for every  $j=\OO(\log(|X|) \cdot (\log \frac{1}{\epsilon})^{\OO(\bar{d}+1)})$:
		\begin{itemize}
			\item [$P_1^{(j)}:$] Find the element with the most occurrences in $Table^{(j)}$, say, $c_j$. 
			\item [$P_2^{(j)}:$] Determine the columns containing $c_{j}$ in $Table^{(j)}$, and construct $Table^{(j+1)}$ if the iteration continues.  %.
		\end{itemize}
		A brute-force method that enumerates all the entries in each $Table^{(j)}$ can work, but it is too expensive. In the following, we aim to derive a more efficient algorithm. 
		
		We need to build $Table^{(1)}$ before iteration. Towards this, we first compute $\overline{X}$ and sort the elements in $\overline{X}$ in ascending order. Without loss of generality, we assume that $\bar{x}_{1} \le \bar{x}_{2} \le \cdots \le \bar{x}_{m}$, i.e., $k_1 \le k_2 \le \cdots \le k_{m}$. %In the meantime, we maintain a one-to-one correspondence $\phi_{X}$ between $\overline{X}$ and $X$: $\phi_{X}(\overline{x}_i) = x_i \ (\forall i)$. 
		Then we compute $\mathscr{H}$ and $\overline{\mathscr{H}}$, and sort the elements in $\overline{\mathscr{H}}$ in descending order. Without loss of generality, we assume that $\overline{H}_{1} \ge \overline{H}_{2} \ge \cdots \ge \overline{H}_{\tau}$, i.e., $s_1 \ge s_2 \ge \cdots \ge s_{\tau}$. %{\color{blue}In the meantime, we maintain a one-to-one correspondence $\phi$ between $\overline{\mathscr{H}}$ and $\mathscr{H}$: $\phi(\overline{H}_i) = H_i \ (\forall i)$.} 
		Finally, $Table^{(1)}$ will be built in an implicit way without specifying each entry, in particular, only $\overline{X}$ and $\overline{\mathscr{H}}$ are stored (from $\overline{X}$ and $\overline{\mathscr{H}}$ it is sufficient to recover the whole table as each entry can be uniquely determined). The total time of building $Table^{(1)}$ is $\OO(|X|\log |X|+\frac{1}{\epsilon}\log \frac{1}{\epsilon})$.

		Now we consider an arbitrary iteration $j$, and solve $P_1^{(j)}$ and $P_2^{(j)}$. At the beginning of iteration $j$, we have obtained $Table^{(j)}$, whose columns are exactly columns $l^j_1,l^j_2,\cdots,l^j_{m_j}$ of $Table^{(1)}$, where $l^j_1 \le l^j_2 \le \cdots \le l^j_{m_j}$. Let $\overline{X}^{j} = \{\bar{x}_{l^j_1},\bar{x}_{l^j_2},\cdots, \bar{x}_{l^j_{m_j}}\}$, then $Table^{(j)}$ is as follows. 
		\begin{center}
			\begin{tabular}{|c|cccc|}
				\hline
				\diagbox{$\overline{\mathscr{H}}$}{$\overline{X}^{j}$} & $\bar{x}_{l^j_1} = (1+\epsilon^{2+\lambda-\alpha})^{k_{l^j_1}}$ & $\bar{x}_{l^j_2}= (1+\epsilon^{2+\lambda-\alpha})^{k_{l^j_{2}}}$  & $\cdots$  & $\bar{x}_{l^j_{m_{j}}}=(1+\epsilon^{2+\lambda-\alpha})^{k_{l^j_{m_{j}}}}$ \\
				\hline
				$\overline{{H}}_1 = (1+\epsilon^{2+\lambda-\alpha})^{s_{1}}$ &${k_{l^j_1}-s_{1}}$ &${k_{l^j_{2}}-s_{1}}$   &$\cdots$   &${k_{l^j_{m_{j}}}-s_{1}}$ \\
				$\overline{{H}}_2 = (1+\epsilon^{2+\lambda-\alpha})^{s_{2}}$ &${k_{l^j_1}-s_{2}}$ &${k_{l^j_{2}}-s_{2}}$   &$\cdots$   &${k_{l^j_{m_{j}}}-s_{2}}$ \\
				$\vdots$ &$\vdots$   & $\vdots$ & $\vdots$  &$\vdots$ \\
				$\overline{{H}}_{\tau}  =(1+\epsilon^{2+\lambda-\alpha})^{s_{\tau}}$ &${k_{l^j_1}-s_{\tau}}$ &${k_{l^j_{2}}-s_{\tau}}$  &$\cdots$   &${k_{l^j_{m_{j}}}-s_{\tau}}$ \\
				\hline
			\end{tabular}
		\end{center} 
		
		Towards $P_1^{(j)}$, we define the following two polynomials:
		$$f^{j}(x) = x^{k_{l^j_1}-s_{1}}+x^{k_{l^j_1}-s_{2}}+x^{k_{l^j_1}-s_{3}}+ \cdots +x^{k_{l^j_1}-s_{\tau}}$$   
		$$
		g^{j}(x) =1+ x^{k_{l^j_{2}}-k_{l^j_{1}}}+x^{k_{l^j_{3}}-k_{l^j_{1}}}+x^{k_{l^j_{4}}-k_{l^j_{1}}}+\cdots +x^{k_{l^j_{m_{j}}}-k_{l^j_{1}}}$$
		Consider the product $f^{j}(x)\cdot g^{j}(x)$ and the coefficient of an arbitrary term $x^c$ in the product. By the definition of polynomial multiplication, the coefficient of $x^c$ counts all the pairs $(a,b)$ such that the $a$-th term of $f^j$ and the $b$-th term of $g^j$ multiply to $x^c$. Equivalently, the coefficient of $x^c$ counts all the pairs $(a,b)$ such that the exponent of $a$-th term of $f^j$, which is $k_{l^j_1}-s_a$, and the exponent of $b$-th term of $g^j$, which is $k_{l_{b}^j}-k_{l^j_1}$, add up to $c$. Observe that the two exponents, $k_{l^j_1}-s_a$ and $k_{l_{b}^j}-k_{l^j_1}$, add up to exactly $k_{l_{b}^j}-s_a$, which is the element in the $b$-th column of Table$^{(j)}$. Hence, the coefficient of $x^c$ counts the number of the occurrences of element $c$ in $Table^{(j)}$.  
		
		Recall that $\mathscr{C} = set \{k_i -s_v \ | \ i= 1,2,\cdots,m \text{ \ and \ } v= 1,2,\cdots, \tau \}$ denotes the set of all distinct entries in $Table^{(1)}$, we claim that $ 0 \le \mathscr{C}^{\min} \le \mathscr{C}^{\max}=\OO(\frac{\alpha+1}{\epsilon^{2+\lambda-\alpha}} \log \frac{1}{\epsilon}) $.  The claim is guaranteed by the followings:
		\begin{itemize}
			\item By the facts that $\alpha \in [1+\lambda,2+\lambda)$ and $\lambda \in [-1,+\infty)$, we have  $\alpha \ge 0$. Recall that  $X \subset [\frac{1}{\epsilon^{2+\lambda}},\frac{2}{\epsilon^{2+\lambda}}]$ and $\mathscr{H} \subset [\frac{1}{4\epsilon^{2+\lambda-\alpha}},\frac{1}{\epsilon^{2+\lambda-\alpha}})$, we have ${\overline{\mathscr{H}}}^{\max} \le \overline{X}^{\min}$.
			\item Notice that $\min_{c\in \mathscr{C} }(1+\epsilon^{2+\lambda-\alpha})^{c} = {  \overline{X}^{\min}}/{ {\overline{\mathscr{H}}}^{\max} }  \ge 1$,  we have $\mathscr{C}^{\min} \ge 0 $.
			\item Recall \eqref{eq:x-val-2} and \eqref{eq:h-val-2},  we have $\max_{c\in \mathscr{C} }(1+\epsilon^{2+\lambda-\alpha})^{c} = {  \overline{X}^{\max}}/{ {\overline{\mathscr{H}}}^{\min} }    \in  [\frac{1}{(1+\epsilon^{2+\lambda-\alpha})\epsilon^{\alpha}}, \frac{8\cdot (1+\epsilon^{2+\lambda-\alpha})}{\epsilon^{\alpha}}]$, which imlpies that $ \mathscr{C}^{\max}=\OO( \frac{\alpha+1}{\epsilon^{2+\lambda-\alpha}} \log \frac{1}{\epsilon}) $. 
		\end{itemize}
		Since ${\overline{\mathscr{H}}}^{\max} \le \overline{X}^{\min}$, we have $k_{l^j_{1}} \ge s_{\tau}$, hence $k_{l^j_{m_j}}-k_{l^j_1}\le k_{l^j_{m_j}} -s_{\tau}$. Notice that $\mathscr{C}^{\min} \le k_{l^j_1} - s_{1} \le k_{l^j_1} - s_{2} \le \cdots \le  k_{l^j_1} - s_{\tau}\le \mathscr{C}^{\max}  $ and $0\le k_{l^j_2}-k_{l^j_1}  \le  k_{l^j_3}-k_{l^j_1}  \le \cdots \le k_{l^j_{m_j}}-k_{l^j_1}\le k_{l^j_{m_j}} -s_{\tau} \le \mathscr{C}^{\max}$.  %引用快速傅立叶的结论， According to \cite{}, 
		Then in %$\OO((k_{l^j_1} - s_{\tau}+  k_{l^j_{m_j}}-k_{l^j_1} )\log (k_{l^j_1} - s_{\tau}+  k_{l^j_{m_j}}-k_{l^j_1} ) ) =  \OO((k_{l^j_{m_j}} - s_{\tau} )\log (k_{l^j_{m_j}} - s_{\tau} ) )= \OO(\mathscr{C}^{\max} \log(\mathscr{C}^{\max} ))=\OO(\frac{\alpha(\alpha+1)}{\epsilon^{2+\lambda-\alpha}}(\log \frac{1}{\epsilon})^2)$ 
		$ \OO(\mathscr{C}^{\max} \log(\mathscr{C}^{\max} ))=\OO(\frac{\alpha(\alpha+1)}{\epsilon^{2+\lambda-\alpha}}(\log \frac{1}{\epsilon})^2)$ 
		processing time, Fast Fourier Transform can return the product $f^{j}(x)\cdot g^{j}(x)$. We pick the term $x^{c_j}$ whose coefficient is the largest in $f^{j}(x)\cdot g^{j}(x)$ (we break tie arbitrarily), then $c_j$ is the element with the most occurrences in $Table^{(j)}$. Note that the time to construct $f^{j}(x)$ and $g^{j}(x)$ is $\OO(|\overline{X}^j|+\frac{1}{\epsilon^{2+\lambda-\alpha}})$, and the time to pick $x^{c_j}$ from $f^{j}(x)\cdot g^{j}(x)$ is  %$\OO( k_{l^j_1} - s_{\tau}+  k_{l^j_{m_j}}-k_{l^j_1}  ) = \OO(k_{l^j_{m_j}} - s_{\tau}  ) = \OO( \mathscr{C}^{\max}) =\OO( \frac{\alpha+1}{\epsilon^{2+\lambda-\alpha}} \log \frac{1}{\epsilon})$
		$\OO( \mathscr{C}^{\max}) =\OO( \frac{\alpha+1}{\epsilon^{2+\lambda-\alpha}} \log \frac{1}{\epsilon})$. To summarize, the total time to find $c_j$, i.e., fine the element with the most occurrences in $Table^{(j)}$, is $\OO(|\overline{X}^j|+\frac{\alpha(\alpha+1)}{\epsilon^{2+\lambda-\alpha}}(\log \frac{1}{\epsilon})^2)$.

		%Since ${\overline{\mathscr{H}}}^{\max} \le \overline{X}^{\min}$, we have $k_{l^j_1} \ge s_v$ for every $v=1,2\cdots,\tau$, thus  $\max_{u}(k_{l^j_u} - k_{l^j_1} ) \le \mathscr{C}^{\max} = \OO( 1+\frac{\alpha}{\epsilon^{2+\lambda-\alpha}} \log \frac{1}{\epsilon})$. Let $\mathscr{C}_j = set\{ k_{l^j_1}-s_v \ |\ u=1,2,\cdots, m_j, v = 1,2,\cdots,\tau \}$. Observe that $\mathscr{C}_j \subset \mathscr{C}$, we have $\mathscr{C}_j^{\min} \ge 0 $ and $\mathscr{C}_j^{\max} = \OO( 1+\frac{\alpha}{\epsilon^{2+\lambda-\alpha}} \log \frac{1}{\epsilon})$.Then according to \cite{}, %引用快速傅立叶的结论，in $\OO(\frac{\alpha}{\epsilon^{2+\lambda-\alpha}}(\log \frac{1}{\epsilon})^2)$ time, Fast Fourier Transform can return the product $f^{j}(x)\cdot g^{j}(x)$, and we pick the term $x^{c_j}$ whose coefficient is the largest in $f^{j}(x)\cdot g^{j}(x)$ (we break tie arbitrarily). Consequently, $c_j$ is the element with the most occurrences in $Table^{(j)}$. 
		
		Now we consider $P_2^{(j)}$. Given $c_{j}$ as the element with the most occurrences in $Table^{(j)}$, we want to identify all columns that contain $c_j$. For every $ \bar{x}_{l^j_u} \in \overline{X}_{j}$, we use binary search to check whether $\frac{\bar{x}_{l^j_u}}{(1+\epsilon^{2+\lambda-\alpha})^{c_{j}}} $ %$\frac{\bar{x}_{l^j_u}}{(1+\epsilon^{2+\lambda-\alpha})^{c_{j}}} = (1+\epsilon^{2+\lambda-\alpha})^{k_{l^j_{u}}-c_{j}}$
		is in $\overline{\mathscr{H}}$. If the answer is ``yes", i.e., there exists $\overline{H}_{v} \in \overline{\mathscr{H}}$ such that $\overline{H}_{v} =  \frac{\bar{x}_{l^j_u}}{(1+\epsilon^{2+\lambda-\alpha})^{c_{j}}}$,  then column $l^j_{u}$ contains $c_j$ and we can round $x_{l^j_u}$ to $(1+\epsilon^{2+\lambda-\alpha})^{c_{j}}\cdot  H_{v}$. Binary search takes logarithmic time, therefore in $\OO(| \overline{X}^{j}|\log \frac{1}{\epsilon})$ time, we can find all columns containing $c_j$ in $Table^{(j)}$, and meanwhile round the elements in $X$ corresponding to these columns accordingly. Moreover, by dropping these columns from $Table^{(j)}$, we can obtain $Table^{(j+1)}$ and go to the next iteration. In general, for $Table^{j+1}$ we only store $\overline{X}^{j+1}$ and $\overline{\mathscr{H}}$. From $\overline{X}^{j+1}$ and $\overline{\mathscr{H}}$ it is sufficient to recover the whole table as each entry can be uniquely determined.
		
		To summarize, there are at most $\OO(\log(|X|)\cdot (\log\frac{1}{\epsilon})^{\OO(\bar{d}+1)})$ iterations, a simple calculation shows that the total processing time is $\OO((|X|+\frac{\alpha(\alpha+1)}{\epsilon}) \cdot 
		\log(|X|) \cdot(\log\frac{1}{\epsilon} )^{\OO(\bar{d}+1)})$. This completes the proof of Lemma~\ref{lemma:smooth_appro}.\qed\end{proof}
	
	\section{Algorithms for computing subset-sums of smooth numbers}\label{sec:alg-smooth}
	
	Given any multiset $X$, recall that the subset-sums of $X$ is $S(X) = set\{\Sigma(Y) \ | \  Y   \subset X\}$, representing the set of all possible subset sums of $X$.  We show that, if all the input numbers have a nice number-theoretic property, then their subset-sums can be computed (approximately) in a more efficient way. More precisely, given $\epsilon>0$, $d\in \mathbb{N}_{+}$ and $\bar{d}\in \mathbb{N}\cap [0, d-1]$, we define $(\epsilon,d,\bar{d})$-\textit{smooth numbers} as the integers that can and have been factorized as $h_1 h_2 \cdots h_{\bar{d}+1}$, where  $h_{\bar{d}+1} \in \mathbb{N}\cap [1, {2\epsilon^{-\frac{1}{d}}}]$ and 
	$h_i \in \mathbb{N}\cap [\frac{1}{2}\epsilon^{-\frac{1}{d}}, 2\epsilon^{-\frac{1}{d}}]$ for $i =1,2,\cdots, \bar{d}$. For ease of presentation, we refer to $(\epsilon,d,\bar{d})$-{smooth numbers} as smooth numbers when $d,\bar{d}$ are clear from the context. The goal of this section is to prove the following Lemma on algorithms for approximating subset-sums of smooth numbers.  
	
	\begin{lemma}\label{lemma:e_apx_sm}
		Given $d \in \mathbb{N}_{+}$ and $\bar{d} \in \mathbb{N}\cap [0,d-1]$. Let $A $ be a multiset of $(\epsilon,d,\bar{d})$-{smooth numbers}, that is, every element in $A$ can and have been factorized as $h_1 h_2 \cdots h_{\bar{d}+1}$, where  $h_{\bar{d}+1}\in \mathbb{N}_{+}\cap [1,2\epsilon^{-\frac{1}{d}}]$ and if $\bar{d} \ge 1$ then $h_{i} \in \mathbb{N}_{+}\cap [\frac{1}{2}\epsilon^{-\frac{1}{d}},{2\epsilon^{-\frac{1}{d}}}]$ for $i=1,2,\cdots,\bar{d}$. 
		
		Then for any $k \in \mathbb{N}\cap [0,\bar{b}]$, in $\OO(d\cdot |A|+\Sigma(A) {\epsilon^{\frac{k}{d}}} \cdot \log (|A|)\cdot \log (\Sigma(A) {\epsilon^{\frac{k}{d}}})+ {\epsilon^{-(1+\frac{k}{d})}}\log \frac{1}{\epsilon})$ processing time, we can 
		\begin{enumerate}
			\item[(i).] Compute an $\OO(\epsilon (\log \frac{1}{\epsilon})^k)$-approximate set with cardinality of $\OO(\frac{1}{\epsilon})$ for $S(A)$;
			\item[(ii).] Meanwhile build a $T^{(1)}$-time oracle for backtracking from this approximates set to $A$, where $T^{(1)}=\OO(\Sigma(A) {\epsilon^{\frac{k}{d}}} \cdot \log (|A|)\cdot \log (\Sigma(A) {\epsilon^{\frac{k}{d}}})+ {\epsilon^{-(1+\frac{k}{d})}}\log \frac{1}{\epsilon})$.
		\end{enumerate}
	\end{lemma}
	\noindent\textbf{Remark.} %Notice that to achieve an optimal time bound in Lemma~\ref{lemma:e_apx_sm}, one need to choose $k$ such that $\max\{\Sigma(A) {\epsilon^{\frac{k}{d}}}, \epsilon^{-(1+\frac{k}{d})} \}$ attains its minimum. 
	We will take a specific value for $k$ when we apply Lemma~\ref{lemma:e_apx_sm} in the future, namely $k = {d}/{4}$, where $d\ge 12$ is divisible by $4$.
	
	Note that Lemma~\ref{lemma:e_apx_sm} consists of two parts: one is computing (approximately) subset-sums	of $A$, and the other is building an oracle for backtracking.  Thus Lemma~\ref{lemma:e_apx_sm} implies a deterministic weak $(1-\tilde\OO(\epsilon))$-approximation algorithm for  SUBSET SUM instance $(A,t)$, where $A$ is a multiset of $(\epsilon,d,\bar{d})$-smooth numbers and $t = \Theta(\Sigma(A))$.

	\begin{proof}[Proof of Lemma~\ref{lemma:e_apx_sm}]
		Denote the set of all different integers in $\mathbb{N}_{+}\cap [\frac{1}{2}\epsilon^{-\frac{1}{d}},2\epsilon^{-\frac{1}{d}}]$ by $\mathcal{P} = \{p_1,p_2,\cdots,p_{m}\}$. Then every element of $A$ is a multiplication of $\bar{d}+1$ integers (which are called factors of the element), with $\bar{d}$ integers belonging to $\mathcal{P}$ and one extra integer $h_{\bar{d}+1}$ that may or may not belong to $\mathcal{P}$.  For simplicity, we fix the order of factors of each element arbitrarily (except that the $(\bar{d}+1)$-th factor must be $h_{\bar{d}+1}$ if $h_{\bar{d}+1}\not\in\mathcal{P}$) and refer to $h_i$ as its $i$-th factor. % where $m \le \epsilon^{-\frac{\alpha}{d}}$. 
		%Assume that $\bar{d}\ge 1$. 
		
		Intuitively, we want to divide integers of $A$ into subsets such that integers that share the same factors are in the same subset. However, since every integer is a multiplication of $\bar{d}+1$ factors, the division will have a layered structure. We give a simple example for when $\bar{d}=2$. Consider 3 integers $p_1 p_2 p_3$, $p_2 p_3 p_4$ and $p_1 p_2 p_4$. We first divide them based on the first factor, that is, integers whose first factor is $p_i$ are put into group $i$. Hence we derive two groups: group 1 containing $p_1p_2p_3$, $p_1p_2p_4$, and group 2 containing $p_2p_3p_4$. We further subdivide group 1 based on the second factor, that is, integers whose second factor is $p_i$ are put into subgroup $i$. We can see that $p_1p_2p_3$ and $p_1p_2p_4$ are still in the same subgroup as their second factors are also the same. We keep the subdivision procedure for $\bar{d}$ times. It is easy to see that if two integers are in the same subgroup after $\zeta$ rounds of subdivisions, then they share the first $\zeta$ factors. 
		
		Now we formally present the division procedure. We build a tree structure of $ \bar{d}+1$ layers for all elements in $A$ as follows: 
		\begin{itemize}
			\item Let the root node contain all elements in $A$. Root node has layer-$1$.				
			\item If $\bar{d} =0$, we stop building and obtain a single node tree structure.  Else if $\bar{d} \ge 1$, we proceed to create $m$ child nodes of root node by subdividing $A$ into $m$ groups such that elements in the $i$-th group share $p_i$ as their first factor (If $A$ does not contain an element with $p_i$ as its first factor, then let the $i$-th group be an empty set).    				
			\item Suppose we have obtained nodes of layer-$\zeta$, where $2
			\le \zeta\le \bar{d}+1$, such that elements in each node share the same $j$-th factor for $1 \le j\le \zeta-1$.  If $\zeta  = \bar{d}+1$, we stop building and obtain a tree structure of $ \bar{d}+1$ layers. Else if $\zeta \le \bar{d}$,  we create nodes of layer-$(\zeta+1)$ as follows. Consider an arbitrary node of layer-$\zeta$, say, $u$.  Let $A_u\subset A$ be the (multi-)set of all the elements contained in node $u$.  We create $m$ child nodes of $u$ by subdividing $A_u$ into $m$ groups such that elements in the $i$-th group share $p_i$ as their $\zeta$-th factor (If $A_u$ does not contain an element with $p_i$ as its $\zeta$-th factor, then let the $i$-th group be an empty set).    %Let $A_u\subset A$ be the (multi-)set of all the elements contained in node $u$. Since they share the same $j$-th factor for $1\le j\le \zeta-1$, we denote by $\pi_u$ the product of these $\zeta-1$ factors. Then $A_u/\pi_u$ is also an integral (multi-)set.
			
		\end{itemize}

		%Observe that elements in layer-$\zeta$ are of the form $h_1 h_2 \cdots h_{d+1-\zeta}$.
		For each layer-$\zeta$ node $u$, let $A_u\subset A$ be the (multi-)set of all the elements contained in node $u$. We define $\pi_u$ as follows: if  $\zeta=1$, then $\pi_u = 1$; else if 
		$2\le \zeta \le \bar{d}+1$, then elements in $A_u$ share the same $j$-th factor for $1\le j \le \zeta-1$, we let $\pi_u$ be the product of these $\zeta-1$ factors.  
		
		Notice that each node in this tree except the leaf node has a degree of $m=\OO({\epsilon^{-\frac{1}{d}}})$, so there are in total  $m^{\zeta-1}=\OO({\epsilon^{-\frac{(\zeta-1)}{d}}})$ nodes at the $\zeta$-th layer (called layer-$\zeta$ nodes).   The time of building this tree is  $\OO({\epsilon^{-\frac{(\bar{d}+1)}{d}}}+\bar{d} \cdot |A|)$.
		
		Let $k \in \mathbb{N} \cap [0, \bar{d}]$. Next, with the help of this tree structure, we will compute an $\tilde\OO(\epsilon)$-approximate set of $S(A)$ and build an oracle for backtracking from this approximate set to $A$. Note that we do not start from the leaves of the tree, but start from layer-$(k+1)$ nodes for some parameter $k$ that can be optimized when we apply Lemma~\ref{lemma:e_apx_sm}. 
		
		\begin{enumerate}
			\item   %If $k=0$, recall Lemma~\ref{lemma:sumset}, in $\OO(\Sigma(A)\log (\Sigma(A)) \log(|A|)) = \OO(\Sigma(A) \epsilon^{\frac{k}{d}}\log (\Sigma(A)\epsilon^{\frac{k}{d}}) \log(|A|)) $ processing time, we can compute $S(A)$ and meanwhile build an $\OO(\Sigma(A) \log (\Sigma(A)) \log(|A|))$-time oracle for backtracking from $S(A)$ to $A$. 

			Start from layer-$(k+1)$.  Consider any node in layer-$(k+1)$, namely node $u$, let $A_u$ be the (multi-)set of elements it contains. %We define $\pi_u$ as follows:  If $1\le k\le \bar{d}$, then elements in $A_u$ share the same $j$-th factor for $1\le j \le k$, we let $\pi_u$ be the product of these $k$ factors.  Else if $k=0$, then there is only one node in layer-$(k+1)$, namely the root node, we let $\pi_u = 1$. 
			%Recall that $\pi_u = 1$ when $k=0$, moreover, when $1\le k \le \bar{d}$,  elements in $A_u$ share the same $j$-th factor for $1\le j \le k$ and $\pi_u$ is the product of these $k$ factors.  
			Let $\tilde{A}_u=A_u/\pi_u$.  Note that  $\tilde{A}_u$ is also an integral (multi-)set.  %we treat each element in $\tilde{A}_u$ as a set, 
			Recall Lemma~\ref{lemma:sumset}, in $\OO(\Sigma(\tilde{A}_u)\log \Sigma(\tilde{A}_u) \log|\tilde{A}_u|)$ processing time, we can compute $S(\tilde{A}_u)$ and build an $\OO(\Sigma(\tilde{A}_u) \log \Sigma(\tilde{A}_u) \log|\tilde{A}_u|)$-time oracle for backtracking from $S(\tilde{A}_u)$ to $\tilde{A}_u$. Define $B_{\tau} =  S(\tilde{A}_u) \cap [(\tau-1)\epsilon \Sigma(\tilde{A}_u),\tau\epsilon \Sigma(\tilde{A}_u))$, where $\tau = 1,2,\cdots,2+\lfloor{1}/{\epsilon}\rfloor$. Let $C_{\tilde{A}_u} = set \{(B_{\tau})^{max}, (B_{\tau})^{min} \ | \ \tau = 1,2,\cdots,2+\lfloor{1}/{\epsilon}\rfloor\} $. 
			$C_{\tilde{A}_u}$ essentially contains the two ``extreme" points (i.e., largest and smallest) of $S({\tilde{A}_u})$ within each subinterval  $[(\tau-1)\epsilon \Sigma({\tilde{A}_u}),\tau\epsilon \Sigma({\tilde{A}_u}))$, and it is thus easy to see that $C_{\tilde{A}_u}$ is an $\epsilon$-approximate set with cardinality of $\OO(1/\epsilon)$ for $S({\tilde{A}_u})$ and the above oracle is an $\OO(\Sigma({\tilde{A}_u}) \log \Sigma({\tilde{A}_u}) \log|{\tilde{A}_u}|)$-time oracle for backtracking from $C_{\tilde{A}_u}$ to ${\tilde{A}_u}$. Note that the time of computing  $C_{\tilde{A}_u}$ is bounded by the cardinality of $S({\tilde{A}_u})$, which is $\OO(\Sigma({\tilde{A}_u}))$. %{\color{red}The time of computing $C_{\tilde{A}_u}$ is bounded by the cardinality of $\Sigma({\tilde{A}_u})$, note that $\Sigma({\tilde{A}_u})$ is a set, thus the processing time  $\OO(\Sigma({\tilde{A}_u}))$.} 
			Observe that  $\pi_{u} C_{{\tilde{A}_u}}$ is an $\epsilon$-approximate set of $ S(\pi_{u} {\tilde{A}_u})= S(A_u)$, and the oracle for backtracking from $C_{{\tilde{A}_u}}$ to ${\tilde{A}_u}$ directly yields an $\OO(\Sigma({\tilde{A}_u}) \log \Sigma({\tilde{A}_u}) \log|{\tilde{A}_u}|)$-time oracle for backtracking from $\pi_{u} C_{{\tilde{A}_u}}$ to $A_u =\pi_{u} {\tilde{A}_u}$.  We associate $\pi_{u} C_{\tilde{A}_u}$ and the oracle for backtracking from $\pi_{u} C_{{\tilde{A}_u}}$ to $A_u$ with node $u$.

			%Now we estimate the total processing time of handling all layer-$(k+1)$ nodes. % for computing $\pi_u C_{\tilde{A}_u}$ for every node of layer-$(k+1)$. 
			Note that the processing time of handling one layer-$(k+1)$ node, say $u$, is $\tilde\OO(\Sigma({\tilde{A}_u}) \log \Sigma({\tilde{A}_u}) \log|{\tilde{A}_u}|)$. Then the total processing time of handling all layer-$(k+1)$ nodes is $\sum_{u}\OO(\Sigma({\tilde{A}_u}) \log \Sigma({\tilde{A}_u}) \log|{\tilde{A}_u}|)$.  Observe that $\pi_u=\Theta(\epsilon^{-k/d})$ and $\sum_{u}\Sigma(A_u)=\Sigma(A)$, % since all layer-$(k+1)$ nodes $u$ form a partition of $A$. 
			we have $\sum_{u}\Sigma(\tilde{A}_u)=\OO(\Sigma(A) {\epsilon^{\frac{ k}{d}}})$ and  $\sum_{u}\OO(\Sigma({\tilde{A}_u}) \log \Sigma({\tilde{A}_u}) \log|{\tilde{A}_u}|)=\OO(\Sigma(A) {\epsilon^{\frac{k }{d}}} \log (|A|)\log (\Sigma(A) {\epsilon^{\frac{k }{d}}}))$.

			\item If $k= 0$, note that we have computed an $\epsilon$-approximate set for  $S(A)$ and built an oracle for backtracking from this approximate set to $A$. In the following, we consider the case that $1\le k \le \bar{d}$. From layer-$k$ to layer-1, we iteratively compute an $\tilde\OO(\epsilon)$-approximate set for each node by using the $\tilde\OO(\epsilon)$-approximate sets of its children: %from the obtained $\tilde\OO(\epsilon)$-approximate sets of its children. 
			%Details are as follows. 
			Consider any layer-$\zeta$ node $u$, where $1\le\zeta \le k$. Let $A_u$ be the (multi-)set of elements $u$ contains, and ${\tilde{A}_u}=A_u/\pi_u$. %(recall that $\pi_u = 1$ when $\zeta=1$, moreover, when $2\le \zeta \le \bar{d}$,  elements in $A_u$ share the same $j$-th factor for $1\le j \le \zeta-1$ and $\pi_u$ is the product of these $\zeta-1$ factors). %(recall that elements in $A_u$ share the same first $\zeta-1$ factors and $\pi_u$ is the product of the first $\zeta-1$ factors). % and ${\tilde{A}_u} = A_u/\pi_u$. 			
			Let $u_i$ denote the $i$-th child node of $u$, $A_{u_i}$ be the set of elements contained in $u_i$ for $i=1,2,\cdots, m$, and  $ {\tilde{A}_{u_i}}=A_{u_i}/ \pi_{u_i}$.
			Suppose we have obtained an $\tilde\OO(\epsilon)$-approximate set with cardinality of $\OO({1}/{\epsilon})$ for each $S({\tilde{A}_{u_i}})$, say $C_{{\tilde{A}_{u_i}}}$, and meanwhile we have built a $T_i$-time oracle for backtracking from $C_{{\tilde{A}_{u_i}}}$ to ${\tilde{A}_{u_i}}$. Then $\pi_{u_i} C_{{\tilde{A}_{u_i}}}$ is an $\tilde\OO(\epsilon)$-approximate set of $ S(A_{u_i})$, and the oracle for backtracking from $C_{{\tilde{A}_{u_i}}}$ to ${\tilde{A}_{u_i}}$ directly yields an $\OO(T_i)$-time oracle for backtracking from $\pi_{u_i} C_{{\tilde{A}_{u_i}}}$ to $A_{u_i}$. Note that $A_{u} = A_{u_1} \dot{\cup} A_{u_2} \dot{\cup} \cdots \dot{\cup} A_{u_m}$. Recall Corollary~\ref{coro:tree-fashion_} and $m = \OO(\epsilon^{-\frac{1}{d}})$, in $\OO(\frac{1}{\epsilon^{1+\frac{1}{d}}}\log \frac{1}{\epsilon})$ processing time, we can compute an $\tilde\OO(\epsilon \log \frac{1}{\epsilon})$-approximate set with cardinality of $\OO( \frac{1}{\epsilon})$ for $S(A_u)$, denote by $C_{A_u}$ this approximate set. At the same time, we have built an $\OO(\Sigma^m_{i=1}T_i + \frac{1}{\epsilon^{1+\frac{1}{d}}}\log\frac{1}{\epsilon})$-time oracle for backtracking from $C_{A_u}$ to $A_u$. Then we associate $C_{A_u}$ and the oracle for backtracking from $C_{A_u}$ to $A_u$ with node $u$. 
			
			Note that there are a total of $\OO(\epsilon^{-\frac{(\zeta-1) }{d}})$ nodes in layer-$\zeta$, thus the overall processing time in layer-$\zeta$ is $\OO(\frac{1}{\epsilon^{1+\frac{\zeta}{d}}}\log \frac{1}{\epsilon})$. Furthermore, the total processing time from layer-$k$ to root is $\OO(\frac{1}{\epsilon^{1+\frac{k}{d}}}\log \frac{1}{\epsilon})$.
		\end{enumerate}   
		We estimate the overall processing time now. The total processing time is 
		\begin{align*}
			T &=\OO({\epsilon^{-\frac{(\bar{d}+1)}{d}}}+\bar{d} \cdot |A|)+ \OO(\Sigma(A) {\epsilon^{\frac{k}{d}}} \cdot \log (|A|)\cdot \log (\Sigma(A) {\epsilon^{\frac{k}{d}}})+ {\epsilon^{-(1+\frac{k}{d})}}\log \frac{1}{\epsilon})\\
			& = \OO(d\cdot |A|+\Sigma(A) {\epsilon^{\frac{k}{d}}} \cdot \log (|A|)\cdot \log (\Sigma(A) {\epsilon^{\frac{k}{d}}})+ {\epsilon^{-(1+\frac{k}{d})}}\log \frac{1}{\epsilon}).
		\end{align*}
		
		At the root node, we will obtain an $\OO(\epsilon (\log \frac{1}{\epsilon})^k)$-approximate set with cardinality of $\OO({1}/{\epsilon})$ for $S(A)$. Denote by $C_{A}$ this approximate set. 
		
		With the help of this tree structure, we derive an oracle for backtracking from $C_{A}$ to $A$ as follows: given any $c\in C_{A}$, use the oracle associated with the root node to return two numbers from its two child nodes. Backtrace recursively from the root node to leaf nodes, for any number in layer-$h$ node, use the oracle associated with this node to return two numbers from its two child nodes.  Let $A'$ be the collection of numbers returned from all leaf nodes.  One can easily prove that  $A '\subset A$ satisfies $|\Sigma(A')-c| \le \OO(\epsilon(\log \frac{1}{\epsilon})^k) \Sigma(A)$. The total processing time for backtracking is $T^{(1)}= \OO(\Sigma(A) {\epsilon^{\frac{k}{d}}} \cdot \log (|A|)\cdot \log (\Sigma(A) {\epsilon^{\frac{k}{d}}})+ {\epsilon^{-(1+\frac{k}{d})}}\log \frac{1}{\epsilon})$, which follows from the following recurrent calculation:  $T^{(k+1)} = \OO(\Sigma(A) {\epsilon^{\frac{k }{d}}} \log (|A|)\log (\Sigma(A) {\epsilon^{\frac{k }{d}}}))$ and $T^{(\zeta)} = \OO(T^{(\zeta+1)}+  \frac{m^{\zeta-1}}{\epsilon^{1+\frac{1}{d}}}\log\frac{1}{\epsilon}) =  \OO(T^{(\zeta+1)}+  \frac{1}{\epsilon^{1+\frac{\zeta}{d}}}\log\frac{1}{\epsilon})$ for $\zeta = k,k-1,k-2,\cdots,1$. \qed

		%Notice that with the help of this tree structure, we have built an $T^{(1)}$-time oracle for backtracking from this approximate set to $A$. Where $T^{(1)}= \OO(\Sigma(A) {\epsilon^{\frac{k}{d}}} \cdot \log (|A|)\cdot \log (\Sigma(A) {\epsilon^{\frac{k}{d}}})+ {\epsilon^{-(1+\frac{k}{d})}}\log \frac{1}{\epsilon})$, which follows from the following recurrent calculation:  $T^{(k+1)} = \OO(\Sigma(A) {\epsilon^{\frac{k }{d}}} \log (|A|)\log (\Sigma(A) {\epsilon^{\frac{k }{d}}}))$ and $T^{(\zeta)} = \OO(T^{(\zeta+1)}+  \frac{m^{\zeta-1}}{\epsilon^{1+\frac{1}{d}}}\log\frac{1}{\epsilon}) =  \OO(T^{(\zeta+1)}+  \frac{1}{\epsilon^{1+\frac{\zeta}{d}}}\log\frac{1}{\epsilon})$ for $\zeta = k,k-1,k-2,\cdots,1$.
		
	\end{proof}
	%\footnote{This is because that the time of backtracking-time oracle for the approximate set in each node increases by the same amount as it takes to compute this $\tilde\OO(\epsilon)$-approximate set.}
	
	%为什么必须要分开？用的时候也是分开用的
	%有capped和没capped的技术区别是什么？为什么两个证明要分开（在框架相似的情况下）
	Sometimes we only care about computing $\omega$-capped subset-sums of smooth numbers and hope to design a customized algorithm whose running time will decrease with the decrease of $\omega$. Towards this, we develop the following Lemma~\ref{lemma:tree-alg} via a similar proof as Lemma~\ref{lemma:e_apx_sm}.
	\begin{lemma}\label{lemma:tree-alg}
		Given $d\in \mathbb{N}_{+}$ and $\bar{d}\in \mathbb{N}\cap [0,d-1]$. Let $A$ be a multiset of $(\epsilon,d,\bar{d})$-smooth numbers, that is,  every element in $A$ have been factorized as $h_1 h_2 \cdots h_{\bar{d}} h_{\bar{d}+1}$, where $h_{\bar{d}+1} \in \mathbb{N}_{+}\cap [1,2\epsilon^{-\frac{1}{d}}]$ and if $\bar{d}\ge 1$ then  $h_{i} \in \mathbb{N}_{+} \cap [\frac{1}{2}\epsilon^{-\frac{1}{d}},2\epsilon^{-\frac{1}{d}}]$ for $i = 1,2,\cdots,\bar{d}$. 
		
		%{\color{red}Fix $\omega \in [0, \Sigma(A)/2]$.}
		Fix $\omega \in [0, \Sigma(A)]$. Then for any $k\in \mathbb{N}\cap [0,\bar{d}]$, in  $\OO\left(\epsilon^{-\frac{(\bar{d}+1)}{d}}+\bar{d} \cdot |A|+ \Sigma(A) {\epsilon^{\frac{k }{d}}} \log (|A|)\log (\Sigma(A) {\epsilon^{\frac{k }{d}}})+ 2^{2k+1} \epsilon^{-\frac{1}{d}}  \omega\log \omega\right)$ processing time, we can  
		%{\color{blue}$\OO(\epsilon^{-\frac{\bar{d}+1}{d}}+\bar{d}\cdot |A|+ \Sigma(A) {\epsilon^{\frac{k }{d}}} \log (|A|)\log (\Sigma(A) {\epsilon^{\frac{k }{d}}})+  2^{k-1}{\epsilon^{\frac{-1}{d}}}{\omega}  \cdot \log (2^{k-1}{\epsilon^{\frac{k-1}{d}}}{\omega}))$ processing time}, we can 
		\begin{itemize}
			\item[(i).] Compute $S(A;[0,\omega])$;
			\item[(ii).] Meanwhile build a $T^1$-time oracle for backtracking from $S(A;[0,\omega])$ to $A$, where $T^1 = \OO(\Sigma(A) {\epsilon^{\frac{k }{d}}} \log (|A|)\log (\Sigma(A) {\epsilon^{\frac{k }{d}}})+ 2^{2k+1} \epsilon^{-\frac{1}{d}}\omega\log \omega)$.
		\end{itemize}
	\end{lemma}
	\noindent\textbf{Remark.} Notice that to achieve an optimal time bound in Lemma~\ref{lemma:tree-alg}, one need to choose $k$ such that $\max\{\Sigma(A) {\epsilon^{\frac{k }{d}}}, 2^{2k+1}{\epsilon^{\frac{-1}{d}}}{\omega} \}$ attains its minimum. We will take a specific value for $k$ when we apply Lemma~\ref{lemma:tree-alg} in the future, namely $k = {d}/{2}$, where $d$ is a positive even number.
	
	Note that Lemma~\ref{lemma:tree-alg} consists of two parts: one is computing capped subset-sums of $A$, and the other is building an oracle for backtracking. Thus Lemma~\ref{lemma:tree-alg} implies an exact algorithm for SUBSET SUM instance $(A,\omega)$, where $A$ is a multiset of $(\epsilon,d,\bar{d})$-smooth numbers and $\omega \in [0, \Sigma(A)]$.  
	
	\begin{proof}[Proof of Lemma~\ref{lemma:tree-alg}]
		Denote the set of all different integers in $\mathbb{N} \cap [\frac{1}{2}\epsilon^{-\frac{1}{d}},2\epsilon^{-\frac{1}{d}}]$ by $\mathcal{P} = \{p_1,p_2,\cdots,p_{m}\}$. Then every element of $A$ is a multiplication of $\bar{d}+1$ integers (which are called factors of the element), with $\bar{d}$ integers belonging to $\mathcal{P}$ and one extra integer $h_{\bar{d}+1}$ that may or may not belong to $\mathcal{P}$.  For simplicity, we fix the order of factors of each element arbitrarily (except that the $(\bar{d}+1)$-th factor must be $h_{\bar{d}+1}$ if $h_{\bar{d}+1}\not\in\mathcal{P}$) and refers to $h_i$ as the $i$-th factor. % where $m \le \epsilon^{-\frac{\alpha}{d}}$. 
		%Assume that $\bar{d}\ge 1$. 

		Now we apply the same method used in the proof of Lemma~\ref{lemma:tree-fashion_} to build a tree structure of $ \bar{d}+1$ layers for all elements in $A$ as follows:
		%Now we formally present the division procedure. We build a tree structure of $ \bar{d}+1$ layers for all elements in $A$ as follows: 
		\begin{enumerate}
			\item Let the root node contains all elements in $A$. Root node has layer-$1$.
			
			\item If $\bar{d} =0$, we stop building and obtain a single node tree structure.  Else if $\bar{d} \ge 1$, we proceed to create $m$ child nodes of root node by subdividing $A$ into $m$ groups such that elements in the $i$-th group share $p_i$ as their first factor (If $A$ does not contain an element with $p_i$ as its first factor, then let the $i$-th group be an empty set).    				
			
			\item Suppose we have obtained nodes of layer-$\zeta$, where $2
			\le \zeta\le \bar{d}+1$, such that elements in each node share the same $j$-th factor for $1 \le j\le \zeta-1$.  If $\zeta  = \bar{d}+1$, we stop building and obtain a tree structure of $ \bar{d}+1$ layers. Else if $\zeta \le \bar{d}$,  we create nodes of layer-$(\zeta+1)$ as follows. Consider an arbitrary node of layer-$\zeta$, say, $u$.  Let $A_u\subset A$ be the (multi-)set of all the elements contained in node $u$.  We create $m$ child nodes of $u$ by subdividing $A_u$ into $m$ groups such that elements in the $i$-th group share $p_i$ as their $\zeta$-th factor (If $A_u$ does not contain an element with $p_i$ as its $\zeta$-th factor, then let the $i$-th group be an empty set).   
			
		\end{enumerate}
		
		For each layer-$\zeta$ node $u$, let $A_u\subset A$ be the (multi-)set of all the elements contained in node $u$. We define $\pi_u$ as follows: if  $\zeta=1$, then $\pi_u = 1$; else if 
		$2\le \zeta \le \bar{d}+1$, then elements in $A_u$ share the same $j$-th factor for $1\le j \le \zeta-1$, we let $\pi_u$ be the product of these $\zeta-1$ factors.  
		
		Notice that each node in this tree except the leaf node has a degree of $m=\OO({\epsilon^{-\frac{1}{d}}})$, so there are in total  $m^{\zeta-1}=\OO({\epsilon^{-\frac{(\zeta-1)}{d}}})$ nodes at the $\zeta$-th layer (called layer-$\zeta$ nodes).   The time of building this tree is  $\OO({\epsilon^{-\frac{(\bar{d}+1)}{d}}}+\bar{d} \cdot |A|)$.
		
		Let $k \in \mathbb{N} \cap [0, \bar{d}]$. Next, with the help of this tree structure, we will compute   $S(A;[0,\omega])$ and build an oracle for backtracking from $S(A;[0,\omega])$ to $A$. Note that we
		do not start from the leaves of the tree but rather start from layer-$(k+1)$ nodes for some parameter $k$ that can be optimized when we apply Lemma~\ref{lemma:tree-alg}.
		
		\begin{enumerate}
			\item  Start from layer-$(k+1)$. Consider any layer-$(k+1)$ node $u$. Let $A_u$ be the (multi-)set of elements it contains and $\tilde{A}_u = A_u /\pi_u$. Recall Lemma~\ref{lemma:sumset}, in $\OO(\Sigma(\tilde{A}_u)\log \Sigma(\tilde{A}_u)\log|\tilde{A}_u|)$ processing time, we can compute $S(\tilde{A}_u)$ and build an $\OO(\Sigma(\tilde{A}_u) \log \Sigma(\tilde{A}_u) \log|\tilde{A}_u|)$-time oracle for backtracking from $S(\tilde{A}_u)$ to $\tilde{A}_u$. Observe that $A_u = \pi_u \tilde{A}_u$, $S({A}_u; [0,\omega])=\pi_u S(\tilde{A}_u) \cap [0,\omega]$, and the oracle for backtracking from $S(\tilde{A}_u)$ to $\tilde{A}_u$ directly yields an $\OO(\Sigma(\tilde{A}_u) \log \Sigma(\tilde{A}_u) \log|\tilde{A}_u|)$-time oracle for backtracking from $S(A_u;[0,\omega])$ to $A_u$. We associate $S(A_u;[0,\omega])$ and the oracle for backtracking from $S(A_u;[0,\omega])$ to $A_u$ with node $u$.

			Note that the processing time of handling one layer-$(k+1)$ node, say $u$, is $\OO(\Sigma(\tilde{A}_u)\log \Sigma(\tilde{A}_u)\log|\tilde{A}_u|)$. Then the total processing time of handling all layer-$(k+1)$ nodes is $\OO(\sum_u \Sigma(\tilde{A}_u)\log \Sigma(\tilde{A}_u)\log|\tilde{A}_u|)$.
			%Now we estimate the total time for computing $\pi_u S(\tilde{A}_u)$ for every node of layer-$(k+1)$. 
			Observe that $\pi_u=\Theta(\epsilon^{-k/d})$ and $\sum_{u}\Sigma(A_u)=\Sigma(A)$, we have $\sum_{u}\Sigma(\tilde{A}_u)=\OO(\Sigma(A) {\epsilon^{\frac{ k}{d}}})$ and $\OO(\sum_u \Sigma(\tilde{A}_u)\log \Sigma(\tilde{A}_u)\log|\tilde{A}_u|)= \OO(\Sigma(A) {\epsilon^{\frac{k }{d}}} \log (|A|)\log (\Sigma(A) {\epsilon^{\frac{k }{d}}}))$.
			
			\item If $k=0$, note that we have computed $S(A;[0,\omega])$ and build an oracle for backtracking from $S(A;[0,\omega])$ to $A$. In the following, we consider the case that $1\le k \le \bar{d}$. From layer-$k$ to layer-1, we iteratively compute the capped subset-sums for each node by using the capped subset-sums of its children: Consider any layer-$\zeta \le k$ node $u$, where $1\le \zeta \le k$. Let $A_u$ be the (multi-)set of elements $u$ contains, and ${\tilde{A}_u}=A_u/\pi_u$. % and ${\tilde{A}_u} = A_u/\pi_u$. 
			For $i=1,2,\cdots, m$, let $u_i$ denote the $i$-th child node of $u$, $A_{u_i}$ be the set of elements contained in $u_i$  and  $ \tilde{A}_{u_i}=A_{u_i}/ \pi_{u_i}$. We have $ \pi_{u_i} = \pi_u  \cdot p_i$ for every $i=1,2,\cdots, m$ and $\dot{\cup}^{m}_{i=1} p_{i}  \tilde{A}_{u_i}  = \tilde{A}_u$. 	%Recall that our goal is to compute capped subset-sums. 
			Note that to compute $S(A_u;[0,\omega])$, we only need to compute $S(\tilde{A}_u;[0,\frac{\omega}{\pi_u}])$.   Observe that  $S(\tilde{A}_u;[0,\frac{\omega}{\pi_u}]) = \left(\oplus^{m}_{i=1} S(p_i\tilde{A}_{u_i};[0,\frac{\omega}{\pi_u}])\right)\cap [0,\frac{\omega}{\pi_u}]$. Recall Lemma~\ref{lemma:cap_sumset}, in $\OO( m \frac{\omega}{ \pi_u } \log \frac{\omega}{\pi_u})$ processing time, we can compute $S(\tilde{A}_u ;[0,\frac{\omega}{\pi_u}])$ and meanwhile build an $\OO(m \frac{\omega}{ \pi_u } \log \frac{\omega}{\pi_u})$-time oracle for backtracking from  $S(\tilde{A}_u;[0,\frac{\omega}{\pi_u}])$ to $\dot{\cup}^{m}_{i=1} S(p_i\tilde{A}_{u_i};[0,\frac{\omega}{\pi_u}])$.  Suppose for each $\tilde{A}_{u_i}$, we have obtained $S(\tilde{A}_{u_i};[0,\frac{\omega}{\pi_{u_i}}])$ and built a $T_i$-time oracle for backtracking from $S(\tilde{A}_{u_i};[0,\frac{\omega}{\pi_{u_i}}])$ to $\tilde{A}_{u_i}$.  Notice that $\left(p_i S(\tilde{A}_{u_i};[0,\frac{\omega}{\pi_{u_i}}])\right) \cap[0,\frac{\omega}{\pi_u}] = S(p_i \tilde{A}_{u_i};[0,\frac{\omega}{\pi_u}])$, moreover, the oracle for backtracking from $S(\tilde{A}_{u_i};[0,\frac{\omega}{\pi_{u_i}}])$ to $\tilde{A}_{u_i}$ directly yields an $\OO(T_i)$-time oracle for backtracking from $S(p_i \tilde{A}_{u_i};[0,\frac{\omega}{\pi_u}])$ to $p_i \tilde{A}_{u_i}$.  Combine with the oracle for backtracking from  $S(\tilde{A}_u;[0,\frac{\omega}{\pi_u}])$ to $\dot{\cup}^{m}_{i=1} S(p_i\tilde{A}_{u_i};[0,\frac{\omega}{\pi_u}])$, we will obtain an $\OO(\sum^{m}_{i=1}T_i+m \frac{\omega}{ \pi_u } \log \frac{\omega}{\pi_u})$-time oracle for backtracking from  $S(\tilde{A}_u;[0,\frac{\omega}{\pi_u}])$ to $\tilde{A}_{u}$. Then we associate $S(\tilde{A}_u;[0,\frac{\omega}{\pi_u}])$  and the oracle for backtracking from  $S(\tilde{A}_u;[0,\frac{\omega}{\pi_u}])$ to $\tilde{A}_u $ with node $u$.

			Recall that $\pi_u \in \mathbb{N}\cap [\frac{\epsilon^{-\frac{\zeta-1}{d}}}{2^{\zeta-1}}, 2^{\zeta-1}\epsilon^{-\frac{\zeta-1}{d}}]$ and $m \in \mathbb{N}\cap[\frac{1}{2}\epsilon^{-\frac{1}{d}},2\epsilon^{-\frac{1}{d}}]$. Note that the processing time of handling one layer-$\zeta$ node, say $u$, is $\OO( m \frac{\omega}{ \pi_u } \log \frac{\omega}{\pi_u})$.  Since there are a total of $m^{\zeta-1}$ nodes in layer-$\zeta$,  the total processing time of handling all layer-$\zeta$ nodes is $\OO(\sum_u m \frac{\omega}{ \pi_u } \log \frac{\omega}{\pi_u}) = \OO( m^{\zeta} \sum_u  \frac{\omega}{ \pi_u } \log \frac{\omega}{\pi_u}) = \OO(m^{\zeta} 2^{\zeta-1} \epsilon^{\frac{\zeta-1}{d}} \omega \log \omega)$, furthermore, the overall processing time of handling nodes from layer-$k$ to root is $\OO(2^{2k+1} \epsilon^{-\frac{1}{d}} \omega\log \omega)$.				
			
		\end{enumerate}
		We estimate the overall processing time now. The total processing time is 
		\begin{align*}
			T &=  \OO(\epsilon^{-\frac{(\bar{d}+1)}{d}}+\bar{d} \cdot |A|)+\OO(\Sigma(A) {\epsilon^{\frac{k }{d}}} \log (|A|)\log (\Sigma(A) {\epsilon^{\frac{k }{d}}}))+\OO(2^{2k+1} \epsilon^{-\frac{1}{d}} \omega\log \omega)\\
			&=\OO(\epsilon^{-\frac{(\bar{d}+1)}{d}}+\bar{d} \cdot |A|+ \Sigma(A) {\epsilon^{\frac{k }{d}}} \log (|A|)\log (\Sigma(A) {\epsilon^{\frac{k }{d}}})+ 2^{2k+1} \epsilon^{-\frac{1}{d}} \omega \log \omega).
		\end{align*}
		
		At root node, we will obtain $S(A;[0,\omega])$.

		With the help of this tree structure, we derive an oracle for backtracking from $S(A;[0,\omega])$ to $A$ as follows: given any $c\in S(A;[0,\omega])$, use the oracle associated with the root node to return two numbers from its two child nodes. Backtrace recursively from the root node to leaf nodes, for any number in layer-$h$ node, use the oracle associated with this node to return two numbers from its two child nodes.  Let $A'$ be the collection of numbers returned from all leaf nodes.  One can easily prove that  $A '\subset A$ satisfies $\Sigma(A')=c$. The total processing time for backtracking is $T^1 = \OO(\Sigma(A) {\epsilon^{\frac{k }{d}}} \log (|A|)\log (\Sigma(A) {\epsilon^{\frac{k }{d}}})+ 2^{2k+1} \epsilon^{-\frac{1}{d}}  \omega\log \omega)$, which follows from the following recurrent calculation:  $T^{k+1} = \OO(\Sigma(A) {\epsilon^{\frac{k }{d}}} \log (|A|)\log (\Sigma(A) {\epsilon^{\frac{k }{d}}}))$ and $T^{\zeta} = \OO(T^{\zeta+1} + m^{\zeta} 2^{\zeta-1} \epsilon^{\frac{\zeta-1}{d}}  \omega\log \omega )  \text{ \ for \ } \zeta = k,k-1,k-2,\cdots,1$. \qed

	\end{proof}

	\section{Preprocessing.}\label{sec:proce}
	%讲一下前面section4 已对特殊情况进行了处理，为了利用这个结果，我们有了下面的preprocessing。

	%Let $\epsilon>0$ be a sufficiently small number. Given any SUBSET SUM instance $(X,t)$ and let $OPT$ be the optimal objective value of $(X,t)$. In this section, we focus on simplifying $(X,t)$ with $OPT \ge t/2$ and show that we can construct an equivalent instance $(F,\bar{t})$, where $F$ and $\bar{t}$ satisfying the following conditions:
	Let $\epsilon>0$ be a sufficiently small number. Given any SUBSET SUM instance $(X,t)$ and let $OPT$ be the optimal objective value of $(X,t)$. In this section, we focus on simplifying $(X,t)$ where $OPT \ge t/2$ and show that we can construct a reduced instance $(F,\hat{t})$. Formally, we have the following lemma.
	
	%lemma 10的用处，informal的语言介绍一下lemma10的内容，以及他为什么重要
	\begin{lemma}\label{obs:obs_pre}
		Given any SUBSET SUM instance $(X,t)$, let $OPT$ be the optimal objective value of $(X,t)$. Assume that $OPT \ge t/2$, then in $\OO((|X|+\frac{1}{\epsilon})(\log |X|)^2 (\log \frac{1}{\epsilon})^{\OO(d)})$ processing time, we can
		\begin{itemize}
			\item[(i).] Obtain a modified SUBSET SUM instance $(F,\hat{t})$ satisfying the following conditions:% where $\hat{t}$ and $F$ satisfy the following conditions:
			\begin{itemize}
				\item [(\uppercase\expandafter{\romannumeral1})] $\Sigma(F) \le \frac{4(1+\epsilon)\Sigma(X)}{\epsilon^3 t}$ and $\hat{t}= \frac{4(1+\epsilon)}{\epsilon^3}$.
				\item [(\uppercase\expandafter{\romannumeral2})] The optimal objective value of $(F,\hat{t})$ is at least $\frac{4(1-\epsilon)}{\epsilon^3\cdot t}\cdot (OPT-\frac{\epsilon t}{2})$.
				\item [(\uppercase\expandafter{\romannumeral3})]  $|F| \le |X|$ and $F$ has been divided into $\OO((\log |X|)^2 (\log \frac{1}{\epsilon})^{\OO(d)})$ subgroups: $F^{(j;i;k)}_p$'s. 
				%$$\{F^{(j;i;k)}_p  \ |  \  1\le p \le \OO(\log |X|), 1\le j \le \OO(\log \frac{1}{\epsilon}),  i=1,2 \text{\ and \ } 1\le k \le \OO(\log|X|  \cdot (\log \frac{1}{\epsilon})^{\OO(d)})\}.$$  
				\item [(\uppercase\expandafter{\romannumeral4})]Each subgroup $F^{(j;i;k)}_p$ satisfies the followings:
				\begin{enumerate}
					\item[(a)] Elements in $F^{(j;i;k)}_p$ are different from each other;
					\item[(b)]$F^{(j;i;k)}_p \subset  [\frac{(1-\epsilon)2^{p+j-1}}{\epsilon^{2}}, \frac{(1+\epsilon)2^{p+j}}{\epsilon^{2}})$;
					\item[(c)] $F^{(j;i;k)}_p = 2^p \rho^j_k \overline{F}^{(j;i;k)}_p$, where $\overline{F}^{(j;i;k)}_p \subset \mathbb{N}_{+}\cap[\frac{1}{4\epsilon},\frac{1}{\epsilon}]$
					and $\overline{F}^{(j;i;k)}_p$ is a set of $(\epsilon,d,d-1)$-{smooth numbers}, that is, every element in $\overline{F}^{(j;i;k)}_p$ has been factorized as $h_1 h_2 \cdots h_{d}$, where $h_{d}\in \mathbb{N}_{+}\cap [1,2\epsilon^{-\frac{1}{d}}]$ and if $d \ge 2$ then $h_{i} \in \mathbb{N}_{+} \cap [\frac{1}{2}\epsilon^{-\frac{1}{d}},2{\epsilon^{-\frac{1}{d}}}]$ for $i=1,2,\cdots,d-1$. 
					
				\end{enumerate}
			\end{itemize}

			%where $\hat{t}$ and $F$ satisfy conditions $(\uppercase\expandafter{\romannumeral1})$ $(\uppercase\expandafter{\romannumeral2})$ $(\uppercase\expandafter{\romannumeral3})$$(\uppercase\expandafter{\romannumeral4})$.
			\item[(ii).] Meanwhile build an oracle for backtracking from $F$ to $X$. Precisely, given subset $U^{(j;i;k)}_p \subset F^{(j;i;k)}_p$ for every $F^{(j;i;k)}_p$ and let $F'$ denote the multiset-union of all $U^{(j;i;k)}_p$'s. Then $F' \subset F$ and in linear time, the oracle will return a subset $X'\subset X$ satisfying $|\frac{\epsilon^3 t}{4} \Sigma(F')-\Sigma(X')|\le \epsilon \Sigma(X')$.
		\end{itemize}
	\end{lemma} 
	\noindent\textbf{Remark.}  Note that optimal objective value of $(F,\hat{t})$ is at least $\frac{4(1-\epsilon)}{\epsilon^3\cdot t}\cdot (OPT-\frac{\epsilon t}{2})$ and $OPT \ge \frac{t}{2}$. The second part of Lemma~\ref{obs:obs_pre} guarantees that towards finding a weak $(1-\tilde\OO(\epsilon))$-approximation of SUBSET SUM instance $(X,t)$, it is sufficient to find a weak $(1-\tilde\OO(\epsilon))$-approximation of modified instance $(F,\hat{t})$.

	The rest of this section is dedicated to proving Lemma~\ref{obs:obs_pre}. We will step-by-step modify the given SUBSET SUM instance $(X,t)$, and Lemma~\ref{obs:obs_pre} follows directly after all the modification operations. 
	
	%{\color{red} Without loss of generality, we may assume the following: $X \subset (0,t)$.}
	\paragraph{Step 1: Handling small elements.} We call an element $x\in X$ a \textit{small} element if and only $x\in [0,\frac{\epsilon  t}{4})$, otherwise we call $x$ a \textit{large} element. Let $X^{s}$ denote all small elements in $X$, then $X \backslash X^s$ is the set of all large elements in $X$. 
	
	We first greedily divide $X^{s}$ into the subgroups $X^{s}_1,X^{s}_2,X^{s}_3,\cdots,X^{s}_{\xi}$, such that $\Sigma(X^{s}_i) \in [\frac{\epsilon  t}{4},\frac{\epsilon  t}{2}) (\forall i)$. Let $\mathcal{L}: = \{X^{s}_1,X^{s}_2,X^{s}_3,\cdots,X^{s}_{\xi}\} \dot{\cup} \left\{  \{x\} : x\in X \backslash X^s  \right\}$, which gives a division of $X$. Then we define multiset $Z := \{z_1,z_2,\cdots,z_{\xi},z_{\xi+1},\cdots, z_{n'}\},$ where  $z_{i} = \Sigma(X^{s}_{i})$ for $1\le i\le \xi$ and $\{z_{\xi+1},\cdots,z_{n'}\} = X \backslash X^s $. Note that $|Z| \le |X|$, $\Sigma(Z) = \Sigma(X)$ and $Z^{\min}\ge \frac{\epsilon t}{4}$.

	The total processing time of Step 1 is $\OO(|X|)$.  We have the following observation.
	
	\begin{observation}\label{obs:han_small}
		Let $Z$ be defined above.	Given any $A\subset X$, there exists $B \subset Z$ such that 
		$\Sigma(A)-\frac{\epsilon t}{2} \le \Sigma(B) \le \Sigma(A)$. Moreover, given any $B'\subset Z$, there exists $A' \subset X$ such that $\Sigma(B') = \Sigma(A')$.

	\end{observation}
	\begin{proof}
		For any $A \subset X$, it holds that $\Sigma(A\cap X^s) \le \Sigma(X^s) = \sum^{\xi}_{i=1}(z_i)$. Since $z_i \in [\frac{\epsilon t}{4},\frac{\epsilon t}{2})$ for $1\le i \le \xi$, we can greedily pick elements in $\{z_1,z_2,\cdots,z_{\xi}\}$ such that the summation of all picked elements is within $[\Sigma(A\cap X^s)-\frac{\epsilon t}{2}, \Sigma(A\cap X^s))$. Denote by $\hat{Z}$ the set of all picked items. It follows that $\hat{Z}\dot{\cup} (A \backslash X^s) \subset Z$ and $\Sigma(\hat{Z}) \in [\Sigma(A\cap X^s)-\frac{\epsilon t}{2}, \Sigma(A\cap X^s))$, moreover, we have
		%$$\Sigma(A) -\frac{\epsilon t}{2}=\Sigma(A\cap X^s)-\frac{\epsilon t}{2} +\Sigma(A \backslash X^s) \le  \Sigma(\hat{Z})+\Sigma(A \backslash X^s)= \Sigma(\hat{Z}\dot{\cup} (A \backslash X^s))\le  \Sigma(A\cap X^s)+\Sigma(A \backslash X^s) = \Sigma(A). $$
		\begin{align*}
			\Sigma(A) -\frac{\epsilon t}{2}
			&=\Sigma(A\cap X^s)-\frac{\epsilon t}{2} +\Sigma(A \backslash X^s)\\
			%&\le 	\Sigma(A) = \Sigma(A\cap X^s) + \Sigma(A \backslash X^s)\\
			&\le  \Sigma(\hat{Z})+\Sigma(A \backslash X^s)= \Sigma(\hat{Z}\dot{\cup} (A \backslash X^s)) \\
			&\le \Sigma(A\cap X^s)+\Sigma(A \backslash X^s) = \Sigma(A)
		\end{align*}
		
		Now consider the second part of the observation. Given any $B'\subset Z$, let $B'' = B'\cap \{z_1,z_2,\cdots,z_\xi \}$. For each $z \in B''$, there exists a corresponding $X_z \subset X$ such that $\Sigma(X_z) = z$. Then we have $(\dot{\cup}_{z\in B''} X_z )\dot{\cup} (B' \backslash B'') \subset X$  and  $
		\Sigma((\dot{\cup}_{z\in B''} X_z)\dot{\cup} (B' \backslash B'')) = \Sigma(\dot{\cup}_{z\in B''} X_z )+\Sigma(B' \backslash B'') = \sum_{z\in B''} (z) + \Sigma(B' \backslash B'') = \Sigma(B'). $\qed \end{proof}
	
	Recall that the optimal objective value of $(X,t)$ is $OPT$. Observation~\ref{obs:han_small} guarantees that there exists subset $Z' \subset Z$ such that $OPT-\frac{\epsilon t}{2} \le \Sigma(Z') \le t$, then  the optimal objective value of SUBSET SUM instance $(Z,t)$ is at least $OPT-\frac{\epsilon t}{2}$. %Recall $OPT \ge \frac{t}{2}$, towards finding a weak $(1-\tilde\OO(\epsilon))$-approximation of $(X,t)$, it is sufficient to find a weak $(1-\tilde\OO(\epsilon))$-approximation of modified instance $(Z,t)$.

	\paragraph{Step 2: Scaling and Grouping.}  We scale $t$ and
	each element in $Z$ by $\frac{\epsilon^3  t}{4}$. To be specific, we scale
	$t$ to $\bar{t} = t \cdot \frac{4}{\epsilon^3\cdot t}$ and scale each $z_i \in Z$ to $\bar{z}_i = z_i \cdot \frac{4}{\epsilon^3\cdot t}$. Let $\bar{Z} = \{\bar{z}_1,\bar{z}_2,\cdots, \bar{z}_{n'}\} \cap [0, \bar{t}]$. Observe that there exists $\bar{Z}'\subset \bar{Z}$ such that $ \frac{4}{\epsilon^3\cdot t} (OPT-\frac{\epsilon t}{2})\le \Sigma(\bar{Z}') \le \bar{t}$, hence the optimal objective value of SUBSET SUM instance $(\bar{Z},\bar{t})$ is at least $\frac{4}{\epsilon^3\cdot t}(OPT-\frac{\epsilon t}{2})$. %Thus  towards finding a weak $(1-\tilde\OO(\epsilon))$-approximation of $(Z,t)$, it is sufficient to find a weak $(1-\tilde\OO(\epsilon))$-approximation of $(\bar{Z},\bar{t})$.

	Recall that $\mathcal{L} = \{X^{s}_1,X^{s}_2,X^{s}_3,\cdots,X^{s}_{\xi}\} \dot{\cup} \left\{  \{x\} : x\in X \backslash X^s  \right\}$ and  each element in $\mathcal{L}$ generates one element of $\bar{Z}$, thus there is a one-to-one correspondence between $\mathcal{L}$ and $\bar{Z}$: $ele \leftrightarrow  \frac{4}{\epsilon^3\cdot t} \Sigma(ele)$ for any $ele \in \mathcal{L}$.  For each $y \in set(\bar{Z})$, let $\mathcal{L}_{y}  := \{ele  \in \mathcal{L} \ : \ \frac{4}{\epsilon^3\cdot t} \Sigma(ele) = y\}$.  We define the mapping $\Phi_1$ from $\bar{Z}$ to $\mathcal{L}$ as follows: given any $\bar{Z}' \subset \bar{Z}$,  for every $y\in set(\bar{Z}')$, recall that $card_{\bar{Z}'}[y]$ refers the multiplicity of $y$ in $\bar{Z}'$, mapping $\Phi_1$ returns any $card_{\bar{Z}'} [y]$ elements in $\mathcal{L}_{y}$. Let $\Phi_1(\bar{Z}')$ be the collection of all elements returned by $\Phi_1$ given $\bar{Z}'$.

	%We build a dictionary $\mathscr{D}_{\bar{Z}}$ as follows: each element in $\mathscr{D}_{\bar{Z}}$ is a  key-value pair $(\bar{z},\mathscr{D}_{\bar{Z}}[\bar{z}])$, where $\bar{z} \in set(\bar{Z})$ and $\mathscr{D}_{\bar{Z}}[\bar{z}]$ is the set of elements in $\mathcal{L}$ that generate the number $\bar{z}$, i.e., $\mathscr{D}_{\bar{Z}}[\bar{z}]= \{ele \in \mathcal{L} :  \Sigma(ele) \cdot \frac{4}{\epsilon^3\cdot t} =  \bar{z}  \}$.  The collection of keys of $\mathscr{D}_{\bar{Z}}$ is $set(\bar{Z})$.   

	Notice that $\Sigma(\bar{Z}) \le \frac{4\Sigma(X)}{\epsilon^3 t}$ and		$|\bar{Z}| \le |Z| \le |X|$. Moreover, note that multiset $\bar{Z} \subset [\frac{1}{\epsilon^2}, \frac{4}{\epsilon^3}]$, we can divide $\bar{Z}$ into $\eta = \OO(\log \frac{1}{\epsilon})$ groups, denoted by $Y_1, Y_2,\cdots, Y_{\eta} $, such that $y \in Y_j $ if and only if $y \in [\frac{2^{j-1}}{\epsilon^2}, \frac{2^{j}}{\epsilon^2})\cap \bar{Z}$.%Let $Y_j = \emptyset$ if $\bar{Z} \cap  [\frac{2^{j-1}}{\epsilon^2}, \frac{2^{j}}{\epsilon^2}) = \emptyset $. 
	
	The total processing time of Step 2 is $\OO(|X| )$.
	
	\paragraph{Step 3: Rounding and Further Grouping.} For each multiset $Y_j$, it holds $Y_j \subset [\frac{2^{j-1}}{\epsilon^2}, \frac{2^{j}}{\epsilon^2})$. We can rewrite $[\frac{2^{j-1}}{\epsilon^2}, \frac{2^{j}}{\epsilon^2})$ as $[\frac{1}{\epsilon^{2+\lambda_j}}, \frac{2}{\epsilon^{2+\lambda_j}})$, where $0\le \lambda_j < 2$. % ($\epsilon>0$ is sufficiently small). 
	Recall Lemma~\ref{lemma:smooth_appro}, given $d \in \mathbb{N}_{+}$ and $\alpha_j = 1+\lambda_j$, %we have $\bar{d}+1 = d$, then 
	in $\OO((|Y_j|+\frac{1}{\epsilon})\cdot\log(|Y_j|) \cdot (\log \frac{1}{\epsilon})^{\OO(d)})$ time, we can obtain a set $\Delta_j  \subset \mathbb{R}$ with $\Delta_j \subset \Theta(\frac{1}{\epsilon^{1+\lambda_j}})$ and $|\Delta_j| = \OO(\log(|Y_j|)\cdot(\log \frac{1}{\epsilon})^{\OO(d)})$, moreover, we can round every $y \in Y_j$ to the form $\rho h_1 h_2 \cdots h_d$, where $\rho\in \Delta_j$ and $h_i$'s satisfy the following conditions: 
	\begin{subequations}
		\begin{align}
			&h_1 h_2 \cdots h_d \in \mathbb{N}\cap [\frac{1}{4\epsilon},\frac{1}{\epsilon}] \text{ \ and \ } h_i \in \mathbb{N}\cap [\frac{1}{2\epsilon^{\frac{1}{d}}},\frac{2}{\epsilon^{\frac{1}{d}}}] \ \text{ for}\ 1\le i \le d; \label{eq:con_1_j}\\
			&|y-\rho h_1 h_2 \cdots h_d| \le \epsilon y. \label{eq:con_2_j}
		\end{align}
	\end{subequations}
	Let $\bar{Y}_j$ denote the set of all such rounded elements obtained from $Y_j$. Consider $\dot{\cup} ^{\eta}_{j=1}\bar{Y}_j$, we have $| \dot{\cup} ^{\eta}_{j=1}\bar{Y}_j| \le \sum^{\eta}_{j=1} | {Y}_j|  = | \bar{Z}| \le |X|$. Recall that the optimal objective value of $(\bar{Z},\bar{t})$ is at least $\frac{4}{\epsilon^3\cdot t} (OPT-\frac{\epsilon t}{2})$. 
	Let $\hat{t} = (1+\epsilon)\bar{t} = \frac{4(1+\epsilon)}{\epsilon^3}$. Condition ~\eqref{eq:con_2_j} guarantees that $\Sigma(\dot{\cup} ^{\eta}_{j=1}\bar{Y}_j) \le \frac{4(1+\epsilon)\Sigma(X)}{\epsilon^3 t}$ and
	the optimal objective value of SUBSET SUM instance $( \dot{\cup} ^{\eta}_{j=1}\bar{Y}_j, \hat{t})$ is at least $\frac{4(1-\epsilon)}{\epsilon^3\cdot t} (OPT-\frac{\epsilon t}{2})$. 
	
	% furthermore, towards finding a weak $(1-\tilde\OO(\epsilon))$-approximation of $(\bar{Z},\bar{t})$, it is sufficient to find a weak $(1-\tilde\OO(\epsilon))$-approximation of $( \dot{\cup} ^{\eta}_{j=1}\bar{Y}_j,(1+\epsilon)\bar{t})$.
	
	Notice that for each $\bar{Y}^j$,  there is a one-to-one correspondence between ${Y}^j$ and $\bar{Y}^j$: $y \leftrightarrow \bar{y}$ for any $y\in Y_j$, where $\bar{y}$ is the factorized form of $y$ obtained after the above rounding procedure.  For each $\bar{y} \in \bar{Y}^j$, let $\mathcal{C}^j_{\bar{y}}$ denote the set of all numbers in $Y_j$ factorized to the form $\bar{y}$.  We define the mapping $\Phi^j_2$ from $ \bar{Y}_j$ to ${Y}_j$ as follows: given any $\bar{Y}'_j \subset \bar{Y}_j$, for every $\bar{y}\in set(\bar{Y}'_j )$, mapping $\Phi^j_2$ returns any $card_{\bar{Y}'_j}[\bar{y}]$ elements in $\mathcal{C}^j_{\bar{y}}$. Let $\Phi^j_2 (\bar{Y}'_j)$ be the collection of all elements returned by $\Phi^j_2$ given $\bar{Y}'_j$.

	For each $\bar{Y}_j$, let $\Delta_j = \{\rho^j_1,\rho^j_2,\cdots, \rho^j_{|\Delta_j|}\}$, where $|\Delta_j| = \OO(\log(|Y_j|)\cdot(\log \frac{1}{\epsilon})^{\OO(d)})$. We further divide $\bar{Y}_j$ into $|\Delta_j|$ groups, denoted by $\bar{Y}^{j}_1,\bar{Y}^{j}_2,\cdots, \bar{Y}^{j}_{|\Delta_j|}$, such that 
	$y\in \bar{Y}^{j}_k$ if and only if $y$ is of the form $\rho^j_k h_1 h_2 \cdots h_d$. %Let $\bar{Y}^{j}_{k} = \emptyset$ if there is no element in $\bar{Y}^{j}$ that is of the form $\rho^j_k h_1 h_2 \cdots h_d$.

	%Notice that $\dot{\cup}^{\eta}_{j=1} \left( \dot{\cup}^{|\Delta_j|}_{k=1} \bar{Y}^j_k\right)$ is actually obtained from $\mathcal{L}$, i.e., each element in $\mathcal{L}$ generates one element of $\dot{\cup}^{\eta}_{j=1} \left( \dot{\cup}^{|\Delta_j|}_{k=1} \bar{Y}^j_k\right)$.Thus there is a one-to-one correspondence between $\mathcal{L}$ and $ \dot{\cup}^{\eta}_{j=1} \left( \dot{\cup}^{|\Delta_j|}_{k=1} \bar{Y}^j_k\right)$.  	Consider each $\bar{Y}^j_k$, for each element $\rho^j_k h_1 h_2 \cdots h_d \in set(\bar{Y}^j_k)$, let $\mathscr{U}^{(j;k)}_{h_1 h_2 \cdots h_d}$ denote the set of elements in $set(Y_j)$ that are rounded to the form $\rho^j_k h_1 h_2 \cdots h_d$ in Step 3.  Recall the definition of $\mathscr{D}_{\bar{Z}}$, $\mathop{ \dot{\cup}}\limits_{u\in \mathscr{U}^{(j;k)}_{h_1 h_2 \cdots h_d}} \mathscr{D}_{\bar{Z}}[u]$ is the set of elements in $\mathcal{L}$ that generate $\rho^j_k h_1 h_2 \cdots h_d$. We build a dictionay $\mathscr{D}_{\bar{Y}^j_k}$  as follows: each element in $\mathscr{D}_{\bar{Y}^j_k}$ is a key-value pair $\left(\rho^j_k h_1 h_2 \cdots h_d, \mathscr{D}_{\bar{Y}^j_k} [\rho^j_k h_1 h_2 \cdots h_d] \right)$, where $\rho^j_k h_1 h_2 \cdots h_d  \in set(\bar{Y}^j_k)$ and $ \mathscr{D}_{\bar{Y}^j_k} [\rho^j_k h_1 h_2 \cdots h_d]  =\mathop{ \dot{\cup}}\limits_{u\in \mathscr{U}^{(j;k)}_{h_1 h_2 \cdots h_d}} \mathscr{D}_{\bar{Z}}[u]$. The collection of keys in $\mathscr{D}_{\bar{Y}^j_k}$ is $set(\bar{Y}^j_k)$.		

	The total processing time of Step 3 is $\OO((|X|+\frac{1}{\epsilon})(\log(|X|))^2 (\log \frac{1}{\epsilon})^{\OO(d+1)})$.
	
	\paragraph{Step 4: From Multiset to (almost) Set.} 
	
	Here, we show that a multiset can be reduced to an alternative multiset with multiplicity at most 2. %The reduction idea is somewhat standard (see \cite{DBLP:books/daglib/0010031}, Section 7.1.1), and first appeared in \cite{lawler1979fast}. We present it here for completeness.
	Towards this, we need the following lemma, which was introduced in \cite{koiliaris2019faster} and also used in \cite{DBLP:conf/soda/MuchaW019}. We copy it here with a slight extension.
	
	%Bringmann et al. \cite{bringmann2021fine} have derived an algorithm for computing $X_1 \oplus X_2 \oplus \cdots \oplus X_{\ell}$, in the following, we show that this algorithm can be extended so that an oracle for backtracking from $\oplus^{\ell}_{i=1}X_{i}$ to $\dot{\cup}^{\ell}_{i=1}X_i$ can be obtained simultaneously. 
	
	%强调 essentialy 是他们的证明，for the completeness of the paper, we present the proof here.
	\begin{lemma}\label{lemma:dis_elem}%[CF. Lemma 4.1 from \cite{DBLP:conf/soda/MuchaW019}] 
		Given a multiset $A$ of positive integers, where %from $\{1,2,\cdots,u\}$, 
		$|A| = n$ and $|set(A)|=n'$. In $\OO(n (\log n)^2)$ processing time, one can
		divide $A$ into $m$ subgroups $A_1, A_2,\cdots,A_{m}$ such that $\Sigma(A_{\iota}) = 2^{p_{\iota}} A^{\max}_{\iota}$, where $p_{\iota} \in \mathbb{N} \cap [0, \log n]$ and $\iota = 1,2\cdots,m$.  Moreover,  $B= \{\Sigma(A_{\iota}) : \iota=1,2,\cdots,m\}$ is a multiset satisfying: (i) $S(A)=S(B)$; (ii) $|B| \leq |A|$ and $|B|=\OO(n'\log n)$; %(iii) each element in $B$ is factorized as the form $2^p x$ where $0\le p \le  \log_2 n$ and $x \in A$; 
		(iii) no element in $B$ has multiplicity exceeding two.
	\end{lemma}
	
	%The proof idea essentially comes from [\cite{DBLP:conf/soda/MuchaW019}, Lemma 4.1]. For the completeness of this paper, we present it here.
	{\begin{proof}[Proof of Lemma~\ref{lemma:dis_elem}.]
			The proof idea essentially comes from [\cite{DBLP:conf/soda/MuchaW019}, Lemma 4.1], in which the method for computing $B$ was provided. For the completeness of this paper, we present it here. %We describe the high-level idea as follows. 
			Given a multiset $A$ of positive integers, where $|A| = n$ and $|set(A)| = n'$. Copy the elements of $A$ into a working multiset $T$ and let $B$ initially be the empty set. For any $x \in T$ with at most $2$ copies, i.e., $card_{T}[x] \le 2$, delete all of $x$ from $T$, then add all of them into $B$. For any $x \in T$ with at least 3 copies, i.e., $card_{T}[x] = 2k+num^x_T$ where $k\in \mathbb{N}_{+}$ and $num^x_T = 1\ \text {or} \ 2$. Delete all of $x$ from $T$, then add $num^x_T$ copies of $x$ into $B$ and add $k$ copies of $2x$ into $T$. Iterate over numbers in $set(T)$ from the smallest one and perform the procedure as described above. %It is easy to prove that $S(A) = S(T\dot{\cup}B)$.
			%Consider any $x \in A$. If the multiplicity of $x$ is $2k + 1$, we can replace it with a single copy of $x$ and $k$ copies of $2x$ while keeping the multiset equivalent. If the multiplicity of $x$ is $2k + 2$, can replace it with 2 copies of $x$ and $k$ copies of $2x$ while keeping the multiset equivalent. We iterate over items from the smallest one and for each with at least 3 copies we perform the replacement as described above. Details of the proof are as follows.
			We remark that the method derived in~\cite{DBLP:conf/soda/MuchaW019} only returns the set $B$, it does not specify the $A_{\iota}$'s. In particular, it requires some extra effort to get the corresponding division of $A$ without exploding the running time. Below we present the details.

			We introduce a special data structure, which is ``dictionary"(see, e.g.~\cite{DBLP:books/daglib/0010343}), to store the division of $A$. Dictionary data structure is used to store data in the key-value pair format. When presented with a key, the dictionary will simply return the associated value. The biggest advantage of this data structure is that the time complexity of inserting, deleting or searching element is $\OO(1)$. We first build two working dictionaries, $\mathscr{D}^*$ and $\mathscr{D}$, where $\mathscr{D}^*$ is the empty dictionary and $\mathscr{D}$ is initialized as follows: each element in $\mathscr{D}$ is a key-value pair $(x,\mathscr{D}[x])$, where $x\in set(A)$ and $\mathscr{D}[x]$ is the multiset satisfying (1). $set(\mathscr{D}[x]) = \left\{ \{x\} \right\}$; (2). $card_{\mathscr{D}[x]}[\{x\}] = card_A [x]$; (3). the collection of all keys in $\mathscr{D}$ is $set(A)$. Note that the multiset-union of all values in $\mathscr{D}$, i.e, $\mathop{\dot{\cup}}\limits_{x\in set(A)} (\mathscr{D}[x])$, gives a division of $A$.

			Let $B$ initially be the empty set. Copy the elements of $A$ into a working multiset $T$. 
			We introduce the min heap data structure (see, e.g.~\cite{DBLP:books/daglib/0010343}) to maintain $T$. ``Min heap" is a specialized tree-based data structure satisfying the following properties: (1). each node is a key-value pair; (2). for any given node $C$, if $P$ is a parent node of $P$, the key of $P$ is less than or equal to the key of $C$. Thus the key of the root node is the smallest among all nodes. It takes linear time to build a min heap from a given array. For a min heap with $m$ nodes, the time complexity of deleting the root node (while keeping the min heap properties) and updating any node is $\OO(\log m)$. %The node at the "top" of the heap (with no parents) is called the root node. 
			We maintain the elements of $T$ in a min heap $D$, where $D$ is initialized as follows: each node in $D$ is a key-value pair $(x,card_{T}[x])$ where $x \in set(T)$, moreover, the collection of all keys in $D$ is $set(T)$. 
			
			In each iteration, extract the root node $(x,card_{T}[x])$ from the heap $D$. 
			
			\begin{itemize}
				\item If $card_T [x] \le 2$, we first delete all $x$ from  $T$ and insert the key-value pair $(x, \mathscr{D}^*[x])$ into dictionary $\mathscr{D}^*$, where $\mathscr{D}^*[x] = \mathscr{D}[x]$. Meanwhile, we add $card_T [x]$ copies of $x$ into $B$. Then we delect key-value pair $(x, \mathscr{D}[x])$ from $\mathscr{D}$.
				
				The algorithm continues to the next iteration. 
				\item Else if $card_T [x] > 2$, then $card_T [x] = 2k+num_T^x$, where $k \in \mathbb{N}_+$ and $num^T_x=1\  \text{or} \ 2$.  We first delete all $x$ from $T$, then add $k$ copies of $2x$ into $T$. Meanwhile, we add $num_T^x$ copies of $x$ into $B$.  We then update Dictionary $\mathscr{D}$, Dictionary $\mathscr{D}^*$ and heap $D$ as follows.
				\begin{itemize}
					\item  Consider the key-value pair $(x,\mathscr{D}[x])$ in Dictionary $\mathscr{D}$. Select any $num^T_x$ elements from  $\mathscr{D}[x]$, denote by $ELE$ the set of these selected elements. Define $\mathcal{L}_{2x}$ as follows: divide $\mathscr{D}[x] \backslash ELE$ into $k$ subgroups $Gr_1,Gr_2,\cdots,Gr_k$ such that each subgroup contains exact 2 elements, then $\mathcal{L}_{2x}: = \{\mathop{\dot{\cup}}\limits_{ele \in Gr_i} ele : i=1,2,\cdots,k\}$. 
					\item Insert the key-value pair $(x, \mathscr{D}^*[x])$ into dictionary $\mathscr{D}^*$, where $\mathscr{D}^*[x] = ELE$. If key $2x$ is already contained in heap $D$, increase the value corresponding to key $2x$ in heap $D$ by $k$, meanwhile update key-value pair $(2x,\mathscr{D}[2x])$ in $\mathscr{D}$ by inserting all elements in $\mathcal{L}_{2x}$ into $\mathscr{D}[2x]$. Else if $2x$ is not already in heap $D$, add $(2x,k)$ into heap $D$ and add key-value pair $(2x,\mathscr{D}[2x])$ into $\mathscr{D}$, where $\mathscr{D}[2x]=\mathcal{L}_{2x}$.  Delect key-value pair $(x, \mathscr{D}[x])$ from $\mathscr{D}$. 
				\end{itemize}

				The algorithm now continues to the next iteration.
			\end{itemize}
			Let $K_{\mathscr{D}}$ and $K_{\mathscr{D}^*}$ denote the collections of keys in $\mathscr{D}$ and $\mathscr{D}^*$, respectively. At the end of each iteration, we have the following observations:
			%Let $K_{\mathscr{D}}$ denote the collection of keys in dictionary $\mathscr{D}$ and let $K_{\mathscr{D}^*}$ denote the collection of keys in dictionary $\mathscr{D}^*$. At the end of each iteration, we have the following observations:
			\begin{enumerate}
				
				\item $S(B\dot{\cup}T) =S(A)$.%, which is guaranteed by equation $\{i \cdot z \ | \ 0 \le i \le 2k + 1\} = \{i \cdot  z + j \cdot 2z \ | \ 0 \le i \le 1 \ \text{and} \ 0 \le j \le k \}$. 
				\item Consider any key-value pair in $\mathscr{D}$, say $(x, \mathscr{D}[x])$. Each $ele\in \mathscr{D}[x]$ is a subset of $A$ satisfying $\Sigma(ele) = x$. 
				\item Consider any key-value pair in $\mathscr{D}^*$, say $(y, \mathscr{D}^*[y])$.  Each  $ele\in\mathscr{D}^*[y]$ is a subset of $A$ satisfying $\Sigma(ele) = 2^p \cdot (ele)^{\max}=y$, where $p \in \mathbb{N} \cap [0,\log_2 |A|]$.  Moreover, $\mathscr{D}^*[y]$ contains at most $2$ elements. 
				\item Elements in $(\mathop{\dot{\cup}}\limits_{x \in K_{\mathscr{D}}} \mathscr{D}[x])\dot{\cup} (\mathop{\dot{\cup}}\limits_{y \in K_{\mathscr{D}^*}} \mathscr{D}^*[y])$ form a division of $A$. That is, elements in $(\mathop{\dot{\cup}}\limits_{x \in K_{\mathscr{D}}} \mathscr{D}[x])\dot{\cup} (\mathop{\dot{\cup}}\limits_{y \in K_{\mathscr{D}^*}} \mathscr{D}^*[y])$ are subsets of $A$ and $\Sigma(A) = \sum\limits_{ele \in (\mathop{\dot{\cup}}\limits_{x \in K_{\mathscr{D}}} \mathscr{D}[x])\dot{\cup} (\mathop{\dot{\cup}}\limits_{y \in K_{\mathscr{D}^*}} \mathscr{D}^*[y])} \Sigma(ele)$. 
			\end{enumerate}
			
			By the time the iteration procedure stops, heap $D$, multiset $T$ and dictionary $\mathscr{D}$ are empty, we will obtain the final dictionary $\mathscr{D}^*$.  Then elements in $ \mathop{\dot{\cup}}\limits_{key \in K_{\mathscr{D}^*}}  \mathscr{D}^*[key] $ form a division of $A$. Notice that $B = \{\Sigma(ele): ele\in \mathop{\dot{\cup}}\limits_{key \in K_{\mathscr{D}}} \mathscr{D}[key] \}$. One can easily prove that $B$ satisfies (i).$S(A)=S(B)$; (ii).$|B| \leq |A|$ and $|B|=\OO(|set(A)|\log_2 |A|)$.

			%Consider each key-value pair $(key, \mathscr{D}^*[key])$ in $\mathscr{D}^*$, notice that $| \mathscr{D}^*[key]| \le 2$ and each $ele \in \mathscr{D}^*[key]$ is a subset of $A$ satisfying $\Sigma(ele) = key$.  Let $K_{\mathscr{D}^*}$ denote the set of keys in $\mathscr{D}^*$. One can easily prove that $\mathop{\dot{\cup}}\limits_{key \in K_{\mathscr{D}}} ( \mathop{\dot{\cup}}\limits_{ele \in \mathscr{D}[key] } ele ) = A$,  which implies that elemens in $ \mathop{\dot{\cup}}\limits_{key \in K_{\mathscr{D}}}  \mathscr{D}[key] $ form a division of $A$. Let $B = \{\Sigma(ele): ele\in \mathop{\dot{\cup}}\limits_{key \in K_{\mathscr{D}}} \mathscr{D}[key] \}$. Apparently, no element in $B$ has multiplicity exceeding two. Moreover, equation $\{i \cdot z \ | \ 0 \le i \le 2k + 1\} = \{i \cdot  z + j \cdot 2z \ | \ 0 \le i \le 1 \ \text{and} \ 0 \le j \le k \}$ guarantees that $S(B) =S(A)$.  Observe that each element in $B$ is of the form $2^p x$ where $0\le p \le  \log_2(|A|)$ and $x\in A$. This yields the bound on $|B|$, which is $\OO(|set(A)| \log |A|)$.  
			
			We estimate the overall processing time now.  Initialization takes $\OO(|A|)$ time. Notice that in each iteration, we will add one key-value pair into $\mathscr{D}^*$,  thus the number of iterations is $\OO(|K_{\mathscr{D}^*}|) = \OO(|set(A)| \log |A|)$.  In each iteration, it takes $O(\log |A|)$ time for updating $D$, $B$, $T$ and $\mathscr{D}^*$, then the total time for handling $D$, $B$, $T$ and $\mathscr{D}^*$ is $\OO(|set(A)| (\log |A|)^2)$. Moreover,  a careful analysis shows that the total time for updating $\mathscr{D}$ through all iterations is $\OO(\sum_{x\in A} card_A[x]\cdot  \log (card_A[x] )) = \OO(|A|\log|A|)$.   To summarize, the overall processing time is $\OO(|A| (\log|A|)^2)$.\qed 	\end{proof}}

	Back to our preprocessing procedure. %where $1\le j \le \eta $ and $1\le k \le |\Delta_j|$
	For each  $\bar{Y}^{j}_k$, according to  Lemma~\ref{lemma:dis_elem}, in $\OO(|\bar{Y}^{j}_k| (\log |\bar{Y}^{j}_k|)^2)$ time, we can divide $\bar{Y}^{j}_k$ into $m^k_j=\OO(|set(\bar{Y}^{j}_k)| \log |\bar{Y}^{j}_k|)$ groups  $\bar{Y}^{(j;1)}_k, \bar{Y}^{(j;2)}_k,\cdots, \bar{Y}^{(j;m^k_j)}_k$ such that
	$\bar{F}^{j}_k = \{ \Sigma(\bar{Y}^{(j;1)}_k),\Sigma(\bar{Y}^{(j;2)}_k), \cdots, \Sigma(\bar{Y}^{(j;m^k_j)}_k)\}$ is a multiset satisfying the followings:
	%such that $\Sigma(\bar{Y}^{(j;\iota)}_k) = 2^{p_{\iota}} \cdot (\bar{Y}^{(j;\iota)}_k)^{\min}$, where ${p_{\iota}} =pow(\frac{\Sigma(\bar{Y}^{(j;\iota)}_k)}{ (\bar{Y}^{(j;\iota)}_k)^{\min}})$ and $\iota = 1,2,\cdots,m^k_j$. Moreover,  $F^{j}_k = \{ \Sigma(\bar{Y}^{(j;1)}_k),\Sigma(\bar{Y}^{(j;2)}_k), \cdots, \Sigma(\bar{Y}^{(j;m^k_j)}_k)\}$ is a multiset satisfying the followings:
	\begin{subequations}
		\begin{align}
			&\Sigma(\bar{Y}^{(j;\iota)}_k) = 2^{p_{\iota}} \cdot (\bar{Y}^{(j;\iota)}_k)^{\max}, \ \text{where}\  {p_{\iota}} \in \mathbb{N} \cap [0, \log_2 |\bar{Y}^{j}_k |] \ \text{and} \ \iota = 1,2,\cdots,m^k_j. \label{eq:iv_set}\\
			&S(\bar{F}^{j}_k)=S(\bar{Y}^{j}_k) \text{ \ and \ } |\bar{F}^{j}_k|\le |\bar{Y}^{j}_k|; \label{eq:i_set} \\ 
			%&\text{elements in ${F^{j}_k}$ are of the form $2^{p} y$, where $y \in \bar{Y}^{j}_k$ and $0\le p \le \log_2 |\bar{Y}^{j}_k|$.} \label{eq:iv_set}\\
			&\text{no element in ${\bar{F}^{j}_k}$ has multiplicity exceeding two}; \label{eq:iii_set}
		\end{align}
	\end{subequations}
	According to \eqref{eq:iv_set}, for each number in $\bar{F}^{j}_k$, we can factorize it to the form $2^p y$, where $p \in \mathbb{N} \cap [0,\log_2 |\bar{Y}^{j}_k |]$ and $y \in \bar{Y}^{j}_k$. Let $F^j_k$ denote the set of all such factorized elements obtained from $\bar{F}^j_k$.

	Let $\mathcal{H}^{(j;k)}:=\{ \bar{Y}^{(j;1)}_k, \bar{Y}^{(j;2)}_k,\cdots, \bar{Y}^{(j;m^k_j)}_k\}$.  Notice that for each $F^{j}_k $, there is a one-to-one correspondence between $\mathcal{H}^{(j;k)}$ and $F^{j}_k$: $ele \leftrightarrow  \Sigma(ele) = 2^{pow(\frac{\Sigma(ele)}{ele^{\max}})}(ele)^{\max}$ for any $ele\in \mathcal{H}^{(j;k)}$.  For each $f\in F^j_k$, let  $\mathcal{H}^{(j;k)}_{f}:=\{ele \in \mathcal{H}^{(j;k)} \ : \ \Sigma(ele) = f\}$. We define the mapping $\Phi^{(j;k)}_{3}$ from $F^j_k$ to $\mathcal{H}^{(j;k)}$ as follows: given any $U^j_k \subset F^j_k$, for every $f\in set(U^j_k)$, mapping $\Phi^{(j;k)}_3$ returns any $card_{U^j_k} [f]$ elements in $\mathcal{H}^{(j;k)}_{f}$. Let $\Phi^{(j;k)}_3$ be the collection of all elements returned by $\Phi^{(j;k)}_3$ given $U^j_k$.

	%For each $F^{j}_k$,  note that $F^{j}_k = \{ \Sigma(\bar{Y}^{(j;1)}_k),\Sigma(\bar{Y}^{(j;2)}_k), \cdots, \Sigma(\bar{Y}^{(j;m^k_j)}_k)\}$ and $pow\left(\frac{\Sigma(\bar{Y}^{(j;\iota)}_k)}{ (\bar{Y}^{(j;\iota)}_k)^{\max}}\right) \le \log_2 (|\bar{Y}^{(j;\iota)}_k|)$ for each $\iota$. We first divide $F^{j}_k$ into 

	For each $F^{j}_k$, we first divide it into two groups $F^{(j;1)}_k$ and $F^{(j;2)}_k$, such that the values of elements in $F^{(j;i)}_k \ (i=1,2)$ are different from each other.  Then for each $F^{(j;i)}_k$, we divide it into $\log_2 |\bar{Y}^{j}_k|$ groups, denoted by $F^{(j;i;k)}_1 ,F^{(j;i;k)}_2,\cdots, F^{(j;i;k)}_{\log_2 |\bar{Y}^{j}_k|}$, such that $x \in  F^{(j;i;k)}_p$ if and only if $x$ has been factorized to the form $2^p y $ where $y \in\bar{Y}^{j}_k$.  It is easy to observe that $F^{(j;i;k)}_p \subset   [\frac{(1-\epsilon)2^{p+j-1}}{\epsilon^{2}}, \frac{(1+\epsilon)2^{p+j}}{\epsilon^{2}})$. Let $F: = \mathop{\dot{\cup}}\limits_{1\le j \le \eta; 1\le k \le |\Delta_j|} F^{j}_k =  \mathop{\dot{\cup}}\limits_{ 0\le p\le \log_2 |\bar{Y}^{j}_k| \atop  i=1,2;1\le j \le \eta; 1\le k \le |\Delta_j| } F^{(j;i;k)}_p $.  Notice that 
	$|F| \le |\dot{\cup}^\eta_{j=1} \bar{Y}_j| \le |X|$. Moreover, note that $F$ is an equivalent multiset of $\dot{\cup}^\eta_{j=1} \bar{Y}_j$, i.e., $S(F)=S(\dot{\cup}^\eta_{j=1} \bar{Y}_j)$. Thus $\Sigma(F) = \Sigma(\dot{\cup} ^{\eta}_{j=1}\bar{Y}_j) \le \frac{4(1+\epsilon)\Sigma(X)}{\epsilon^3 t}$
	and the optimal objective value of SUBSET SUM instance $(F, \hat{t})$ is at least $\frac{4(1-\epsilon)}{\epsilon^3\cdot t} (OPT-\frac{\epsilon t}{2})$. 
	
	Observe that the total time to obtain $F^{(j;i;k)}_p$'s from $\bar{Y}^{j}_k$ is $\OO(|\bar{Y}^{j}_k| (\log |\bar{Y}^{j}_k|)^2)$. Thus the total processing time of Step 4 is $\OO(\mathop{\Sigma}\limits_{1\le j\le \eta; 1\le k \le |\Delta_j|} (|\bar{Y}^{j}_k| (\log |\bar{Y}^{j}_k|)^2)) = \OO(|X| (\log |X|)^2)$.
	
	\paragraph{\textbf{Modified instance after preprocessing.}} To summarize, we have reduced the instance $(X,t)$ to a modified instance $(Y,\hat{t})$, where $Y$ and $\hat{t}$ satisfying the following conditions:
	\begin{itemize}
		\item [(\uppercase\expandafter{\romannumeral1})] $\Sigma(F) \le \frac{4(1+\epsilon)\Sigma(X)}{\epsilon^3 t}$ and $\hat{t}= \frac{4(1+\epsilon)}{\epsilon^3}$.
		\item [(\uppercase\expandafter{\romannumeral2})] The optimal objective value of $(F,\hat{t})$ is at least $\frac{4(1-\epsilon)}{\epsilon^3\cdot t}\cdot (OPT-\frac{\epsilon t}{2})$.
		\item [(\uppercase\expandafter{\romannumeral3})]  $|F| \le |X|$ and $F$ has been divided into $\OO((\log |X|)^2 (\log \frac{1}{\epsilon})^{\OO(d)})$ subgroups: $F^{(j;i;k)}_p$'s. 
		%$$\{F^{(j;i;k)}_p  \ |  \  1\le p \le \OO(\log |X|), 1\le j \le \OO(\log \frac{1}{\epsilon}),  i=1,2 \text{\ and \ } 1\le k \le \OO(\log|X|  \cdot (\log \frac{1}{\epsilon})^{\OO(d)})\}.$$  
		\item [(\uppercase\expandafter{\romannumeral4})]
		Each subgroup $F^{(j;i;k)}_p$ satisfies the followings:
		\begin{enumerate}
			\item[(a)] Elements in $F^{(j;i;k)}_p$ are different from each other;
			\item[(b)]$F^{(j;i;k)}_p \subset  [\frac{(1-\epsilon)2^{p+j-1}}{\epsilon^{2}}, \frac{(1+\epsilon)2^{p+j}}{\epsilon^{2}})$;
			\item[(c)] $F^{(j;i;k)}_p = 2^p \rho^j_k \overline{F}^{(j;i;k)}_p$, where $\overline{F}^{(j;i;k)}_p \subset \mathbb{N}_{+}\cap[\frac{1}{4\epsilon},\frac{1}{\epsilon}]$
			and $\overline{F}^{(j;i;k)}_p$ is a set of $(\epsilon,d,d-1)$-{smooth numbers}, that is, every element in $\overline{F}^{(j;i;k)}_p$ has been factorized as $h_1 h_2 \cdots h_{d}$, where $h_{d}\in \mathbb{N}_{+}\cap [1,2\epsilon^{-\frac{1}{d}}]$ and if $d \ge 2$ then $h_{i} \in \mathbb{N}_{+} \cap [\frac{1}{2}\epsilon^{-\frac{1}{d}},2{\epsilon^{-\frac{1}{d}}}]$ for $i=1,2,\cdots,d-1$. 
			
		\end{enumerate}
	\end{itemize}   
	
	Till now, we have completed the preprocessing procedure for modifying SUBSET SUM instance $(X,t)$. The total processing time is $\OO((|X|+\frac{1}{\epsilon})(\log |X|)^2 (\log \frac{1}{\epsilon})^{\OO(d)}).$ This accomplishes the first half (i.e., item (i)) of Lemma~\ref{obs:obs_pre}. In the following, we will present the oracle for backtracking from $F$ to $X$, which is the second half (item (ii)) of Lemma~\ref{obs:obs_pre}.
	
	\paragraph{\textbf{Oracle for backtracking from $F$ to $X$.}} We now present the oracle $\textbf{Ora}^{X}_{F}$ for backtracking from $F$ to $X$.  Given $U^{(j;i;k)}_p \subset F^{(j;i;k)}_p$ for every $F^{(j;i;k)}_p$, oracle $\textbf{Ora}^{X}_{F}$ works as follows: %where $ j \in \mathbb{N} \cap [1, \eta]$, $k \in \mathbb{N} \cap [1,|\delta_j|]$, $i=1,2$ and $ p \in \mathbb{N} \cap [0,  \log_2 |  \log_2 (|\bar{Y}^{j}_k|)|]$
	\begin{itemize}
		\item Let $U^{j}_k:= \mathop{\dot{\cup}}\limits_{p} \left( \mathop{\dot{\cup}}\limits_i  U^{(j;i;k)}_p \right)$ for every $j$ and $k$.  %Consider each $U^{j}_k$.
		Note that $U^{j}_k \subset F^j_k$ and $\Phi^{(j;k)}_{3}$ is a mapping from $F^j_k$ to $\mathcal{H}^{(j;k)}$. $\textbf{Ora}^{X}_{F}$ first uses $\Phi^{(j;k)}_{3}$ to obtain $\Phi^{(j;k)}_{3} (U^{j}_k)$.  Let $V^j_k = \mathop{\dot{\cup}}\limits_{ele \in \Phi^{(j;k)}_{3} (U^{j}_k)} ele$. Observe that $\Sigma(V^j_k) = \Sigma(U^j_k) =\sum_{p} \sum_i (\Sigma(U^{(j;i;k)}_p))$ and $V^j_k  \subset  \bar{Y}^j_k$.
		\item Let $V_j = \mathop{\dot{\cup}}\limits_{k} (V^j_k )$ for every $j$. Notice that %\mathop{\dot{\cup}}\limits_{k} (\mathop{\dot{\cup}}\limits_{ele \in \Phi^{(j;k)}_{3} [U^{j}_k] } ele ) \subset \bar{Y}_j
		$V_j\subset \bar{Y}_j$ and $\Phi^j_2$ is a mapping from $ \bar{Y}_j$ to ${Y}_j$. Then $\textbf{Ora}^{X}_{F}$ uses $\Phi^j_2$ to obtain $\Phi^j_2 (V_j)$. %$\Phi^j_2 \left[\mathop{\dot{\cup}}\limits_{k} (\mathop{\dot{\cup}}\limits_{ele \in \Phi^{(j;k)}_{3} [U^{j}_k] } ele)\right]$. 
		Let $W=\mathop{\dot{\cup}}\limits_{j} \Phi^j_2 (V_j)$. Observe that $|\Sigma(W)- \sum_{j}\Sigma(V_j)| = |\sum_{j} \Sigma(\Phi^j_2 (V_j))- \sum_{j}\Sigma(V_j)| \le \sum_j \epsilon \Sigma(\Phi^j_2 (V_j)) = \epsilon \Sigma(W)$ and $W \subset \bar{Z}$.
		\item  %Note that $\Phi^j_2 [\mathop{\dot{\cup}}\limits_{k} (V^j_k )] \subset Y_j$ and $\mathop{\dot{\cup}}\limits_{j} \Phi^j_2 [\mathop{\dot{\cup}}\limits_{k} (V^j_k )] \subset \bar{Z}$. 
		Recall that $\Phi_1$ is a mapping from $\bar{Z}$ to $\mathcal{L}$. Finally, $\textbf{Ora}^{X}_{F}$ uses $\Phi_1$ to obtain $\Phi_1(W)$ and returns  $\mathop{\dot{\cup}}\limits_{ele \in \Phi_1(W)}ele$. Note that $\mathop{\dot{\cup}}\limits_{ele \in \Phi_1(W)}ele \subset X$ and we have $\Sigma(\mathop{\dot{\cup}}\limits_{ele \in \Phi_1(W)}ele) = \frac{4}{\epsilon^3 t}\Sigma(W)$.
		
	\end{itemize}
	The above backtracking procedure only takes linear time.  Let $F'$ be the multiset-union of all $U^{(j;i;k)}_p$'s.  To summarize, we have $|\frac{4}{\epsilon^3 t}\Sigma(F')-\Sigma(\mathop{\dot{\cup}}\limits_{ele \in \Phi_1(W)}ele)| \le \epsilon \Sigma(\mathop{\dot{\cup}}\limits_{ele \in \Phi_1(W)}ele)$.  
	
	Till now, we complete the proof of Lemma~\ref{obs:obs_pre}.
	
	So far we have obtained all the prerequisites. In the subsequent 3 sections, we will present our main results in this paper.

	\section{An $\tilde\OO(n+\epsilon^{-\frac{5}{4}})$-time FPTAS for PARTITION.}\label{sec:5/4_main}
	The goal of this section is to prove the following theorem.
	
	\begin{theorem}\label{the:5/4-fptas}
		There is an $\tilde\OO(n+\epsilon^{-\frac{5}{4}})$ deterministic FPTAS for PARTITION.
	\end{theorem}
	
	Given a multiset $X = \{x_1,x_2,\cdots,x_n\}\subset \mathbb{N}$, PARTITION is the same as SUBSET SUM on target $t = \frac{\Sigma(X)}{2}$. Let $OPT$ be the optimal objective value of SUBSET SUM instance $(X,\frac{\Sigma(X)}{2})$. By Lemma~\ref{lemma:prea}, %{we may assume $X\subset (0, t)$} and 
	we may assume $OPT \ge \frac{\Sigma(X)}{4}=\Theta(\Sigma(X))$. Notice that once we have found a subset $Y \subset X$ satisfying 
	$|\Sigma(Y)-OPT|\le \tilde\OO(\epsilon)\Sigma(X)$, then one of $Y$ and $X \backslash Y$ is a $(1-\tilde\OO(\epsilon))$-approximation solution of $(X,\frac{\Sigma(X)}{2})$, thus the following lemma implies Theorem~\ref{the:5/4-fptas} directly.

	\begin{lemma}\label{lemma:o(e)-apx-set}
		Given a multiset $X= \{x_1,x_2,\cdots,x_n\}\subset \mathbb{N}$, in $\tilde\OO(n+\epsilon^{-\frac{5}{4}})$ processing time, we can 
		\begin{itemize}
			\item[(i).] Compute an $\tilde\OO(\epsilon)$-approximate set with cardinality of $\OO({\epsilon}^{-1})$ for $S(X)$;
			\item[(ii).] Meanwhile build an %$\tilde\OO(\epsilon)$-approximate
			$\tilde\OO(n+\epsilon^{-\frac{5}{4}})$-time oracle for backtracking from this approximate set to $X$.
	\end{itemize}  \end{lemma}
	
	Note that an $\tilde\OO(\epsilon)$-approximate set for $S(X)$ admits an additive error of $\OO(\epsilon\Sigma(X))$, which solves PARTITION but not SUBSET SUM with target $t\ll \Sigma(X)$. Nevertheless, Lemma~\ref{lemma:o(e)-apx-set} also implies the following corollary.
	
	\begin{corollary}
		There is an $\tilde\OO(n+\epsilon^{-\frac{5}{4}})$ deterministic weak $(1-\epsilon)$-approximation algorithm for SUBSET SUM if $t=\Theta(\Sigma(X))$, i.e., 
		the target is some constant fraction of the total summation of elements.  
	\end{corollary}

	Given a SUBSET SUM instance $(X,t)$ where $t=\Theta(\Sigma(X))$ and let $OPT$ be the optimal objective value of $(X,t)$. Let $d \in \mathbb{N}_{+}$ be a constant to be fixed later (in particular, we will choose $d=12$).  Recall Lemma~\ref{obs:obs_pre}, in $\OO((|X|+\frac{1}{\epsilon})(\log |X|)^2 (\log \frac{1}{\epsilon})^{\OO(d)}) = \tilde\OO(|X|+\frac{1}{\epsilon})$ time, we can reduce $(X,t)$ to a SUBSET SUM instance $(F,\hat{t})$ satisfying conditions $(\uppercase\expandafter{\romannumeral1})(\uppercase\expandafter{\romannumeral2})(\uppercase\expandafter{\romannumeral2})(\uppercase\expandafter{\romannumeral4})$ and meanwhile build an oracle for backtracking from $F$ to $X$ (see Lemma~\ref{obs:obs_pre} in Section~\ref{sec:proce}). Condition $(\uppercase\expandafter{\romannumeral2})$ claims that the optimal objective value of $(F,\hat{t})$ is at least $\frac{4(1-\epsilon)}{\epsilon^3\cdot t}\cdot (OPT-\frac{\epsilon t}{2})$, recall that $OPT \ge t/2$, then given any weak $(1-\tilde\OO(\epsilon))$-approximation of $(F,\hat{t})$, in linear time, the oracle will return a weak $(1-\tilde\OO(\epsilon))$-approximation of $(X,t)$. It thus suffices to consider $(F,\hat{t})$. %{\color{red}Condition $(\uppercase\expandafter{\romannumeral2})$ and the backtracking oracle guarantee that towards proving Lemma~\ref{lemma:o(e)-apx-set}, it is sufficient to consider $(F,\hat{t})$ (here needs copying Condition $(\uppercase\expandafter{\romannumeral2})$ and explain a bit )}. 
	That is, Lemma~\ref{lemma:o(e)-apx-set}, and hence Theorem~\ref{the:5/4-fptas}, follow from the following Lemma~\ref{lemma:5/4-sub-pro}.

	\begin{lemma}\label{lemma:5/4-sub-pro}
		Given any SUBSET SUM instance $(X,t)$, where $t=\Theta(\Sigma(X))$ and the optimal objective value of $(X,t)$ is at least ${t}/{2}$. Let $d \ge 12$ be an integer divisible by 4 and let $(F,\hat{t})$ be a modified instance returned by Lemma~\ref{obs:obs_pre}, where $F$ satisfies the followings:
		\begin{enumerate}
			\item %[(\uppercase\expandafter{\romannumeral1})]
			$\Sigma(F)=\Theta(\epsilon^{-3})$.
			\item %[(\uppercase\expandafter{\romannumeral2})]
			$|F| \le |X|$ and $F$ has been divided into $\OO((\log |X|)^2 (\log \frac{1}{\epsilon})^{\OO(d)})$ subgroups: $F^{(j;i;k)}_p$'s. 
			\item%[(\uppercase\expandafter{\romannumeral2})]
			Each $F^{(j;i;k)}_p$ satisfies the followings:
			\begin{itemize}
				\item Elements in $F^{(j;i;k)}_p$ are distinct;
				\item $F^{(j;i;k)}_p \subset  [\frac{(1-\epsilon)2^{p+j-1}}{\epsilon^{2}}, \frac{(1+\epsilon)2^{p+j}}{\epsilon^{2}})$;
				\item $F^{(j;i;k)}_p = 2^p \rho^j_k \overline{F}^{(j;i;k)}_p$, where $\overline{F}^{(j;i;k)}_p \subset \mathbb{N}_{+}\cap[\frac{1}{4\epsilon},\frac{1}{\epsilon}]$
				and $\overline{F}^{(j;i;k)}_p$ is a set of $(\epsilon,d,d-1)$-{smooth numbers}, that is, every element in $\overline{F}^{(j;i;k)}_p$ has been factorized as $h_1 h_2 \cdots h_{d}$, where $h_{d}\in \mathbb{N}_{+}\cap [1,2\epsilon^{-\frac{1}{d}}]$ and $h_{i} \in \mathbb{N}_{+} \cap [\frac{1}{2}\epsilon^{-\frac{1}{d}},2{\epsilon^{-\frac{1}{d}}}]$ for $i=1,2,\cdots,d-1$. 
			\end{itemize}
		\end{enumerate}
		Then in $\tilde\OO(\epsilon^{-\frac{5}{4}})$ processing time, we can 
		\begin{itemize}
			\item[(i).] Compute an $\tilde\OO(\epsilon)$-approximate set $C$ with cardinality of $\OO(\frac{1}{\epsilon})$ for $S(F)$.
			\item[(ii).] Meanwhile build an			$\tilde\OO(\epsilon^{-\frac{5}{4}})$-time oracle for backtracking from $C$ to $F$. Here $F$ is the union of $F^{(j;i;k)}_p$'s and the oracle actually works as follows: given any $c\in C$, in  $\tilde\OO(\epsilon^{-\frac{5}{4}})$ processing time, the oracle will return $U^{(j;i;k)}_p \subset F^{(j;i;k)}_p$ for every $F^{(j;i;k)}_p$. Let $F'$ be the multiset-union of all $U^{(j;i;k)}_p$'s, we have $F'\subset F$ and $|\Sigma(F')-c|\le \tilde\OO(\epsilon)\Sigma(F)$.
		\end{itemize}
	\end{lemma}
	The rest of this section is dedicated to proving Lemma~\ref{lemma:5/4-sub-pro}. %According to Lemma~\ref{obs:obs_pre}, $F$ is divided into $\OO((\log |X|)^2 (\log \frac{1}{\epsilon})^{\OO(d)})$ groups: $F_{p}^{(j;i;k)}$'s, and {each $F_{p}^{(j;i;k)}$ satisfies condition $(\uppercase\expandafter{\romannumeral4})$ (see Lemma~\ref{obs:obs_pre} in Section~\ref{sec:proce}). Moreover, we have $\Sigma(F_{p}^{(j;i;k)})\le \Sigma(F)= \Theta(\hat{t}) = \Theta(\epsilon^{-3})$. 
		Recall Corollary~\ref{coro:tree-fashion_}, which allows us to build the approximation set of the union of $F^{(j;i;k)}_p$'s from the approximate set of each $F^{(j;i;k)}_p$. Towards proving Lemma~\ref{lemma:5/4-sub-pro}, we only need to derive an algorithm that can solve the following  \textbf{problem}-$\mathscr{P}$ in $\tilde\OO(\epsilon^{-\frac{5}{4}})$ time: 
		
		\textbf{problem}-$\mathscr{P}$: for every $F_{p}^{(j;i;k)}$, within an additive error of $\tilde\OO(\epsilon\hat{t})$, 
		compute an $\tilde\OO(\epsilon)$-approximate set with cardinality of $\tilde\OO(\epsilon^{-\frac{5}{4}})$ for $S(F_{p}^{(j;i;k)})$ and build an $\tilde\OO(\epsilon^{-\frac{5}{4}})$-time oracle for backtracking from this approximate set to $F_{p}^{(j;i;k)}$. 
		
		To solve \textbf{problem}-$\mathscr{P}$, we only need to prove the following lemma~\ref{lemma:small_gro_p}, where $M$ represents an arbitrary $F_{p}^{(j;i;k)}$. 
		%. Lemma~\ref{lemma:tree-fashion_} guarantees that we can reduce the problem on $F$ to problems on smaller groups: $F^{(i;j;k)}$'s.
		
		%Now back to the proof of Lemma~\ref{lemma:5/4-sub-pro}. Recall that $F$ has been divided into $\OO((\log n)^2 (\log \frac{1}{\epsilon})^{\OO(d)})$ % = \tilde\OO(1) disjoint groups $F^{(i;j;k)}$'s,

		\begin{lemma}\label{lemma:small_gro_p}
			Given $M\subset \mathbb{R}_{\ge 0}$ satisfying the following conditions:
			\begin{enumerate}
				\item $\Sigma(M) = \OO(\epsilon^{-3})$;
				\item Each element in $M$ is distinct;
				\item $M  = \beta \tilde{M}$, where $\beta$ and $\tilde{M}$ satisfy the followings: 
				\begin{itemize}
					\item $\beta \in \Theta(\frac{1}{\epsilon^{1+\lambda}})$ and $\beta \tilde{M} \subset [\frac{1-\epsilon}{\epsilon^{2+\lambda}}, \frac{2(1+\epsilon)}{\epsilon^{2+\lambda}}]$, where $\lambda \ge 0$;
					%$\beta \in [\frac{1-\epsilon}{\epsilon^{1+\lambda}},\frac{8(1+\epsilon)}{\epsilon^{1+\lambda}}]$, where $\lambda \ge 0$;
					\item  $\tilde{M}\subset \mathbb{N}_{+}\cap [\frac{1}{4\epsilon},\frac{1}{\epsilon}]$ and $\tilde{M}$ is a set of $(\epsilon,d,d-1)$-{smooth numbers}, that is, every element in $\tilde{M}$ has been factorized as $h_1 h_2 \cdots h_{d}$, where $h_{d}\in \mathbb{N}_{+}\cap [1,2\epsilon^{-\frac{1}{d}}]$ and $h_{i} \in \mathbb{N}_{+} \cap [\frac{1}{2}\epsilon^{-\frac{1}{d}},2{\epsilon^{-\frac{1}{d}}}]$ for $i=1,2,\cdots,d-1$. 
				\end{itemize}
				%elements in $M$ are of the form $\beta  h_1 h_2 \cdots h_d$ where $\beta $ and $h_i$'s satisfying the following: (i) $\beta  h_1 h_2 \cdots h_d \in [\frac{1-\epsilon}{\epsilon^{2+\lambda}}, \frac{2(1+\epsilon)}{\epsilon^{2+\lambda}} )$ where $\lambda \ge 0$; % and $\beta \in \Theta(\frac{1}{\epsilon^{1+\lambda}})$; 
				%(ii) $h_1 h_2 \cdots h_d \in \mathbb{N}\cap [\frac{1}{4\epsilon},\frac{1}{\epsilon}]$ and $h_i \in \mathbb{N}\cap [\frac{1}{2\epsilon^{\frac{1}{d}}},\frac{1}{\epsilon^{\frac{1}{d}}}]$ for $i = 1,2,\cdots,d$.
			\end{enumerate}
			Here $d \ge 12$ is an integer divisible by 4 (e.g., $d = 12$). Then in $\tilde\OO(\epsilon^{-\frac{5}{4}})$ processing time, within an additive error of $\tilde\OO(\epsilon^{-2})$, we can compute an $\tilde\OO(\epsilon)$-approximate set with cardinality of $\OO({\epsilon}^{-1})$ for $S(M)$ and build an $\tilde\OO(\epsilon^{-\frac{5}{4}})$-time oracle for backtracking from this approximate set to $M$. 
		\end{lemma}
		\noindent\textbf{Remark.} Here an additive error of $\tilde\OO(\epsilon^{-2})$ is acceptable since $\hat{t} = \Theta(\epsilon^{-3})$ and $F$ is only divided into polylogarithmic groups, whereas an additive $\tilde\OO(\epsilon^{-2})$ error per group gives overall $\tilde\OO(\epsilon\hat{t})$ error.%It is easy to observe that $M$ contains at most $\OO(1/\epsilon)$ elements. 

		%\noindent\textbf{Remark.} It is easy to observe that (i) $\beta \in \Theta(\frac{1}{\epsilon^{1+\lambda}})$; (ii) $M$ contains at most $\OO(1/\epsilon)$ elements.

		We give a very high-level description of the proof. To obtain an efficient algorithm, our main tool is Lemma~\ref{lemma:e_apx_sm}, whose running time depends on the summation of the input numbers, and thus it is important to reduce this sum. Consider $M$, since all of its elements share $\beta$ as the common divisor, we can restrict our attention to $\tilde{M}$. Note that $\beta=\Theta(\frac{1}{\epsilon^{1+\lambda}})$. If $\lambda \ge 1/2$, and hence $\beta$ is large, then $\Sigma(\tilde{M})$ is small, and we can apply Lemma~\ref{lemma:e_apx_sm} directly. Otherwise, $\beta$ is small, then $\Sigma(\tilde{M})=\OO(\frac{1}{\epsilon^{2-\lambda}})$ is large, implying that $\tilde{M}$ consists of many smooth numbers. In this case, we can exploit the additive combinatoric result from \cite{DBLP:journals/siamcomp/GalilM91},  which roughly says that if there are sufficiently many distinct integers, then SUBSET SUM can be solved efficiently if the target $t$ is in the ``medium" range, that is, $t$ is close to half of the total sum of integers (see Theorem~\ref{the:dense} for a formal description of the additive combinatoric result. This result has also been leveraged before by Mucha et al.~\cite{DBLP:conf/soda/MuchaW019}). Then what if $t$ is out of the medium range? In this case, $t$ is either very small or very large, and by symmetry it suffices to consider the case when $t$ is small. The crucial observation is that a small target value can only be the sum of a few numbers in $\tilde{M}$. Hence, we may adopt a ``coarse" rounding to re-round the numbers, that is, instead of only introducing $\OO(\epsilon)$-multiplicative error to each input number, we may introduce a larger error. Although per number the multiplicative error is $\Omega(\epsilon)$, but since only $o(\frac{1}{\epsilon})$ numbers will be selected, we can guarantee that the summation of selected numbers gives additive $\tilde\OO(\epsilon^{-2})$ error in total by a careful parameterized analysis.

		In the following, we will give detailed proof of Lemma~\ref{lemma:small_gro_p}. Note that 
		$M = \beta \tilde{M} \subset [\frac{1-\epsilon}{\epsilon^{2+\lambda}}, \frac{2(1+\epsilon)}{\epsilon^{2+\lambda}}]$ and $\Sigma(M) = \OO(\epsilon^{-3})$, we have $|M| = |\tilde{M}| = \OO(\epsilon^{-1})$. We call ${M}$ a large-value group if $\lambda \ge 1/2$
		, otherwise we call ${M}$ a small-value group. Large-value groups and small-value groups will be handled separately.
		
		%. Lemma~\ref{lemma:tree-fashion_} guarantees that we can reduce the problem on $F$ to problems on smaller groups: $F^{(i;j;k)}$'s.
		%Now back to the proof of Lemma~\ref{lemma:5/4-sub-pro}. Recall that $F$ has been divided into $\OO((\log n)^2 (\log \frac{1}{\epsilon})^{\OO(d)})$ % = \tilde\OO(1) disjoint groups $F^{(i;j;k)}$'s,

		\subsection{Handling Large-Value Group.} 
		
		We first consider the case that $M \subset  [\frac{(1-\epsilon)}{\epsilon^{2+\lambda}}, \frac{2(1+\epsilon)}{\epsilon^{2+\lambda}}]$ where $\lambda \ge 1/2$. 
		
		Notice that elements in $\tilde{M}$ are $(\epsilon,d,d-1)$-smooth numbers. According to Lemma~\ref{lemma:e_apx_sm}, for any $k \in \mathbb{N} \cap [0,d-1]$, in $T =\OO(d \cdot |\tilde{M}|+\Sigma(\tilde{M}) {\epsilon^{\frac{k}{d}}} \cdot \log (|\tilde{M}|)\cdot \log (\Sigma(\tilde{M}) {\epsilon^{\frac{k}{d}}})+ {\epsilon^{-(1+\frac{k}{d})}}\log \frac{1}{\epsilon})$ processing time, we can compute an $\OO(\epsilon(\log \frac{1}{\epsilon})^{k})$-approximate set with cardinality of $\OO(\frac{1}{\epsilon})$ for $S(\tilde{M})$. Denote by $C_{\tilde{M}}$ this approximate set, in the meantime, we have built a $T^{(1)}$-time oracle for backtracking from $C_{\tilde{M}}$ to $\tilde{M}$, where $T^{(1)} = \OO(\Sigma(\tilde{M}) {\epsilon^{\frac{k}{d}}} \cdot \log (|\tilde{M}|)\cdot \log (\Sigma(\tilde{M}) {\epsilon^{\frac{k}{d}}})+ {\epsilon^{-(1+\frac{k}{d})}}\log \frac{1}{\epsilon})$. It is easy to observe that $\beta C_{\tilde{M}}$ is an $\OO(\epsilon(\log \frac{1}{\epsilon})^{k})$-approximate set with cardinality of $\OO( \frac{1}{\epsilon})$ for $S({M})$ and the oracle for backtracking from from $C_{\tilde{M}}$ to $\tilde{M}$ directly yields an $\OO(T^{(1)})$-time oracle for backtracking from $\beta C_{\tilde{M}}$ to $M$. Note that $|M| = \OO(\frac{1}{\epsilon})$ and $\Sigma(\tilde{M}) = \OO(\epsilon^{-\frac{3}{2}})$ if $\lambda \ge \frac{1}{2}$. Recall that $d\ge 12$ is an integer divisible by 4, by fixing $k$ to $\frac{d}{4}$, we have $\OO(\epsilon(\log \frac{1}{\epsilon})^{k})= \tilde\OO(\frac{1}{\epsilon})$, $T = \OO(\frac{d}{\epsilon}+ \epsilon^{-\frac{5}{4}} (\log \frac{1}{\epsilon})^2) = \tilde\OO(\epsilon^{-\frac{5}{4}})$ and $T^{(1)} = \OO(\epsilon^{-\frac{5}{4}} (\log \frac{1}{\epsilon})^2)=\tilde\OO(\epsilon^{-\frac{5}{4}})$. Thus Lemma~\ref{lemma:small_gro_p} has been proved for case that $M$ is a large-value group. 
		
		It remains to prove Lemma~\ref{lemma:small_gro_p} for the case that $M$ is a small-value group, which is the task of the following Subsection~\ref{sec:handle_small-gro-}.

		%In particular, combined with the discussion in the following Subsection~\ref{sec:handle_small-gro-}, $d$ will be fixed to $12$, it follows that $T^{(1)}=\tilde\OO(\epsilon^{-\frac{5}{4}})$ and $T^(1) = \tilde\OO(\epsilon^{-\frac{5}{4}})$, thus Lemma~\ref{ }
		
		%\subsection{Handling Large-Value Group}\label{subset:lar-va-g}
		%In this section, we aim to prove Lemma~\ref{lemma:small_gro_p} for the case that $M \subset [\frac{(1-\epsilon)}{\epsilon^{2+\lambda}}, \frac{2(1+\epsilon)}{\epsilon^{2+\lambda}}]$ where $\lambda \ge 1/2$. 

		\subsection{Handling Small-Value Group}\label{sec:handle_small-gro-}
		
		The goal of this section is to prove Lemma~\ref{lemma:small_gro_p} for the case that $M \subset [\frac{1-\epsilon}{\epsilon^{2+\lambda}}, \frac{2(1+\epsilon)}{\epsilon^{2+\lambda}}]$ where $\lambda < 1/2$.  
		
		Note that $M$ contains at most $\OO(\frac{1}{\epsilon})$ elements and these elements are different from each other, we call $M$ a \textit{dense} group if $|M| =  \Omega(\epsilon^{-\frac{1}{2}} \log \frac{1}{\epsilon})$, otherwise, we call $M$ a \textit{sparse} group. When $M$ is sparse, i.e., $|M| = \OO(\epsilon^{-\frac{1}{2}} \log \frac{1}{\epsilon})$, note that $\tilde{M} \subset [\frac{1}{4\epsilon}, \frac{1}{\epsilon}]$, we have $\Sigma(\tilde{M}) \le |M|\cdot \frac{1}{\epsilon}= \tilde\OO(\epsilon^{-\frac{3}{2}})$, same as the discussion for large-value group, Lemma~\ref{lemma:e_apx_sm} guarantees the correctness of Lemma~\ref{lemma:small_gro_p}. Thus we only need to consider the dense group. 
		
		In the following, we assume $M \subset  [\frac{1-\epsilon}{\epsilon^{2+\lambda}}, \frac{2(1+\epsilon)}{\epsilon^{2+\lambda}}]$ with $\lambda < \frac{1}{2}$ and $|M| = \Omega(\epsilon^{-\frac{1}{2}}\log \frac{1}{\epsilon})$. Define $ L:=\frac{100\cdot \Sigma (M) \cdot \sqrt{\frac{1}{\epsilon}} \log \frac{1}{\epsilon}}{|M|}$. We approximate $S(M)$ and build an oracle for backtracking in the following three steps: (1). handle $S(M)\cap (L, \Sigma(M)-L)$; (2). handle $S(M) \cap [0,L]$ and $S(M) \cap [\Sigma(M)-L,\Sigma(M)]$; (3). handle $S(M).$
		
		\paragraph{Step 1: Handling $S(M) \cap (L,\Sigma(M)-L)$.} { }~\\
		
		In this step, we aim to prove the following claim.
		\begin{claim} \label{obs:mid_obs}
			In $T_{mid}=\OO(\frac{1}{\epsilon} \log \frac{1}{\epsilon})$ processing time, we can
			\begin{itemize} 
				\item[(i).] Compute a subset $\Sigma_{mid} \subset  S(M) \cap (L,\Sigma(M)-L) $ satisfying the following conditions: (i). $|\Sigma_{mid}| = \OO(\frac{1}{\epsilon})$; (ii). given any $s \in S(M) \cap (L,\Sigma(M)-L)$, there exists $s'\in \Sigma_{mid}$ such that $|s'-s|\le \epsilon \Sigma(M)$. 
				\item[(ii).] Meanwhile build an $\OO(1)$-time oracle for backtracking from $\Sigma_{mid}$ to $M$. That is, given any $c\in \Sigma_{mid}$, in $\OO(1)$ time, the oracle will return $M_c \subset M$ such that $\Sigma(M_c) = c$.
			\end{itemize}
			
		\end{claim}
		
		\iffalse{\begin{observation}\label{obs:mid_obs}
				Given any $s\in S(M) \cap (L,\Sigma(M)-L)$, there exists $s'\in \Sigma_{mid}$ such that $|s'-s|\le \epsilon \Sigma(M)$. Moreover, we have built a structure such that given any $c\in \Sigma_{mid}$, in constant time, we can determine a subset $M'\subset M$ satisfying $c=\Sigma(M')$. 
		\end{observation}}\fi
		
		Before proving Claim~\ref{obs:mid_obs}, we first import the following theorem, which is derived by Galil and Margalit~\cite{DBLP:journals/siamcomp/GalilM91}.
		\begin{theorem}[CF. Theorem 6.1 from \cite{DBLP:journals/siamcomp/GalilM91}]\label{the:dense} Let $Z$
			be a set of $m$ distinct numbers in the interval $(0,l]$ such
			that
			$$
			m>1000 \cdot \sqrt{l} \log l.
			$$
			and let $L:=\frac{100 \cdot \sum(Z) \sqrt{l} \log l}{m}$.
			Then in $\mathcal{O}\left(m+(\frac{l}{m} \log l)^{2}+ \frac{\sum(Z)\sqrt{l}(\log l)^{2} }{m^{2}}\right)$ preprocessing time we can build a structure which allows us to solve the subset-sum problem for any given integer $t$ in the interval $\left(L, \sum(Z)-L\right)$. Solving means finding a subset $B \subset Z$ such that $\sum (B) \leq t$ and there is no subset $C \subset Z$ such that $\sum{(B)}<\sum{(C)} \leq t$. An optimal subset $B$ is built in $O(\log l)$ time per target number. 
		\end{theorem}
		
		Now we are ready to prove Claim~\ref{obs:mid_obs}.
		\begin{proof}[Proof of Claim~\ref{obs:mid_obs}]
			Consider $\tilde{M} = (M)/\beta$. Recall that elements in $\tilde{M}$ are different from each other and $|\tilde{M}|=|M|=\Omega(\epsilon^{-\frac{1}{2}}\log \frac{1}{\epsilon})$, moreover, $\tilde{M} \subset \mathbb{N}\cap [\frac{1}{4\epsilon}, \frac{1}{\epsilon}]$. Let %$\frac{M}{\beta} = \tilde{M}$ and let 
			$\tilde{L} = \frac{L}{\beta} =  \frac{100 \cdot \Sigma(\tilde{M}) \sqrt{\frac{1}{\epsilon}} \log \frac{1}{\epsilon}}{|\tilde{M}|}$. According to Theorem~\ref{the:dense}, in $\mathcal{O}(\frac{1}{\epsilon} \log\frac{1}{\epsilon})$ time, we can construct a data structure which allows us to solve the  SUBSET SUM problem $(\tilde{M},c)$ in $\OO(\log \frac{1}{\epsilon})$ time, where $c$ is any integer in $(\tilde{L},\Sigma (\tilde{M})- \tilde{L})$. Let $\mathcal{C}=\{c_1,c_2,\cdots,c_{\iota}\}$ denote the set of all different integers in $\left\{\lfloor\epsilon \Sigma (\tilde{M}) \rfloor, \lfloor 2\epsilon \Sigma (\tilde{M}) \rfloor, \cdots, \lfloor \Sigma (\tilde{M}) \rfloor, \lfloor  \Sigma (\tilde{M})-\tilde{L} \rfloor \right \} \cap (\tilde{L}, \Sigma (\tilde{M})-\tilde{L})$. We solve the SUBSET SUM problem $(\tilde{M},c)$ for every integer $c \in \mathcal{C}$. Denote by $E_i$ a solution of the SUBSET SUM problem $(\tilde{M},c_i)$, let $\mathcal{E}: = \{E_1,E_2,\cdots,E_{\iota}\}$ and let $\Sigma_{mid}: =\{\beta \Sigma(E_1), \beta \Sigma(E_2),\cdots,\beta \Sigma(E_{\iota})\}$. We build a dictionary $\mathscr{D}$ as follows: each element in $\mathscr{D}$ is a key-value pair $(\beta \Sigma(E_i), \mathscr{D}[\beta \Sigma(E_i)])$, where $\mathscr{D}[\beta \Sigma(E_i)] = \beta E_i$ and $i=1,2,\cdots,\iota$. The collection of keys of $\mathscr{D}$ is $\Sigma_{mid}$.
			
			%We use a dictionary to  maintain a one-to-one correspondence between $\mathcal{E}$ and $\Sigma_{mid}$: $E_i \leftrightarrow \beta \Sigma(E_i)$. The total processing time to obtain $\Sigma_{mid}$ and build this dictionary is $T_{mid} = \OO(\frac{1}{\epsilon} \log\frac{1}{\epsilon})$. 
			
			In the next, we show that $\Sigma_{mid}$ satisfies the conditions required in Claim~\ref{obs:mid_obs}. Apparently, $|\Sigma_{mid}| = \OO(\frac{1}{\epsilon})$. 
			Consider any $s \in S(M)\cap (L,\Sigma(M)-L)$, it follows that $\tilde{L} <\frac{s}{\beta} \le \lfloor  \Sigma (\tilde{M})-\tilde{L} \rfloor$ and  there exists some integer $k \in [0,\frac{1}{\epsilon}]$ such that $k\epsilon \Sigma(\tilde{M}) < \frac{s}{\beta} \le (1+k)\epsilon \Sigma(\tilde{M})$, thus we have
			$$\lfloor k\epsilon \Sigma(\tilde{M}) \rfloor < \frac{s}{\beta} \le \min \left\{\lfloor (1+k)\epsilon \Sigma(\tilde{M}) \rfloor,  \lfloor  \Sigma (\tilde{M})-\tilde{L} \rfloor\right\}.$$ 
			Observe that $ \min \left\{\lfloor (1+k)\epsilon \Sigma(\tilde{M}) \rfloor,  \lfloor  \Sigma (\tilde{M})-\tilde{L} \rfloor\right\} \in \mathcal{C} $, which guarantees that there exists some $E_k\in \mathcal{E}$ satisfying $\Sigma(E_k) \in [k\epsilon \Sigma(\tilde{M}), (1+k)\epsilon \Sigma(\tilde{M})]$, then $\beta \Sigma(E_k) \in \Sigma_{mid}$ and $|s-\beta \Sigma(E_k)|\le \beta \cdot\epsilon \Sigma(\tilde{M}) = \epsilon \Sigma({M})$. 
			
			An $\OO(1)$-time oracle for backtracking from $\Sigma_{mid}$ to $M$  works as follows: %Recall that we have built a dictionary to maintain a one-to-one correspondence between $\mathcal{E}$ and $\Sigma_{mid}$: $E_i \leftrightarrow \beta \Sigma(E_i)$. 
			given any $c\in \Sigma_{mid}$, the oracle uses dictionary $\mathscr{D}$ to return $\mathscr{D}[c]$. It is easy to see that $\mathscr{D}[c] \subset M$ and $ \Sigma(\mathscr{D}[c] )=c$.\qed
			%in $\OO(1)$ processing time, we can determine $E_k \subset \tilde{M}$ satisfying $c= \beta \Sigma(E_k) =  \Sigma(\beta E_k)$. Note that $\beta E_k \subset M$. Thus this dictionary is exactly an $\OO(1)$-time oracle for backtracking from $\Sigma_{mid}$ to $M$.
		\end{proof}

		\paragraph{Step 2:  Handling $S(M) \cap [0,L]$ and $S(M) \cap [\Sigma(M)-L,\Sigma(M)]$.} { }~\\
		
		In this step, we aim to prove the following claim.
		
		\begin{claim}\label{obs:obs_lr}
			In $T_{LR}=\tilde\OO(\epsilon^{-\frac{5}{4}})$ processing time, we can 
			\begin{itemize}
				\item[(i).] Compute a set $\Sigma_{LR} \subset \mathbb{R}_{\ge 0}$ satisfying the following conditions: (i). $|\Sigma_{LR}| =\OO(\frac{1}{\epsilon})$; (ii). given any $s\in \left(S(M)\cap[0,L]\right) \cup  \left(S(M)\cap[\Sigma(M)-L,\Sigma(M)] \right)$, there exists $s'\in \Sigma_{LR}$ such that $|s-s'|\le \tilde\OO(\epsilon)\Sigma(M)+\tilde\OO(\epsilon^{-2})$.
				\item[(ii).] Meanwhile build an $\tilde\OO(\epsilon^{-\frac{5}{4}})$-time oracle for backtracking from $\Sigma_{LR}$ to $M$. That is,  given any $c\in \Sigma_{LR}$,  in $\tilde\OO(\epsilon^{-\frac{5}{4}})$ time, the oracle will return $M_c\subset M$ such that $|c- \Sigma(M')|\le \tilde\OO(\epsilon)\Sigma(M)+\tilde\OO(\epsilon^{-2})$.
				
			\end{itemize}
		\end{claim}
		\begin{proof}[Proof of Claim~\ref{obs:obs_lr}]
			Divide $M$ into the following three disjoint groups: $M_1 = M \cap [\frac{1-\epsilon}{\epsilon^{2+\lambda}}, \frac{1}{\epsilon^{2+\lambda}})$, $M_2 = M\cap [\frac{1}{\epsilon^{2+\lambda}}, \frac{2}{\epsilon^{2+\lambda}})$ and $M_3 =  M\cap[\frac{2}{\epsilon^{2+\lambda}}, \frac{2(1+\epsilon)}{\epsilon^{2+\lambda}}]$. Since $\epsilon>0$ is sufficiently small, we coulld assume that $\epsilon \le \frac{1}{2}$, then $2M_1, \frac{1}{2} M_3\subset  [\frac{1}{\epsilon^{2+\lambda}}, \frac{2}{\epsilon^{2+\lambda}})$. Let $V_1$, $V_2$ and $V_3$ denote $2M_1$, $M_2$ and $\frac{1}{2} M_3$, respectively.

			Recall that $|M| = \OO(\frac{1}{\epsilon})$, then $|V_i| =\OO(\frac{1}{\epsilon})$ for $i=1,2,3$. Given $d \in \mathbb{N}_{+}$ and let $\bar{d}< d$ be the integer such that $\lambda+\frac{1}{2} \in (\frac{\bar{d}}{d},\frac{\bar{d} +1}{d}]$. Recall Lemma~\ref{lemma:smooth_appro}, for each $V_i$, in $\OO(\frac{1}{\epsilon}(\log \frac{1}{\epsilon})^{\OO(\bar{d}+1)})$ processing time, we can obtain a set $\Delta'_{i} \subset \Theta(\epsilon^{-\frac{3}{2}})$ with $|\Delta'_{i}| = \OO((\log \frac{1}{\epsilon})^{\OO(\bar{d}+1)})$ and round every $y \in V_i $ to the form $\rho' h'_1 h'_2 \cdots h'_{\bar{d}+1}$, where $\rho'\in \Delta'_{i}$ and $h'_{j}$'s satisfy the followings: (i) $h'_1 h'_2 \cdots h'_{\bar{d}+1} \in \mathbb{N}_{+}\cap [\frac{1}{4\epsilon^{\lambda+\frac{1}{2}}},\frac{1}{\epsilon^{\lambda+\frac{1}{2}}}]$; (ii) $h'_{\bar{d}+1}\in \mathbb{N}_{+}\cap [\frac{1}{2\epsilon^{\lambda+\frac{1}{2}-\frac{\bar{d}}{d}}},\frac{2}{\epsilon^{\lambda+\frac{1}{2}-\frac{\bar{d}}{d}}}]$ and if $\bar{d}>0$ then $h'_{j} \in \mathbb{N}_{+}\cap [\frac{1}{2\epsilon^{\frac{1}{d}}},\frac{2}{\epsilon^{\frac{1}{d}}}]$ for every $1\le i\le \bar{d}$; (iii) $|y-\rho h'_1 h'_2 \cdots h'_{\bar{d}+1}| \le \epsilon^{\lambda+\frac{1}{2}} y$. Denote by $\overline{V}_i$ the set of all such rounded elements obtained from $V_i$.
			
			%In the following, we aim to prove that for each $\overline{V}_i$, in $\tilde\OO(\epsilon^{\epsilon^{-\frac{5}{4}}})$, we can compute an $\tilde\OO(\epsilon)$-approximate set with cardinality of $\OO(\frac{1}{\epsilon})$ for $S(\overline{V}_i)$ and build an $\tilde\OO(\epsilon^{-\frac{5}{4}})$-time oracle for backtracking from this approximate set to $\overline{V}_i$.
			
			Let $\overline{V} = (\frac{1}{2} \overline{V}_1 )\dot{\cup} \overline{V}_2 \dot{\cup} (2\overline{V}_2)$. We have the following observation. 
			\begin{observation}\label{obs:sub_o_1}
				In $\tilde\OO(\epsilon^{-\frac{5}{4}})$ time, we can compute an $\tilde\OO(\epsilon)$-approximate set with cardinality of $\OO(\frac{1}{\epsilon})$ for $S(\overline{V})$ and meanwhile build an $\tilde\OO(\epsilon^{-\frac{5}{4}})$-time oracle for backtracking from this approximate set to $\overline{V}$.
			\end{observation}
			\begin{proof}%[Proof of Observation~\ref{claim:sub_o_1}]
				Given any multiset $X$ and any $\gamma \in \mathbb{R}_{\ge 0}$, note that if $C$ is an $\tilde\OO(\epsilon)$-approximate set of 
				$S(X)$, then $\gamma C$ is an $\tilde\OO(\epsilon)$-approximate set of 
				$S(\gamma X)$, moreover, a $T$-time oracle for backtracking from $C$ to $X$ directly yields an $\OO(T)$-time oracle for backtracking from $\gamma C$ to $\gamma X$. 
				
				Recall Corollary~\ref{coro:tree-fashion_}, towards Observation~\ref{obs:sub_o_1}, it is sufficient to prove that for each $\overline{V}_i$, in $\tilde\OO(\epsilon^{-\frac{5}{4}})$ processing time, one can compute an $\tilde\OO(\epsilon)$-approximate set with cardinality of $\OO(\frac{1}{\epsilon})$ for $S(\overline{V}_i)$ and build an $\tilde\OO(\epsilon^{-\frac{5}{4}})$-time oracle for backtracking from this approximate set to $\overline{V}_i$.  Consider each $\overline{V}_i$. Let $\Delta'_i := \{\rho'_1,\rho'_2,\cdots, \rho'_{|\Delta'_i|}\}$, we can divide $\overline{V}_i$ into $|\Delta'_i|= \OO((\log \frac{1}{\epsilon})^{\OO(\bar{d}+1)})$ groups $\overline{V}_i^1,\overline{V}_i^2,\cdots, \overline{V}_i^{|\Delta'|}$ such that $y\in \overline{V}_i^k$ if and only if $y$ is of the form $\rho'_k h'_1 h'_2 \cdots h'_{\bar{d}+1}$.  Again by Corollary~\ref{coro:tree-fashion_}, it is sufficent to prove that for each $\overline{V}^k_i$, in $\tilde\OO(\epsilon^{-\frac{5}{4}})$ processing time, one can compute an $\tilde\OO(\epsilon)$-approximate set with cardinality of $\OO(\frac{1}{\epsilon})$ for $S(\overline{V}^k_i)$ and build an  $\tilde\OO(\epsilon^{-\frac{5}{4}})$-time oracle for backtracking from this approximate set to $\overline{V}^k_i$.  
				
				Note that elements in $\frac{\overline{V}_i^k}{\rho'_k}$ are $(\epsilon,d,\bar{d})$-smooth numbers. Recall Lemma~\ref{lemma:e_apx_sm}, for any $\tau  \in \mathbb{N} \cap [0,\bar{d}]$, in $T^k_i =\OO\left(d \cdot \left|\frac{\overline{V}_i^k}{\rho'_k}\right|+\Sigma(\frac{\overline{V}_i^k}{\rho'_k}) {\epsilon^{\frac{\tau }{d}}} \cdot \log (\left|\frac{\overline{V}_i^k}{\rho'_k}\right|)\cdot \log (\Sigma(\frac{\overline{V}_i^k}{\rho'_k}) {\epsilon^{\frac{\tau }{d}}})+ {\epsilon^{-(1+\frac{\tau }{d})}}\log \frac{1}{\epsilon}\right)$ processing time, we can compute an $\OO(\epsilon (\log\frac{1}{\epsilon})^k)$-approximate set with cardinality of $\OO(\frac{1}{\epsilon})$ for $S(\frac{\overline{V}_i^k}{\rho'_k})$. Denote by $C_{{\overline{V}_i^k}}$ this approximate set, in the meantime, we have built a $\tilde{T}^k_i$-time oracle for backtracking from $C_{{\overline{V}_i^k}}$ to $\frac{\overline{V}_i^k}{\rho'_k}$, where $\tilde{T}^k_i = \OO\left(\Sigma(\frac{\overline{V}_i^k}{\rho'_k}) {\epsilon^{\frac{\tau }{d}}} \cdot \log (\left|\frac{\overline{V}_i^k}{\rho'_k}\right|)\cdot \log (\Sigma(\frac{\overline{V}_i^k}{\rho'_k}) {\epsilon^{\frac{\tau }{d}}})+ {\epsilon^{-(1+\frac{\tau }{d})}}\log \frac{1}{\epsilon}\right)$. It is easy to observe that ${\rho'_k} C_{{\overline{V}_i^k}}$ is an $\OO(\epsilon (\log\frac{1}{\epsilon})^k)$-approximate set of $S(\overline{V}_i^k)$ and the oracle for backtracking from $C_{{\overline{V}_i^k}}$ to $\frac{\overline{V}_i^k}{\rho'_k}$ directly yields an $\OO(\tilde{T}^k_i)$-time oracle for backtracking from ${\rho'_k}C_{{\overline{V}_i^k}}$ to ${\overline{V}_i^k}$. 
				It remains to determine $k$ such that $T^k_i$ and $\tilde{T}^k_i$ attain their minimum values.  Recall that $|M|= \OO(\frac{1}{\epsilon})$, then $\left|\frac{\overline{V}_i^k}{\rho'_k}\right| \le |M| = \OO(\frac{1}{\epsilon})$. Notice that $\Sigma(\overline{V}^k_i) \le (1+ {\epsilon^{\lambda+\frac{1}{2}}}) \Sigma(M) = \OO(\epsilon^{-3})$ and $\rho'_k \in \Theta(\epsilon^{-\frac{3}{2}})$, then $\Sigma(\frac{\overline{V}_i^k}{\rho'_k}) = \OO(\epsilon^{-\frac{3}{2}})$. Recall that $\lambda \ge 0$, $\frac{\bar{d}+1}{d} \ge \lambda +\frac{1}{2}$ and $d \ge 12$ is an integer divisible by 4, we have $\frac{d}{4} \le \bar{d}$. By setting $\tau  = \frac{d}{4}$, we have $T^k_i = \OO(\epsilon^{-\frac{5}{4}} (\log \frac{1}{\epsilon})^2)$  and $\tilde{T}^k_i= \OO(\epsilon^{-\frac{5}{4}} (\log \frac{1}{\epsilon})^2)$. \qed %Observe that $|\Delta'_i| = \OO((\log \frac{1}{\epsilon})^{\OO(\bar{d}+1)})$.  Recall Corollary~\ref{coro:tree-fashion_}, in overall $\OO((\epsilon^{-\frac{5}{4}}+ \frac{1}{\epsilon})(\log \frac{1}{\epsilon})^{\OO(\bar{d}+1)}) = \tilde\OO(\epsilon^{-\frac{5}{4}})$ processing time, we can compute an $\tilde\OO(\epsilon)$-approximate set with cardinality of $\OO(\frac{1}{\epsilon})$ for $S(\overline{V}_i)$ and build an $\tilde\OO(\epsilon^{-\frac{5}{4}})$-time oracle for backtracking from this approximate set to $\overline{V}_i$.
			\end{proof}
			
			According to Observation~\ref{obs:sub_o_1}, in $\tilde\OO(\epsilon^{-\frac{5}{4}})$ processing time, we can compute an $\tilde\OO(\epsilon)$-approximate set with cardinality of $\OO(\frac{1}{\epsilon})$ for $S(\overline{V})$. Denote by $C_{\overline{V}}$ this approximate set. In the meantime, we have built an $\tilde\OO(\epsilon^{-\frac{5}{4}})$-time oracle for backtracking from $C_{\overline{V}}$ to $\overline{V}$. Denote by $\textbf{Ora}_{\overline{V}}$ this oracle. 
			
			%Recall that $L=\frac{100\cdot \Sigma (M) \cdot \sqrt{\frac{1}{\epsilon}} \log \frac{1}{\epsilon}}{|M|}$.	
			Define $\Sigma_{L}: = C_{\overline{V}} \cap [0,(1+\tilde\OO(\epsilon^{\lambda+\frac{1}{2}}))L+\tilde\OO(\epsilon)\Sigma(\overline{V})]$ and $\Sigma_{R}: = \{\Sigma(M)-s \ | \  s\in \Sigma_{L}\}$. Let $\Sigma_{LR} = \Sigma_{L}\cup \Sigma_{R}$. To summarize, by setting $d = 12$, in total $\tilde\OO(\epsilon^{-\frac{5}{4}})$ processing time, we can obtain $\Sigma_{LR}$ with $|\Sigma_{LR}| = \OO(\frac{1}{\epsilon})$. Moreover, we have the following observation.
			
			\begin{observation}\label{obs:sub_o_2}
				Given any $s\in \left(S(M)\cap[0,L]\right) \cup  \left(S(M)\cap[\Sigma(M)-L,\Sigma(M)]\right)$, there exists $s'\in \Sigma_{LR}$ such that $|s-s'|\le \tilde\OO(\epsilon)\Sigma(M)+\tilde\OO(\epsilon^{-2})$.   
			\end{observation}
			\begin{proof}
				Note that $\overline{V}$ is obtained by rounding elements in $M$. For any $\overline{V}' \subset \overline{V}$, let ${M}'\subset M$ be the original subset corresponding to $\overline{V}'$, it holds that $|\Sigma(\overline{V}')-\Sigma({M}')| \le   \epsilon^{\lambda + \frac{1}{2}}\Sigma({M}')$, then $|\Sigma(\overline{V}')-\Sigma({M}')| \le   \OO(\epsilon^{\lambda + \frac{1}{2}})\Sigma(\overline{V}')$. %In particular, we have  $\Sigma(\overline{V}) \le (1+ \OO(\epsilon^{\lambda + \frac{1}{2}}))\Sigma(M)$ and $\Sigma(M) \le (1+ \OO(\epsilon^{\lambda + \frac{1}{2}}))\Sigma(\overline{V})$. 
				
				%Moreover, we have $\Sigma(\overline{V}) \le (1+\epsilon^{\lambda +\frac{1}{2}})\Sigma(M)$. % < 2\Sigma(M) 
				
				We first show that given any $\bar{s}\in \Sigma_L$  and any $\overline{V}' \subset \overline{V}$ satisfying $|\Sigma(\overline{V}')-\bar{s}|\le \tilde\OO(\epsilon)\Sigma(\overline{V})$, let ${M}'\subset M$ be original subset corresponding to $\overline{V}'$, it always holds that $|\bar{s}-\Sigma({M}')| \le \tilde\OO(\epsilon)\Sigma(M)+ \tilde\OO(\epsilon^{-2})$. It suffices to observe the following:
				\begin{align*}
					|\bar{s}-\Sigma({M}')| &\le |\bar{s}-\Sigma(\overline{V}')|+|\Sigma(\overline{V}')-\Sigma({M}')|\\
					&\le \tilde\OO(\epsilon)\Sigma(\overline{V})+ \OO(\epsilon^{\lambda+1/2})\Sigma(\overline{V}')\\
					&\le \tilde\OO(\epsilon)\Sigma(\overline{V})+ \OO(\epsilon^{\lambda+1/2})(\bar{s}+\tilde\OO(\epsilon)\Sigma(\overline{V}))\\
					&\le  \tilde\OO(\epsilon)\Sigma(\overline{V})+ \OO(\epsilon^{\lambda+1/2})\left((1+\tilde\OO(\epsilon^{\lambda+1/2}))L+\tilde\OO(\epsilon)\Sigma(\overline{V})\right)\\
					&= \tilde\OO(\epsilon)\Sigma(\overline{V})+ \OO(\epsilon^{\lambda+1/2})L\le \tilde\OO(\epsilon)\Sigma(M)+ \OO(\epsilon^{\lambda+1/2})L= \tilde\OO(\epsilon)\Sigma(M)+ \tilde\OO(\epsilon^{-2}).
				\end{align*}
				Where the last inequality holds by the facts that $L=\frac{100\cdot \Sigma (M) \cdot \sqrt{\frac{1}{\epsilon}} \log \frac{1}{\epsilon}}{|M|} = \tilde\OO(\frac{1}{\epsilon^{5/2+\lambda}})$.

				Now, we are ready to prove Observation~\ref{obs:sub_o_2}. Given any $s\in \left(S(M)\cap[0,L]\right) \cup  \left(S(M)\cap[\Sigma(M)-L,\Sigma(M)]\right)$. If $s\in S(M)\cap[0,L]$, then there exists $s'\in S(\overline{V})$ such that $|s-s'|\le \epsilon^{\lambda+1/2} s \le \epsilon^{\lambda+1/2} L.$ % = \tilde\OO(\epsilon^{-2}).
				For $s'\in S(\overline{V})$, recall that $C_{\overline{V}}$ is an $\tilde\OO(\epsilon)$-approximate set of $S(\overline{V})$, thus there exists $s''\in C_{\overline{V}}$ such that $|s''-s'|\le \tilde\OO(\epsilon)\Sigma(\overline{V})$. %\le \tilde\OO(\epsilon)\Sigma(M)
				To summarize, we have $$|s''-s|\le \OO(\epsilon^{\lambda+1/2})L+ \tilde\OO(\epsilon)\Sigma(\overline{V}) \le \tilde\OO(\epsilon^{-2})+\tilde\OO(\epsilon)\Sigma(M).
				$$
				It follows that $s''\le s+ \OO(\epsilon^{\lambda+1/2})L+ \tilde\OO(\epsilon)\Sigma(\overline{V}) \le (1+\tilde\OO(\epsilon^{\lambda+1/2}))L+ \tilde\OO(\epsilon)\Sigma(\overline{V})$, thus $s''\in \Sigma_{L} \in \Sigma_{LR}$. 
				
				Else if $s \in S(M)\cap[\Sigma(M)-L,\Sigma(M)]$, we have $\Sigma(M)-s \in S(M)\cap[0,L]$, then according to the above discussion, there exists $s''\in \Sigma_L$ such that 
				$$|(\Sigma(M)-s'')-s|= |(\Sigma(M)-s)-s''|\le \OO(\epsilon^{\lambda+1/2})L+ \tilde\OO(\epsilon)\Sigma(V)\le \tilde\OO(\epsilon^{-2})+\tilde\OO(\epsilon)\Sigma(M).
				$$ 
				It is easy to see that  $\Sigma(M)-s'' \in \Sigma_{R}\in \Sigma_{LR}$. \qed \end{proof}

			It remains to give the oracle for backtracking from $\Sigma_{LR}$ to $M$. Recall that $\Sigma_{L} \subset C_{\overline{V}}$ and we have built an $\tilde\OO(\epsilon^{-\frac{5}{4}})$-time oracle $\textbf{Ora}_{\overline{V}}$ for backtracking from $C_{\overline{V}}$ to $\overline{V}$. An $\tilde\OO(\epsilon^{-\frac{5}{4}})$-time oracle $\textbf{Ora}_{LR}$ for backtracking from $\Sigma_{LR}$ to $M$ works as follows. Given any $c\in \Sigma_{LR}$. If $c\in \Sigma_{L}$,  oracle $\textbf{Ora}_{LR}$ first uses $\textbf{Ora}_{\overline{V}}$ to return a subset $V_c\subset \overline{V}$ such that $|\Sigma(V_c)-c|\le \tilde\OO(\epsilon)\Sigma(\overline{V})$. %The processing time is $\tilde\OO(\epsilon^{-\frac{5}{4}})$. 
			Then $\textbf{Ora}_{LR}$ returns $M_c$, which is the original subset in $M$ that is corresponding to $V_c$. According to the above discussion, we have $$|c-\Sigma(M_c)| \le\tilde\OO(\epsilon)\Sigma(M)+ \tilde\OO(\epsilon^{-2}).
			$$
			Else if $c \in \Sigma_{R}$, we have $\Sigma(M)-c \in \Sigma_{L}$. Oracle $\textbf{Ora}_{LR}$ first uses $\textbf{Ora}_{\overline{V}}$ to return a subset $V'\subset \overline{V}$ such that $|\Sigma(V' )-(\Sigma(M)-c)|\le \tilde\OO(\epsilon)\Sigma(\overline{V})$. Let $M'$ be the original subset in $M$ that is corresponding to $V'$. Then $\textbf{Ora}_{LR}$ returns $M \backslash M'$. According to the above discussion, we have 
			$$|\Sigma(M \backslash M')-c| =
			|(\Sigma(M)-c)-\Sigma(M')|\le \tilde\OO(\epsilon)\Sigma(M)+ \tilde\OO(\epsilon^{-2}).$$
			Till now, we complete the proof of Claim~\ref{obs:sub_o_1}.\qed \end{proof}

		\paragraph{Step 3: Handling  $S(M)$.}~\newline
		Given $\Sigma_{mid}$ and $\Sigma_{LR}$ obtained in Step 1 and Step 2, respectively. Let $\Sigma_{union} = \Sigma_{mid} \cup \Sigma_{LR}$. Then we have the following Claim~\ref{obs:dense_gro_}, which follows directly from Claim~\ref{obs:mid_obs} and Claim~\ref{obs:obs_lr}.
		%Claim~\ref{obs:mid_obs} and Claim~\ref{obs:obs_lr}, respectively, 
		%of the previous two steps together, we get the following conclusion. 
		\begin{claim}\label{obs:dense_gro_}
			With an additive error of $\tilde\OO({\epsilon^{-2}})$, $\Sigma_{union}$ is an $\tilde\OO(\epsilon)$-approximate set of $S({M})$. Moreover, we have built an $\tilde\OO(\epsilon^{-\frac{5}{4}})$-time oracle for backtracking from $\Sigma_{union}$ to $M$.
		\end{claim}

		Recall that the time to obtain $\Sigma_{mid}$ and  $\Sigma_{LR}$ is $\tilde\OO(\epsilon^{-\frac{5}{4}})$. Moreover, we have $|\Sigma_{mid}|+|\Sigma_{LR}| = \OO(1/\epsilon)$. Thus the overall time to compute $\Sigma_{union}$ is $\tilde\OO(\epsilon^{-\frac{5}{4}})$, and we have $|\Sigma_{union}|\le | \Sigma_{mid}|+ | \Sigma_{LR}|  = \OO(1/\epsilon)$. Then combine with Claim~\ref{obs:dense_gro_}, we have proved Lemma~\ref{lemma:small_gro_p} for the case that $M$ is a dense small-value group. This completes the proof of Lemma~\ref{lemma:small_gro_p} for the case $M$ is a small-value group, furthermore, completes the proof of Lemma~\ref{lemma:small_gro_p}. Till now, we have completed the proof of Theorem~\ref{the:5/4-fptas}.

		%{\color{blue}We call $M$ a \textit{large-value} group if $\lambda \ge \frac{1}{2}$. Otherwise, we call $M$ a \textit{small-value} group. We will prove Lemma~\ref{lemma:small_gro_p} by designing algorithms separately for large-value group and small-value group. For large-value group, we adopt the algorithm derived in Lemma~\ref{lemma:e_apx_sm} to handle it. For small-value group, we leverage the additive combinatoric result from \cite{DBLP:journals/siamcomp/GalilM91}, which states that  As for the small-value-group, we deal with the dense and sparse cases separately. For sparse case, we again use the algorithm derived in Lemma~\ref{lemma:e_apx_sm} to handle. For the dense case, we divide the subset-sums into three intervals for approximation separately, then we get Claim~\ref{obs:obs_lr}, Claim~\ref{obs:mid_obs} and Observation~\ref{obs:dense_gro_}, which together prove the dense case. Details of the proof are omitted here, we refer the reader to Appendix~\ref{appsec:proof_lemma:small_gro_p}.}
		
		\section{An $\tilde\OO(n+\epsilon^{-\frac{3}{2}})$-time weak $(1-\tilde\OO(\epsilon))$-approximation algorithm for SUBSET SUM.}\label{sec:3/2} 
		%算法分两部分：一是求出每个子集的e-近似；二是求出之后得到总的集合的近似。
		
		%在这里说明之前randomized的结果已经做出来（Introduction里也要强调）
		The goal of this section is to prove the following theorem.
		\begin{theorem}\label{the:sub_sum}
			Let $d \in \mathbb{N}_{+}$ be an arbitrary %fixed 
			even integer. There is a deterministic weak $(1-\epsilon)$-approximation algorithm for SUBSET SUM running in $\OO((n+\epsilon^{-\frac{3}{2}-\frac{1}{d}})(\log n)^2 (\log \frac{1}{\epsilon})^{\OO(d)})$ time.
		\end{theorem}
		
		\noindent\textbf{Remark.}  Taking $d=\OO(\sqrt{\frac{\log \frac{1}{\epsilon}}{\log\log \frac{1}{\epsilon}}})$, the overall running time is $\tilde{\OO}(n+\epsilon^{-\frac{3}{2}})$ where $\tilde{\OO}$ hides a factor of $2^{\OO(\sqrt{\log \frac{1}{\epsilon}\log\log\frac{1}{\epsilon}})}=(\frac{1}{\epsilon})^{o(1)}$. 
		
		Given a multiset $X = \{x_1,x_2,\cdots,x_n \}\subset \mathbb{N}$ and a target $t \in \mathbb{R}_{\ge 0}$, we let $X^*$ and $OPT$ be an optimal solution and the optimal objective value of SUBSET SUM instance $(X,t)$, respectively. By Lemma~\ref{lemma:prea}, we may assume that $OPT \ge {t}/{2}$. Then a subset $Y\subset X$ satisfying $|\Sigma(Y)-OPT| \le \tilde\OO(\epsilon)t$ is a weak $(1-\tilde\OO(\epsilon))$-approximation solution of $(X,t)$. Thus the following lemma implies Theorem~\ref{the:sub_sum} directly. 
		
		\begin{lemma}\label{lemma:sumset-app-stru_}
			Let $d \in \mathbb{N}_{+}$ be an arbitrary %fixed 
			even number. Given a multiset $X = \{x_1,x_2,\cdots,x_n\} \subset \mathbb{N}$, in $\tilde\OO(n+\epsilon^{-\frac{3}{2}-\frac{1}{d}})$ processing time, we can 
			\begin{itemize}
				\item[(i).] Compute an $(\tilde\OO(\epsilon),t)$-approximate set with cardinality of $\tilde\OO(\epsilon^{-\frac{3}{2}})$
				for $S(X)$;
				\item[(ii).] Meanwhile build an $\tilde\OO(n+\epsilon^{-\frac{3}{2}-\frac{1}{d}})$-time oracle for backtracking from this approximate set to $X$.
			\end{itemize}
		\end{lemma}
		
		%Given a SUBSET SUM instance $(X,t)$, let $d \in \mathbb{N}_{+}$ be an arbitrary fixed even number. 
		Recall Lemma~\ref{obs:obs_pre},  in $\OO((|X+\frac{1}{\epsilon})(\log |X)^2 (\log \frac{1}{\epsilon})^{\OO(d)})$ time, we can reduce $(X,t)$ to a SUBSET SUM instance $(F,\hat{t})$ satisfying conditions $(\uppercase\expandafter{\romannumeral1})(\uppercase\expandafter{\romannumeral2})(\uppercase\expandafter{\romannumeral2})(\uppercase\expandafter{\romannumeral4})$ 
		and meanwhile build an oracle for backtracking from $F$ to $X$ (see Lemma~\ref{obs:obs_pre} in Section~\ref{sec:proce}). Condition $(\uppercase\expandafter{\romannumeral2})$ claims that the optimal objective value of $(F,\hat{t})$ is at least $\frac{4(1-\epsilon)}{\epsilon^3\cdot t}\cdot (OPT-\frac{\epsilon t}{2})$, recall that $OPT \ge t/2$, then given any weak $(1-\tilde\OO(\epsilon))$-approximation of $(F,\hat{t})$, in linear time, the oracle will return a weak $(1-\tilde\OO(\epsilon))$-approximation of $(X,t)$. It thus suffices to consider $(F,\hat{t})$, that is, Lemma~\ref{lemma:sumset-app-stru_} and hence Theorem~\ref{the:sub_sum}, follows from the following Lemma~\ref{lemma:sumset-app-stru}. 	%Therefore, it suffices to consider $(F,\bar{t})$. More precisely, the following Lemma~\ref{lemma:sumset-app-stru} implies Lemma~\ref{lemma:sumset-app-stru_}, and hence Theorem~\ref{the:sub_sum}.
		
		\iffalse{followings:
			\begin{itemize}
				\item [(\uppercase\expandafter{\romannumeral1})] $\hat{t}= \frac{4(1+\epsilon)}{\epsilon^3}$;
				\item [(\uppercase\expandafter{\romannumeral2})] The optimal objective value of $(F,\hat{t})$ is at least $\frac{4(1-\epsilon)}{\epsilon^3\cdot t}\cdot (OPT-\frac{\epsilon t}{2})$;
				\item [(\uppercase\expandafter{\romannumeral3})] $F$ has been divided into $\OO((\log |X|)^2 (\log \frac{1}{\epsilon})^{\OO(d)})$ subgroups: $F^{(j;i;k)}_p$'s. 
				%$$\{F^{(j;i;k)}_p  \ |  \  1\le p \le \OO(\log |X|), 1\le j \le \OO(\log \frac{1}{\epsilon}),  i=1,2 \text{\ and \ } 1\le k \le \OO(\log|X|  \cdot (\log \frac{1}{\epsilon})^{\OO(d)})\}.$$  
			\end{itemize}
			In the meantime, we can build an oracle for backtrackig from $F$ to $X$, %Here $F$ is the multiset union of $F^{(j;i;k)}_p$'s and 
			where the oracle works as follows: given $U^{(j;i;k)}_p \subset F^{(j;i;k)}_p$ for every $F^{(j;i;k)}_p$, let $F'\subset F$ be the multiset-union of all $U^{(j;i;k)}_p$'s, then in linear time, the oracle will return a subset $X' \subset X$ satisfying  $|\frac{\epsilon^3 t}{4} \Sigma(F')-\Sigma(X')|\le \OO(\epsilon) \Sigma(X')$.   
		}\fi

		\begin{lemma}\label{lemma:sumset-app-stru}
			Given any SUBSET SUM instance $(X,t)$, where the optimal objective value of $(X,t)$ is at least $t/2$. Let $d\in\mathbb{N}_{+}$ be an arbitrary fixed even number and let $(F,\hat{t})$ a modified instance returned by Lemma~\ref{obs:obs_pre}, where $F$ and $\hat{t}$ satisfy the followings:
			\begin{enumerate}
				\item%[(\uppercase\expandafter{\romannumeral1})] 
				$\Sigma(F) \le \frac{4(1+\epsilon)\Sigma(X)}{\epsilon^3 t}$ and $\hat{t}= \frac{4(1+\epsilon)}{\epsilon^3}$.
				%\item [(\uppercase\expandafter{\romannumeral2})]The optimal objective value of $(F,\hat{t})$ is at least $\frac{4(1-\epsilon)}{\epsilon^3\cdot t}\cdot (OPT-\frac{\epsilon t}{2})$.
				\item %[(\uppercase\expandafter{\romannumeral3})]  
				$|F| \le |X|$ and $F$ has been divided into $\OO((\log |X|)^2 (\log \frac{1}{\epsilon})^{\OO(d)})$ subgroups: $F^{(j;i;k)}_p$'s. 
				\item %[(\uppercase\expandafter{\romannumeral4})]
				Each subgroup $F^{(j;i;k)}_p$ satisfies the followings:
				\begin{itemize}
					\item each element in $F^{(j;i;k)}_p$ is distinct;
					\item $F^{(j;i;k)}_p \subset  [\frac{(1-\epsilon)2^{p+j-1}}{\epsilon^{2}}, \frac{(1+\epsilon)2^{p+j}}{\epsilon^{2}})$;
					\item $F^{(j;i;k)}_p = 2^p \rho^j_k \overline{F}^{(j;i;k)}_p$, where $\overline{F}^{(j;i;k)}_p \subset \mathbb{N}_{+}\cap[\frac{1}{4\epsilon},\frac{1}{\epsilon}]$ and $\overline{F}^{(j;i;k)}_p$ is a set of $(\epsilon,d,d-1)$-{smooth numbers}, that is, every element in $\overline{F}^{(j;i;k)}_p$ has been factorized as $h_1 h_2 \cdots h_{d}$, where $h_{d}\in \mathbb{N}_{+}\cap [1,2\epsilon^{-\frac{1}{d}}]$ and $h_{i} \in \mathbb{N}_{+} \cap [\frac{1}{2}\epsilon^{-\frac{1}{d}},2{\epsilon^{-\frac{1}{d}}}]$ for $i=1,2,\cdots,d-1$. 
				\end{itemize}
			\end{enumerate} 	
			Then in $\tilde\OO(\epsilon^{-\frac{3}{2}-\frac{1}{d}})$ processing time, we can
			\begin{itemize}
				\item[(i).]  Compute an $(\tilde\OO(\epsilon),\hat{t})$-approximate set $C$ with cardinality of $\tilde\OO(\epsilon^{-\frac{3}{2}})$ for $S(F)$.
				\item[(ii).] Meanwhile build an $\tilde\OO(\epsilon^{-\frac{3}{2}-\frac{1}{d}})$-time oracle for backtracking from $C$ to $F$. %Here $F$ is the multiset union of $F^{(j;i;k)}_p$'s 
				The oracle actually works as follows: given any $c\in C$, in $\tilde\OO(\epsilon^{-\frac{3}{2}-\frac{1}{d}})$ processing time, the oracle will return $U^{(j;i;k)}_p \subset F^{(j;i;k)}_p$ for every $F^{(j;i;k)}_p$. Let $F'$ be the multiset-union of all $U^{(j;i;k)}_p$'s, we have $F' \subset F$ and $|\Sigma(F')-c| \le \tilde\OO(\epsilon)\hat{t}$.
			\end{itemize}
		\end{lemma} 
		The rest of this section is dedicated to proving Lemma~\ref{lemma:sumset-app-stru}.  %According to Lemma~\ref{obs:obs_pre}, $F$ is divided into $\OO((\log |X|)^2 (\log \frac{1}{\epsilon})^{\OO(d)})$ groups: $ F^{(j;i;k)}_p$'s, and each $ F^{(j;i;k)}_p$ satisfies condition $(\uppercase\expandafter{\romannumeral4})$ (see Section~\ref{sec:proce}). Note that $\hat{t}= \Theta(\epsilon^{-3})$. 
		Recall Corollary~\ref{coro:sub_sum_top}, which allows us to build the approximation set of the union of $F^{(j;i;k)}_p$'s from the approximate set of each $F^{(j;i;k)}_p$. 
		Towards proving Lemma~\ref{lemma:sumset-app-stru}, we only need to derive an algorithm that can solve the following \textbf{problem}-$\mathscr{Q}$ in $\tilde\OO(\epsilon^{-\frac{3}{2}-\frac{1}{d}})$-time.
		
		\textbf{problem}-$\mathscr{Q}$: for every $F_p^{(j;i;k)}$,  compute an $(\tilde\OO(\epsilon),\hat{t})$-approximate set with cardinality of $\tilde\OO(\epsilon^{-\frac{3}{2}})$ for $S(F_p^{(j;i;k)})$ and build an $\tilde\OO(\epsilon^{-\frac{3}{2}-\frac{1}{d}})$-time oracle for backtracking from this approximate set to $F_p^{(j;i;k)}$. 
		
		To solve \textbf{problem}-$\mathscr{Q}$, we only need to prove the following lemma~\ref{lemma:apx_small-group}, where $M$ represents an arbitrary $F_p^{(j;i;k)}$.
		
		\begin{lemma}\label{lemma:apx_small-group}
			Suppose we are given $M\subset \mathbb{R}_{\ge 0}$ satisfying the following conditions:
			\begin{enumerate}
				\item Each element in $M$ is distinct; 
				\item $M = \beta \tilde{M}$, where $\beta$ and $\tilde{M}$ satisfy the followings:
				\begin{itemize}
					\item   $\beta \in \Theta(\frac{1}{\epsilon^{1+\lambda}})$ and $\beta\tilde{M} \subset [\frac{(1-\epsilon)}{\epsilon^{2+\lambda}}, \frac{2(1+\epsilon)}{\epsilon^{2+\lambda}} ]$ where $\lambda \ge 0$; 
					\item  $\tilde{M} \subset \mathbb{N}_{+}\cap [\frac{1}{4\epsilon},\frac{1}{\epsilon}]$ and $\tilde{M}$ is a set of $(\epsilon,d,d-1)$-smooth numbers, that is , every element in $\tilde{M}$ has been factorized as $h_1 h_2 \cdots h_d$, where $ h_d \in  \mathbb{N}\cap[1,{2\epsilon^{-\frac{1}{d}}}]$ and 
					$h_i \in \mathbb{N}\cap [\frac{\epsilon^{-\frac{1}{d}}}{2},{2\epsilon^{-\frac{1}{d}}}] $ for $i = 1,2,\cdots,d-1$.
				\end{itemize}
			\end{enumerate}
			Here $d\in \mathbb{N}_{+}$ is an arbitrary fixed 
			even integer.   Then given any $\omega = \Theta(\epsilon^{-3})$, in  $\OO(\frac{d}{\epsilon}+ \frac{2^{{d}}}{\epsilon^{\frac{3}{2}+\frac{1}{d}}} (\log \frac{1}{\epsilon})^3)$ processing time, we can compute an $(\epsilon,\omega)$-approximate set with cardinality of $\OO(\epsilon^{-\frac{3}{2}}\log \frac{1}{\epsilon})$ for $S(M)$ and build an $\OO(\frac{2^{{d}}}{\epsilon^{\frac{3}{2}+\frac{1}{d}}} (\log \frac{1}{\epsilon})^3)$-time oracle for backtracking from this approximate set to $M$.
		\end{lemma}
		%\noindent\textbf{Remark.} It is easy to observe that $|M| = \OO(\frac{1}{\epsilon})$.  We call $M$ a \textit{dense} group if $|M| = \Omega(\epsilon^{-\frac{1}{2}} \log \frac{1}{\epsilon})$, otherwise, we call $M$ a \textit{sparse} group. 

		The high-level proof idea of Lemma~\ref{lemma:apx_small-group} resembles that of Lemma~\ref{lemma:small_gro_p}. In particular, if $M$ only contains a few distinct numbers, then we leverage Lemma~\ref{lemma:sumset} to handle it. Otherwise, $M$ contains many numbers, then we again leverage the additive combinatoric result from \cite{DBLP:journals/siamcomp/GalilM91} to deal with the case when the target $t$ is in the medium range. If $t$ is out of the medium range, we leverage Lemma~\ref{lemma:tree-alg}. Note that for general SUBSET SUM the target may be sufficiently smaller than the total sum of input numbers, therefore, unlike Lemma~\ref{lemma:small_gro_p}, we cannot tolerate an additive error of $\OO(\epsilon^{-2})$ anymore, which means the idea of ``coarse" rounding in proving Lemma~\ref{lemma:small_gro_p} is inapplicable here. Therefore, we can only guarantee the running time of $\tilde{\OO}(\epsilon^{-3/2})$.

		%The crucial observation is that a small target value can only be the sum of a few smooth numbers in $M/\beta$. Hence, we may adopt a ``course" rounding, that is, instead of only introducing $\OO(\epsilon)$-multiplicative error to each input number, we may introduce a larger error. Although per number the multiplicative error is $\Omega(\epsilon)$, but since only $o(\frac{1}{\epsilon})$ numbers will be selected, we can guarantee that their summation gives additive $\OO(\epsilon^{-2})$ error in total by a careful parameterized analysis. 
		
		%{\color{blue} We prove Lemma~\ref{lemma:apx_small-group} separately for the sparse group and the dense group. We use the algorithm derived in Lemma~\ref{lemma:sumset} to handle the sparse group. As for the dense group, we also divide the subset-sums into three intervals for approximation separately, then we get Observation~\ref{obs:3/2_mid_1}, Observation~\ref{obs:3/2_mid_} and Observation~\ref{obs:dense_gro}, which together prove the dense case. The details of the proof are presented in Appendix~\ref{appsec:proof_apx_small-group}.}
		
		It is easy to observe that $|M| = \OO(\frac{1}{\epsilon})$.  We call $M$ a \textit{dense} group if $|M| = \Omega(\epsilon^{-\frac{1}{2}} \log \frac{1}{\epsilon})$, otherwise, we call $M$ a \textit{sparse} group. Sparse groups and dense groups will be handled separately.

		\subsection{Handling Sparse Group.} 
		
		We first consider the case that $|M| = \OO(\epsilon^{-\frac{1}{2}}\log \frac{1}{\epsilon})$. 
		
		Notice that $\tilde{M} \subset [\frac{1}{4\epsilon}, \frac{1}{\epsilon}]$ and $|\tilde{M}| = \OO(\epsilon^{-\frac{1}{2}}\log \frac{1}{\epsilon})$, then we have $\Sigma(\tilde{M}) = \OO(\epsilon^{-\frac{3}{2}}\log \frac{1}{\epsilon})$. If we regard each element in $\tilde{M}$ as a set containing this single element, then according to Lemma~\ref{lemma:sumset}, in $\OO(\epsilon^{-\frac{3}{2}}(\log \frac{1}{\epsilon})^3)$ time, we can compute $S(\tilde{M})$ and build an $\OO(\epsilon^{-\frac{3}{2}}(\log \frac{1}{\epsilon})^3)$-time oracle for backtracking from $S(\tilde{M})$ to $\tilde{M}$.  Observe that $\left(\beta S(\tilde{M}) \right)\cap [0,\omega] = S(M; [0,\omega])$ and $S(M; [0,\omega])$ is an $(\epsilon,\omega)$-approximate set for $S(M)$. Moreover, the oracle for backtracking from $S(\tilde{M})$ to $\tilde{M}$ directly yields an $\OO(\epsilon^{-\frac{3}{2}}(\log \frac{1}{\epsilon})^3)$-time oracle for backtracking from $S(M; [0,\omega])$ to $M$. Note that $S(\tilde{M})$ is a set, we have $|S(\tilde{M})|\le (S(\tilde{M}))^{\max}  = \Sigma(\tilde{M}) = \OO(\epsilon^{-\frac{3}{2}}\log\frac{1}{\epsilon})$. The time to compute $S(M; [0,\omega])$ from $S(\tilde{M})$ is $\OO(\epsilon^{-\frac{3}{2}}\log\frac{1}{\epsilon})$,  then the overall processing time is $\OO(\epsilon^{-\frac{3}{2}}(\log \frac{1}{\epsilon})^3)$. Thus Lemma~\ref{lemma:apx_small-group} has been proved for the case that $M$ is a sparse group. 
		
		It remains to prove Lemma~\ref{lemma:apx_small-group} for the case that $M$ is a dense group, which is the task of the following Subsection~\ref{sec:dense_gro_3/2}.

		\subsection{Handling Dense Group}\label{sec:dense_gro_3/2}
		The goal of this section is to prove Lemma~\ref{lemma:apx_small-group} for the case that $|M|=\Omega({\epsilon}^{-\frac{1}{2}}\log \frac{1}{\epsilon})$.

		In the following, we assume that $|M|=\Omega({\epsilon}^{-\frac{1}{2}}\log \frac{1}{\epsilon})$. Define $ L:=\frac{100\cdot \Sigma (M) \cdot \sqrt{\frac{1}{\epsilon}} \log \frac{1}{\epsilon}}{|M|}$. We approximate $S(M)$ and build an oracle for backtracking in the following three steps:
		(1). handle $S(M) \cap [0,L]$ and $S(M) \cap [\Sigma(M)-L,\Sigma(M)]$; (2). handle $S(M)\cap (L, \Sigma(M)-L)$; (3). handle $S(M)$.
		
		\paragraph{Step 1: Handling $S(M) \cap [0,L]$ and $S(M) \cap [\Sigma(M)-L,\Sigma(M)]$.} { }~\\
		
		In this step, we aim to prove the following claim.
		\begin{claim}\label{obs:3/2_mid_1}
			In $ \OO(\frac{d}{\epsilon}+ \frac{2^{{d}}}{\epsilon^{\frac{3}{2}+\frac{1}{d}}} (\log \frac{1}{\epsilon})^2)$ processing time, we can:
			\begin{itemize}
				\item[(i).] Compute $S(M) \cap \left([0,L]\cup [\Sigma(M)-L,\Sigma(M)]\right)$;
				\item[(ii).] Meanwhile build an $\OO( \frac{2^{{d}}}{\epsilon^{\frac{3}{2}+\frac{1}{d}}} (\log \frac{1}{\epsilon})^2)$-time oracle for backtracking from $S(M) \cap \left([0,L]\cup [\Sigma(M)-L,\Sigma(M)]\right)$ to $M$. That is, given any $c\in S(M) \cap \left([0,L]\cup [\Sigma(M)-L,\Sigma(M)]\right)$, in $\OO( \frac{2^{{d}}}{\epsilon^{\frac{3}{2}+\frac{1}{d}}} (\log \frac{1}{\epsilon})^2)$ time, the oracle will return $M_c \subset M$ such that $\Sigma(M_c) = c$.
			\end{itemize}
			Moreover, it holds that $\left|S(M) \cap \left([0,L]\cup [\Sigma(M)-L,\Sigma(M)]\right)\right| = \OO(\epsilon^{-\frac{3}{2}}\log \frac{1}{\epsilon})$. 
		\end{claim}
		
		\begin{proof}
			
			We first consider $S(M;[0,L])$. Notice that elements in $\tilde{M}$ are $(\epsilon,d,d-1)$-smooth numbers. Recall Lemma~\ref{lemma:tree-alg}, for any $k\in \mathbb{N}\cap [0,d-1]$, in $T_L = \OO(\frac{1}{\epsilon}+(d-1) \cdot |\tilde{M}|+ \Sigma(\tilde{M}) {\epsilon^{\frac{k }{d}}} \log (|\tilde{M}|)\log (\Sigma(\tilde{M}) {\epsilon^{\frac{k }{d}}})+  2^{2k+1}{\epsilon^{\frac{-1}{d}}}{\frac{L}{\beta}}  \log ({\frac{L}{\beta}}))$ processing time, we can compute $S(\tilde{M};[0,\frac{L}{\beta}])$, and meanwhile build a $T'_L$-time oracle for backtracking from $S(\tilde{M};[0,\frac{L}{\beta}])$ to $\tilde{M}$, where $T'_L = \OO( \Sigma(\tilde{M}) {\epsilon^{\frac{k }{d}}} \log (|\tilde{M}|)\log (\Sigma(\tilde{M}) {\epsilon^{\frac{k }{d}}})+  2^{2k+1}{\epsilon^{\frac{-1}{d}}}{\frac{L}{\beta}}  \log ({\frac{L}{\beta}})))$. It is easy to observe that $S({M};[0,{L}])=\beta S(\tilde{M};[0,\frac{L}{\beta}])$ and the oracle for backtracking from $ S(\tilde{M};[0,\frac{L}{\beta}])$ to $\tilde{M}$ directly yields an $\OO(T'_L)$-time oracle for backtracking from $S({M};[0,{L}])$ to $M$. 	Recall that $L = \frac{100\cdot \Sigma (M) \cdot \sqrt{\frac{1}{\epsilon}} \log \frac{1}{\epsilon}}{|M|}$, then $\frac{L}{\beta} = \frac{100\cdot \Sigma (\tilde{M}) \cdot \sqrt{\frac{1}{\epsilon}} \log \frac{1}{\epsilon}}{|\tilde{M}|}$. Since $\tilde{M} \subset [\frac{1}{4\epsilon}, \frac{1}{\epsilon}]$ and $|\tilde{M}| = |M| = \OO(\frac{1}{\epsilon})$, we have $\frac{L}{\beta} = \OO(\epsilon^{-\frac{3}{2}}\log \frac{1}{\epsilon})$ and $\Sigma(\tilde{M}) = \OO(\frac{1}{\epsilon^2})$.
			Note that $d\in \mathbb{N}_{+}$ is divisible by $2$, by fixing $k$ to $\frac{d}{2}$, we have $T_L = \OO(\frac{d}{\epsilon}+ \frac{2^{{d}}}{\epsilon^{\frac{3}{2}+\frac{1}{d}}} (\log \frac{1}{\epsilon})^2)$ and $T'_L = \OO(\frac{2^{{d}}}{\epsilon^{\frac{3}{2}+\frac{1}{d}}} (\log \frac{1}{\epsilon})^2)$. Note that $S(\tilde{M};[0,\frac{L}{\beta}])$ is a set, thus $|S(M;[0,L])| = |S(\tilde{M};[0,\frac{L}{\beta}])|= \OO(\frac{L}{\beta})=\OO(\epsilon^{-\frac{3}{2}}\log \frac{1}{\epsilon})$. 
			
			Note that given any $E \subset M$ with $\Sigma(E) \in [0,L]$, then $M \backslash E$ satisfies $\Sigma(M \backslash E) \in [\Sigma(M)-L,\Sigma(M)]$, and conversely, given any $E' \subset M$ with $\Sigma(E') \in [\Sigma(M)-L,\Sigma(M)]$, then $M \backslash E'$ satisfies $\Sigma(M \backslash E') \in [0,L]$.  Thus $S(M)\cap [\Sigma(M)-L,\Sigma(M)]=\left\{\Sigma(M)-c : c \in S(M; [0,L]) \right\}$, which implies that given  $S(M; [0,L])$,  in $|S(M; [0,L])| = \OO(\epsilon^{-\frac{3}{2}}\log \frac{1}{\epsilon})$ processing time, we can obtain $S(M)\cap [\Sigma(M)-L,\Sigma(M)]$ from $S(M; [0,L])$. Observe that 
			$\left |S(M)\cap [\Sigma(M)-L,\Sigma(M)]\right| = \left|S(M; [0,L])\right| = \OO(\epsilon^{-\frac{3}{2}}\log \frac{1}{\epsilon})$, we have $\left|S(M) \cap \left([0,L]\cup [\Sigma(M)-L,\Sigma(M)]\right)\right| = \OO(\epsilon^{-\frac{3}{2}}\log \frac{1}{\epsilon})$. 
			Moreover, given any $T'_L$-time oracle for backtracking from $S(M; [0,L])$ to $M$, then this oracle directly yields an $\OO(T'_L)$-time oracle for backtracking from $S(M)\cap [\Sigma(M)-L,\Sigma(M)]$ to $M$. 
			
			To summarize, in overall $\OO(\frac{d}{\epsilon}+ \frac{2^{{d}}}{\epsilon^{\frac{3}{2}+\frac{1}{d}}} (\log \frac{1}{\epsilon})^2)$ time, we can compute $S(M) \cap \left([0,L]\cup [\Sigma(M)-L,\Sigma(M)]\right)$ and build an $\OO(\frac{2^{{d}}}{\epsilon^{\frac{3}{2}+\frac{1}{d}}} (\log \frac{1}{\epsilon})^2)$-time oracle for backtracking from $S(M) \cap \left([0,L]\cup [\Sigma(M)-L,\Sigma(M)]\right)$ to $M$. \qed\end{proof}

		\paragraph{Step 2: Handling $S(M)\cap (L,\Sigma (M) - L)$.}{ }~\\
		In this step, we aim to prove the following claim.
		\begin{claim}\label{obs:3/2_mid_}
			In $\OO(\frac{1}{\epsilon}\log \frac{1}{\epsilon})$ processing time, we can:
			\begin{itemize}
				\item[(i).] Compute a subset $\Gamma_{mid} \subset  S(M) \cap (L,\Sigma(M)-L)$ satisfying the following conditions: (1). $|\Gamma_{mid} | = \OO(\frac{1}{\epsilon})$; (2). $\Gamma_{mid}  \subset [0,\omega]$; (3). Given any $s \in S(M)\cap (L,\Sigma(M)-L)\cap [0,\omega]$, there exists $s'\in \Gamma_{mid}$ such that $|s'-s|\le \epsilon \omega $.
				\item[(ii).] Meanwhile build an $\OO(1)$-time oracle for backtracking from $\Gamma_{mid}$ to $M$. That is, given any $c\in \Gamma_{mid}$, in $\OO(1)$ time, the oracle will return $M_c \subset M$ such that $\Sigma(M_c) = c$.
			\end{itemize}
		\end{claim}
		\begin{proof}
			Recall that $\tilde{M} \subset [\frac{1}{4\epsilon}, \frac{1}{\epsilon}]$ and each element in $\tilde{M}$ is distinct, moreover, $\frac{100\cdot \Sigma (\tilde{M}) \cdot \sqrt{\frac{1}{\epsilon}} \log \frac{1}{\epsilon}}{|\tilde{M}|}=\frac{L}{\beta}$. Notice that we now consider the case that $|\tilde{M}|= |M| = \Omega(\epsilon^{-\frac{1}{2}}\log \frac{1}{\epsilon})$, recall Theorem~\ref{the:dense}, in $\OO(\frac{1}{\epsilon} \log\frac{1}{\epsilon})$ time, we can construct a data structure which allows us to solve the SUBSET SUM problem $(\tilde{M},c)$ in $\OO(\log \frac{1}{\epsilon})$ time, where $c$ is any integer in $(\frac{L}{\beta},\Sigma(\tilde{M})- \frac{L}{\beta})$.   
			
			We define a set $\mathcal{C}$ as follows: if $\frac{\omega}{\beta} \le \frac{L}{\beta}$, let $\mathcal{C} = \emptyset$; else let $\mathcal{C}=\{c_1,c_2,\cdots,c_{\iota}\}$ denote the set of all different integers in $\left\{\lfloor\epsilon \frac{\omega}{\beta} \rfloor, \lfloor 2\epsilon \frac{\omega}{\beta} \rfloor, \cdots, \lfloor\frac{\omega}{\beta} \rfloor, \lfloor \min\{ \frac{\omega}{\beta}, \Sigma (\tilde{M})-\frac{L}{\beta}\} \rfloor \right\} \cap \left(\frac{L}{\beta},\min\{ \frac{\omega}{\beta} , \Sigma (\tilde{M})-\frac{L}{\beta}\}\right)$. Then we solve SUBSET SUM problem $(\tilde{M},c)$ for every integer $c \in \mathcal{C}$. Denote by $E_i$ the solution of SUBSET SUM problem $(\tilde{M},c_i)$, let $\mathcal{E}: = \{E_1,E_2,\cdots,E_{\iota}\}$ and let $\Gamma_{mid}: =\{\beta\Sigma(E_1), \beta\Sigma(E_2),\cdots,\beta\Sigma(E_{\iota})\}$. We build a dictionary $\mathscr{D}$ as follows: each element in $\mathscr{D}$ is a key-value pair $(\beta\Sigma(E_i), \mathscr{D}[\beta\Sigma(E_i)])$, where $\mathscr{D}[\beta\Sigma(E_i)]=\beta E_i$ and $i=1,2,\cdots,\iota$. The collection of keys of $\mathscr{D}$ is $\Gamma_{mid}$.
			%to maintain a one-to-one correspondence between $\mathcal{E}$ and $\Sigma_{mid}$: $E_i \leftrightarrow \beta \Sigma(E_i)$. The total processing time to obtain $\Sigma_{mid}$ and build this dictionary is $T_{mid} = \OO(\frac{1}{\epsilon} \log\frac{1}{\epsilon})$.
			
			In the next, we show that $\Gamma_{mid}$ satisfies the conditions required in Claim~\ref{obs:3/2_mid_}. Apparently, $|\Gamma_{mid}| = \OO(\frac{1}{\epsilon})$ and $\Gamma_{mid} \subset [0,\omega]$. Consider any $s\in S(M)\cap (L,\Sigma(M)-L) \cap [0,\omega]$, it follows that  $\frac{s}{\beta} \in S(\tilde{M}) \cap (\frac{L}{\beta},\Sigma(\tilde{M})-\frac{L}{\beta})\cap [0,\frac{\omega}{\beta}]$%.  Note that $0<\frac{s}{\beta} \le \frac{\omega}{\beta}$, 
			and there exists some integer $k \in [0,\frac{1}{\epsilon}]$ such that $k\epsilon \frac{\omega}{\beta} < \frac{s}{\beta} \le (1+k)\epsilon \frac{\omega}{\beta}$, thus we have
			$$\left\lfloor k\epsilon  \frac{\omega}{\beta}\right\rfloor < \frac{s}{\beta} \le \min \left\{\left\lfloor (1+k)\epsilon  \frac{\omega}{\beta} \right\rfloor,  \left\lfloor \min  \left\{\frac{\omega}{\beta} , \sum \left(\tilde{M}\right)-\frac{L}{\beta}\right\}\right \rfloor \right\}.$$
			Observe that $ \min \left\{\left\lfloor (1+k)\epsilon   \frac{\omega}{\beta} \right\rfloor,  \left\lfloor \min  \left\{\frac{\omega}{\beta} , \sum \left(\tilde{M}\right)-\frac{L}{\beta}\right\}\right \rfloor \right\}\in \mathcal{C}$,  which guarantees that there exists some $E_k \in \mathcal{E}$ satisfying $\Sigma(E_k) \in [k\epsilon \frac{\omega}{\beta}, (k+1)\epsilon \frac{\omega}{\beta}]$, then $\beta \Sigma(E_k) \in \Gamma_{mid}$ and $|s-\beta \Sigma(E_k)| \le \beta \cdot \epsilon\frac{\omega}{\beta} = \epsilon \omega$.

			An $\OO(1)$-time oracle for backtracking from $\Gamma_{mid}$ to $M$ works as follows: given any $c\in \Gamma_{mid}$, the oracle uses dictionary $\mathscr{D}$ to return $\mathscr{D}[c]$. It is easy to see that $\mathscr{D}[c] \subset M$ and $\Sigma(\mathscr{D}[c]) = c$.\qed
			%Recall that we have created a dictionary to maintain a one-to-one correspondence between $\mathcal{E}$ and $\Sigma_{mid}$: $E_i \leftrightarrow \beta \Sigma(E_i)$. Thus given any $c\in \Sigma_{mid}$, in $\OO(1)$ processing time, we can determine $E_k \subset \mathcal{E}$ satisfying $c= \beta \Sigma(E_k) = \Sigma(\beta E_K)$. Note that $\beta E_k \subset M$. Thus this dictionary is exactly an $\OO(1)$-time oracle for backtracking from $\Sigma_{mid}$ to $M$.
		\end{proof}
		
		\paragraph{Step 3: Handling $S(M)$.} %We are now ready to approximate $S(M)$.
		
		%According to Observation~\ref{obs:3/2_mid_1}, in $T_{LR} = \OO(\frac{d}{\epsilon}+ \frac{2^{{d}/{2}}}{\epsilon^{\frac{3}{2}+\frac{1}{d}}} (\log \frac{1}{\epsilon})^2)$ processing time, we can compute $S(M) \cap ([0,L]\cup [\Sigma(M)-L,\Sigma(M)])$. According to Observation~\ref{obs:3/2_mid_}, in $T_{mid} = \OO(\frac{1}{\epsilon}\log \frac{1}{\epsilon})$ processing time, we obtain $\Sigma_{mid} \subset S(M; (L,\Sigma(M)-L))$. 
		
		Given $\Gamma_{mid}$, which is obtained in Step 2. Let $\Gamma_{LR} = \left (S(M) \cap ([0,L]\cup [\Sigma(M)-L,\Sigma(M)])\right) \cap [0,\omega]$ and let $\Gamma = \Gamma_{mid} \cup \left (S(M) \cap ([0,L]\cup [\Sigma(M)-L,\Sigma(M)])\right)$. Then we have the
		following Claim~\ref{obs:dense_gro}, which follows directly from Claim~\ref{obs:3/2_mid_1} and  Claim~\ref{obs:3/2_mid_}.
		\begin{claim}\label{obs:dense_gro}
			$\Gamma$ is an $(\epsilon, \omega)$-approximate set of $S(M)$. Moreover, we can build an 
			$\OO(\frac{2^{{d}}}{\epsilon^{\frac{3}{2}+\frac{1}{d}}} (\log \frac{1}{\epsilon})^2)$-time oracle for backtracking from $\Gamma$ to $M$.
		\end{claim}
		
		According to Claim~\ref{obs:3/2_mid_1} and Claim~\ref{obs:3/2_mid_}, we have $|\Gamma_{LR}| = \OO(\epsilon^{-\frac{3}{2}}\log \frac{1}{\epsilon})$ and $|\Gamma_{mid}| = \OO(\frac{1}{\epsilon})$, it follows that $|\Gamma| = \OO(\epsilon^{-\frac{3}{2}}\log \frac{1}{\epsilon})$. Recall that the total time for computing $\Gamma_{LR}$ and $\Gamma_{mid}$ is $\OO(\frac{d}{\epsilon}+ \frac{2^{{d}}}{\epsilon^{\frac{3}{2}+\frac{1}{d}}} (\log \frac{1}{\epsilon})^2)$, thus the overall time for computing $\Gamma$ is  $\OO(\frac{d}{\epsilon}+ \frac{2^{{d}}}{\epsilon^{\frac{3}{2}+\frac{1}{d}}} (\log \frac{1}{\epsilon})^2)$. Combine with Claim~\ref{obs:dense_gro}, we have proved Lemma~\ref{lemma:apx_small-group} for the case that $M$ is dense,  which completes the proof of Lemma~\ref{lemma:apx_small-group}, furthermore, completes the proof of Lemma~\ref{lemma:sumset-app-stru}. Till now, we have completed the proof of  Theorem~\ref{the:sub_sum}.

		\section{An $\tilde\OO(n+\epsilon^{-1})$-time weak $(1-\tilde\OO(\epsilon))$-approximation algorithm for UNBOUNDED SUBSET SUM.}\label{sec:1/3_weak_apx}
		%In this section, we consider weak approximation for UNBOUNDED SUBSET SUM. More precisely, given an arbitrary UNBOUNDED SUBSET SUM instance $(X,t)$, where $X = \{x_1,x_2,\cdots,x_n\}$ is a set of different positive integers and $t \ge 0$ is a fixed constant. Denote by $OPT$ the optimal objective value of the instance $(X,t)$, our goal is to find a multiset $X'$ with $set(X') \subset X$ such that$$(1-\tilde\OO(\epsilon))OPT\le \Sigma(X') \le (1+\tilde\OO(\epsilon))t,$$ moreover, the time to determine $X'$ is $\tilde\OO(n+\epsilon^{-1})$.
		
		The goal of this section is to prove the following Theorem.
		\begin{theorem}\label{the:un_boubded_sum}
			There is a deterministic weak $(1-\epsilon)$-approximation algorithm for UNBOUNDED SUBSET SUM running in $\tilde\OO(n+{\epsilon}^{-1})$ time.
		\end{theorem}
		Given any UNBOUNDED SUBSET SUM instance un-$(X,t)$, where $X=\{x_1,x_2,\cdots,x_n\}$ is a set of different positive integers and $t >0$ is a fixed constant. Let $OPT$ be the optimal objective value of the instance un-$(X,t)$, our goal is to find a multiset $X'$ with $set(X') \subset X$ such that
		$$(1-\tilde\OO(\epsilon))OPT\le \Sigma(X') \le (1+\tilde\OO(\epsilon))t,$$ moreover, the time to determine $X'$ is $\tilde\OO(n+\epsilon^{-1})$.
		
		Without loss of generality, we can assume that $0<x \le (1+\tilde\OO(\epsilon))t$ for every $x \in X$. Moreover, we have the following useful preliminary lemma.
		
		\begin{lemma}\label{lemma:pre_2_un}
			Given any UNBOUNDED SUBSET SUM instance un-$(X,t)$, where $0< x \le (1+\tilde\OO(\epsilon))t$ for every $x\in X$. Let $OPT$ be the optimal objective value of un-$(X,t)$. In $\OO(|X|)$ processing time, we can either find a weak $(1-\tilde\OO(\epsilon))$-approximation solution for un-$(X,t)$, or assert that $\epsilon t < x < t$ for every $x\in X$.
		\end{lemma}
		\begin{proof}
			Note that if there exists some $x\in X$ with $t\le x \le (1+\tilde\OO(\epsilon)) t$, then $\{x\}$ yields a weak $(1-\tilde\OO(\epsilon))$-approximation solution for un-$(X,t)$. If there exists some $x \in X$ with $0<x \le \tilde\OO(\epsilon) t$, then a multiset $X'$ with $set(X') =\{x\}$ and $|X'| =\lfloor\frac{t}{x}\rfloor+1$ yields a weak $(1-\tilde\OO(\epsilon))$-approximation solution for un-$(X,t)$.
			
			It is easy to see that, in $\OO(|X|)$ time, we can determine whether there exists some $x\in X$ satisfying $t \le x \le (1+\tilde\OO(\epsilon)) t$ or $0<x \le \tilde\OO(\epsilon) t$. \qed
		\end{proof}
		
		The rest of this section is dedicated to proving Theorem~\ref{the:un_boubded_sum}. Consider any UNBOUNDED SUBSET SUM instance un-$(X,t)$, where $\epsilon t < x < t$ for every $x \in X$. In Section~\ref{subsec:unboun_pre}, we focus on simplifying $\text{un-}(X,t)$ to obtain a BOUNDED SUBSET SUM instance $(Y,\hat{t})$, where $Y$ is divided into $\OO(\log|X|(\log \frac{1}{\epsilon})^2)$ subgroups: $Y^j_k$'s.  In particular, we derive Lemma~\ref{lemma:pre_unbo} which guarantees that towards approximating un-$(X,t)$, it is sufficient to consider $(Y,\hat{t})$. In Section~\ref{subsec:un_2_app}, we focus on approximating $(Y,\hat{t})$. Corollary~\ref{coro:sub_sum_top} guarantees that towards approximating $(Y,\hat{t})$, it is sufficient to consider each $Y^j_k$. Subsection~\ref{subsec:un_g_1} and Subsection~\ref{subsec:m>_2} are dedicated to approximating $Y^j_k$ for two separate cases.

		Before we move on to the details,	we briefly present the main idea. The crucial observation that leads to an almost linear time algorithm for UNBOUNDED SUBSET SUM is the following sparsification lemma by Klein~\cite{DBLP:conf/soda/Klein22} on the exact algorithm for UNBOUNDED SUBSET SUM.  
		\begin{lemma}[CF. Corollary 1. from~\cite{DBLP:conf/soda/Klein22}]\label{lemma:supp}
			Given a set $X = \{x_1, x_2,\cdots,x_n\}\subset \mathbb{N}$. If there is a feasible solution to the following integer program: 
			\begin{equation*} 
				\begin{aligned}
					\mathbf{IP_t:}  & \sum^{n}_{i=1} v_i x_i  = t \\  
					& \vev=(v_1,v_2,\cdots,v_n) \in \mathbb{N}^n, 
				\end{aligned}
			\end{equation*}
			Then there is a solution $\vev$ to $\mathbf{IP_t}$ such that $\text{supp}(\vev)\le \log_{2}(X^{min})+1$, where $\text{supp}(\vev)$ denotes the number of nonzero coordinates of $\vev$, i.e., $\text{supp}(\vev)=|\{i:v_i\neq 0\}|$.  
			%   whose number of non-zero components does not exceed $\log_{2}(X^{min})+1.$
		\end{lemma}
		%$\{m_1,m_2,\cdots,m_n\} \in \mathbb{N}^{n}$ in $\tilde\OO(n+\epsilon^{-1})$ time, such that $$
		%(1-\tilde\OO(\epsilon))OPT\le \sum^{n}_{i=1} x_i m_i \le (1+\tilde\OO(\epsilon))t.$$
		Lemma~\ref{lemma:supp} implies that, for UNBOUNDED SUBSET SUM problem, any subset-sum, including $OPT$, can be achieved by a sparse solution $\vev$ in the sense its support $\text{supp}(\vev)$ is bounded by a logarithmic value. More precisely, given $(X,t)$ as an input, we let $T(X):=\{\sum^{n}_{i=1} v_i x_i|\vev \in \mathbb{N}^n\}$. Suppose $x_i=(1/\epsilon)^{\OO(1)}$ for simplicity, %(in particular, by Lemma~\ref{the:un_boubded_sum} we may assume $x_i\le 1/\epsilon$), 
		then $T(X)=T_{sparse}(X)$ by Lemma~\ref{lemma:supp}, where $T_{sparse}(X):=\{\sum^{n}_{i=1} v_i x_i|\vev \in \mathbb{N}^n, \text{supp}(\vev)=\OO(\log \frac{1}{\epsilon}) \}$. %and  $\text{supp}(\vev)$ denotes the support of $\vev$ (i.e., the number of non-zero coordinates of $\vev$). 
		Now consider an arbitrary $\sum^{n}_{i=1} v_i x_i\in T_{sparse}(X)$. Those nonzero $v_i$'s do not necessarily take the same value. However, by writing each nonzero $v_i$ into a binary number, $v_i=\sum_h \gamma_h(v_i)\cdot 2^h$ where $\gamma_h(v_i)\in\{0,1\}$, we have the following equivalent expression:
		$$\sum^{n}_{i=1} v_i x_i=\sum_h 2^h\sum^{n}_{i=1} \gamma_h(v_i) x_i,  \quad\quad \text{supp}(\gamma_h(v_1),\gamma_h(v_2),\cdots,\gamma_h(v_n))=\OO(\log\frac{1}{\epsilon}).$$
		Therefore, let $G_{\log}$ stands for the set of sums of at most $\OO(\log\frac{1}{\epsilon})$ elements of $X$, then the above equation implies that $T_{sparse}(X)\subset G_{\log}\oplus 2G_{\log}\oplus\cdots\oplus 2^hG_{\log}\cdots$. On the other hand, $G_{\log}\oplus 2G_{\log}\oplus\cdots\oplus 2^hG_{\log}\cdots\subset T(X)=T_{sparse}(X)$, hence $T_{sparse}(X)= G_{\log}\oplus 2G_{\log}\oplus\cdots\oplus 2^hG_{\log}\cdots$. Using the above equation and the fact that $G_{\log}$ can be approximated in nearly linear time, Theorem~\ref{the:un_boubded_sum} can be proved. In the subsequent subsections, we provide the details. %$C_{\log}=\underbrace{X\oplus X\oplus\cdots\oplus X}_{\OO(\log\frac{1}{\epsilon})}$ 

		\subsection{Preprocessing the UNBOUDED SUBSET SUM instance.}\label{subsec:unboun_pre}
		%In this section, we focus on simplifying $(X,t)$. 
		%We will step by step modify the instance, and Lemma~\ref{lemma:pre_unbo} follows directly after all the modification operations.
		
		%Let $\epsilon>0$ be a sufficiently small number.  
		%Note that our goal is to find a weak $(1-\tilde\OO(\epsilon))$-approximation solution for un-$(X,t)$. Without loss of generality, we can assume that $0<x \le (1+\tilde\OO(\epsilon))t$ for every $x \in X$. 
		
		%In the following, we assume that $\epsilon t < x < t$ for every $x\in X$. We focus on simplifying $(X,t)$ and show that we can construct a reduced instance $(Y,\hat{t})$. Formally, we have the following lemma.

		Given any UNBOUNDED SUBSET SUM instance un-$(X,t)$, where $X = \{x_1,x_2,\cdots,x_n\}$ is a set of distinct positive integers and $t>0$ is a fixed constant. Assuming that $\epsilon t<x < t$ for every $x \in X$, the goal of this section is to simplify $(X,t)$. %and show that we can construct a reduced instance $(Y,\hat{t})$. 
		Formally, we have the following lemma.
		
		\begin{lemma}\label{lemma:pre_unbo}
			Given any UNBOUNDED SUBSET SUM instance un-$(X,t)$, where $\epsilon t < x < t$ for every $x\in X$. 
			Let $OPT$ be the optimal objective value of un-$(X,t)$.  In $\OO\left((|X|+ \frac{1}{\epsilon})(\log |X| )^2(\log \frac{1}{\epsilon})^2\right)$ time, the followings can be achieved: 
			\begin{enumerate}
				\item[(i).] We can obtain a modified BOUNDED SUBSET SUM instance $(Y,\hat{t})$ satisfying the following conditions:%, where $Y$ and $\hat{t}$ satisfy the following conditions:
				\begin{itemize}
					\item[(A)] $\hat{t}=(1+\epsilon)\epsilon^{-3}$. 
					\item[(B)] The optimal objective value of $(Y,\hat{t})$ is at least $\frac{(1-\epsilon)  OPT}{\epsilon^3 t}$.
					\item[(C)]  $Y$ is divided into $\OO(\log |X| (\log \frac{1}{\epsilon})^2)$ groups: $Y^j_{k}$'s, i.e., $Y = \mathop{\dot{\cup}}\limits_{j}(\mathop{\dot{\cup}}\limits_{k} Y^j_k)$. %$\{Y^j_{k}\ |1\le j \le \OO(\log \frac{1}{\epsilon}) \text{ \ and\ } 1\le  k \le \OO(\log n\log \frac{1}{\epsilon})\}$
					\item[(D)]  $set(Y^j_{k})$ is explicitly given for every $Y^j_k$ and $ \mathop{\sum}\limits_{j}\mathop{\sum}\limits_{k} |set(Y^j_k)| \le |X|$.
					\item [(E)]  Each subgroup $Y^j_{k}$ satisfies the following conditions:
					
					\begin{enumerate}
						\item $Y^j_{k}\subset [\frac{2^{j-1}(1-\epsilon)}{\epsilon^2},\frac{2^{j}(1+\epsilon)}{\epsilon^2}]$;
						\item $Y^j_{k}= \rho^j_{k} \overline{Y}^j_k$, where $\rho^j_{k} = \Omega(\epsilon^{-1})$ and $\overline{Y}^j_k \subset \mathbb{N}_{+}\cap [\frac{1}{2\epsilon}, \frac{1}{\epsilon}]$;
						\item  Let $ n^j_k = \lfloor \frac{\hat{t}}{(Y^j_k)^{\min}}\rfloor$ and let
						\begin{align*}
							l^j_k =
							\begin{cases}
								&2^{1+pow(\log_{2}(\frac{1}{\epsilon})+1)} \hspace{9mm} \  \text{if } n^j_k  \le 2^{1+pow(\log_{2}(\frac{1}{\epsilon})+1)}\\
								&2n^j_k\cdot 2^{1+pow(\log_{2}(\frac{1}{\epsilon})+1)} \hspace{2mm} \  \text{if } n^j_k > 2^{1+pow(\log_{2}(\frac{1}{\epsilon})+1)}
							\end{cases}
							%	, \ \text{where} \ n^j_k = \lfloor \frac{(1+\tilde\OO(\epsilon))\hat{t}}{(Y^j_k)^{\min}}\rfloor.
						\end{align*}
						then	$Y^j_{k}$ consists of $l^j_k$ copies of $set(Y^j_{k})$, i.e., every element of $Y^j_k$ has the same multiplicity, which is $l^j_k$.
					\end{enumerate}
				\end{itemize}

				\item[(ii).]  Meanwhile, we can build an oracle for backtracking from $Y$ to $X$. Precisely, for each $Y^j_k$, given any $y \in Y^j_k$, in $\OO(1)$ time,  the oracle will return $x' \in X$ such that $|x' - \epsilon^3 t y | \le \epsilon x'$.
				
				%given subset $U^j_k \subset Y^j_k$ for every $Y^j_k$ and let $U'$ denote the multiset union of all $U^j_k$'s. Then $U' \subset Y$ and in linear time, the oracle will return a multiset $X'$ satisfying $set(X') \subset X$ and $| \Sigma(X') - \epsilon^3 t \Sigma(U')| \le \epsilon\Sigma(X')$.
				
			\end{enumerate}
		\end{lemma}
		\noindent\textbf{Rermark.}  We briefly explain how Lemma~\ref{lemma:pre_unbo} is leveraged to prove Theorem~\ref{the:un_boubded_sum}. Roughly speaking, Lemma~\ref{lemma:pre_unbo} reduces the UNBOUNDED SUBSET SUM to the bounded version (but has a special structure) so that we may utilize techniques developed in previous sections. More precisely, consider the second part of Lemma~\ref{lemma:pre_unbo}, for any subset $U^j_k \subset Y^j_k$, if $card_{U^j_k} [y]$ is known for every $y\in set(Y^j_k)$, then in $\OO(|set(Y^j_k)|)$ time, the oracle will return a multiset $X^j_k$ such that $set(X^j_k)\subset X$ and  $| \Sigma(X^j_k) - \epsilon^3 t \Sigma(U^j_k)| \le \epsilon\Sigma(X^j_k)$. Furthermore, given subset $U^j_k \subset Y^j_k$ for every $Y^j_k$ and let $U'$ be the multiset-union of all these $U^j_k$'s. For each $U^j_k$, if $card_{U^j_k} [y]$ is known for every $y\in set(Y^j_k)$, then in $\OO(\sum_{j}\sum_{k}|set(Y^j_k)|) = \OO(|X|)$ time, the oracle will return a multiset $X'$ such that $set(X')\subset X$ and  $| \Sigma(X') - \epsilon^3 t \Sigma(U')| \le \epsilon\Sigma(X')$. Moreover, note that the optimal objective value of $(Y,\hat{t})$ is at least $\frac{(1-\epsilon) OPT}{\epsilon^3 t} $, if $U'$ is a weak $(1-\tilde\OO(\epsilon))$-approximation solution of $(Y,\hat{t})$, then $X'$ is a weak $(1-\tilde\OO(\epsilon))$-approximation solution of un-$(X,t)$. Thus towards proving Theorem~\ref{the:un_boubded_sum}, it is sufficient to consider $(Y,\hat{t})$, which shall be handled in the next subsection, Subsection~\ref{subsec:un_2_app}.  %an $\tilde\OO(|X|+\epsilon^{-1})$-time algorithm for computing a weak $(1-\tilde\OO(\epsilon))$-approximation solution of $(Y,\hat{t})$ directly yields an $\tilde\OO(|X|+\epsilon^{-1})$-time algorithm for computing a weak $(1-\tilde\OO(\epsilon))$-approximation solution of un-$(X,t)$.

		%Note that optimal objective value of $(Y,\hat{t})$ is at least $\frac{1-\epsilon}{\epsilon^3 t} \cdot OPT$. The second part of Lemma~\ref{lemma:pre_unbo} guarantees that towards finding a weak $(1-\tilde\OO(\epsilon))$-approximation of UBNBOUNDED SUBSET SUM instance un-$(X,t)$, it is sufficient to find a weak $(1-\tilde\OO(\epsilon))$-approximation of modified BOUNDED SUBSET SUM instance $(Y,\hat{t})$.
		
		\iffalse{	Now, we have completed the preprocessing of the instance $(X,t)$. To summarize, we have the following lemma.
			\begin{lemma}%\label{lemma:pre_unbo}
				Given an UNBOUNDED SUBSET SUM instance $(X,t)$, in $\OO((n+\epsilon^{-1})\log n (\log \frac{1}{\epsilon})^2)$ preprocessing time, we can obtain a modified BOUNDED SUBSET SUM instance $(Y,\bar{t})$ satisfying conditions $(A)$ and $(B)$ such that there exists a 
				weak $(1-\tilde\OO(\epsilon))$-approximation of $(X,t)$ if and only there exists a weak $(1-\tilde\OO(\epsilon))$-approximation of $(Y,\bar{t})$. Moreover, using the same processing time, we can build a mapping 
				from $Y$ to $X$, such that given any $Y'\subset Y$, in linear time, 
				we can determine the corresponding multiset $X'$ with $set(X')\subset X$ and  $X'$ satisfies $|\Sigma(X')-\epsilon^3 t \Sigma(Y')|\le \OO(\epsilon)\Sigma(X')$.
		\end{lemma}}\fi
		
		The rest of this section is dedicated to proving Lemma~\ref{lemma:pre_unbo}. We will step-by-step modify the given UNBOUNDED SUBSET SUM instance un-$(X,t)$, and Lemma~\ref{lemma:pre_unbo} follows directly after all the modification operations. 
		
		\paragraph{Step 1: Scaling and Grouping.}  We scale $t$ and each element in $X$ by $\epsilon^3 t$. To be specific, we scale $t$ to $\bar{t} = t \cdot \frac{1}{\epsilon^3 t}$ and scale each $x_i \in X$ to $\bar{x}_i = x_i \cdot \frac{1}{\epsilon^3 t}$. Let $\overline{X}=\{\bar{x}_1,\bar{x}_2,\cdots,\bar{x}_n\}$. Notice that the optimal objective value of UNBOUNDED SUBSET SUM instance un-$(\overline{X},\bar{t})$ is at least $\frac{OPT}{\epsilon^3 t}$. 
		
		Observe that there is a one-to-one correspondence between $X$ and $\overline{X}$: $x \leftrightarrow   \frac{x}{\epsilon^3 t}$ for any $x\in X$. We define the mapping $\Psi_1$ from $\overline{X}$ to $X$ as follows: given any $y\in \overline{X}$, mapping $\Psi_1$ returns $\epsilon^3 t y$. Denote by $\Psi_1(y)$ the element returned by $\Psi_1$ given $y\in \overline{X}$.
		
		Note that $ \overline{X} \subset [\frac{1}{\epsilon^2},\frac{1}{\epsilon^3})$, we can divide $\overline{X}$ into $\eta =  \OO(\log \frac{1}{\epsilon})$ groups, denoted by $\overline{X}_1, \overline{X}_2, \cdots \overline{X}_{\eta}$, such that $\bar{x} \in \overline{X}_j$ if and only if $x \in [\frac{2^{j-1}}{\epsilon^2},\frac{2^{j}}{\epsilon^2})\cap \overline{X}$. %Let $\overline{X}_j = \emptyset$ if $\overline{X}\cap [\frac{2^{j-1}}{\epsilon^2},\frac{2^{j}}{\epsilon^2})=\emptyset$.
		
		The total processing time of Step 1 is $\OO(|X|)$.
		
		\paragraph{Step 2. Rounding and Further Grouping.} Consider each $\overline{X}_j$, note that $\overline{X}_j \subset  [\frac{2^{j-1}}{\epsilon^2},\frac{2^{j}}{\epsilon^2}]=[\frac{1}{\epsilon^{2+\lambda_j}}, \frac{2}{\epsilon^{2+\lambda_j}}]$, where $0 \le \lambda_j < 1$. %, we can rewrite $[\frac{2^{j-1}}{\epsilon^2},\frac{2^{j}}{\epsilon^2}]$ as $[\frac{1}{\epsilon^{2+\lambda'}}, \frac{2}{\epsilon^{2+\lambda'}}]$, where $\lambda' \ge 0$. 
		Recall Lemma~\ref{lemma:smooth_appro}, given $d = 1$ and $\alpha_j = 1+\lambda_j$, in $\OO\left((|\overline{X}_j|+\frac{1}{\epsilon})\log |\overline{X}_j|\log\frac{1}{\epsilon}\right)$ time, we can obtain a set $\Delta_j \subset \mathbb{R}$ with 	
		$\Delta_j \subset \Theta(\frac{1}{\epsilon^{1+\lambda_{j}}})$ and $|\Delta_j| = \OO(\log |\overline{X}_j|\log\frac{1}{\epsilon})$, moreover, we can round every $\bar{x} \in \overline{X}_j$ to the form $\rho h$ where $\rho \in \Delta_j$ and $h$ satisfying the following conditions: 
		\begin{subequations}
			\begin{align}
				&h\in \mathbb{N}_{+}\cap[\frac{1}{2\epsilon},\frac{1}{\epsilon}]; \label{eq:con_1_un}\\
				&|\bar{x}-\rho h| \le \epsilon \bar{x}. \label{eq:con_2_un}
			\end{align}
		\end{subequations}
		Denote by $\hat{X}_j$ the set of all such rounded elements obtained from $\overline{X}_j$. Let $Z_j = set\{ \hat{X}_j\}$ and let $\hat{t} = (1+\epsilon)\bar{t}$. We have $Z_j \subset [\frac{2^{j-1}(1-\epsilon)}{\epsilon^2},\frac{2^{j}(1+\epsilon)}{\epsilon^2}]$. Recall that the optimal objective value of UNBOUNDED SUBSET SUM instance un-$(\overline{X},\bar{t})$ is at least $\frac{OPT}{\epsilon^3 t}$.  Condition~\eqref{eq:con_2_un} guarantees that the optimal objective value of UNBOUNDED SUBSET SUM instance un-$(set({\cup}^{\eta}_{j=1} Z_j), \hat{t})$ is at least $\frac{(1-\epsilon)OPT}{\epsilon^3 t}$. %Towards finding a weak ()

		Notice that for each $\hat{X}_j$, there is a one-to-one correspondence between $\overline{X}_j$ and $\hat{X}_j$: $\bar{y} \leftrightarrow \hat{y}$ for any $\bar{y}\in \overline{X}_j$, where $\hat{y}$ is the factorized form of $\bar{y}$ obtained after the above rounding procedure. For each $\hat{y} \in Z_j$, let $\mathcal{M}^j_{\hat{y}}$ denote the set of all numbers in $\overline{X}_j$ factorized to the form $\hat{y}$. We define the mapping $\Psi^j_2$ from $Z_j$ to $\overline{X}_j$ as follows: given any $\hat{y} \in Z_j$, mapping $\Psi^j_2$ returns any one element in $\mathcal{M}^j_{\hat{y}}$. Denote by $\Psi^j_2(\hat{y})$ the element returned by $\Psi^j_2$ given $\hat{y} \in Z_j$.

		For each $Z_j$, let $\Delta_j = \{\rho^j_1,\rho^j_2,\cdots, \rho^j_{|\Delta_j|}\}$, where $|\Delta_j|= \OO(\log |\overline{X}_j|\log\frac{1}{\epsilon})$. We further divide $Z_j$ into $|\Delta_j|$ groups, denoted by $Z^j_1,Z^j_2,\cdots,Z^j_{|\Delta_{j}|}$, such that $y \in Z^j_k$ if and only $y$ is of the form $\rho^j_k h$.

		%Let $\hat{x}$ denote the rounded form of $\bar{x}\in \overline{X}_j$ and let $\hat{X}_j$ denote the set of all such rounded elements obtained from $\overline{X}_j$. We maintain a one-to-one correspondence between $X$ and $\dot{\cup}^{\eta}_{j=1}\hat{X}_j$: $x \leftrightarrow \hat{x}$. Note that a multiset $X'$ is a weak $(1-\tilde\OO(\epsilon))$-approximation solution to the UNBOUNDED SUBSET SUM instance $(X,t)$ if and only if $\frac{X'}{\epsilon^3 t}$ is a weak $(1-\tilde\OO(\epsilon))$-approximation solution to the UNBOUNDED SUBSET SUM instance $(\overline{X},\bar{t})$. Condition (ii) guarantees a multiset $\overline{X}'$ is a weak $(1-\tilde\OO(\epsilon))$-approximation solution to the UNBOUNDED SUBSET SUM instance $(\overline{X},\bar{t})$ if and only if $\hat{X}' = \{\hat{x} \ | \ \bar{x}\in \overline{X}'\}$ is a weak $(1-\tilde\OO(\epsilon))$-approximation solution to the UNBOUNDED SUBSET SUM instance $(\dot{\cup}^{\eta}_{j=1} \hat{X}_j,\bar{t})$. 
		
		The total preprocessing time of Step 2 is $\OO((|X|+\frac{1}{\epsilon})(\log |X|)^2 (\log \frac{1}{\epsilon})^2)$.
		
		\paragraph{Step 3. Reduce the UNBOUNDED SUBSET SUM to BOUNDED SUBSET SUM.} Note that $\hat{t}  = \frac{1+\epsilon}{\epsilon^3}$ and $Z^j_k \in [\frac{2^{j-1}(1-\epsilon)}{\epsilon^2}, \frac{2^{j}(1+\epsilon)}{\epsilon^2}]$ for every $Z^{j}_k$, thus any optimal solution of $\text{un-}(set(\dot{\cup}^{\eta}_{j=1} Z_j), \hat{t})$ contains at most $n^j_k = \lfloor \frac{\hat{t}}{(Z^j_k)^{\min}} \rfloor$ elements from $Z^j_k$. Note that elements in $Z^j_k$ are different from each other. For every $Z^j_k$, we define a multiset $Y^j_k$ such that $set\{Y^j_k\} = Z^j_k$ and elements in $Y^j_k$ have the same multiplicity of $l^j_k$, where
		
		\begin{align*}
			l^j_k =
			\begin{cases}
				&2^{1+pow(\log_{2}(\frac{1}{\epsilon})+1)} \hspace{9mm} \ \text{if } n^j_k  \le 2^{1+pow(\log_{2}(\frac{1}{\epsilon})+1)}\\
				&2n^j_k\cdot 2^{1+pow(\log_{2}(\frac{1}{\epsilon})+1)} \hspace{2mm} \ \text{if } n^j_k > 2^{1+pow(\log_{2}(\frac{1}{\epsilon})+1)}
			\end{cases}
		\end{align*}
		Let $Y = \mathop{\dot{\cup}}\limits_{1\le j\le \eta; 1\le k\le |\Delta_j|} Y^j_k$. Note that any optimal solution of BOUNDED SUBSET SUM instance $(Y, \hat{t})$ is a optimal solution of UNBOUNDED SUBSET SUM instance $(set({\cup}^{\eta}_{j=1} Z_j), \hat{t})$, vice versa. Thus the optimal objective value of $(Y,\hat{t})$ is at least $\frac{(1-\epsilon)OPT}{\epsilon^3 t}$.

		%Moreover, observe that any weak $(1-\tilde\OO(\epsilon))$-approximation solution of BOUNDED SUBSET SUM instance $(Y, \hat{t})$ is a weak $(1-\tilde\OO(\epsilon))$-approximation solution of UNBOUNDED SUBSET SUM instance $(set({\cup}^{\eta}_{j=1} Z_j), \hat{t})$. 
		
		Observe that the time to obtain $Y^j_k$ from $Z^j_k$ is $\OO(|Z^j_k| \log |Z^j_k|)$. Thus the total processing time of Step 3 is $\OO(|X| \log |X|)$.
		
		Till now, we complete the preprocessing procedure and the overall time for preprocessing UNBOUNDED SUBSET SUM instance $(X, t)$ is $\OO((|X|+\frac{1}{\epsilon})(\log |X|)^2 (\log\frac{1}{\epsilon})^2)$. This finishes the first half (item (i)) of Lemma~\ref{lemma:pre_unbo}. In the following, we prove the second half (item (ii)) of Lemma~\ref{lemma:pre_unbo}.

		\paragraph{\textbf{Modified instance after preprocessing.}} To summarize, we have reduced the UNBOUNDED SUBSET SUM instance un-$(X,t)$ to a BOUNDED SUBSET SUM instance $(Y,\hat{t})$ %, where $Y$ and $\hat{t}$ 
		satisfying the following conditions:
		\begin{itemize}
			\item[(A)] $\hat{t}=(1+\epsilon)\epsilon^{-3}$ 
			\item[(B)] The optimal objective value of $(Y,\hat{t})$ is at least $\frac{(1-\epsilon) OPT}{\epsilon^3 t}$.
			\item[(C)]  $Y$ has been divided into $\OO(\log |X| (\log \frac{1}{\epsilon})^2)$ groups: $Y^j_{k}$'s, i.e., $Y = \mathop{\dot{\cup}}\limits_{j} (\mathop{\dot{\cup}}\limits_{k} Y^j_k)$. %$\{Y^j_{k}\ |1\le j \le \OO(\log \frac{1}{\epsilon}) \text{ \ and\ } 1\le  k \le \OO(\log n\log \frac{1}{\epsilon})\}$
			\item[(D)]  $set(Y^j_{k})$ is explicit given for every $Y^j_k$ and $ \sum_{j}\sum_{k} |set(Y^j_k)| \le |X|$.
			\item [(E)]  Each subgroup $Y^j_{k}$ satisfies the following conditions:
			
			\begin{enumerate}
				\item [(a)] $Y^j_{k}\subset [\frac{2^{j-1}(1-\epsilon)}{\epsilon^2},\frac{2^{j}(1+\epsilon)}{\epsilon^2}]$;
				\item [(b)] $Y^j_{k}= \rho^j_{k} \overline{Y}^j_k$, where $\rho^j_{k} = \Omega(\epsilon^{-1})$ and $\overline{Y}^j_k \subset \mathbb{N}_{+}\cap [\frac{1}{2\epsilon}, \frac{1}{\epsilon}]$;
				\item [(c)]   Let $ n^j_k = \lfloor \frac{\hat{t}}{(Y^j_k)^{\min}}\rfloor$ and let
				\begin{align*}
					l^j_k =
					\begin{cases}
						&2^{1+pow(\log_{2}(\frac{1}{\epsilon})+1)} \hspace{9mm} \text{if } n^j_k  \le 2^{1+pow(\log_{2}(\frac{1}{\epsilon})+1)}\\
						&2n^j_k\cdot 2^{1+pow(\log_{2}(\frac{1}{\epsilon})+1)} \hspace{2mm} \text{if } n^j_k > 2^{1+pow(\log_{2}(\frac{1}{\epsilon})+1)}
					\end{cases}
					%	, \ \text{where} \ n^j_k = \lfloor \frac{(1+\tilde\OO(\epsilon))\hat{t}}{(Y^j_k)^{\min}}\rfloor.
				\end{align*}
				$Y^j_{k}$ consists of $l^j_k$ copies of $set(Y^j_{k})$, i.e., every element in $Y^j_k$ has the same multiplicity, which is $l^j_k$.
			\end{enumerate}
		\end{itemize}

		\paragraph{\textbf{Oracle for backtracking from $Y$ to $X$.}}  We now present the oracle %$\textbf{Or}^{X}_{Y}$ 
		for backtracking from $Y$ to $X$. For each $Y^j_{k}$, given $y\in Y^j_{k}$, the oracle %$\textbf{Or}^{X}_{Y}$ 
		works as follows:
		\begin{itemize}
			\item 	Recall that $Y^j_k \subset Z_j$ and $\Psi^j_2$ is a mapping from $Z_j$ to $\overline{X}_j$. The oracle %$\textbf{Or}^{X}_{Y}$
			uses $\Psi^j_2$ to first obtain $ \Psi^j_2(y)$. Observe that  $|\Psi^j_2(y) - y| \le \epsilon  \Psi^j_2(y)$  and $\Psi^j_2(y) \in \overline{X}$.
			
			\item Recall that $\Psi_1$ is a mapping from $ \overline{X}$ to $X$. Then the oracle %$\textbf{Or}^{X}_{Y}$ 
			uses  $\Psi_1$ to obtain and return $\Psi_1(\Psi^j_2(y))$. Note that $\Psi_1(\Psi^j_2(y)) \in X$ and $\Psi_1(\Psi^j_2(y)) = \epsilon^3 t \Psi^j_2(y)$.
		\end{itemize}
		Let $x'= \Psi_1(\Psi^j_2(y))$. The above backtracking procedure only takes $\OO(1)$ time. To summarize, we have $| x' - \epsilon^3 t y| = \epsilon^3 t \cdot |\Psi^j_2(y) -y | \le \epsilon^3 t \cdot \epsilon  \Psi^j_2(y) =\epsilon x'$. 
		
		Till now, we have completed the proof of Lemma~\ref{lemma:pre_unbo}.

		\subsection{Algorithm for UNBOUNDED SUBSET SUM.}\label{subsec:un_2_app}
		The goal of this section is to prove Theorem~\ref{the:un_boubded_sum}. Given any UNBOUNDED SUBSET SUM instance $\text{un-}(X,t)$, where $X = \{x_1,x_2,\cdots,x_n\}$ is a set of distinct positive integers and $t>0$ is a fixed constant. Let $OPT$ be the optimal objective value of $\text{un-}(X,t)$. Our goal is to find a weak $(1-\tilde\OO(\epsilon))$-approximation solution for $\text{un-}(X,t)$. Without loss of generality, we may assume that $0< x \le (1+\tilde\OO(\epsilon))t$ for every $x\in X$. Furthermore, according to Lemma~\ref{lemma:pre_2_un}, we may assume that $\epsilon t < x <t$ for every $x\in X$. Then by Lemma~\ref{lemma:pre_unbo}, in $\OO((n+\frac{1}{\epsilon}) (\log n )^2(\log \frac{1}{\epsilon})^2)$ time, we can obtain a bounded SUBSET SUM instance $(Y,\hat{t})$ satisfying conditions (A)(B)(C)(D)(E) and meanwhile build an oracle for backtracking from $Y$ to $X$ (see Lemma~\ref{lemma:pre_unbo} in Section~\ref{subsec:unboun_pre}). Condition (B) and the backtracking oracle guarantee that towards proving Theorem~\ref{the:un_boubded_sum}, it is sufficient to consider $(Y,\hat{t})$. Precisely, Theorem~\ref{the:un_boubded_sum} follows directly from the following Lemma~\ref{lemma:re_un_1}.

		\begin{lemma}\label{lemma:re_un_1}
			Give any UNBOUNDED SUBSET SUM instance $\text{un-}(X,t)$ where $\epsilon t < x < t$ for every $x\in X$. Let $(Y,\hat{t})$ be the modified instance returned by Lemma~\ref{lemma:pre_unbo}, where $Y$ is the multiset-union of all $Y^j_k$'s.
			Then in  $\tilde\OO(|X| +\epsilon^{-1})$ processing time, we can determine $U^j_k \subset Y^j_k$ for every $U^j_k$ such that 
			\begin{enumerate}
				\item[(i).]  For each $U^j_k$, $card_{U^j_k }[y]$ is explicit given for every $y\in set(Y^j_k)$.
				\item[(ii).] The multiset-union of all $U^j_k$'s is a weak $(1-\tilde\OO(\epsilon))$-approximation solution for $(Y,\hat{t})$.						
			\end{enumerate}				
		\end{lemma}

		Let $Y_{\le \hat{t}} = \{y\in Y  :  y \le \hat{t}\}$ and let $Y^{(j;k)}_{\le \hat{t}} = \{y\in Y^j_k : y \le \hat{t} \}$ for every $Y^j_k$, it holds that $Y_{\le \hat{t}} = \mathop{\dot{\cup}}\limits_j (\mathop{\dot{\cup}}\limits_k Y^{(j;k)}_{\le \hat{t}})$. Recall Lemma~\ref{lemma:pre_unbo}, $Y^j_k$ is $l^j_k$ copies of $set(Y^j_k)$, then $card_{Y^{(j;k)}_{\le \bar{t}} }[y] = l^j_k$ for every $y\in Y^{(j;k)}_{\le \bar{t}}$.  Note that $set(Y^{(j;k)}_{\le \hat{t}}) = set(Y^j_k) \cap [0,\hat{t}]$.
		Thus in $\OO(\mathop{\sum}\limits_j \mathop{\sum}\limits_k |set(Y^j_k)|)=\OO(|X|)$ time, we can obtain $Y^{(j;k)}_{\le \hat{t}}$ for every $Y^j_k$, furthermore, we can obtain $\Sigma(Y_{\le \hat{t}})=\mathop{\sum}\limits_j \mathop{\sum}\limits_k \left( l^j_k  \Sigma(set(Y^{(j;k)}_{\le \hat{t}}))\right)$.  
		
		Let $OPT^*$ be the optimal objective value of $(Y,\hat{t})$.  Recall Lemma~\ref{lemma:prea}, if $\Sigma(Y_{\le \hat{t}}) < {\hat{t}}/{2}$, then $OPT^*  < {\hat{t}}/{2}$ and $Y_{\le \hat{t}}$ is an optimal solution of $(Y,\hat{t})$, else if $\Sigma(Y_{\le \hat{t}}) \ge  {\hat{t}}/{2}$, we can assert that $OPT^* \ge {\hat{t}}/{2}$. Thus we only need to consider the case that $OPT^* \ge {\hat{t}}/{2}$.

		In the following, we assume that $OPT^* \ge {\hat{t}}/{2}$. Then a subset $Y'\subset Y$ satisfying $|\Sigma(Y') -OPT^*| \le \tilde\OO(\epsilon) \hat{t}$ is a weak $(1-\tilde\OO(\epsilon))$-approximation solution of $(Y,\hat{t})$. Recall Lemma~\ref{lemma:pre_unbo}, $set(Y^j_k)$ is explicit given for every $Y^j_k$ and $\sum_j\sum_k |set(Y^j_k)| \le |X|$. Thus the following Lemma~\ref{lemma:un-subset-sum-l} implies Lemma~\ref{lemma:re_un_1} directly.

		\begin{lemma}\label{lemma:un-subset-sum-l}
			Give any UNBOUNDED SUBSET SUM instance $\text{un-}(X,t)$ where $\epsilon t < x < t$ for every $x\in X$. Let $(Y,\hat{t})$ be the modified instance returned by Lemma~\ref{lemma:pre_unbo}, where $Y$ is the multiset-union of all $Y^j_k$'s.
			Then in $\tilde\OO(\epsilon^{-1})$ processing time, we can 
			\begin{itemize}
				\item[(i).] Compute an $(\tilde\OO(\epsilon),\hat{t})$-approximate set $C$ with cardinality of $\OO(\epsilon^{-1})$ for $S(Y)$.
				\item[(ii).] Meanwhile build an $\tilde\OO(\epsilon^{-1})$-time oracle for backtracking from $C$ to $Y$. Here $Y$ is the multiset-union of all $Y^j_k$'s and the oracle actually works as follows: given any $c\in C$, in $\tilde\OO(\epsilon^{-1})$ time, the oracle will return $U^j_k \subset Y^k_j$ for every $Y^k_j$, where $card_{U^j_k }[y]$ is explicit given for every $y\in set(Y^j_k)$.  Let $Y'$ be the multiset-union of all $U^j_k$'s, we have $Y' \subset Y$ and $|\Sigma(Y') -c| \le \tilde\OO(\epsilon)\hat{t}$.
			\end{itemize}
		\end{lemma}
		The rest of this section is dedicated to proving Lemma~\ref{lemma:un-subset-sum-l}. %Note that $\hat{t}= \Theta(\epsilon^{-3})$ and $Y$ is divided into $\OO(\log|X| (\log\frac{1}{\epsilon})^2) $ groups: $Y^j_k$'s, where each $Y^j_k$ satisfies condition (E) (see Section~\ref{subsec:unboun_pre}). 
		Recall Corollary~\ref{coro:sub_sum_top}, towards proving Lemma~\ref{lemma:un-subset-sum-l}, we only need to derive an algorithm that can solve the following \textbf{problem}-$\mathscr{F}$ in $\tilde\OO(\epsilon^{-1})$-time.
		
		\textbf{problem}-$\mathscr{F}$: for every $Y^j_k$, compute a  $(\tilde\OO(\epsilon),\hat{t})$-approximate set with cardinality of $\OO(\epsilon^{-1})$ for $S(Y^j_k)$, and build an $\tilde\OO(\epsilon^{-1})$-time oracle for backtracking from this approximate set to $Y^j_k$. 
		
		Note that all $Y^j_k$'s satisfy Condition (B) (see Lemma~\ref{lemma:pre_unbo}). To solve 	 \textbf{problem}-$\mathscr{F}$, we only need to prove the following Lemma~\ref{lemma:sam_mul}, where $M$ represents an arbitrary $Y^j_k$.
		
		\begin{lemma}\label{lemma:sam_mul}
			Given $\omega \in \Theta(\epsilon^{-3})$ and a multiset $M\subset \mathbb{R}_{\ge 0}$  satisfying the following conditions:
			\begin{enumerate}
				\item %[(i)] 
				$M= \beta \overline{M}$, where $\beta \in \Omega(\frac{1}{\epsilon})$ and $\overline{M}\subset \mathbb{N}_{+}\cap [\frac{1}{2\epsilon}, \frac{1}{\epsilon}]$.
				\item %[(ii)] 
				$set(M)$ is explicitly given.
				\item %[(iii)] 
				Let $m = \lfloor\frac{\omega}{M^{\min}}\rfloor$ and let 
				\begin{align*}
					l_M =
					\begin{cases}
						&2^{1+pow(\log_{2}(\frac{1}{\epsilon})+1)} \hspace{9mm}\text{if } m  \le 2^{1+pow(\log_{2}(\frac{1}{\epsilon})+1)}\\
						&2m\cdot 2^{1+pow(\log_{2}(\frac{1}{\epsilon})+1)} \hspace{2mm}\text{if } m > 2^{1+pow(\log_{2}(\frac{1}{\epsilon})+1)}
					\end{cases}
				\end{align*}
				$M$ consists of $l_M$ copies of $set(M)$, i.e., every element in $M$ has the same multiplicity, which is $l_M$.
			\end{enumerate}
			In $\tilde\OO(\epsilon^{-1})$ processing time, we can 
			\begin{itemize}
				\item[(i).] Compute an $(\tilde\OO(\epsilon),\omega)$-approximate set $C_M$ with cardinality of $\OO({\epsilon^{-1}})$ for $S(M)$.
				\item[(ii).] Meanwhile build an $\tilde\OO(\epsilon^{-1})$-time oracle for backtracking from $C_M$ to $M$. That is, given any $c\in C_M$, in $\tilde\OO(\epsilon^{-1})$ time, the oracle will return $M' \subset M$ such that $|\Sigma(M')-c| \le \tilde\OO(\epsilon)\omega$. Moreover,  $card_{M'}[y]$ for every $y \in set(M)$ are also returned by the oracle. 
			\end{itemize}
		\end{lemma}
		%Towards this, we first present the following useful Observation~\ref{obs:G_m} and Observation~\ref{obs:G+G}.

		Let $\upsilon = \frac{\omega}{\beta}$ and let $G := set(\overline{M})=\frac{set(M)}{\beta}$, it is easy to see that $\upsilon = \OO(\frac{1}{\epsilon^{2}})$ and $\lfloor\frac{\upsilon}{G^{\min}}\rfloor = \lfloor\frac{\omega}{M^{\min}} \rfloor = m = \OO(\frac{1}{\epsilon})$. 
		Observe that given any $(\tilde\OO(\epsilon),\upsilon)$-approximate set of $S(\overline{M})$, say $C_{\overline{M}}$, then $\beta C_{\overline{M}}$ is an $(\tilde\OO(\epsilon),\omega)$-approximate set of $S(M)$, furthermore, a $T$-time oracle for backtracking from $C_{\overline{M}}$ to $\overline{M}$ will directly yield an $\OO(T)$-time oracle for backtracking from $\beta C_{\overline{M}}$ to $M$. It thus suffices to consider $\overline{M}$ and $\upsilon$. Precisely, we only need to derive an $\tilde\OO(\epsilon^{-1})$-time algorithm for computing an $(\tilde\OO(\epsilon), \upsilon)$-approximate set with cardinality of $\OO(\epsilon^{-1})$ for $S(\overline{M})$ and meanwhile build an $\tilde\OO(\epsilon^{-1})$-time oracle for backtracking from this approximate set to $\overline{M}$.  Two cases, $m \le 2^{1+pow(\log_2 (\frac{1}{\epsilon})+1)}$ and $m > 2^{1+pow(\log_2 (\frac{1}{\epsilon})+1)}$, will be considered in Section~\ref{subsec:un_g_1} and Section~\ref{subsec:m>_2} separately.
		
		Let $G_{\log} =  \underbrace{G\oplus G \oplus \cdots \oplus G}_{2^{1+pow(\log_{2}(\frac{1}{\epsilon})+1)}}$. Then $2^k G_{\log} = \{2^k  x : x \in G_{\log} \}$ for $k =1,2,\cdots,pow(m)$.  Before proceeding, we first present the following useful observations.
		
		\begin{observation}\label{obs:G_m}
			When $m \le 2^{1+pow(\log_2 (\frac{1}{\epsilon})+1)}$, %$l_M  = 2^{1+pow(\log_{2}(\frac{1}{\epsilon})+1)}$, 
			we have $G_{\log}\subset S(\overline{M})$, moreover, we have $G_{\log }\cap [0,\upsilon] =  S(\overline{M})\cap [0,\upsilon]$.
		\end{observation}
		\begin{proof}
			When $m \le 2^{1+pow(\log_2 (\frac{1}{\epsilon})+1)}$, we have $l_M= 2^{1+pow(\log_{2}(\frac{1}{\epsilon})+1)}$. Note that $G=set(\overline{M})$ and elements in $\overline{M}$ have the same multiplicity of $l_M$. It follows that $G_{\log}= set\{\Sigma(\overline{M}') : \overline{M}'\subset \overline{M} \text{ \ and\ } |\overline{M}'|\le2^{1+pow(\log_{2}(\frac{1}{\epsilon})+1)} \}$.  Thus $G_{\log}\subset S(\overline{M})$, furthermore, we have $G_{
				\log}\cap [0,\upsilon] \subset  S(\overline{M})\cap [0, \upsilon]$. 
			
			It remains to prove that $S(\overline{M})\cap [0,\upsilon] \subset  G_{\log}\cap [0,\upsilon]$. Consider any $s\in S(\overline{M})\cap [0,\upsilon]$, there exists $\overline{M}'' \subset \overline{M}$ such that $s = \Sigma(\overline{M}'')$. Notice that $|\overline{M}''|\cdot G^{\min}\le \Sigma(\overline{M}'') \le \upsilon$, we have $|\overline{M}''|\le \lfloor\frac{\upsilon}{G^{\min}}\rfloor= m\le 2^{1+pow(\log_{2}(\frac{1}{\epsilon})+1)}$, which implies that $s=\Sigma(\overline{M}'') \in G_{\log}$. Hence $ S(\overline{M}) \cap [0,\upsilon] \subset G_{\log}\cap [0, \upsilon]$. \qed \end{proof}

		\begin{observation}\label{obs:G+G}
			When $m >  2^{1+pow(\log_2 (\frac{1}{\epsilon})+1)}$, %$l_M  = 2m\cdot 2^{1+pow(\log_{2}(\frac{1}{\epsilon})+1)}$, 
			we have $(G_{\log}\oplus 2G_{\log} \oplus\cdots \oplus 2^{pow(m)}G_{\log})\subset S(\overline{M})$. Moreover, we have $S(\overline{M})\cap [0, \upsilon] = (G_{\log}\oplus 2G_{\log} \oplus\cdots \oplus 2^{pow(m)}G_{\log}) \cap [0, \upsilon]$.
		\end{observation}
		\begin{proof}
			When $m >  2^{1+pow(\log_2 (\frac{1}{\epsilon})+1)}$, we have $l_M =2m\cdot 2^{1+pow(\log_{2}(\frac{1}{\epsilon})+1)}$. Note that $G=set(\overline{M})$ and elements in $\overline{M}$ have the same multiplicity of $l_M$. It follows that $$\underbrace{G \oplus G \oplus \cdots \oplus G}_{2m\cdot 2^{1+pow(\log_{2}(\frac{1}{\epsilon})+1)}} = set\{\Sigma(\overline{M}')  :  \overline{M}'\subset \overline{M} \text{ \ and\ } |\overline{M}'|\le 2m\cdot 2^{1+pow(\log_{2}(\frac{1}{\epsilon})+1)} \} \subset S(\overline{M}).$$ 
			Observe that $G_{\log}\oplus 2G_{\log} \oplus\cdots \oplus 2^{pow(m)}G_{\log} \subset \underbrace{G_{\log} \oplus G_{\log}  \oplus \cdots \oplus G_{\log}}_{2m} = \underbrace{G \oplus G \oplus \cdots \oplus G}_{2m\cdot 2^{1+pow(\log_{2}(\frac{1}{\epsilon})+1)}}\subset S(\overline{M})$, furthermore,  we have  $ (G_{\log}\oplus 2G_{\log} \oplus\cdots \oplus2^{pow(m)}G_{\log})\cap [0,\upsilon] \subset  S(\overline{M})\cap [0, \upsilon]$.

			\iffalse{
				Towards this, we need the following lemma derived by Klein~\cite{DBLP:conf/soda/Klein22}. We copy it here with a slight change of notation.%proving the second part of the observation, we import the following lemma, which is derived by Klein~\cite{DBLP:conf/soda/Klein22}. 
				\newtheorem{L_1}{Lemma~\ref{lemma:supp}}
				\begin{L_1}[CF. Corollary 1. from~\cite{DBLP:conf/soda/Klein22}]
					Given a set $X = \{x_1, x_2,\cdots,x_n\}\subset \mathbb{N}$ and a target $t \in \mathbb{N}$. If there is a feasible solution to the following integer program: 
					\begin{equation*} 
						\begin{aligned}
							\mathbf{IP_t:}  & \sum^{n}_{i=1} v_i x_i  = t \\  
							& \vev=(v_1,v_2,\cdots,v_n) \in \mathbb{N}^n, 
						\end{aligned}
					\end{equation*}
					Then there is a solution to $\mathbf{IP_t}$ whose number of non-zero components does not exceed $\log_{2}(X^{\min})+1.$
				\end{L_1}
				Back to the proof of Observation~\ref{obs:G+G}. 
			}\fi

			%the first part of the observation has been proved. %Since $G_{\log}\oplus 2G_{\log} \oplus\cdots \oplus 2^{pow(m)}G_{\log} \subset  S(\overline{M})$, we have $ (G_{\log}\oplus 2G_{\log} \oplus\cdots \oplus2^{pow(m)}G_{\log})\cap [0,(1+\epsilon)\upsilon] \subset  S(\overline{M})\cap [0, (1+\epsilon)\upsilon]$.
			It remains to prove that $S(\overline{M})\cap [0, \upsilon] \subset   (G_{\log}\oplus 2G_{\log} \oplus\cdots \oplus2^{pow(m)}G_{\log})\cap [0,\upsilon]$. %Towards this, we need Lemma~\ref{lemma:supp}. 
			Let $G = \{g_1,g_2,\cdots,g_{|G|}\}$. For any $s \in S(\overline{M})\cap [0, \upsilon]$, note that $G=set(\overline{M})$, then $s \in S(G)\cap [0, \upsilon]$. Recall Lemma~\ref{lemma:supp}, there exists $\ve y=\{y_1, y_2,\cdots, y_{|G|}\} \in \mathbb{N}^{|G|}$ such that $s=\sum^{|G|}_{i=1}g_i y_i $, and $\ve y$ contains at most $\log_{2}(G^{\min})+1$ non-zero components. Notice that $G^{\min} \le \frac{1}{\epsilon}$, we have $\log_{2}(G^{\min})+1\le \log_{2}(\frac{1}{\epsilon})+1$. Observe that $\sum^{|G|}_{i=1}y_i \le \lfloor \frac{\sum^{|G|}_{i=1}g_i y_i}{G^{\min}}\rfloor  \le \lfloor \frac{\upsilon}{G^{\min}}\rfloor = m$, we have $y_i \le m$ for every $i$. Then each $y_i$ can be written as $y_i = \sum^{pow(m)}_{k=0} \gamma_{k}(y_i) \cdot 2^{k}$, %=\gamma_{0}(y_i)+\gamma_{1}(y_i)2^1+\gamma_{2}(y_i)2^2+\cdots + \gamma_{pow(m)}(y_i)2^{pow(m)}$, 
			where $\gamma_{k}(y_i) \in \{0,1\} $ for $k =0,1,2,\cdots,pow(m)$. Note that if $y_i =0$, then $\gamma_{k}(y_i) = 0$ for every $k$. Thus the fact that $\ve y$ contains at most $(\log_{2}(\frac{1}{\epsilon})+1)$ non-zero components implies that $(\gamma_{k}(y_1), \gamma_{k}(y_2),\cdots, \gamma_{k}(y_{|G|}))$ contains at most $(\log_{2}(\frac{1}{\epsilon})+1)$ non-zero components, futher implies that $\Sigma^{|G|}_{i=1}\gamma_{k}(y_i) g_i \in  G_{\log}$.
			Then we have $s = \sum^{|G|}_{i=1}g_i y_i  = \sum^{pow(m)}_{k=0} 2^k (\Sigma^{|G|}_{i=1}\gamma_{k}(y_i) g_i) \in (G_{\log}\oplus 2G_{\log} \oplus\cdots \oplus2^{pow(m)}G_{\log})\cap [0,\upsilon]$. \qed \end{proof}
		
		%Now, we are ready to  prove Lemma~\ref{lemma:sam_mul}. 
		
		%Let $\upsilon = \frac{\omega}{\beta}$, it is easy to see that $\upsilon = \OO({\epsilon^{-2}})$. Observe that given any $(\tilde\OO(\epsilon),\upsilon)$-approximate set of $S(\overline{M})$, say $C_{\overline{M}}$, then $\beta C_{\overline{M}}$ is an $(\tilde\OO(\epsilon),\omega)$-approximate set of $S(M)$, furthermore, a $T$-time oracle for backtracking from $C_{\overline{M}}$ to $\overline{M}$ will directly yield an $\OO(T)$-time oracle for backtracking from $\beta C_{\overline{M}}$ to $M$. It thus suffices to consider $\overline{M}$ and $\upsilon$. Precisely, we only need to derive an $\tilde\OO(\epsilon^{-1})$-time algorithm for computing an $(\tilde\OO(\epsilon), \upsilon)$-approximate set with cardinality of $\OO(\epsilon^{-1})$ for $S(\overline{M})$ and meanwhile build an $\tilde\OO(\epsilon^{-1})$-time oracle for backtracking from this approximate set to $\overline{M}$.  Two cases, $m \le 2^{1+pow(\log_2 (\frac{1}{\epsilon})+1)}$ and $m > 2^{1+pow(\log_2 (\frac{1}{\epsilon})+1)}$, will be considered separately in the following.

		\subsubsection{Handling the case that $m \le 2^{1+pow(\log_2 (\frac{1}{\epsilon})+1)}$.}\label{subsec:un_g_1}
		
		%In the following, we first approximate $2^k G^{\log}$ for $k=0,1,2,\cdots,2^{pow(m)}$, and then approximate $G^{\log}\oplus 2G^{\log} \oplus\cdots \oplus 2^{pow(m)}G^{\log}$. In the meantime, we build a structure for backtracking. 
		In this section, we consider the case that $m \le 2^{1+pow(\log_2 (\frac{1}{\epsilon})+1)}$. %Then elements in $\overline{M}$ have the same multiplicity of $l_M$, where $l_M = 2^{1+pow(\log_2 (\frac{1}{\epsilon})+1)}$. 
		The following Claim~\ref{claim:c_G_M} guarantees that towards proving Lemma~\ref{lemma:sam_mul} for the case $m \le 2^{1+pow(\log_2 (\frac{1}{\epsilon})+1)}$, it is sufficient to consider $G_{\log}$.
		\begin{claim}\label{claim:c_G_M}
			Assume that $m \le 2^{1+pow(\log_2 (\frac{1}{\epsilon})+1)}$. Let $C$ be an $(\tilde\OO(\epsilon), \upsilon)$-approximate set for $G_{\log} =  \underbrace{G\oplus G \oplus \cdots \oplus G}_{2^{1+pow(\log_{2}(\frac{1}{\epsilon})+1)}}$. Let $\textbf{Ora}$ be a $T$-time oracle for backtracking from $C$ to $\underbrace{G\dot{\cup} G \dot{\cup} \cdots\dot{\cup} G}_{2^{1+pow(\log_{2}(\frac{1}{\epsilon})+1)}}$, that is, given any $c \in C$, within $T$ time, $\textbf{Ora}$ will return $x_1,x_2,\cdots,x_{2^{1+pow(\log_{2}(\frac{1}{\epsilon})+1)}}\in G\cup \{0\}$
			%$\{ x_1,x_2,\cdots,x_{2^{1+pow(\log_{2}(\frac{1}{\epsilon})+1)}}\}$ 
			such that $|c-{\sum}^{2^{1+pow(\log_{2}(\frac{1}{\epsilon})+1)}}_{i=1} x_i| \le \tilde\OO(\epsilon) \upsilon$. %, where $x_i \in G\cup \{0\}$ for every $i=1,2,\cdots,2^{1+pow(\log_{2}(\frac{1}{\epsilon})+1)}$.
			
			Then $C$ is an $(\tilde\OO(\epsilon), \upsilon)$-approximate set for $S(\overline{M})$. Moreover,  $\textbf{Ora}$ directly yileds an $\OO(T+\frac{1}{\epsilon})$-time oracle for backtracking from $C$ to $\overline{M}$. That is, given any $c \in C$, by calling $\textbf{Ora}$, in $\OO(T+\frac{1}{\epsilon})$ time, we can obtain $\overline{M}' \subset \overline{M}$ such that $|c-\Sigma(\overline{M})'| \le \tilde\OO(\epsilon) \upsilon$, in particular, $card_{\overline{M}' }[x]$ is obtained for every $x\in set(\overline{M})$.  \end{claim}

		\begin{proof}
			When $m\le 2^{1+pow(\log_2 (\frac{1}{\epsilon})+1)}$, elements in $\overline{M}$ have the same multiplicity of $l_M = 2^{1+pow(\log_2 (\frac{1}{\epsilon})+1)}$.  Moreover, recall Observation~\ref{obs:G_m}, %when $m \le 2^{1+pow(\log_2 (\frac{1}{\epsilon})+1)}$,
			we have $G_{\log}\subset S(\overline{M})$ and  $G_{\log }\cap [0,\upsilon] =  S(\overline{M})\cap [0,\upsilon]$.  	
			
			Since $C$ is an $(\tilde\OO(\epsilon), \upsilon)$-approximate set of $G_{\log}$, we observe the followings:
			\begin{itemize}
				\item $C \subset [0, (1+\tilde\OO(\epsilon))\upsilon]$. 
				\item For any $c\in C$, there exist $x_1,x_2,\cdots,x_{2^{1+pow(\log_{2}(\frac{1}{\epsilon})+1)}} \in G\cup \{0\}$ such that $|c-\sum^{2^{1+pow(\log_{2}(\frac{1}{\epsilon})+1)}}_{i=1} x_i| \le \tilde\OO(\epsilon) \upsilon$. %$\{ x_1,x_2,\cdots,x_{2^{1+pow(\log_{2}(\frac{1}{\epsilon})+1)}}\}$ such that $|c-\sum^{2^{1+pow(\log_{2}(\frac{1}{\epsilon})+1)}}_{i=1} x_i| \le \tilde\OO(\epsilon) \upsilon$, where $x_i \in G\cup \{0\}$ for every $i=1,2,\cdots,2^{1+pow(\log_{2}(\frac{1}{\epsilon})+1)}$. 
				Note that $G = set(\overline{M})$ and elements in $\overline{M}$ have the same multiplicity of $2^{1+pow(\log_2 (\frac{1}{\epsilon})+1)}$, we have $\{ x_1,x_2,\cdots,x_{2^{1+pow(\log_{2}(\frac{1}{\epsilon})+1)}}\} \subset \overline{M}$.
				\item For any $s \in S(\overline{M})\cap[0,\upsilon]$, since $G_{\log }\cap [0,\upsilon] =  S(\overline{M})\cap [0,\upsilon]$, we have $s \in G_{\log }\cap [0,\upsilon]$. Then there exists $s''\in C$ such that $|s-s''|\le \tilde\OO(\epsilon)\upsilon$.
			\end{itemize} 
			Thus $C$ is an $(\tilde\OO(\epsilon), \upsilon)$-approximate set of $S(\overline{M})$.
			
			Given $\textbf{Ora}$ defined in Claim~\ref{claim:c_G_M}.  For any $c\in C$, within $T$ time, $\textbf{Ora}$ will return $x_1,x_2,\cdots,x_{2^{1+pow(\log_{2}(\frac{1}{\epsilon})+1)}} \in G\cup \{0\}$ such that $|c-\sum^{2^{1+pow(\log_{2}(\frac{1}{\epsilon})+1)}}_{i=1} x_i| \le \tilde\OO(\epsilon) \upsilon$.  %$\{ x_1,x_2,\cdots,x_{2^{1+pow(\log_{2}(\frac{1}{\epsilon})+1)}}\}$ such that $|c-\sum^{2^{1+pow(\log_{2}(\frac{1}{\epsilon})+1)}}_{i=1} x_i| \le \tilde\OO(\epsilon) \upsilon$, where $x_i \in G\cup \{0\}$ for every $i=1,2,\cdots,2^{1+pow(\log_{2}(\frac{1}{\epsilon})+1)}$. 
			Let $\overline{M}' = \{ x_1,x_2,\cdots,x_{2^{1+pow(\log_{2}(\frac{1}{\epsilon})+1)}}\}$. It is easy to see that $\overline{M}' \subset \overline{M}$ and the total time to determine $card_{\overline{M}'}[x]$ for every $x\in set(\overline{M})$ is $\OO(2^{1+pow(\log_{2}(\frac{1}{\epsilon})+1)}+ |G|) = \OO(\frac{1}{\epsilon})$. Thus $\textbf{Ora}$ yields an $\OO(T+\frac{1}{\epsilon})$-time oracle for backtracking from $C$ to $\overline{M}$.\qed \end{proof}

		%According to Claim~\ref{claim:c_G_M}, towards proving Lemma~\ref{lemma:sam_mul} for the case $m \le 2^{1+pow(\log_2 (\frac{1}{\epsilon})+1)}$, we only need to consider $G_{\log}$.
		
		In the following, we will design an iterative approach to compute an approximate set for $G_{\log}$ and meanwhile build an oracle for backtracking. % and build a linked list structure for backtracking. 
		\paragraph{\textbf{Approximating $\mathbf{G_{\log}}$.}}
		Given any multiset $A$, recall that $A\dot{\cup} A$ is the multiset that duplicates each element in $A$. We build a linked list structure as follows:
		\begin{enumerate}
			\item At iteration-1, we use Observation~\ref{obs:cap_apx_ss} to compute an $(\epsilon, \min\{2G^{\max},\upsilon\})$-approximate set with cardinality of $\OO(\frac{1}{\epsilon})$ for $G\oplus G$ and meanwhile derive an $\OO(\frac{1}{\epsilon} \log\frac{1}{\epsilon})$-time oracle for backtracking from this approximate set to $G \dot{\cup} G$. Denote by $U^1$ this approximate set.  
			We create a head node, and let the head node contain $U^1$ and the oracle for backtracking from $U^1$ to $G \dot{\cup} G$.  Let $f^{1}(\epsilon) := \epsilon$.
			
			Note that $G\subset [\frac{1}{2\epsilon}, \frac{1}{\epsilon}]$ and elements in $G$ are different from each other, thus $|G| \le \frac{1}{\epsilon}$. According to Observation~\ref{obs:cap_apx_ss}, the total processing time at iteration-1 is $\OO(\frac{1}{\epsilon}\log \frac{1}{\epsilon})$.
			
			\item Before proceeding to iteration-$h$, where $h \ge 2$, we assume that the following things have been down:
			\begin{itemize}
				\item Let $U^0$ denote $G$. %, and let $f^0(\epsilon) = 0$ $\{U^{k} :  k =1,2,\cdots, h-1\}$, 
				We have obtained $U^{0}, U^{1}, U^{2},\cdots, U^{h-1}$, where $U^{k}$ is an $\left(\epsilon, \min\{2^{k}G^{\max},\upsilon\}+f^{k-1}(\epsilon)\min\{2^k G^{\max},2\upsilon\}\right)$-approximate set with cardinality of $\OO(\frac{1}{\epsilon})$ for $U^{k-1}\oplus U^{k-1}$ and $k=1,2,\cdots, h-1$. 
				
				Functions in $\{f^{\iota} (\epsilon) :  \iota=0,1,\cdots,h-1 \}$ are defined by the following recurrence relation:
				$f^{0} (\epsilon) = 0$ and $f^{\iota}(\epsilon) = \epsilon+2(1+\epsilon)f^{\iota-1}(\epsilon) \text{ \ for \ } \iota =1,2,\cdots, k.$ 
				
				\item For each $U^k$, where $k= 1,\cdots, h-1$, we have built an $\OO(\frac{1}{\epsilon}\log\frac{1}{\epsilon} )$-time oracle for backtracking from $U^{k}$ to $U^{k-1} \dot{\cup}U^{k-1}$. That is, given any $y \in U^{k}$, in $\OO(\frac{1}{\epsilon}\log\frac{1}{\epsilon})$
				time, the oracle will return $y_1 \in U^{k-1} \cup \{0\}$ and $y_2 \in U^{k-1}\cup \{0\}$ such that $$|y-(y_1+y_2)| \le \epsilon \left(\min\{2^{k}G^{\max},\upsilon\}+f^{k-1}(\epsilon)\min\{2^k G^{\max},2\upsilon\}\right).$$
				Moreover, a node is created for containing $U^k$ and the oracle for backtracking from $U^{k}$ to $U^{k-1} \dot{\cup}U^{k-1}$.
			\end{itemize}
			
			Now we start iteration-$h$. We use Observation~\ref{obs:cap_apx_ss} to compute an $\left(\epsilon, \min\{2^h G^{\max},\upsilon\}+f^{h-1}(\epsilon) \min\{2^{h}G^{\max},2\upsilon\}\right)$-approximate set with cardinality of $\OO(\frac{1}{\epsilon})$ for $U^{h-1} \oplus U^{h-1}$ and meanwhile build an $\OO(\frac{1}{\epsilon}\log\frac{1}{\epsilon})$-time oracle for backtracking from this approximate set to $U^{h-1} \dot{\cup} U^{h-1}$. Denote by $U^{h}$ this approximate set.  Then we create a node behind the $(h-1)$-th node. Let this new node contain $U^{h}$ and the oracle for backtracking from  $U^{h}$ to $U^{h-1} \dot{\cup} U^{h-1}$. Let $f^{h}(\epsilon) := \epsilon+2(1+\epsilon)f^{h-1}(\epsilon)$. 
			
			Notice that $|U^{h-1}| = \OO(\frac{1}{\epsilon})$. According to Observation~\ref{obs:cap_apx_ss}, the total processing time at iteration-$h$ is $\OO(\frac{1}{\epsilon}\log \frac{1}{\epsilon})$. 
			
			\item Using the same approach in iteration-$h$ recursively, and iteratively create linked nodes. 
			
			Let $h_p: = 1+pow(1+\log_{2}\frac{1}{\epsilon})$. After $h_p$ such rounds we stop and have built a linked list whose tail node contains (i). $U^{h_p}$, which is an $\left(\epsilon, \min\{2^{h_p} G^{\max},\upsilon\}+f^{h_p-1}(\epsilon)\min\{2^{h_p}G^{\max},2\upsilon\}\right)$-approximate set for $U^{h_p-1}\oplus U^{h_p-1}$;
			(ii). an $\OO(\frac{1}{\epsilon}\log\frac{1}{\epsilon})$-time oracle for backtracking from $U^{h_p}$ to $U^{h_p-1} \dot{\cup} U^{h_p-1}$. That is, given any $y \in U^{h_p}$, in $\OO(\frac{1}{\epsilon}\log\frac{1}{\epsilon})$
			time, the oracle will return $y_1 \in U^{h_p-1} \cup \{0\}$ and $y_2 \in U^{h_p-1}\cup \{0\}$ such that $$|y-(y_1+y_2)| \le \epsilon \left(\min\{2^{h_p}G^{\max},\upsilon\}+f^{h_p-1}(\epsilon)\min\{2^{h_p} G^{\max},2\upsilon\}\right).$$

		\end{enumerate}
		
		To summarize, the overall processing time is $\OO(\frac{h_{p} }{\epsilon}\log \frac{1}{\epsilon}) = \OO(\frac{1}{ \epsilon} (\log \frac{1}{\epsilon})^2)$. Note that we have obtained $U^0, U^1,U^2, \cdots, U^{h_p}$ and defined $f^{0}(\epsilon), f^{1}(\epsilon),\cdots, f^{h_p}(\epsilon)$. For every integer $1\le h \le h_p$, $U^{h}$ is an $\left(\epsilon, \min\{2^{h} G^{\max},\upsilon\}+f^{h-1}(\epsilon)\cdot \min\{2^{h}G^{\max},2\upsilon\}\right)$-approximate set for $U^{h-1}\dot{\cup} U^{h-1}$, moreover, $|U^h|=\OO(\frac{1}{\epsilon})$.  Functions in $\left\{f^{h}(\epsilon) : h=0,1,2,\cdots,h_p\right\}$ are defined by the following recurrence relation:
		$f^{0} (\epsilon) = 0$ and $f^{h}(\epsilon) = \epsilon+2(1+\epsilon)f^{h-1}(\epsilon) $ for $ h=1,2,\cdots, h_p$.  A simple calculation shows that $f^{h_p}(\epsilon)=\OO(\epsilon (\log \frac{1}{\epsilon})^2)$. %=f^{pow(1+\log_{2}\frac{1}{\epsilon})+1}(\epsilon)  
		
		We claim that $U^{h_p}$ is an $(\tilde\OO(\epsilon), \upsilon)$-approximate set with cardinality of $\OO(\frac{1}{\epsilon})$ for $G_{\log}$. In particular, we claim the following.
		
		\begin{claim}\label{claim:cl_f_un}
			For every $h =1,2,\cdots, h_p$, $U^h$ is an $(f^{h}(\epsilon), \min\{2^{h}G^{\max},\upsilon\})$-approximate set for $\underbrace{G \oplus G\oplus \cdots \oplus G}_{2^{h}}$. In particular, $U^{h}$ is an $(f^{h}(\epsilon), \upsilon)$-approximate set for $\underbrace{G \oplus G\oplus \cdots \oplus G}_{2^{h}}$. 
		\end{claim}
		\begin{proof}
			We prove Claim~\ref{claim:cl_f_un} by induction.  
			
			Note that $\upsilon \ge \min \{ 2^{h} G^{\max} ,\upsilon \} $ and $\underbrace{G \oplus G \oplus \cdots \oplus G}_{2^{h}} \cap [0,\upsilon] =  \underbrace{G \oplus G \oplus \cdots \oplus G}_{2^{h}} \cap [0,\min\{2^h G^{\max},\upsilon\}]$ hold for every $h=1,2\cdots,h_p$. It is straightforward that an $(f^{h}(\epsilon), \min\{2^{h}G^{\max},\upsilon\})$-approximate set for $\underbrace{G \oplus G\oplus \cdots \oplus G}_{2^{h}}$ is automatically an $(f^{h}(\epsilon), \upsilon)$-approximate set for $\underbrace{G \oplus G\oplus \cdots \oplus G}_{2^{h}}$. So we only need to prove the first half of Claim~\ref{claim:cl_f_un}.

			For $h=1$. Recall that $U^1$ is an $\left(\epsilon, \min\{2^{1} G^{\max},\upsilon\}+f^{0}(\epsilon)\cdot \min\{2^{1}G^{\max},2\upsilon\}\right)$-approximate set for $U^{0}\dot{\cup} U^{0}$, where $f^{0}(\epsilon) = 0$ and $U^{0} = G$.   Apparently, $U^1$ is an $\left(\epsilon, \min\{2 G^{\max},\upsilon\}\right)$-approximate set for $G\dot{\cup} G$.
			
			Then we consider any $h =2,3,\cdots, h_p$.  Assume that $U^{h-1}$ is  $(f^{h-1}(\epsilon),\min\{2^{h-1}G^{\max},\upsilon\})$-approximate set of $\underbrace{G \oplus G\oplus \cdots \oplus G}_{2^{h-1}}$. In the following, we prove that $U^h$ is an $(f^{h}(\epsilon), \min\{2^{h}G^{\max},\upsilon\})$-approximate set for $\underbrace{G \oplus G\oplus \cdots \oplus G}_{2^{h}}$. Recall that $U^{h}$ is an $\left(\epsilon, \min\{2^h G^{\max},\upsilon\}+f^{h-1}(\epsilon)\min\{2^{h}G^{\max},2\upsilon\}\right)$-approximate set for $U^{h-1} \oplus U^{h-1}$ and $f^{h}(\epsilon) = \epsilon + 2(1+\epsilon)f^{h-1}(\epsilon)$. We have
			\begin{align*}
				(U^{h})^{\max}
				&\le (1+\epsilon) \left( \min\{2^h G^{\max},\upsilon\}+f^{h-1}(\epsilon)\min\{2^{h}G^{\max},2\upsilon\}\right)\\
				&= (1+\epsilon)  \left(\min\{2^h G^{\max},\upsilon\}+2f^{h-1}(\epsilon)\min\{2^{h-1}G^{\max},\upsilon\} \right)\\
				&\le (1+\epsilon)  \left(\min\{2^h G^{\max},\upsilon\}+2f^{h-1}(\epsilon)\min\{2^{h}G^{\max},\upsilon\} \right) \\
				&=  (1+f^{h}(\epsilon))\min \{ 2^{h} G^{\max} ,\upsilon \}.
			\end{align*}
			Thus $U^{h}\subset [0,(1+f^{h}(\epsilon ))\min \{ 2^{h} G^{\max} ,\upsilon \}]$.

			Given any $a \in \underbrace{G \oplus G \oplus \cdots \oplus G}_{2^{h}} \cap [0,\min \{ 2^{h} G^{\max} ,\upsilon \}] = \underbrace{G \oplus G \oplus \cdots \oplus G}_{2^{h}} \cap [0,\upsilon]$, there exist $a_1,a_2 \in \underbrace{G \oplus G \oplus \cdots \oplus G}_{2^{h-1}}\cap[0,\upsilon]=\underbrace{G \oplus G \oplus \cdots \oplus G}_{2^{h-1}} \cap [0,\min \{ 2^{h-1} G^{\max} ,\upsilon \}] $ such that $a= a_1+a_2$.  
			Recall that $U^{h-1}$ is $(f^{h-1}(\epsilon),\min\{2^{h-1}G^{\max},\upsilon\})$-approximate set for $\underbrace{G \oplus G\oplus \cdots \oplus G}_{2^{h-1}}$, thus there exists $a'_i \in U^{h-1}$ such that $|a_i -a'_i| \le f^{h-1}(\epsilon) \min\{2^{h-1}G^{\max},\upsilon\}$ for $i=1,2$. %and $|a_2 -a'_2| \le f^{h-1}(\epsilon) \min\{2^{h-1}G^{\max},\upsilon\}$. 
			Then we have %$|a-a'|\le f^{h-1}(\epsilon) \min\{2^{h}G^{\max},2\upsilon\}$ and then
			$a'_1+a'_2 \le a+ f^{h-1}(\epsilon) \min\{2^{h}G^{\max},2\upsilon\}\le \min\{2^{h}G^{\max}, \upsilon\}+f^{h-1}(\epsilon) \min\{2^{h}G^{\max},2\upsilon\}.$
			Thus $a'_1+a'_2\in (U^{h-1}\oplus U^{h-1} )\cap [0, \min\{2^h G^{\max},\upsilon\}+f^{h-1}(\epsilon)\min\{2^{h}G^{\max},2\upsilon\}]$. Let $a' = a'_1+a'_2$. Recall that $U^h$ is an $(\epsilon, \min\{2^h G^{\max},\upsilon\}+f^{h-1}(\epsilon)\min\{2^{h}G^{\max},2\upsilon\})$-approximate set for $U^{h-1} \oplus U^{h-1}$, then there exists $a'' \in U^{h}$ such that
			$|a''-a'|\le \epsilon \left(\min\{2^{h}G^{\max}, \upsilon\}+f^{h-1}(\epsilon) \min\{2^{h}G^{\max},2\upsilon\}\right).$ To summarize, given $a \in \underbrace{G \oplus G \oplus \cdots \oplus G}_{2^{h}} \cap [0,\upsilon]$, there exist $a'' \in U^{h}$ such that 
			\begin{align*}
				|a''-a |&\le |a''-a'|+|a'-a|\le  |a''-a'| + |a'_1-a_1|+|a'_2-a_2|\\
				&\le \epsilon \left(\min \{ 2^{h} G^{\max} ,\upsilon\}+f^{h-1}(\epsilon)  \min \{ 2^{h} G^{\max} ,2\upsilon\}\right)
				+   2f^{h-1}(\epsilon)\min \{ 2^{h-1}G^{\max},\upsilon\}\\
				&\le \left(\epsilon +2(\epsilon+1)f^{h-1}(\epsilon)\right) \min \{ 2^{h} G^{\max} ,\upsilon\}=f^{h}(\epsilon )\min \{ 2^{h} G^{\max} ,\upsilon \}.
			\end{align*}
			
			Given any $c \in U^{h}$, note that $U^{h}$ is an $\left(\epsilon, \min\{2^h G^{\max},\upsilon\}+f^{h-1}(\epsilon)\min\{2^{h}G^{\max},2\upsilon\}\right)$-approximate set for $U^{h-1} \oplus U^{h-1}$, then there exist $c'_1,c'_2 \in U^{h-1}\cup\{0\}$ such that 
			$|c-(c'_1+c'_2)| \le \epsilon \left(\min\{2^h G^{\max},\upsilon\}+f^{h-1}(\epsilon)\min\{2^{h}G^{\max},2\upsilon\}\right)$. Consider each $c'_i$. If $c'_i \in U^{h-1}$, recall that $U^{h-1}$ is an $(f^{h-1}(\epsilon),\min\{2^{h-1}G^{\max},\upsilon\})$-approximate set for $\underbrace{G \oplus G\oplus \cdots \oplus G}_{2^{h-1}}$, then there exist $\bar{c}^{(i;1)},\bar{c}^{(i;2)},\cdots,\bar{c}^{(i;2^{h-1})}  \in G \cup \{0\}$ such that $|c'_i - \sum^{2^{h-1}}_{j=1}\bar{c}^{(i;j)}| \le f^{h-1}(\epsilon) \min\{2^{h-1}G^{\max},\upsilon\}$. Else if  $c'_i = 0$, let $\bar{c}^{(i;1)} =\bar{c}^{(i;2)}=\cdots =\bar{c}^{(i;2^{h-1})} = 0$, apparently, $|c'_i - \sum^{2^{h-1}}_{j=1}\bar{c}^{(i;j)}| \le f^{h-1}(\epsilon) \min\{2^{h-1}G^{\max},\upsilon\}$. 
			To summarize, given $c \in U^{h}$, there exist $\bar{c}^{(1;1)},\bar{c}^{(1;2)},\cdots,\bar{c}^{(1;2^{h-1})}; \bar{c}^{(2;1)},\bar{c}^{(2;2)},\cdots,\bar{c}^{(2;2^{h-1})} \in G \cup \{0\}$ such that 
			\begin{align*}
				|c-\sum^2_{i=1}\sum^{2^{h-1}}_{j=1}\bar{c}^{(i;j)}| 
				&\le |c-(c'_1+c'_2)|+|c'_1 - \Sigma^{2^{h-1}}_{j=1}\bar{c}^{(1;j)}| + |c'_2 - \Sigma^{2^{h-1}}_{j=1}\bar{c}^{(2;j)}| \\ 
				&\le \epsilon \left(\min\{2^h G^{\max},\upsilon\}+f^{h-1}(\epsilon)\min\{2^{h}G^{\max},2\upsilon\}\right)+2f^{h-1}(\epsilon)\min\{2^{h-1}G^{\max},\upsilon\}\\
				&\le f^{h}(\epsilon)\min\{2^{h}G^{\max},\upsilon\}.
			\end{align*} 
			Till now, we have proved that $U^h$ is an $(f^{h}(\epsilon), \min\{2^{h}G^{\max},\upsilon\})$-approximate set of $\underbrace{G \oplus G\oplus \cdots \oplus G}_{2^{h}}$, which implies that $U^h$ is an $(f^{h}(\epsilon), \upsilon)$-approximate set for $\underbrace{G \oplus G\oplus \cdots \oplus G}_{2^{h}}$.
			%Notice that $[0,(1+f^{h}(\epsilon ))\min \{ 2^{h} G^{\max} ,\upsilon \}]\subset [0,(1+f^{h}(\epsilon ))\upsilon]$ and $\underbrace{G \oplus G \oplus \cdots \oplus G}_{2^{h}} \cap [0,\min \{ 2^{h} G^{\max} ,\upsilon \}] = \underbrace{G \oplus G \oplus \cdots \oplus G}_{2^{h}} \cap [0,\upsilon]$ hold for every $h=1,2,\cdots,h_p$, thus $U^h$ is an $(f^{h}(\epsilon), \upsilon)$-approximate set of $\underbrace{G \oplus G\oplus \cdots \oplus G}_{2^{h}}$. 
			Then Claim~\ref{claim:cl_f_un} follows by induction.\qed
		\end{proof}
		
		It remains to give the backtracking oracle. With the help of the linked list structure, an $\tilde\OO(\epsilon^{-1}) $-time oracle $\textbf{Ora}$ for backtracking from $U^{h_p}$ to  $\underbrace{G \dot{\cup} G\dot{\cup}\cdots \dot{\cup} G}_{2^{h_p}}$ is derived and works as follows: 
		\begin{itemize}
			\item For any $c\in U^{h_p}$, we backtrace from the tail node to the head node. Note that 	$U^{h_p}$ is an $\left(\epsilon, \min\{2^{h_p} G^{\max},\upsilon\}+f^{h_p-1}(\epsilon) \min\{2^{h_p}G^{\max},2\upsilon\}\right)$-approximate set for $U^{h_p-1}\oplus U^{h_p-1}$, and the tail node contains an $\OO(\frac{1}{\epsilon}\log\frac{1}{\epsilon})$-time oracle for backtracking from $U^{h_p}$ to $U^{h_{p} -1} \dot{\cup} U^{h_{p} -1}$. Thus in $\OO(\frac{1}{\epsilon}\log\frac{1}{\epsilon})$ time, we can determine $c^{(h_{p}-1;1)}, c^{(h_{p}-1;2)}\in U^{h_p -1}\cup\{0\}$ such that
			$$|c- (c^{(h_{p}-1;1)}+c^{(h_{p}-1;2)})| \le \epsilon \left( \min\{2^{h_p} G^{\max},\upsilon\}+f^{h_p-1}(\epsilon)  \min\{2^{h_p}G^{\max},2\upsilon\}\right).$$
			\item Back track recursively.
			
			Given $1\le h \le h_{p}-1$, assume that we have determined $c^{(h;j_h)} \in U^{h} {\cup}\{0\}$ for every $j_h = 1,2,3,\cdots,2^{h_{p}-h}$ such that 
			$$|c-\sum^{2^{h_{p}-h}}_{j_{h}=1} c^{(h;j_h)}| \le 	\sum^{h_p}_{k=h+1} 2^{h_p-k} \epsilon   \left( \min\{2^{k} G^{\max},\upsilon\}+f^{k-1}(\epsilon)  \min\{2^{k}G^{\max},2\upsilon\}\right).$$
			
			For each $c^{(h;j_h)} \in  U^{h} {\cup}\{0\}$, note that  $U^h$ is an $\left(\epsilon, \min\{2^{h} G^{\max},\upsilon\}+f^{h-1}(\epsilon) \min\{2^{h}G^{\max},2\upsilon\}\right)$-approximate set for $U^{h-1}\oplus U^{h-1}$, and the node containing $U^h$ also contains an $\OO(\frac{1}{\epsilon}\log\frac{1}{\epsilon})$-time oracle for backtracking from $U^{h}$ to $U^{h -1} \dot{\cup} U^{h-1}$.  Thus in $\OO(\frac{1}{\epsilon}\log\frac{1}{\epsilon})$ time, we can determine $c^{(h-1;2 j_h-1)}, c^{(h-1;2j_h)}\in U^{h -1}\cup\{0\}$ such that
			$$
			|c^{(h;j_h)}- (c^{(h-1;2j_h-1)}+c^{(h-1;2j_h)})| \le \epsilon \left( \min\{2^{h} G^{\max},\upsilon\}+f^{h-1}(\epsilon)  \min\{2^{h}G^{\max},2\upsilon\}\right).
			$$
			Then in total $\OO( \frac{2^{h_{p}-h}}{\epsilon} \log\frac{1}{\epsilon})$ time, we can determine $c^{(h-1;1)}, c^{(h-1;2)},c^{(h-1;3)},\cdots,c^{(h-1;2^{h_{p}-h+1})} \in U^{h -1}\cup\{0\}$ such that 
			\begin{align*}
				| \sum^{2^{h_{p}-h}}_{j_h =1} c^{(h;j_h)}-\sum^{2^{h_{p}-h+1}}_{j_{h-1}=1} c^{(h-1;j_{h-1})} | &\le \sum^{2^{h_{p}-h}}_{j_h = 1} | c^{(h;j_h)}- (c^{(h-1;2j_h-1)}+c^{(h-1;2j_h)})| \\
				&\le 2^{h_{p} -h} \epsilon\left( \min\{2^{h} G^{\max},\upsilon\}+f^{h-1}(\epsilon)  \min\{2^{h}G^{\max},2\upsilon\}\right).
			\end{align*}
			It follows that 
			$$|c-\sum^{2^{h_{p}-h+1}}_{j_{h-1}=1} c^{(h-1;j_{h-1})}| \le \sum^{h_p}_{k=h} 2^{h_p-k} \epsilon \cdot \left( \min\{2^{k} G^{\max},\upsilon\}+f^{k-1}(\epsilon)  \min\{2^{k}G^{\max},2\upsilon\}\right).$$
			
			\item After $h_p$ such rounds, we stop and have determined $c^{(1;1)}, c^{(1;2)}, c^{(1;3)},\cdots,c^{(1;2^{h_p})} \in G \cup \{0\}$ such that
			$$
			|c- \sum^{2^{h_p}}_{j_{1}=1}c^{(1;j_1)}| \le  \sum^{h_p}_{k=1} 2^{h_p-k} \epsilon \cdot \left( \min\{2^{k} G^{\max},\upsilon\}+f^{k-1}(\epsilon)  \min\{2^{k}G^{\max},2\upsilon\}\right)
			$$
			Observe that $\epsilon \left(\min\{2^k G^{\max},\upsilon\}+f^{k-1}(\epsilon)\min\{2^{k}G^{\max},2\upsilon\}\right)+2f^{k-1}(\epsilon)\min\{2^{k-1}G^{\max},\upsilon\} \le (\epsilon +2(\epsilon+1)f^{k-1}(\epsilon))\min \{ 2^{k} G^{\max} ,\upsilon\} =f^{k}(\epsilon)\min\{2^{k}G^{\max},\upsilon\}$ holds for every $k=1,2,\cdots, h_p$. 
			\iffalse{\begin{align*}
					&\epsilon \left(\min\{2^k G^{\max},\upsilon\}+f^{k-1}(\epsilon)\min\{2^{k}G^{\max},2\upsilon\}\right)+2f^{k-1}(\epsilon)\min\{2^{k-1}G^{\max},\upsilon\} \\
					&\le (\epsilon +2(\epsilon+1)f^{k-1}(\epsilon))\min \{ 2^{k} G^{\max} ,\upsilon\} =f^{k}(\epsilon)\min\{2^{k}G^{\max},\upsilon\}
				\end{align*}
				holds for every $k=1,2,\cdots, h_p$.}\fi
			Then we have
			\begin{align*}
				&|c- \sum^{2^{h_p}}_{j_{1}=1}c^{(1;j_1)}| \le  \sum^{h_p}_{k=1} 2^{h_p-k} \epsilon \cdot \left( \min\{2^{k} G^{\max},\upsilon\}+f^{k-1}(\epsilon)  \min\{2^{k}G^{\max},2\upsilon\}\right)\\
				& = \sum^{h_p}_{k=2} 2^{h_p-k} \epsilon \cdot \left( \min\{2^{k} G^{\max},\upsilon\}+f^{k-1}(\epsilon)  \min\{2^{k}G^{\max},2\upsilon\}\right)+2^{h_{p}-1} f^1(\epsilon) \min \{ 2^1 G^{\max},\upsilon\}\\
				& \le  \sum^{h_p}_{k=3} 2^{h_p-k}  \epsilon \cdot \left( \min\{2^{k} G^{\max},\upsilon\}+f^{k-1}(\epsilon)  \min\{2^{k}G^{\max},2\upsilon\}\right) + 2^{h_{p}-2} f^{2}(\epsilon)\min\{2^{2}G^{\max},\upsilon\}\\
				&\cdots \\
				&\le f^{h_p}(\epsilon)\min\{2^{h_p}G^{\max},\upsilon\} \le \tilde\OO(\epsilon) \upsilon
			\end{align*}
			
		\end{itemize}
		To summarize, given any $c\in C$, in total $\OO(\sum^{h_p}_{k=1} \frac{2^{h_{p}-k}}{\epsilon} \log\frac{1}{\epsilon})=\OO(\frac{1}{\epsilon}(\log\frac{1}{\epsilon})^2)$ processing time, $\textbf{Ora}$ will return $c^{(1;1)}, c^{(1;2)}, c^{(1;3)},\cdots,c^{(1;2^{h_p})} \in G \cup \{0\}$ such that
		$|c- \sum^{2^{h_p}}_{j_{1}=1}c^{(1;j_1)}| \le \tilde\OO(\epsilon) \upsilon$.

		In conclude, within $\OO(\frac{1}{\epsilon}(\log\frac{1}{\epsilon})^2)$ time, we will build a linked list structure whose tail node contains an $(\tilde\OO(\epsilon), \upsilon)$-approximate set with cardinality of $\OO(\frac{1}{\epsilon})$ for $\underbrace{G \oplus G\oplus \cdots \oplus G}_{2^{1+pow(\log_{2}(\frac{1}{\epsilon})+1)}} = G_{\log}$. Meanwhile, with the help of this 
		linked list structure, an $\OO(\frac{1}{\epsilon}(\log\frac{1}{\epsilon})^2)$-time oracle $\textbf{Ora}$ for backtracking from this approximate set to $\underbrace{G \dot{\cup} G\dot{\cup} \cdots \dot{\cup} G}_{2^{1+pow(\log_{2}(\frac{1}{\epsilon})+1)}}$ is derived.   Then by Claim~\ref{claim:c_G_M}, we have proved Lemma~\ref{lemma:sam_mul} for the case that $m \le 2^{1+pow(\log_{2}(\frac{1}{\epsilon})+1)}$. %It remains to consider the case that $m > 2^{1+pow(\log_{2}(\frac{1}{\epsilon})+1)}$.
		
		\subsubsection{Handling the case that $m > 2^{1+pow(\log_{2}(\frac{1}{\epsilon})+1)}$.}\label{subsec:m>_2}
		In this section, we aim to prove Lemma~\ref{lemma:sam_mul} for the case that $m > 2^{1+pow(\log_{2}(\frac{1}{\epsilon})+1)}$.
		%, i.e., elements in $\overline{M}$ have the same multiplicity of $l_M$ where $l_M = 2m \cdot 2^{1+pow(\log_{2}(\frac{1}{\epsilon})+1)}$.

		Recall that in Section~\ref{subsec:un_g_1}, we have obtained $U^{h_p}$, which is an $(f^{h_p}(\epsilon),\min\{2^{h_p} G^{\max},\upsilon\})$-approximate set with cardinality of $\OO(\frac{1}{\epsilon})$ for $G_{\log}$. Meanwhile, we have also derived $\textbf{Ora}$, which is an $\OO(\frac{1}{\epsilon}(\log\frac{1}{\epsilon})^2)$-time oracle  for backtracking from $U^{h_p}$ to $\underbrace{G \dot{\cup} G\dot{\cup} \cdots \dot{\cup} G}_{2^{1+pow(1+\log_{2}\frac{1}{\epsilon})}} $.  Here $h_p = 1+pow(1+\log_{2}\frac{1}{\epsilon})$ and $f^{h_p}(\epsilon) = \OO(\epsilon (\log \frac{1}{\epsilon})^2)$. We have the following observation.
		\begin{observation}\label{obs:G_k_log_}
			Given any $k=1,2,\cdots, pow(m)$, $2^k U^{h_p}$ is an $\left(f^{h_p}(\epsilon), 2^k \min \{(G_{\log})^{\max},\upsilon\}\right)$-approximate set for $2^k G_{\log}$. Moreover, $\textbf{Ora}$ directly yields an $\OO(\frac{1}{\epsilon}(\log\frac{1}{\epsilon})^2))$-time oracle for backtracking from $2^k U^{h_p}$ to $2^k\underbrace{ G \dot{\cup}  G \dot{\cup} \cdots \dot{\cup}  G}_{2^{h_p}}$. That is, given any $c\in 2^k U^{h_p}$, by calling $\textbf{Ora}$, within $\OO(\frac{1}{\epsilon}(\log\frac{1}{\epsilon})^2))$ processing time, we can obtain $c'_1,c'_2,\cdots,c'_{2^{h_p}} \in  G \cup \{0\}$ such that $|c- 2^k\sum^{ 2^{h_p}}_{\iota=1} c'_{\iota}| \le f^{h_p}(\epsilon) 2^k \min \{(G_{\log})^{\max},\upsilon\}$.
		\end{observation}
		\begin{proof}
			Note that $(G_{\log})^{\max} = 2^{h_p} G^{\max}$. We first show that $2^k U^{h_p}$ is an $(f^{h_p}(\epsilon), 2^k \min\{(G_{\log})^{\max}, \upsilon\})$-approximate set for $2^k G_{\log}$. It is sufficient to observe the followings:
			(i) $(2^k U^{h_p})^{\max} = 2^k(U^{h_p})^{\max}\le  \left(1+f^{h_p}(\epsilon)\right) 2^k \min \{(G_{\log})^{\max}, \upsilon\}$; (ii) for any $c\in 2^k U^{h_p}$, we have $\frac{c}{2^k} \in U^{h_p}$, then there exist $c_1,c_2,\cdots,c_{2^{h_p}} \in G\cup \{0\}$ such that $ |\frac{c}{2^k} - \sum^{2^{h_p}}_{\iota = 1} c_{\iota}| \le f^{h_p}(\epsilon)  \min \{(G_{\log})^{\max},\upsilon\}$, i.e., $|c-2^k\sum^{2^{h_p}}_{\iota = 1} c_{\iota}| \le 2^k f^{h_p}(\epsilon)  \min \{(G_{\log})^{\max},\upsilon\}$; (iii) for any $x \in 2^k G_{\log} \cap [0, 2^k \min \{(G_{\log})^{\max},\upsilon\}]$,we have $\frac{x}{2^k} \in G_{\log} \cap [0, \min \{(G_{\log})^{\max},\upsilon\}]$, then there exists $y \in U^{h_p}$ such that $|y- \frac{x}{2^k}| \le f^{h_p}(\epsilon)  \min \{(G_{\log})^{\max},\upsilon\}$, i.e., $|2^k y- x| \le 2^kf^{h_p}(\epsilon)  \min \{(G_{\log})^{\max},\upsilon\}$. %Then by the definition of approximate set, one can easily prove that $2^k U^{h_p}$ is an $(f^{h_p}(\epsilon), 2^k \min\{(G_{\log})^{\max}, \upsilon\})$-approximate set for $2^k G_{\log}$.
			
			For the second part of the observation, consider any $c\in 2^k U^{h_p}$, we have $\frac{c}{2^k} \in U^{h_p}$. Recall that $\textbf{Ora}$ is an $\OO(\frac{1}{\epsilon}(\log\frac{1}{\epsilon})^2)$-time oracle  for backtracking from $U^{h_p}$ to $\underbrace{G \dot{\cup} G\dot{\cup} \cdots \dot{\cup} G}_{2^{h_p}} $.  Thus in $\OO(\frac{1}{\epsilon}(\log\frac{1}{\epsilon})^2)$ time, $\textbf{Ora}$ will return $c'_1,c'_2,\cdots,c'_{2^{h_p}} \in G \cup \{0\}$ such that $|\frac{c}{2^k} -\sum^{2^{h_p}}_{\iota = 1} c'_{\iota}| \le f^{h_p}(\epsilon)\min \{(G_{\log})^{\max},\upsilon\}$. Then we have  $|c -2^k  \sum^{2^{h_p}}_{\iota = 1} c'_{\iota} | \le f^{h_p}(\epsilon) 2^k \min \{(G_{\log})^{\max},\upsilon\}$. \qed \end{proof}

		The following claim guarantees that towards proving Lemma~\ref{lemma:sam_mul} for the case that $m > 2^{1+pow(\log_{2}(\frac{1}{\epsilon})+1)}$, it is sufficient to consider $U^{h_p} \oplus 2U^{h_p} \oplus \cdots \oplus 2^{pow(m)}U^{h_p}=\oplus^{pow(m)}_{k=0} 2^k U^{h_p}$.
		
		\begin{claim}\label{claim:c_G_M_}
			Assume that $m > 2^{1+pow(\log_2 (\frac{1}{\epsilon})+1)}$. Given $C$, which is an $(\tilde\OO(\epsilon), (1+\tilde\OO(\epsilon))\upsilon)$-approximate set for $\oplus^{pow(m)}_{k=0} 2^k U^{h_p}$. %={G_{\log}\oplus 2G_{\log} \oplus \cdots \oplus 2^k G_{\log}\oplus \cdots \oplus 2^{pow(m)}G_{\log}}
			Let $\textbf{Ora*}$ be a $T$-time oracle for backtracking from $C$ to $ U^{h_p}\dot{\cup} 2 U^{h_p} \dot{\cup} \cdots\dot{\cup} 2^{pow(m)}U^{h_p}$, that is, given any $c \in C$, in $T$ processing time, $\textbf{Ora*}$ will return $\{ c'_0,c'_1,c'_2,\cdots,c'_{pow(m)}\}$ such that $|c-\sum^{pow(m)}_{k=0} c'_k| \le \tilde\OO(\epsilon) \upsilon$, where $c'_k \in 2^k U^{h_p}\cup \{0\}$ for every $k=0,1,2,\cdots,pow(m)$.
			
			Then $C$ is an $(\tilde\OO(\epsilon), \upsilon)$-approximate set of $S(\overline{M})$. Moreover,  $\textbf{Ora*}$ and $\textbf{Ora}$ directly yiled an $\OO(T+\frac{1}{\epsilon}(\log\frac{1}{\epsilon})^3)$-time oracle for backtracking from $C$ to $\overline{M}$, that is, given any $c \in C$, by calling $\textbf{Ora*}$ and $\textbf{Ora}$, within $\OO(T+\frac{1}{\epsilon}(\log\frac{1}{\epsilon})^3)$ processing time, we can compute a multiset $\overline{M}' \subset \overline{M}$ such that $|c-\Sigma(\overline{M}')| \le \tilde\OO(\epsilon) \upsilon$. %in particular, $card_{\overline{M}'} [g]$ is obtained for every $g \in set(\overline{M}')$.
			
		\end{claim}
		\begin{proof}
			When $m >  2^{1+pow(\log_2 (\frac{1}{\epsilon})+1)}$, elements in $\overline{M}$ have the same multiplicity of $l_M = 2m\cdot 2^{1+pow(\log_2 (\frac{1}{\epsilon})+1)}$. Moreover, recall Observation~\ref{obs:G+G}, we have $(G_{\log}\oplus 2G_{\log} \oplus\cdots \oplus 2^{pow(m)}G_{\log})\subset S(\overline{M})$ and $S(\overline{M})\cap [0, \upsilon] = (G_{\log}\oplus 2G_{\log} \oplus\cdots \oplus 2^{pow(m)}G_{\log}) \cap [0, \upsilon]$.  Note that $G = set(\overline{M}) \subset  [\frac{1}{2\epsilon},\frac{1}{\epsilon}]$ and $m = \lfloor \frac{\upsilon}{G^{\min}}\rfloor$, we have
			$G^{\max} \le 2G^{\min}\le \frac{2\upsilon}{m}$%, furthermore, we have $f^{h_p} (\epsilon)\sum^{pow(m)}_{k=0} 2^k\min \{(G_{\log})^{\max},\upsilon\} \le f^{h_p}(\epsilon) 2^{pow(m)+1} (G_{\log})^{\max} \le m f^{h_p}(\epsilon) 2^{h_p+1} G^{\max}\le \OO(\epsilon (\log \frac{1}{\epsilon})^3) \upsilon$.
			
			Let $C$ be an $(\tilde\OO(\epsilon), (1+\tilde\OO(\epsilon))\upsilon)$-approximate set for $\oplus^{pow(m)}_{k=0} 2^k U^{h_p}$. We observe the followings:
			\begin{itemize}
				\item $C \subset [0, (1+\tilde\OO(\epsilon))\upsilon]$.
				
				\item For any $c \in C$, there exists $\{c'_0,c'_1,c'_2,\cdots,c'_{pow(m)}\}$ such that $|c-\mathop{\sum}\limits^{pow(m)}_{k=0} c'_k| \le \tilde\OO(\epsilon) \upsilon$, where $c'_k \in 2^k U^{h_p}\cup \{0\}$ for every $k=0,1,2,\cdots,pow(m)$. Consider each $c'_k$. Recall Observation~\ref{obs:G_k_log_}, $2^k U^{h_p}$ is an $\left(f^{h_p}(\epsilon), 2^k \min \{(G_{\log})^{\max},\upsilon\}\right)$-approximate set for $2^k G_{\log}$. Thus there exists $c^{(k;1)},c^{(k;2)}, \cdots,c^{(k;2^{h_p})}\in G \cup \{0\}$ such that $|c'_k -2^k \sum^{2^{h_p}}_{\iota =1} c^{(k;\iota)}| \le f^{h_p}(\epsilon)\cdot 2^k \min \{(G_{\log})^{\max},\upsilon\}$.

				Let $\mathcal{C}_k:=\{c^{(k;1)},c^{(k;2)}, \cdots,c^{(k;2^{h_p})}\}$ for every $k=0,1,2,\cdots,pow(m)$. We define $\overline{M}'$ as follows:
				\begin{itemize}
					\item For each $g \in G$ and each integer $0\le k \le pow(m)$, let $n^k_g$ denote the multiplicity of $g$ in $\mathcal{C}_k$, i.e., $n^k_g = card_{\mathcal{C}_k}[g]$.  Let $\mathcal{R}_{g}$ be the multiset consists of $\sum^{pow(m)}_{k=0} 2^k n^k_g$ copies of $g$, that is, $set(\mathcal{R}_{g}) = \{g\}$ and $|\mathcal{R}_{g}| = \sum^{pow(m)}_{k=0} 2^k  n^k_g$. 
					\item Define $\overline{M}' := \dot{\cup}_{g\in G} \mathcal{R}_{g}$
				\end{itemize}
				Notice that $\Sigma(\overline{M}') = \mathop{\sum}\limits_{g\in G} g  |\mathcal{R}_{g}|   =   \mathop{\sum}\limits_{g\in G}g \mathop{\sum}\limits^{pow(m)}_{k=0} 2^k  n^k_g  = \mathop{\sum}\limits^{pow(m)}_{k=0} 2^k ( \mathop{\sum}\limits_{g\in G} n^k_g  g) = \mathop{\sum}\limits^{pow(m)}_{k=0} 2^k  (\mathop{\sum}\limits^{2^{h_p}}_{\iota = 1} c^{(k;\iota)})$. Recall that $G^{\max} \le \frac{2\upsilon}{m}$, then we have  $|\sum^{pow(m)}_{k=0}c'_k-\Sigma(\overline{M}')|  \le  \sum^{pow(m)}_{k=0} |c'_k - \sum^{2^{h_p}}_{\iota = 1} c^{(k;\iota)} | \le f^{h_p} (\epsilon)\sum^{pow(m)}_{k=0} 2^k\min \{(G_{\log})^{\max},\upsilon\} \le m f^{h_p}(\epsilon) 2^{h_p+1} G^{\max}\le \OO(\epsilon (\log \frac{1}{\epsilon})^3) \upsilon$. We claim that $\overline{M}' \subset \overline{M}$, which is sufficient to observe the followings: (1). $set(\overline{M}') \subset G =set(\overline{M})$; (2). for any $g \in set(\overline{M}')$, the multiplicity of $g$ in $\overline{M}'$ is $|\mathcal{R}_{g}| =   \sum^{pow(m)}_{k=0} 2^k  n^k_g \le 2m \cdot 2^{1+pow(\log_2 \frac{1}{\epsilon}+1)}$; (3). elements in $\overline{M}$ have the same multiplicity of $l_M = 2m \cdot 2^{1+pow(\log_2 \frac{1}{\epsilon}+1)}$. 
				
				To summarize, given any $c\in C$, there exists $\overline{M}' \subset \overline{M}$ such that $|c-\Sigma(\overline{M}')| \le |c-\mathop{\sum}\limits^{pow(m)}_{k=0} c'_k| + |\sum^{pow(m)}_{k=0}c'_k-\Sigma(\overline{M}')|  \le \tilde\OO(\epsilon) \upsilon.$

				\item For any $s \in S(\overline{M})\cap [0,\upsilon]$, since $S(\overline{M})\cap [0, \upsilon] = (G_{\log}\oplus 2G_{\log} \oplus\cdots \oplus 2^{pow(m)}G_{\log}) \cap [0, \upsilon]$, we have $s\in (G_{\log}\oplus 2G_{\log} \oplus\cdots \oplus 2^{pow(m)}G_{\log}) \cap [0, \upsilon]$. Then there $\{s_1,s_2,\cdots,s_{pow(m)}\}$ such that $s = \sum^{pow(m)}_{k=1} s_{k}$, where $s_k \in 2^k G_{\log} \cap [0,\upsilon]$. Consider each $s_k$. Note that $2^k G_{\log} \cap [0,\upsilon] \subset 2^k G_{\log} \cap [0, 2^k\min\{ (G_{\log} )^{\max}, \upsilon \}]$,  then $s_k \in 2^k G_{\log} \cap [0, 2^k\min\{ (G_{\log} )^{\max}, \upsilon \}]$. Recall Observation~\ref{obs:G_k_log_}, $2^k U^{h_p}$ is an $\left(f^{h_p}(\epsilon), 2^k \min \{(G_{\log})^{\max},\upsilon\}\right)$-approximate set for $2^k G_{\log}$, thus there exist $s'_k \in 2^k U^{h_p}$ such that $|s'_k -s_k| \le f^{h_p}(\epsilon) 2^k \min \{(G_{\log})^{\max},\upsilon\}$. Let $s' = \sum^{pow(m)}_{k=0}s'_k$.  Recall that $G^{\max} \le \frac{2\upsilon}{m}$, then $|s-s'|\le f^{h_p}(\epsilon)\sum^{pow(m)}_{k=0} 2^k\min \{(G_{\log})^{\max},\upsilon\}\le f^{h_p}(\epsilon) 2^{pow(m)+1} (G_{\log})^{\max} \le m f^{h_p}(\epsilon) 2^{h_p+1} G^{\max}\le \OO(\epsilon (\log \frac{1}{\epsilon})^3) \upsilon$. Note that $s' \le s+ \tilde\OO(\epsilon)\upsilon \le (1+\tilde\OO(\epsilon))\upsilon$, we have $s' \in  \oplus^{pow(m)}_{k=0} 2^k U^{h_p} \cap [0, (1+\tilde\OO(\epsilon))\upsilon]$, then there exists $s''\in C$ such that $|s'' -s'| \le \tilde\OO(\epsilon) \upsilon$.  To summarize, for any $s \in S(\overline{M})\cap [0,\upsilon]$, there exists $s''\in C$ such that 
				$|s-s''| \le |s-s'|+|s'-s''| \le \tilde\OO(\epsilon)\upsilon.$

			\end{itemize}
			Thus $C$ is an $(\tilde\OO(\epsilon), \upsilon)$-approximate set for $S(\overline{M})$.
			
			In the next, we show that $\textbf{Ora*}$ and $\textbf{Ora}$ directly yields an $\OO(T+\frac{1}{\epsilon}(\log\frac{1}{\epsilon})^2)$-time oracle for backtracking from $C$ to $\overline{M}$. Here $\textbf{Ora*}$ is a $T$-time oracle for backtracking from $C$ to $ U^{h_p}\dot{\cup} 2 U^{h_p} \dot{\cup} \cdots\dot{\cup} 2^{pow(m)}U^{h_p}$, and $\textbf{Ora}$ is an $\OO(\frac{1}{\epsilon}(\log\frac{1}{\epsilon})^2)$-time oracle  for backtracking from $U^{h_p}$ to $\underbrace{G \dot{\cup} G\dot{\cup} \cdots \dot{\cup} G}_{2^{h_p}}$.  Given any $c\in C$, within $T$ processing time, $\textbf{Ora*}$ will return $\{c'_0,c'_1,c'_2,\cdots,c'_{pow(m)}\}$ such that $|c-\mathop{\sum}\limits^{pow(m)}_{k=0} c'_k| \le \tilde\OO(\epsilon) \upsilon$, where $c'_k \in 2^k U^{h_p}\cup \{0\}$ for every $k=0,1,2,\cdots,pow(m)$. Consider each $c'_k$. Recall Observation~\ref{obs:G_k_log_}, $\textbf{Ora}$ directly yields an $\OO(\frac{1}{\epsilon}(\log\frac{1}{\epsilon})^2))$-time oracle for backtracking from $2^k U^{h_p}$ to $2^k\underbrace{ G \dot{\cup}  G \dot{\cup} \cdots \dot{\cup}  G}_{2^{h_p}}$. Thus with the help of $\textbf{Ora}$, within $\OO(\frac{1}{\epsilon}(\log\frac{1}{\epsilon})^2))$ processing time, we can obtain $c^{(k;1)},c^{(k;2)}, \cdots,c^{(k;2^{h_p})}\in G \cup \{0\}$ such that $|c'_k -2^k \sum^{2^{h_p}}_{\iota =1} c^{(k;\iota)}| \le f^{h_p}(\epsilon) 2^k \min \{(G_{\log})^{\max},\upsilon\}$. 	Let $\mathcal{C}_k:=\{c^{(k;1)},c^{(k;2)}, \cdots,c^{(k;2^{h_p})}\}$ for every $k=0,1,2,\cdots,pow(m)$. Note that $m = \lfloor \frac{\omega}{M^{\min}}\rfloor = \OO(\frac{1}{\epsilon})$. The total time to obtain $\mathcal{C}_{k}$ for every $k$ is $\OO(pow(m)\frac{1}{\epsilon} (\log \frac{1}{\epsilon})^2) = \OO(\frac{1}{\epsilon} (\log \frac{1}{\epsilon})^3)$. 
			
			Define $\overline{M}'$ as follows:
			\begin{itemize}
				\item For each $g \in G$ and each integer $0\le k \le pow(m)$, let $n^k_g$ denote the multiplicity of $g$ in $\mathcal{C}_k$, i.e., $n^k_g = card_{\mathcal{C}_k}[g]$.  Let $\mathcal{R}_{g}$ be the multiset consists of $\sum^{pow(m)}_{k=0} 2^k n^k_g$ copies of $g$, that is, $set(\mathcal{R}_{g}) = \{g\}$ and $|\mathcal{R}_{g}| = \sum^{pow(m)}_{k=0} 2^k  n^k_g$. 
				\item Define $\overline{M}' := \dot{\cup}_{g\in G} \mathcal{R}_{g}$
			\end{itemize}
			Same as the previous discussion, we have $\overline{M}' \subset \overline{M}$ and $|c-\Sigma(\overline{M}')| \le \tilde\OO(\epsilon) \upsilon$. Observe that $|\mathcal{R}_{g}|$ is the multiplicity of $g$ in $\overline{M}'$. We now estimate the total time for determining $|\mathcal{R}_{g}|$ for every $g\in set(\overline{M}')$. Given $k$, the time to determine $n^k_g$ for every $g\in G$ is $\OO(2^{h_p})$. Then the total time to determine all $n^k_g$'s is $\OO(pow(m) 2^{h_p})$. Given $g \in G$ and given $n^k_g$ for every $k$, the time to determine $|\mathcal{R}_{g}|$ is $\OO(pow(m))$.  To summarize, the total time to determine $|\mathcal{R}_{g}|$ for every $g \in G$ is $\OO( pow(m) 2^{h_p} + pow(m)  |G|) = \OO(\frac{1}{\epsilon}(\log \frac{1}{\epsilon})^2)$.  
			
			In conclude, given any $c \in C$, with the help of $\textbf{Ora*}$ and $\textbf{Ora}$, in overall $\OO(T+\frac{1}{\epsilon}(\log\frac{1}{\epsilon})^3)$ processing time, we can obtain $\overline{M}' \subset \overline{M}$ such that $|c-\Sigma(\overline{M}')| \le \tilde\OO(\epsilon) \upsilon$, moreover, $card_{\overline{M}'}[g]$ is obtained for every $g \in set(\overline{M}')$.\qed \end{proof}

		%According to Claim~\ref{claim:c_G_M_}, towards proving Lemma~\ref{lemma:sam_mul} for the case that $m > 2^{1+pow(\log_2 \frac{1}{\epsilon}+1)}$, we only need to consider $G_{\log} \oplus 2G_{\log} \oplus \cdots \oplus 2^{pow(m)}G_{\log}$.
		
		%Corollary~\ref{coro:sub_sum_top}	
		
		In the following, we will approximate $\oplus^{pow(m)}_{k=0} 2^k U^{h_p}$ and meanwhile build an oracle for backtracking.
		
		\paragraph{\textbf{Approximating $\oplus^{pow(m)}_{k=0} 2^k U^{h_p}$.}} %In this step, we aim to approximate  $G_{\log} \oplus 2G_{\log} \oplus \cdots \oplus 2^{pow(m)}G_{\log}$. %Recall that we have obtained $2^k C_{\log}$, which is an $(\tilde\OO(\epsilon), \upsilon)$-approximate set with cardinality of $\OO({\epsilon^{-1}})$ for each $2^k G^{\log}$, and meanwhile, we have built an $\tilde\OO(\epsilon)$-time oracle for backtracking from $2^k G^{\log}$ to $\underbrace{2^k G \dot{\cup} 2^k G \dot{\cup} \cdots \dot{\cup} 2^k G}_{2^{pow(1+\log_{2}\frac{1}{\epsilon})+1}}$. %, moreover, we have the following observation.\begin{observation}\end{observation}
		
		Let $\tilde{\upsilon} = (1+\tilde\OO(\epsilon))\upsilon$. Note that $pow(m) = \OO(\log \frac{1}{\epsilon})$. Recall Lemma~\ref{lemma:sub_sum_top}, in $\tilde\OO({\epsilon^{-1}})$ processing time, we can compute an $(\tilde\OO(\epsilon),\tilde{\upsilon})$-approximate set with cardinality of $\OO({\epsilon^{-1}})$ for $\oplus^{pow(m)}_{k=0} 2^k U^{h_p}$. Denote by $C$ this approximate set. At the same time, we can build an $\tilde\OO({\epsilon^{-1}})$-time oracle $\textbf{Ora*}$ for backtracking from $C$ to $U^{h_p}\dot{\cup} 2U^{h_p}\dot{\cup} \cdots \dot{\cup} 2^{pow(m)}U^{h_p}$, that is, given any $c \in C$, within $\tilde\OO({\epsilon^{-1}})$ time, $\textbf{Ora*}$ will return $\{ c'_0,c'_1,c'_2,\cdots,c'_{pow(m)}\}$ such that $|c-\sum^{pow(m)}_{k=0} c'_k| \le \tilde\OO(\epsilon) \upsilon$, where $c'_k \in 2^k U^{h_p}\cup \{0\}$ for every $k=0,1,2,\cdots,pow(m)$.

		Then by Claim~\ref{claim:c_G_M_}, we have proved Lemma~\ref{lemma:sam_mul} for the case that $m > 2^{1+pow(\log_2 \frac{1}{\epsilon}+1)}$, which completes the proof of Lemma~\ref{lemma:sam_mul},  furthermore, completes the proof of Lemma~\ref{lemma:un-subset-sum-l}. Till now, we have completed the proof of Theorem~\ref{the:un_boubded_sum}.

		\section{Conclusion}
		In this paper we present improved approximation schemes for (unbounded) SUBSET SUM and PARTITION. In particular, we give the first subquadratic deterministic weak approximation scheme for SUBSET SUM. %For weak approximating unbounded SUBSET SUM, the running time is essentially the best possible. 
		However, it is not clear whether a better weak approximation scheme exists for SUBSET SUM and PARTITION. The existing results as well as our results seem to imply that  PARTITION admits a better approximation scheme. This is primarily due to that the target $t$ in PARTITION is $\OO(\Sigma(X))$, and can be leveraged to reduce the running time of FFT. It is not clear, however, whether this is the key fact that determines the fine-grained complexity and makes weak approximating SUBSET SUM harder than PARTITION.

			\bibliographystyle{plain}
			\bibliography{ref-soda}

		\end{document}